%% file: main.tex
\documentclass[11pt,english,notitlepage]{report}

\input{header.tex}

\begin{document}

\title{Deterministic Decremental SSSP and Approximate Min-Cost Flow in Almost-Linear
Time}
\author{
	Aaron Bernstein\thanks{Supported by NSF CAREER Grant 1942010.} \\ Rutgers University New Brunswick, USA \\ bernstei@gmail.com   \and
	Maximilian Probst Gutenberg\thanks{The author is supported by a start-up grant of Rasmus Kyng at ETH Zurich. Work was partially done while at the University of Copenhagen where the author was supported by Basic Algorithms Research Copenhagen (BARC), supported by Thorup's Investigator Grant from the Villum Foundation under Grant No. 16582.} \\ ETH Zurich, Switzerland\\ maximilian.probst@outlook.com	\and
	Thatchaphol Saranurak\thanks{Work was partially done while at Toyota Technological Institute at Chicago.} \\ University of Michigan, USA\\ thsa@umich.edu
}
\date{}
\maketitle

\input{abstract}

\pagenumbering{gobble}
\setcounter{tocdepth}{1}  
\tableofcontents{}

\pagebreak
\pagenumbering{arabic}

\chapter{Extended Abstract}
\label{part:intro}
\input{aaron_intro}

\input{overview}

\pagebreak
\chapter[Distance-only Dynamic Shortest Paths]{Distance-only Dynamic Shortest \\Paths}
\label{part:dynamicShortestPaths}
\input{prelim}

\input{components}

\input{robust_core}

\input{covering}

\input{MES}

\input{together_distance_new}

\pagebreak
\newpage
\chapter[Path-reporting Dynamic Shortest Paths]{Path-reporting Dynamic Shortest\\ Paths}
\label{part:augmented-queries}
\input{prelim_path}

\input{components_path}

\input{ball_path}

\input{robust_core_path}

\input{together_path_2.tex}

\pagebreak
\chapter{Approximate Min-Cost Flow}
\label{part:minCostFlow}

\input{flow}

\section*{Acknowledgements}
Aaron and Thatchaphol thank Shiri Chechik for a useful discussion at an early stage of this work. 

\pagebreak
\appendix

\chapter{Appendix}
\input{related_work}

\input{appendix_for_overview}

\section[Appendix of Part II]{Appendix of \Cref{part:dynamicShortestPaths}}

\input{certify_core}

\input{embedwitness}
\input{appendix_path}
\input{flow_appendix}

\pagebreak
\bibliographystyle{alpha}
\bibliography{bibliography}

\end{document}

%% file: header.tex
\usepackage[T1]{fontenc}
\usepackage[latin9]{inputenc}
\usepackage{geometry}
\usepackage{verbatim}
\usepackage{float}
\usepackage{amsmath}
\usepackage{amsthm}
\usepackage{amssymb}
\usepackage{graphicx}
\usepackage{todonotes}
\usepackage{lmodern}
\usepackage{babel}
\usepackage{enumitem}
\usepackage{array}
\usepackage{wrapfig}

\usepackage[nottoc,numbib]{tocbibind}

\usepackage[procnumbered,ruled,vlined,linesnumbered]{algorithm2e}
\usepackage{xcolor}
\usepackage{nameref}
\definecolor{ForestGreen}{rgb}{0.1333,0.5451,0.1333}
\definecolor{DarkRed}{rgb}{0.8,0,0}
\definecolor{Red}{rgb}{1,0,0}
\usepackage[linktocpage=true,
pagebackref=true,colorlinks,
linkcolor=DarkRed,citecolor=ForestGreen,
bookmarks,bookmarksopen,bookmarksnumbered]
{hyperref}

\usepackage{cleveref}

\usepackage{mathtools}
\usepackage{thmtools}
\usepackage{thm-restate}
\usepackage[normalem]{ulem}

\usepackage{titlesec}
\renewcommand{\thechapter}{\Roman{chapter}}
\usepackage{titlesec}
\titleformat{\chapter} 
{\normalfont\huge\bfseries}{Part \thechapter:}{0.5em}{} 

\crefname{chapter}{part}{parts}
\Crefname{chapter}{Part}{Parts}

\declaretheorem[numberwithin=section]{theorem}
\declaretheorem[numberlike=theorem]{lemma}
\declaretheorem[numberlike=theorem,name=Lemma]{lem}

\declaretheorem[numberlike=theorem,name=Proposition]{prop}
\declaretheorem[numberlike=theorem]{corollary}

\declaretheorem[numberlike=theorem]{definition}
\declaretheorem[numberlike=theorem,name=Definition]{defn}
\declaretheorem[numberlike=theorem]{claim}

\declaretheorem[numberlike=theorem,name=Invariant]{inv}

\declaretheorem[numberlike=theorem,style=remark,name=Remark]{rem}
\declaretheorem[numberlike=theorem,style=definition]{goal}
\theoremstyle{remark}
\newtheorem*{rem*}{Remark}

\global\long\def\defeq{\triangleq}

\global\long\def\cC{\mathcal{C}}

\global\long\def\cP{\mathcal{P}}

\global\long\def\Phat{\widehat{P}}
\global\long\def\Vhat{\widehat{V}}
\global\long\def\Vtil{\widetilde{V}}
\global\long\def\Ghat{\widehat{G}}
\global\long\def\Hhat{\widehat{H}}
\global\long\def\Htil{\widetilde{H}}
\global\long\def\Ttil{\widetilde{T}}
\global\long\def\Ptil{\widetilde{P}}
\global\long\def\Wstar{W^{*}}
\global\long\def\Cinit{C^{init}}

\global\long\def\Kinit{K^{init}}
\global\long\def\Xinit{X^{init}}
\global\long\def\dtil{\widetilde{d}}

\global\long\def\wtil{\widetilde{w}}

\global\long\def\ball{\mathrm{ball}}

\global\long\def\core{\textsc{core}}
\global\long\def\shell{\textsc{shell}}
\global\long\def\oshell{\overline{\textsc{shell}}}
\global\long\def\cover{\textsc{cover}}
\global\long\def\Clevel{\ell_{\core}}
\global\long\def\lenapsp{h_{\textsc{apsp}}}

\global\long\def\poly{\mathrm{poly}}
\global\long\def\polylog{\mathrm{polylog}}
\global\long\def\dist{\mathrm{dist}}
\global\long\def\distInduc{\widetilde{\mathrm{dist}}}
\global\long\def\diam{\mathrm{diam}}
\global\long\def\size{\mathrm{size}}
\global\long\def\path{\mathrm{path}}
\global\long\def\vol{\mathrm{vol}}

\global\long\def\val{\mathrm{val}}
\global\long\def\len{\mathrm{len}}
\global\long\def\overhead{\mathrm{overhead}}
\global\long\def\near{\mathrm{near}}
\global\long\def\peel{\mathrm{peel}}
\global\long\def\diam{\mathrm{diam}}
\global\long\def\deg{\mathrm{deg}}

\global\long\def\Ohat{\widehat{O}}
\global\long\def\Otil{\tilde{O}}
\global\long\def\Omegahat{\widehat{\Omega}}
\global\long\def\Omegatil{\tilde{\Omega}}

\global\long\def\scatter{\delta_{\mathrm{scatter}}}
\global\long\def\phicmg{\phi_{\textsc{cmg}}}

\global\long\def\epswit{\epsilon_{\textsc{wit}}}
\global\long\def\eps{\epsilon}
\global\long\def\stretch{\mathrm{str}}
\global\long\def\kappafinal{\kappa_{\textrm{final}}}

\global\long\def\Apxball{\textsc{ApxBall}}
\global\long\def\Core{\textsc{RobustCore}}

\global\long\def\Covering{\textsc{Covering}}
\global\long\def\SSSP{\textsc{SSSP}}

\global\long\def\estree{\textsc{ESTree}}
\global\long\def\gammacompress{\gamma}

\global\long\def\CertifyCore{\textsc{CertifyCore}}
\global\long\def\EmbedCore{\textsc{EmbedWitness}}
\global\long\def\EmbedMatching{\textsc{EmbedMatching}}
\global\long\def\Prune{\textsc{Prune}}

\global\long\def\updatelevel{\textsc{UpdateLevel}}
\global\long\def\MES{\textsc{MES}}

\global\long\def\distScale{\texttt{ds}}
\global\long\def\answer{\texttt{ans}}

\global\long\def\roundup#1#2{\left\lceil #1\right\rceil _{#2}}
\global\long\def\ceiling#1{\left\lceil #1\right\rceil }
\global\long\def\stage#1#2{#1^{(#2)}}

\def\thatchaphol#1{\marginpar{$\leftarrow$\fbox{T}}\footnote{$\Rightarrow$~{\sf\textcolor{purple}{#1 --Thatchaphol}}}}

\DeclareUnicodeCharacter{221A}{$\sqrt{}$}
\DeclareUnicodeCharacter{3F5}{$\epsilon$}

\DontPrintSemicolon
\SetKw{KwAnd}{and}
\SetProcNameSty{textsc}
\SetFuncSty{textsc}

\global\long\def\shortestSTPathVar{\pi(s,t)}

\global\long\def\flowVar{f}
\global\long\def\flowEstVar{\hat{f}}
\global\long\def\costEstVar{\hat{c}}
\global\long\def\initialDifferenceMWU{\delta}
\global\long\def\flowRouted{\Upsilon}

\global\long\def\weightFunctionVar{w}
\global\long\def\weightFunctionEstVar{\Hat{w}}
\global\long\def\weightFunctionCombinedVar{\tilde{w}}

\global\long\def\costFunctionVar{\varphi}
\global\long\def\costFunctionEstVar{\hat{\varphi}}

\global\long\def\inflow{\textrm{in}}
\global\long\def\outflow{\textrm{out}}

\global\long\def\edgeSensitivity{\sigma}
\global\long\def\sensitivity{\sigma}

\global\long\def\neighborhood{\mathcal{N}}

\global\long\def\objFunctionVar{\textsc{ObjVal}}
\global\long\def\objFunctionEstVar{\widehat{\textsc{ObjVal}}}

\global\long\def\timeAugSSSPVar{\mathcal{T}_{SSSP^{\pi}}}

\global\long\def\minCapIndexMWUVar{\Lambda}
\global\long\def\forLoopIndexMWU{\lambda}

\global\long\def\Otil{\tilde{O}}
\global\long\def\Ohat{\widehat{O}}
\global\long\def\poly{\mathrm{poly}}

\newcommand{\shortestSTPath}[1][]{
	\ifthenelse{\equal{#1}{}}{\shortestSTPathVar}{\shortestSTPathVar_{#1}}
}

\newcommand{\costFunction}[1][]{
	\ifthenelse{\equal{#1}{}}{\costFunctionVar}{\costFunctionVar_{#1}}
}

\newcommand{\costFunctionEst}[1][]{
	\ifthenelse{\equal{#1}{}}{\costFunctionEstVar}{\costFunctionEstVar_{#1}}
}

\newcommand{\weightFunction}[1][]{
	\ifthenelse{\equal{#1}{}}{\weightFunctionVar}{\weightFunctionVar_{#1}}
}

\newcommand{\weightFunctionEst}[1][]{
	\ifthenelse{\equal{#1}{}}{\weightFunctionEstVar}{\weightFunctionEstVar_{#1}}
}

\newcommand{\weightFunctionCombined}[1][]{
	\ifthenelse{\equal{#1}{}}{\weightFunctionCombinedVar}{\weightFunctionCombinedVar_{#1}}
}

\newcommand{\flow}[1][]{
	\ifthenelse{\equal{#1}{}}{\flowVar}{\flowVar_{#1}}
}

\newcommand{\flowEst}[1][]{
	\ifthenelse{\equal{#1}{}}{\flowEstVar}{\flowEstVar_{#1}}
}

\newcommand{\costEst}[1][]{
	\ifthenelse{\equal{#1}{}}{\costEstVar}{\flowEstVar_{#1}}
}

\newcommand{\minCapIndexMWU}[1][]{
	\ifthenelse{\equal{#1}{}}{\minCapIndexMWUVar}{\minCapIndexMWUVar_{#1}}
}

\newcommand{\objFunction}[1][]{
	\ifthenelse{\equal{#1}{}}{\objFunctionVar}{\objFunctionVar_{#1}}
}

\newcommand{\objFunctionEst}[1][]{
	\ifthenelse{\equal{#1}{}}{\objFunctionEstVar}{\objFunctionEstVar_{#1}}
}

\newcommand{\timeAugSSSP}[8]{
	\timeAugSSSPVar(#1, #2, #3, #4, #5, #6,#7, #8)
}

\newcommand{\randomThreshold}[1][]{
	\ifthenelse{\equal{#1}{}}{\gamma}{\gamma_{#1}}
}

\newcommand{\ignore}[1]{}
\newcommand{\kappaold}{\kappa^{old}}
\newcommand{\kappanew}{\kappa^{new}}

\def\ShowComment{True} %

\ifdefined\ShowComment
\def\thatchapholtext#1{\textcolor{purple}{[TS]: #1}}
\def\thatchaphol#1{\marginpar{$\leftarrow$\fbox{T}}\footnote{$\Rightarrow$~{\sf\textcolor{purple}{#1 --Thatchaphol}}}}

\def\note#1{#1}
\def\alert#1{\textcolor{red}{#1}}

\else

\def\thatchaphol#1{}
\def\thatchapholtext#1{}
\def\note#1{} 
\def\alert#1{}
\fi

%% file: abstract.tex
\begin{abstract}
In the decremental single-source shortest paths problem, the goal is to maintain distances from a fixed source $s$ to every vertex $v$ in an $m$-edge graph undergoing edge deletions. 
In this paper, we conclude a long line of research on this problem by 
showing a near-optimal \emph{deterministic} data structure that maintains $(1+\epsilon)$-approximate distance estimates and runs in $m^{1+o(1)}$ total update time.

Our result, in particular, removes the \emph{oblivious adversary} assumption required by the previous breakthrough result by Henzinger~et~al.~[FOCS'14], which leads to our second result:
the first almost-linear time algorithm for $(1-\epsilon)$-approximate min-cost flow in undirected graphs where capacities and costs can be taken over edges \emph{and} vertices. 
Previously, algorithms for max flow with vertex capacities, or min-cost flow with any capacities required super-linear time.
Our result essentially completes the picture for approximate flow in undirected graphs. 

The key technique of the first result is a novel framework that allows us to treat low-diameter graphs like expanders. This allows us to harness expander properties while bypassing shortcomings of expander decomposition, which almost all previous expander-based algorithms needed to deal with. 
For the second result, we break the notorious flow-decomposition barrier from the multiplicative-weight-update framework using randomization.

\end{abstract}

%% file: aaron_intro.tex
\section{Introduction}
\label{sec:introduction}

One of the most fundamental problems in graph algorithms is the single-source shortest paths (SSSP) problem where given a source vertex $s$ and a undirected, weighted graph $G=(V,E,w)$ with $n = |V|, m = |E|$, we want to find the shortest paths from $s$ to every vertex in the graph. This problem has been studied since the 1950s \cite{shimbel1954structure, dijkstra1959note} and can be solved in linear time \cite{thorup1999undirected}.

A natural extension of SSSP is to consider a dynamic graph $G$ that is changing over time. The most natural model is the fully dynamic one, where edges can be inserted and deleted from $G$. Unfortunately, recent progress on conditional lower bounds \cite{abboud2014popular, henzinger2015unifying, GutenbergWW20} essentially rules out any fully dynamic algorithm with small update and query times for maintaining distances from $s$. For this reason, most research has focused on the \emph{decremental} setting, where the graph $G$ only undergoes edge deletions.
In addition to being a natural relaxation of the fully dynamic model, the decremental setting is extremely well-motivated for the SSSP problem in particular: a fast data structure for decremental SSSP can be used as a subroutine within the multiplicative weighted update (MWU) framework to speed up algorithms for various (static) flow problems. 

Our main contribution is an almost-optimal data structure for decremental SSSP, which we in turn use to develop the first almost-optimal algorithms for approximate vertex-capacitated max flow and min-cost flow.

\subsection{Previous Work}
\label{subsec:previousWork}
For our discussion of related work, we assume for $(1+\epsilon)$-approximations that $\epsilon > 0$ is constant to ease the discussion. We use $\Otil$- and $\Ohat$-notation to suppress logarithmic and subpolynomial factors in $n$, respectively. We include a broader discussion of related work in \Cref{sec:relatedWork}.

\paragraph{Decremental Single-Source Shortest Paths (SSSP).} A seminal result for decremental SSSP is an algorithm by Even and Shiloach \cite{EvenS} with total update time $O(mn)$ over the entire sequence of updates in unweighted graphs. Conditional lower bounds indicate that this is near-optimal \cite{roditty2004dynamic, abboud2014popular,henzinger2015unifying, GutenbergWW20}. But Bernstein and Roditty showed \cite{bernstein2011improved} that there exist faster algorithms if one allows for a $(1+\epsilon)$-approximation on the distances (and the corresponding shortest paths). This line of research culminated in a breakthrough result by Forster, Henzinger and Nanongkai \cite{henzinger2014decremental} (see also \cite{LackiN20}) who showed how to maintain $(1+\epsilon)$-approximate SSSP in total update time $\Ohat(m \cdot \polylog(W))$, where $W$ is the maximum weight ratio.

\paragraph{Towards Efficient Adaptive Data Structures.} Although it has near-optimal update time, the $\Ohat(m)$ result of \cite{henzinger2014decremental} suffers from a crucial shortcoming: it is randomized and only works against an \emph{oblivious} adversary, i.e.~an adversary that fixes the entire sequence in advance. For this reason, the result of \cite{henzinger2014decremental}  cannot be used as a black-box data structure, and in particular cannot be incorporated into the MWU framework for flow algorithms mentioned above.

Over the last years, there has been significant effort towards designing \emph{adaptive}, or even better deterministic, algorithms with comparable update time guarantees \cite{bernstein2016deterministic, bernstein2016deterministic, bernstein2017deterministicweighted, Chuzhoy:2019:NAD:3313276.3316320, gutenberg2020deterministic, bernstein2020fully, ChuzhoyS20_apsp}. But the best total update time remains $\Ohat(\min\{m\sqrt{n}, n^2\} \polylog(W))$.

\paragraph{Max flow and Min-cost Flow.}
Max flow and min-cost flow problems have been studied extensively since the 1950s \cite{Dantzig51,FordF56,Dinic70,GoldbergT88,goldberg1998beyond,daitch2008faster,Madry13,lee2014path,CohenMSV17,liu2020fasterDiv} and can be solved exactly in time $\Otil((m+n^{1.5})\log^2(UC))$ \cite{BrandLLSSSW20} and, for unit-capacity graphs, $\Ohat(m^{4/3}\log(C))$ \cite{axiotis2020circulation}
 where $U$ is the maximum capacity ratio
and $C$ is the maximum cost ratio. Although enormous effort has been directed towards these fundamental problems, in directed sparse graphs, the fastest algorithms are still far from achieving almost-linear time.

Therefore, an exciting line of work \cite{christiano2011electrical, LeeRS13, sherman2013nearly, kelner2014almost, RackeST14,peng2016approximate} emerged with the goal of obtaining faster \emph{approximation}
algorithms on \emph{undirected} graphs. This culminated in
$\Otil(m \cdot \polylog(U))$-time algorithms for $(1+\eps)$-approximate
max flow \cite{sherman2013nearly, kelner2014almost,peng2016approximate} and $\Otil(m\cdot \polylog(C))$-time algorithms for min-cost flow
when all capacities are infinite \cite{Sherman17,li2020faster,andoni2020parallel}, both of which require only near-linear time.

\paragraph{Limitations of Existing Approaches.} Unfortunately, none of the near-linear-time algorithms above handle vertex capacities or can be generalized to min-cost flow with finite capacities. This severely limits the range of applications of these algorithms. 

This limitation seems inherent to the existing algorithms. The most successful approach for approximate max flow \cite{sherman2013nearly, kelner2014almost} is based on obtaining fast $n^{o(1)}$-competitive oblivious routing schemes for the $\ell_{\infty}$-norm (or $\ell_1$-norm in the case of \cite{sherman2017generalized}). But for both oblivious routing in vertex-capacitated graphs \cite{HajiaghayiKRL07} and min-cost flow oblivious routing\footnote{By min-cost flow oblivious routing, we mean an oblivious routing scheme that is competitive at the same time with the best routing in terms of $\ell_1$-norm \emph{and} $\ell_{\infty}$-norm, respectively.} \cite{aspnes2006eight, ghaffari2020hopconstrained} there are lower bounds of $\Omega(\sqrt{n})$ for the possible competitiveness. This would lead to an additional polynomial overhead for these algorithms. There are also some alternative approaches to flow problems, but currently they do not lead to almost-linear time algorithms even for regular edge-capacitated max-flow (see e.g. \cite{christiano2011electrical, LeeRS13, KyngPSW19}).

\paragraph{Max Flow and Min-Cost Flow via MWU and Decremental SSSP.} In order to overcome limitations in the previous approaches, a line of attack emerged that was originally suggested by \cite{madry2010faster} and was recently reignited by Chuzhoy and Khanna \cite{Chuzhoy:2019:NAD:3313276.3316320}. The idea is that the MWU framework for solving min-cost flow (see e.g. \cite{garg2007faster,fleischer2000approximating}) can be sped up with a fast \emph{adaptive} decremental SSSP data structure.  In \cite{Chuzhoy:2019:NAD:3313276.3316320}, Chuzhoy and Khanna obtained promising results via this approach: an algorithm for max flow with vertex capacities only in $\Ohat(n^2 \polylog(U))$ time.
But this approach currently has two major challenges towards an $\Ohat(m)$ time algorithm: 
\begin{itemize}
    \item Obtaining a fast adaptive decremental SSSP data structure has proven to be an extremely difficult challenge that even considerable effort could not previously resolve  \cite{bernstein2016deterministic, bernstein2016deterministic, bernstein2017deterministicweighted, Chuzhoy:2019:NAD:3313276.3316320, gutenberg2020deterministic, bernstein2020fully, ChuzhoyS20_apsp}.
    
    \item Even given such a data structure, the MWU framework is designed to successively route flows along paths from a source $s$ to a sink $t$. But this implies that the \emph{flow decomposition barrier} applies to the MWU framework, which might have to send flow on $\Omega(mn)$ edges over the course of the algorithm (or $\Omega(n^2)$ edges when only vertex capacities are present). 
\end{itemize}
In this article, we overcome both challenges and complete this line of work.

\subsection{Our Results}

\paragraph{Decremental SSSP.} Our main result is the first deterministic data structure for the decremental SSSP problem in undirected graph with almost-optimal total update time.

\begin{restatable}[Decremental SSSP]{theorem}{mainSSSPResult}\label{thm:mainSSSPResult}
Given an undirected, decremental graph $G=(V,E,w)$, a fixed source vertex $s \in V$, and any $\epsilon > 1/\polylog(n)$, we give a \emph{deterministic} data structure that maintains a $(1+\epsilon)$-approximation of the distance from $s$ to every vertex $t$ in $V$ explicitly in total update time $m^{1+o(1)} \polylog W$. The data structure can further answers queries for an $(1+\epsilon)$-approximate shortest $s$-to-$t$ path $\shortestSTPath$ in time $|\shortestSTPath|n^{o(1)}$.
\end{restatable}

This result improves upon the state-of-the-art $\Ohat(\min\{m\sqrt{n}, n^2\} \polylog(W))$ total update time time in the deterministic (or even adaptive) setting and resolves the central open problem in this line of research. 

\paragraph{Mixed-Capacitated Min-Cost Flow.} Given our new deterministic SSSP data structure, it is rather straight-forward using MWU-based techniques from \cite{fleischer2000approximating, garg2007faster, ChuzhoyS20_apsp} to obtain unit-capacity min-cost flow in almost-linear time. We are able to generalize these techniques significantly to work for arbitrary vertex and edge capacities.

\begin{restatable}[Approximate Mixed-Capacitated Min-Cost Flow]{theorem}{minCostMain}\label{thm:MainMinCost}
For any $\epsilon > 1/\polylog(n)$, consider undirected graph $G=(V,E,c,u)$, where cost function $c$ and capacity function $u$ map each edge and vertex to a non-negative real. Let $s, t \in V$ be source and sink vertices. Then, there is an algorithm that in $m^{1+o(1)}\log \log C$ time returns a feasible flow $f$ that sends a $(1-\epsilon)$-fraction of the max flow value from $s$ to $t$ with cost at most equal to the min-cost flow.\footnote{We can also route an arbitrary demand vector, see an alternative statement in \Cref{subsec:alternativeStatement}.} The algorithm runs correctly with high probability.
\end{restatable}

Our result resolves one of the three key challenges for the max flow/ min-cost problem according to a recent survey by Madry \cite{mkadry2018gradients}.\footnote{We point out, however, that our dependency on $\epsilon$ is significantly worse than formulated in \cite{mkadry2018gradients}.} The state-of-the-art for this problem \cite{BrandLLSSSW20} solved the exact version of this problem in directed graphs and hence obtains significantly slower running time $\Otil((m + n^{1.5}) \cdot \polylog(UC))$ which is still super-linear in sparse graphs. 

\subsection{Applications} 
\label{subsec:applications}

Our two main results have implications for a large number of interesting algorithmic problems. See \Cref{appendix:ApplicationDiscussion} for a more detailed statements and a discussion of how to obtain the results below.

\paragraph{Applications of Mixed-Capacitated Min-Cost Flow.} 
\begin{itemize}
	\item Using a reduction of \cite{KhandekarRV09}, our result for vertex-capacitated flow yields a $O(\log^2(n))$ approximation to sparsest vertex cut in undirected graphs in $\Ohat(m)$ times. This is the first almost-linear-time algorithm for the problem with $\polylog(n)$ approximation.
	\item Combined with another reduction in \cite{bodlaender1995approximating}, our result for sparsest vertex cut yields an $O(\log^3(n))$-approximate algorithm for computing tree-width (and the corresponding tree decomposition) in $\Ohat(m)$ time. This is again the first almost-linear-time algorithm with $\polylog(n)$ approximation, except for the special cases where the tree-width is itself sub-polynomial \cite{fomin2018fully} or the graph is extremely dense \cite{ChuzhoyS20_apsp}. (See other work on computing tree-width in \cite{robertson1995graph, bodlaender1996linear, amir2001efficient, amir2010approximation, bodlaender2016c,bodlaender1995approximating, amir2001efficient, amir2010approximation, feige2008improved, Chuzhoy:2019:NAD:3313276.3316320}.) 
	\item The above algorithm then leads to improvement for algorithms that relied on computing an efficient tree decomposition. For example, we speed-up the high-accuracy LP solver by Dong, Lee and Ye \cite{dong2020nearly} that is parameterized by treewidth; we reduce the running time to $\Ohat(m \cdot \textrm{tw}(G_A)^2 \log(1/\epsilon))$, improving upon the previous dependency of $\textrm{tw}(G_A)^4$. 
	\item Given any graph $G=(V,E)$ (with associated incidence matrix $B$), $\epsilon > 1/\polylog(n)$, a demand vector $\chi \in \mathbb{R}^n$, (super)-linear functions $c_e,c_v : \mathbb{R}_{\geq 0} \rightarrow \mathbb{R}_{\geq 0}$ for each $e \in E$ and $v \in V$. Let $f^*$ be some flow minimizing 
    \[
    	\min_{B^\top f = \chi} c(f) = \sum_{e\in E} c_e(|f_e|) + \sum_{v \in V} c_v((B^\top |f|)_v).
    \]
	Then, we can compute a $(1+\epsilon)$ approximate flow $f$ with $c(f) \leq (1+\epsilon) c(f^*)$ that routes demand $\chi$ in almost-linear time. In particular, this is the first  almost-linear time algorithm for flow in the weighted $p$-norm $\|W^{-1} f\|_p$ (since we can minimize $\|W^{-1} f\|_p^p$ by $c_e(x) = (\frac{x}{w_e})^p$).
\end{itemize}

\paragraph{Applications of Decremental SSSP.} There is currently a large gap between the best-known dynamic graph algorithms against oblivious adversaries and adaptive ones. Much of this gap stems from the problem of finding a deterministic counterpart to picking a random source. Plugging in, either our decremental SSSP as a black-box subroutine or our some techniques that we obtain along the way, we obtain various new adaptive algorithms:
\begin{itemize}
    \item  Decremental $(1+\epsilon)$-approximate all-pairs shortest paths (APSP) in total update time $\Ohat(mn)$. (Previous adaptive results only worked in unweighted graphs \cite{henzinger2016dynamic, gutenberg2020deterministic}.)
       \item Decremental $\Ohat(1)$-approximate APSP with total update time $\Ohat(m)$. Even in unweighted graphs, all previously adaptive algorithms for decremental APSP (for \emph{any} approximation) had total update time at least $\Omega(n^2)$ \cite{henzinger2016dynamic, gutenberg2020deterministic,ChuzhoyS20_apsp,EvaldFGW20}; for weighted graphs they were even slower. Our result is analogous to the \emph{oblivious} algorithm of Chechik, though she achieves a stronger $O(\log(n))$-approximation \cite{chechik2018near}.  
    \item Fully-dynamic $(2+\epsilon)$ approximate all-pairs shortest paths with $\Ohat(m)$ update time, matching the oblivious result of \cite{bernstein2009fully}.
\end{itemize}

\subsection{Technical Contributions}

From a technical perspective, our dynamic SSSP result in Theorem \ref{thm:mainSSSPResult} is by far our more significant contribution. It requires several new ideas, but we would like to highlight one technique in particular that is of independent interest and might have applications far beyond our result:

\paragraph{Key Technique: Converting any Low-Diameter Graph into an Expander}
Several recent papers on dynamic graph algorithms start with the observation that many problems are easy to solve if the underlying graph $G$ is an expander, as one can then apply powerful tools such as expander pruning and flow-based expander embeddings. All of these papers then generalize their results to arbitrary graphs by using expander decomposition: they decompose $G$ into expander subgraphs and then apply expander tools separately to each subgraph. Unfortunately, expander decomposition necessarily involves a large number of crossing edges (or separator vertices) that do not belong to any expander subgraph and need to be processed separately. This difficulty has been especially prominent for decremental shortest paths, where expander-based algorithms had previously been unable to achieve near-linear update time \cite{Chuzhoy:2019:NAD:3313276.3316320,detDiSSSPAndSCC,ChuzhoyS20_apsp,bernstein2020fully}.

Our key technical contribution is showing how to \emph{apply expander-based tools without resorting to expander decomposition}. In a nut-shell, we show that given \emph{any} low-diameter graph $G$, one can in almost-linear time compute a capacity $\kappa(v)$ for each vertex such that the total vertex capacity is small and such that the graph $G$ weighted by capacities effectively corresponds to a weighted vertex expander. We can then apply tools such as expander pruning directly to the low-diameter graph $G$. This allows the algorithm to avoid expander decomposition and instead focus on the much simpler task of computing low-diameter subgraphs. We believe that this technique has the potential to play a key role in designing other dynamic algorithms against an adaptive adversary. 

\paragraph{Breaking the Flow Decomposition Barrier for MWU.} We also briefly mention our technical contribution for the min-cost flow algorithm of Theorem \ref{thm:MainMinCost}. Plugging our new data structure into the MWU framework is not by itself sufficient, because as discussed above, existing implements of MWU necessarily encounter the \emph{flow decomposition barrier} (see for example \cite{madry2010faster}), as they  repeatedly send flow down an entire $s$-$t$ path. We propose a new (randomized) scheme that maintains an \emph{estimator} of the flow. While previous schemes have used estimators for the weights \cite{chekuri2018randomized, chekuri2020fast, chekuri2020fast2}, we are the first to directly maintain only an estimator of the \emph{solution}, i.e. of the flow itself. This poses various new problems to be considered: a more refined analysis of MWU is needed, a new type of query operation for the decremental SSSP data structure is necessary, and the flow estimator we compute is only a pseudoflow. We succeed in tackling these issues and provide a broad approach that might inspire more fast algorithms via the MWU framework.

\subsection{A Paper in Three Parts.}
The article effectively contains three separate papers. Part \ref{part:dynamicShortestPaths} contains our decremental SSSP data structure (Theorem \ref{thm:mainSSSPResult}). We consider this part to be our main technical contribution; it is entirely self-contained and can be read as its own paper on dynamic shortest paths. Part \ref{part:augmented-queries} shows how to extend the data structure from Part \ref{part:dynamicShortestPaths} to answer threshold sub-path queries, which are required for our min-cost flow algorithm. Finally, Part \ref{part:minCostFlow} contains our min-cost flow result (Theorem \ref{thm:MainMinCost}); it is also entirely self-contained and can be read separately. In fact, Part III has zero overlap in techniques with the previous parts. The only reason we include it in the same paper is because it uses the data structure from Parts \ref{part:dynamicShortestPaths}/\ref{part:augmented-queries} as a black box.

Before \Cref{part:dynamicShortestPaths}, we include a detailed overview of techniques in the sections below.

%% file: overview.tex
\chapter*{Overview of Techniques} Most of the overview focuses on the dynamic SSSP algorithm itself (\Cref{part:dynamicShortestPaths}), as we consider this to be the main technical contribution. We give a short overview of the min-cost flow algorithm (\Cref{part:minCostFlow}) at the end. For ease of exposition, many of the definitions and lemmas in the overview sections sweep technical details under the rug; we restate our entire result more formally in the main body of the paper. 

\section{Overview for Part \ref{part:dynamicShortestPaths}: Dynamic Shortest Paths}
\label{sec:overviewPartII}

We now outline our framework for the dynamic algorithm of \Cref{thm:mainSSSPResult}. Our algorithm builds upon many existing techniques: the MES-tree from \cite{henzinger2014decremental}, dynamic graph covers from \cite{henzinger2016dynamic}, the deterministic hopset construction from \cite{gutenberg2020deterministic}, congestion balancing from \cite{detDiSSSPAndSCC}, and others. We first review some of the existing techniques we need. After that, the main goal of the overview is to highlight the crucial building block that previous approaches were not able to solve, and to introduce our new techniques for solving it.

\emph{For simplicity, we assume throughout this entire section that the graph $G$ is unweighted}. So every update to $G$ is just an edge deletion. The extension to graphs with positive weights involves a few technical adjustments, but is conceptually the same. We also assume that all vertices in $G$ have maximum degree $3$; see \Cref{prop:simplify-1} in the main body for justification.

\subsection{Existing Techniques}
\label{subsec:existingTechniquesInOverview}

\paragraph{Hop emulators.} A classic algorithm from 1981 by Even and Shiloach -- denoted ES-tree -- shows that decremental SSSP is easy to solve if we only care about short distances. 
Although we assume for this overview that the input graph $G$ is unweighted, our algorithm will create new weighted graphs. A simple scaling technique extends the ES-tree to weighted graphs in the following way. Define $\dist^{(h)}(s,v)$ to be the length of the shortest $s-v$ that uses at most $h$ edges. Then, $\estree(G,s,h)$ maintains a $(1+\eps)$-approximation to $\dist^{(h)}(s,v)$, for all $v \in V$, in total time $\Otil(mh)$ \cite{bernstein2009fully}.

Note that  $\estree(G,s,h)$ returns a $(1+\eps)$-approximation to $\dist(s,v)$ if $\dist^{(h)}(s,v) \sim \dist(s,v)$ -- that is, if there exists a $(1+\eps)$-approximate shortest from $s$ to $v$ with at most $h$ edges. A common technique in dynamic shortest paths is to construct a \emph{weighted} graph $H$ that has approximately the same shortest distances as $G$, but for which the above property holds for \emph{all} pairs of vertices.

\begin{defn}
Given a graph $G = (V,E)$, we say that graph $H = (V_H, E_H)$ is a $(h,\eps)$-emulator of $G$ if: {\bf 1)} $V_G \subseteq V_H$, {\bf 2)} for every pair $(u,v) \in V_G$, $\dist_G(u,v) \leq \dist_H(u,v)$, and {\bf 3)} for every pair $(u,v) \in V_G$, there exists a path $P_H(u,v) \in H$ such that $w(P_H(u,v)) \leq (1+\eps)\dist_G(u,v)$ and 
the number of edges on $P_H(u,v)$ is at most $h$.
\end{defn}

Observe that if $H$ is a $(h,\eps)$-emulator of $G$ then running $\estree(H,s,h)$ returns $(1+\eps)$-approximate distances in $G$. All efficient algorithms for sparse graphs, including ours, follow the same basic approach: maintain a $(h,\eps)$-emulator $H$, and then maintain
 $\estree(H,s,h)$. Observe that this returns $(1+\eps)^2$-approximate distances in $G$. The time to run the ES-tree in $H$ is $\Otil(mh)$. The harder step is maintaining the $(h,\eps)$-emulator $H$.

There is a huge amount of work on maintaining hopsets in decremental graphs. If the adversary is \emph{oblivious}, Henzinger et al. \cite{henzinger2014decremental} showed an essentially optimal algorithm: they maintain a $(n^{o(1)}, \eps)$-emulator in total time $\Ohat(m)$. But as we discuss below, there is a crucial obstacle to obtaining such guarantees against an \emph{adaptive} adversary. The state-of-the art adaptive algorithm by Probst Gutenberg and Wulff-Nilsen \cite{gutenberg2020deterministic} still suffers from polynomial overhead: they maintain a $(\sqrt{n},\eps)$-emulator in $\Ohat(m\sqrt{n})$ total update time.

\paragraph{Layered construction of hop emulators.}

The standard way of constructing a hop-emulator is to add edge $(u,v)$ of weight $\dist(u,v)$ for some select pairs $(u,v)$. The difficulty is that this requires knowing $\dist(u,v)$, which is precisely the problem we are trying to solve. To overcome this, many algorithms use a layered approach. Let $\gammacompress$ be some parameter that is $n^{o(1)}$ but bigger than $\polylog(n)$. The idea of layering is to first use a regular ES-tree to maintain $\dist(u,v)$ for some nearby pairs in $G$ with $\dist(u,v) \leq \gammacompress$. By adding the corresponding edges $(u,v)$ to an emulator, one can then construct a $(\diam(G)/\gammacompress,\eps)$-emulator $H_1$ of $G$; intuitively, shortest paths in $H_1$ have fewer edges than those in $G$ by a $n^{o(1)}$ factor. The next step is to construct an emulator $H_2$ that further compresses the number of edges on shortest paths. Observe that by construction of our emulator, for any pair of vertices $(x,y)$, $\dist^{(\gammacompress)}_{H_1}(x,y) \sim \dist_G(x,y)$ as long as $\dist_G(x,y) \leq \gammacompress^2$. Thus, running $\estree(H_1,\cdot,\gammacompress)$ actually gives us distances up to $\gammacompress^2$ in $G$; these distances can then be used to construct a $(\diam(G)/\gammacompress^2,\eps)$-emulator $H_2$ of $G$ (see Figure \ref{fig:overview-hop-compression}). Continuing in this way, after $q = \log_{\gammacompress}(n) = o(\log(n))$ iterations, the emulator $H_q$ will be a $(n^{o(1)},\eps)$-emulator, as desired.

\begin{figure}[!ht]
\centering
\includegraphics[scale=.33]{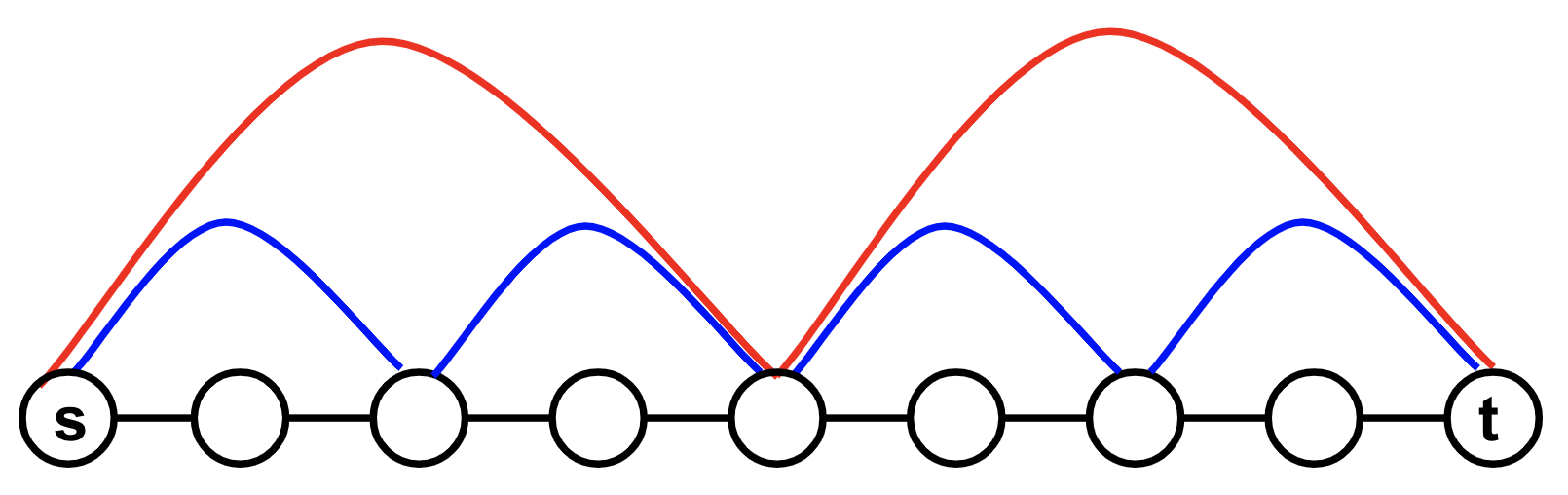}
\caption{The graph $G$ (black edges) has initially large diameter. But $H_1$ (black and blue edges) compresses the graph and reduces the number of edges on shortest paths by factor $2$. Finally, $H_2$ (black, blue and red edges) compresses the graph even further (also roughly by factor $2$).}
\label{fig:overview-hop-compression}
\end{figure}

Our algorithm follows the same layered approach. For ease of exposition, we focus this overview on the goal below, which corresponds to constructing the first hop-emulator $H_1$ of the layering; the crucial obstacle to adaptive algorithms is already present in this simplified problem.

\begin{goal}[Hop Compression] 
\label{goal:overview} Given a decremental graph $G$ with large diameter and a paramemter $\gammacompress = n^{o(1)}$, maintain a $(\diam(G)/\gammacompress,\eps)$-emulator of $G$ with $\Ohat(m)$ edges in total update time $\Ohat(m)$.
\end{goal}

\subsection{Dynamic Hop-Emulator via Covering}
We now describe the basic structure of the emulator $H$ that we construct.

\begin{defn} [Covering] (highly idealized version of Definition \ref{def:Covering})
\label{def:overview-covering} 
Fix parameters $d,D$, where $d,D$ and $D/d$ are all $n^{o(1)}$. We say that an algorithm maintains a \emph{covering} of a decremental graph $G$ if it maintains cores $C_1,\ldots,C_q$, where each $C_i \subset V$, with the following properties:
\begin{enumerate}
	\item \label{item:overview-decremental} The algorithm can create new cores, but once a core $C_i$ is created it only shrinks over time.
	\item \label{item:overview-diameter} Each core $C_i$ has weak diameter $\diam_G(C_i) \defeq \max_{x,y\in C_i} \dist_G(x,y) \leq d$. (Actually, different cores have slightly different diameters, but we omit this complexity in the overview.)
	\item \label{item:overview-cover} Each vertex $v$ is near some $C_i$; formally, it is in $\ball(C_i,4d)$. 
	\item \label{item:overview-shells} Throughout the entire course of the algorithm, each vertex $v$ belongs to only $n^{o(1)}$ different $\shell(C_i)$, where $\shell(C_i) = \ball(C_i, D)$.
\end{enumerate}
\end{defn}

This covering is similar to one used by the previous algorithms of \cite{henzinger2014decremental,chechik2018near, gutenberg2020deterministic}, with the crucial difference that those papers used single vertices $c_i$ instead of low-diameter cores $C_i$. We will need our more general version for our new approach to maintaining such a covering. 

\paragraph{Emulator via Covering.} We outline why such a covering leads to the desired $(\gamma/n^{o(1)},\eps)$-emulator $H$, as outlined in Goal \ref{goal:overview}. The algorithm maintains SSSP from each set $C_i$ up to diameter $D$ using an ES-tree. This is done by simply adding a dummy source $s_i$ with edges to every vertex in $C_i$; if a vertex w is removed from $C_i$, the edge $(s_i,w)$ is deleted (vertices are never added to $C_i$ by Property \ref{item:overview-decremental}). All these ES-trees can be maintained efficiently because by property \ref{item:overview-shells}, the sum of $|\shell(C_i)| = |\ball(C_i,D)|$ is small. The algorithm then constructs emulator $H$ as follows. In addition to the vertices of $G$, $H$ contains a vertex $v_{C_i}$ for each core $C_i$. That is, $V_H = V \cup \{v_{C_1},\ldots,v_{C_q}\}$. For every vertex $v \in C_i$ we add an edge $(v,v_{C_i})$ to $E_H$ of weight $d$. Finally, for every vertex $w \in \shell(C_i)$ we add an edge $(w,v_{C_i})$ of weight $\dist(w,C_i)$. Note that $\dist(w,C_i)$ and the corresponding edge-weight in $H$ can change as vertices in $G$ are deleted; if $w$ leaves $\shell(V_i)$ then the edge $(w,v_{C_i})$ is deleted from $H$. 

\paragraph{Analysis.} We now argue that hop-distances in $H$ are compressed by a factor of about $D/4 = n^{o(1)}$. We will show that for any vertices $u,v \in V$ with $\dist_G(u,v) = D/2$, there is an approximate shortest path $P_H(u,v)$ in $H$ with only two edges. If this property holds, then given any shortest path $\pi_G(x,y)$ in $G$ with $\ell$ edges, one can break $\pi_G(x,y)$ into $2\ell / D$ segments of length $D/2$ and then traverse each segment in $H$ using only $2$ edges, leading to an $x-y$ path in $H$ with $4\ell/D$ edges.

To prove the above property for $u,v$ with $\dist_G(u,v) = D/2$, observe that Property \ref{item:overview-cover} guarantees the existence of some $v_{C_i}$ such that $\dist(u,C_i) \leq 4d \ll D$. We thus have that $\dist(v,C_i) \leq D/2 + 4d < D$, so both $(u,v_{C_i})$ and $(v,v_{C_i})$ are edges in $H$. Now consider the 2-hop path $P_H(u,v) = u \rightarrow v_{C_i} \rightarrow v$ in $H$. We have that $w(P_H(u,v)) = \dist(u,v_{C_i}) + \diam_G(C_i) + \dist(v_{C_i},w) \leq 4d + d + (D/2 + 4d)< (1+\eps)D/2$, where the before-last inequality follows from Property \ref{item:overview-diameter}.

\subsection{The Crucial Building Block: Maintaining Low-Diameter Sets}
The difficult part of maintaining a covering is maintaining the cores $C_1,\ldots,C_q$. After the algorithm initially computes some core $C_i$, it might need to remove vertices from $C_i$ as $G$ undergoes deletions, in order to maintain the property that $C_i$ has small diameter (Property \ref{item:overview-diameter}). We can abstract this goal from the specifics of cores/shells and define the following crucial building block:

\begin{definition}[Crucial Building Block; highly simplified version of Robust Core in Definition \ref{def:Core}]
\label{def:overview-core}
Say that we are given a set $\Kinit \subset V(G)$ with weak diameter $\diam_G(\Kinit) \leq d = n^{o(1)}$ and that the graph $G$ is subject to edge deletions. The goal is to maintain a set $K \subseteq \Kinit$ with the following properties:
\begin{itemize}
	\item {\bf Decremental Property:} the set $K$ is decremental, i.e. it only shrinks over time.
	\item {\bf Diameter Property:}  $\diam_{G}(K) \leq dn^{o(1)}$.
	\item {\bf Scattering Property:} For every vertex $v \in \Kinit \setminus K$, $|\ball_G(v,2d)\cap\Kinit|\le(1-\scatter)\cdot|\Kinit|$, where $\scatter = 1/n^{o(1)}$. 
\end{itemize}
\end{definition}

We refer to the above building block as the Robust Core problem. An algorithm for Robust Core leads to a relatively straightforward algorithm for efficiently maintaining the covering in Definition \ref{def:overview-covering}. Loosely speaking, when a new core $C_i$ is initially created it corresponds to $\Kinit$, while the larger graph $G$ in robust core corresponds to $\shell(C_i)$. The set $K$ then corresponds to the core $C_i$ that is maintained as the graph undergoes edge deletions. 
The decremental property of Robust Core corresponds to Property \ref{item:overview-decremental} of Definition \ref{def:overview-covering}. The diameter property corresponds to Property \ref{item:overview-diameter}. Finally, the scattering property ensures that every time a vertex leaves a core, its neighborhood shrinks by a significant fraction; intuitively, such shrinking can only occur a small number of times in total, so a vertex can only participate in a small number of cores (and hence a small number of shells), which ensures Property \ref{item:overview-shells}.

We now leave aside the details of cores and shells and focus on the abstraction of Robust Core.

\paragraph{Previous Approaches to Robust Core (and their Limitations).}
Although it is not typically stated as such, the Robust Core problem distills the most basic version of a building block that is solved by almost all decremental SSSP algorithms for sparse graphs. This building block has also served as the primary obstacle to progress on this problem. We briefly outline previous approaches.
\begin{itemize}[leftmargin=*]
	\item {\bf Non-Adaptive Adversaries: Random Source.} Robust Core is quite simple to solve with a randomized algorithm that assumes an oblivious adversary: pick a random source $k \in \Kinit$ and maintain $\ball(k,7d) \cap \Kinit$ using an ES-tree. The algorithm keeps this ES-tree as long as $|\ball(k,4d) \cap \Kinit| \geq |\Kinit|/2$. Note that this property ensures that if a vertex $v$ leaves the ES-tree, i.e. if $\dist(k,v)$ becomes larger than $7d$, then $\ball(v,2d)\cap\Kinit$ is disjoint from $\ball(k,4d)\cap\Kinit$, so $v$ can removed from $K$ according to the scattering property. Whenever $|\ball(k,4d) \cap \Kinit|$ becomes too small, the algorithm removes $k$ from $K$ and picks a different random source. One can show that the algorithm only needs one single source in expectation, and $O(\log(n))$ with high probability. Loosely speaking, the argument is that because source $k$ is chosen at random from $K$, the fact that $\ball(k,4d) \cap \Kinit$ has become small implies that, in expectation, $\ball(k,4d) \cap \Kinit$ has become small for half the vertices $v \in K$, which in turn implies that $\ball(v,2d) \cap \Kinit$ has become small for \emph{all} vertices in $K$, so by the scattering property, all vertices can be removed from $K$.
	
	\hspace{3ex} Although the idea of picking a random source is very simple, it is also extremely powerful and leads to a total update time of $\Ohat(m)$ for Robust Core. Unfortunately it has zero utility against adaptive adversaries, because the randomness of the source is no longer independent from the sequence of updates, so the adversary can easily disconnect the source while leaving the rest of the core intact. This one technique, along with a natural generalization to random hitting sets, accounts for much of the gap between adaptive and oblivious algorithms for dynamic SSSP, as well as for related problems such as dynamic strongly connected components (see e.g. \cite{baswana2007improved,roditty2008improved,roditty2012dynamic,henzinger2014decremental,bernstein2016maintaining,chechik2016decremental, chechik2018near, bernstein2019decremental,GutenbergW20a,bernstein2020near}). 
	
	\item {\bf Adaptive Adversaries: Many Sources.} The best-known \emph{adaptive} algorithms for the building block are much slower. Since one can no longer pick a random source, two recent algorithms run an ES-tree from \emph{every} vertex in $K$ \cite{bernstein2017deterministic,gutenberg2020deterministic}. A trivial implementation leads to total update time $O(mn)$, but those papers use sophisticated density arguments to limit the size of ES-trees. These ideas lead to total update time $\Ohat(m\sqrt{n})$ \cite{gutenberg2020deterministic}, but as noted in both papers,  $\Ohat(m\sqrt{n})$ is hard barrier for this approach.

	\item {\bf Adaptive Adversaries: Rooting at an Expander.} Some very recent work on related problems \cite{Chuzhoy:2019:NAD:3313276.3316320,ChuzhoyS20_apsp,detDiSSSPAndSCC} suggests that one can go beyond  $\Ohat(m\sqrt{n})$ with expander tools. Say that the set $\Kinit$ is a $\phi$-(vertex)-expander for $\phi = 1/n^{o(1)}$. (See Definition \ref{def:overview-expander} in the subsection below). Any $\phi$-expander has small diameter. Because expanders are highly robust to deletions, the algorithm can efficiently maintain a large expander $X \in \Kinit$ using standard expander pruning (see Theorem \ref{thm:overview-pruning} in subsection below). The algorithm then maintains $\ball(X,10d)$ and removes from $K$ any vertex that is not in this ball. Intuitively, the algorithm replaces a random source with a deterministic expander, as both have the property of being robust to deletions.
	
	\hspace{3ex} The issue is that even though $\Kinit$ has small diameter, it might not be an expander. The natural solution is to maintain a decomposition of the graph into expanders and handle each expander separately. Unfortunately, such a decomposition must necessarily allow for up to $\phi n$ separator vertices that do not belong to any expander. If $\phi = 1/n^{o(1)}$ then the number of separator vertices is large, and it is unclear how to handle them efficiently.
	We suspect that setting $\phi$ to be a small polynomial, one could combine this expander approach with the density arguments from \cite{bernstein2017deterministic,gutenberg2020deterministic} mentioned above to achieve total update time $O(mn^{1/2-\delta})$. But because $\phi$ is a polynomial, such an approach could not lead to $\Ohat(m)$ total update time.
\end{itemize}

\subsection{Turning a Non-expander into an Expander}
\label{subsec:overview-capacitated}
We now outline our approach to the crucial building block above. In a nutshell, we show the first dynamic algorithm that uses expander tools while bypassing expander decomposition. As above, we assume that $G$ has constant degree.

\paragraph{Expander Preliminaries.}

In this overview, expander always refer to a \emph{vertex} expander. Our expansion factor will always be $1/n^{o(1)}$. 

\begin{definition} \label{def:overview-expander}
Consider an unweighted, undirected graph $G$ and a set $X \subseteq V$. We say that $(L,S,R)$ is a vertex cut with respect to $X$ if $L,S,R$ partition $V(G)$, $|L \cap X| \leq |R \cap X|$, and there are no edges between $L$ and $R$. We say that $(L,S,R)$ is a \emph{sparse} vertex cut with respect to $X$ if $|S| \leq |L \cap X|/n^{o(1)}$. We say that $X \subseteq V(G)$ forms an expander in $G$ if there exists no sparse vertex cuts with respect to $X$. (Note that setting $X = V(G)$ gives the standard definitions of sparse vertex cut and vertex expansion.)
\end{definition}

The key feature of expanders for our purposes is that they are robust to edge deletions. In particular, if $\Xinit$ initially forms an expander in $G$, then even after a large number of edge deletions, there is a guarantee to exist a large $X \subseteq \Xinit$ that forms an expander in $G$, and $X$ can be maintained efficiently. This is known as expander pruning.\footnote{The theorem from \cite{SaranurakW19} is actually stated for edge expanders, and for technical reasons related to expander embedding, we only prune edge expanders in the main body of the paper as well. But we effectively use it to prune vertex expanders, so for simplicity that is how we state it for the overview.}
	
\begin{theorem}[Pruning in Vertex Expanders \cite{SaranurakW19}]
	\label{thm:overview-pruning}
	Let $G$ be a graph subject to edge deletions, and consider a set $\Xinit \subseteq V(G)$ such that $\Xinit$ initially forms an expander in $G$. There exists an algorithm $\Prune(G,\Xinit)$ that can process up to $|\Xinit|/n^{o(1)}$ edge deletions while maintaining a decremental set $X \subseteq \Xinit$ such that $|X| \geq |\Xinit|/2$ and $X$ forms an expander in $G$. The total running time is $\Ohat(|\Xinit|)$.
\end{theorem} 

\paragraph{Capacitated Expanders.}
We argued above that Robust Core can be solved efficiently using standard tools if $\Kinit$ is initially an expander, because each edge deletion would have low impact. But the Robust Core problem only has the much weaker guarantee that $\Kinit$ has small (weak) diameter. To develop some intuition for our approach, consider the example where $\Kinit = V(G)$ and $G$ consists of two  expanders $A,B$ with a single crossing edge $(u,v)$. Note that $G$ (and hence $\Kinit$) has small diameter but is far from being an expander. In particular, it is clear that $u,v$ serve as bottlenecks, in that deleting the $O(1)$ edges incident to $u$ and $v$ would immediately disconnect the graph and cause the scattering property to hold for all vertices. By contrast, deleting  all edges incident to some random vertex $z \in A$ would have low impact, because $A$ is an expander. 

We thus see that in a non-expander, some vertices are much more critical than others. Quantitatively speaking, the vertices $u$ and $v$ are about $n$ times more critical than a random vertex $z \in A$, since their deletions would scatter $n$ vertices. The $O(1)$ neighbors of $u$ and $v$ are also highly critical, since deleting all of their incident edges would again scatter the graph. Criticality then drops off exponentially as we go further from $u$ and $v$. See also Figure \ref{fig:overview-critical-edges-upper-bound}.

\begin{figure}[h!]
\centering
\includegraphics[scale=.5]{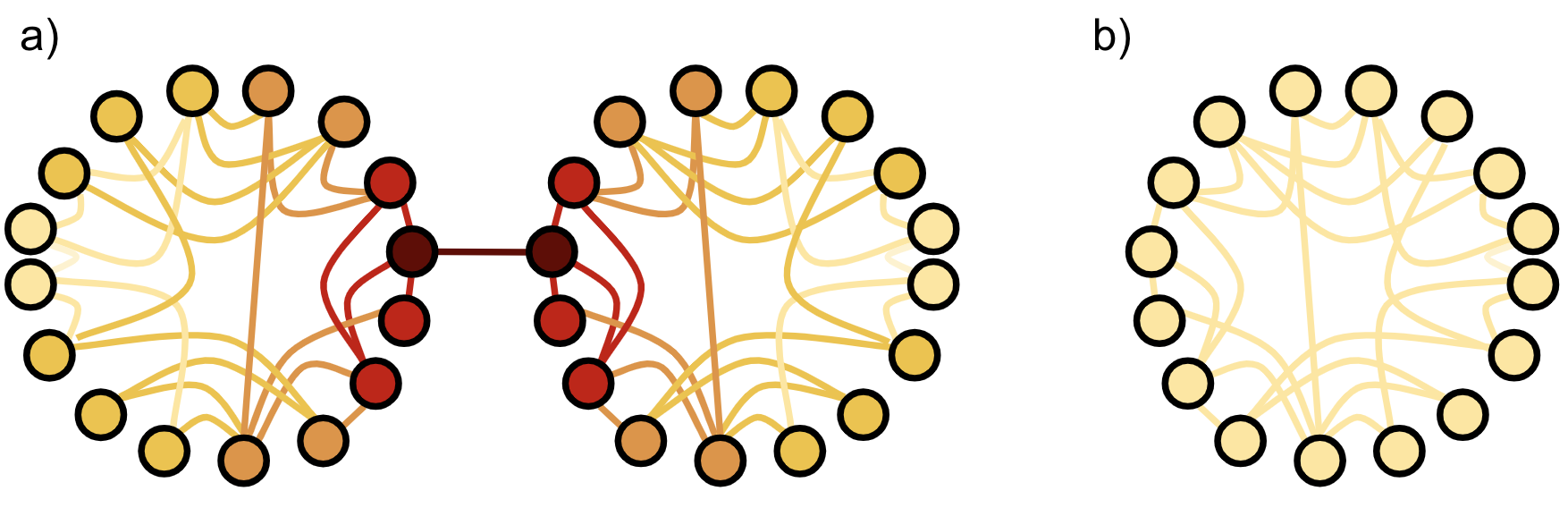}
\caption{Criticality in two graph examples (where red is very critical, yellow not critical). {\bf a)} two expanders $A,B$ joined by a single edge $(u,v)$. The vertices $u,v$ are extremely critical but parts of the graph further away are relatively uncritical. {\bf b)} an expander graph. Here, no vertex is critical.}
\label{fig:overview-critical-edges-upper-bound}
\end{figure}

Our key contribution is an algorithm that computes a criticality score $\kappa(v)$ for each vertex such that the graph weighted by $\kappa$ effectively corresponds to an expander. We now formalize this notion.

\begin{defn}
\label{def:overview-capacitated-expander}
Let $G$ be a graph with vertex capacities $\kappa$, where $\kappa(v) \geq 1$. For any $X \subseteq V(G)$,
we say that $(L,S,R)$ forms a spare \emph{capacitated} vertex cut with respect to $X,\kappa$ if $(L,S,R)$ is a vertex cut with respect to $X$ and  $\sum_{v \in S} \kappa(v) \leq |L \cap X|/n^{o(1)}$. We say that $X,\kappa$ forms a capacitated expander in $G$ if there are no sparse capacitated vertex cuts with respect to $X,\kappa$.
\end{defn}

Note that any connected graph can be made into a capacitated vertex expander by setting $\kappa(v) = n$ for all vertices in $V$. But we want to keep to total vertex capacity small because
our algorithm will decrementally maintain a capacitated expander using pruning, and pruning on capacitated expanders will incur update time proportional to capacities. Intuitively, the reason for this is that by definition of capacitated expander, to disconnect $\beta$ vertices from the graph the adversary has to delete edges with $\sum_{e} \kappa(e) = \Omegahat(\beta)$. 

\begin{lem}[Capacitated Expander Pruning -- implied by \Cref{lem:bound num phase}]
\label{lem:overview-capacity-pruning} Say that we are given a decremental graph $G$, a set $\Xinit \subseteq V(G)$ and a function $\kappa$ such that $(\Xinit,\kappa)$ forms a capacitated expander in $G$. Then, there is an algorithm $\Prune(G,\Xinit,\kappa)$ that can process any sequence of edge deletions in $G$ that satisfy $\sum_e \kappa(e) = O(|\Xinit|/n^{o(1)})$, while maintaining a decremental set $X \subset \Xinit$ such that $|X| \geq |\Xinit|/2$ and $(X,\kappa)$ remains a capacitated expander in $G$. The total running time is $\Ohat(|\Xinit|)$.
\end{lem}

To prove Lemma \ref{lem:overview-capacity-pruning}, we do not need to modify standard pruning from Theorem \ref{thm:overview-pruning}. Instead, we are able to show that one can replace the capacitated expander by a regular uncapacitated one, on which we can then run standard pruning.

\paragraph{Ensuring Small Total Capacity.}
Note that our capacitated pruning terminates after $O(|\Xinit|/n^{o(1)})$ total edge capacity is deleted, at which point we need to reinitialize the pruning algorithm if we want to keep maintaining an expander. Thus, to avoid doing many reinitializations, we want the average edge capacity to be small.
Note that because we assume the main graph $G$ has constant degree, $\sum_{e \in E(G)} \kappa(e) \sim \sum_{v \in V(G)} \kappa(v)$. Our goal can thus be summarized as follows: given graph $G$ and some core $K$, find a capacity function $\kappa$ that turns $K$ into a capacitated expander while minimizing $\sum_{v\in V(G)} \kappa(v)$. One of the highlights of our paper is the following structural lemma, which shows that this minimum $\sum_{v\in V(G)} \kappa(v)$ is directly related to the (weak) diameter of $K$. This lemma is implicitly proved in Section \ref{sec:Core} or Part \ref{part:dynamicShortestPaths}; for an explicit proof see Appendix \ref{sec:overview-proof-kappa}. 

\begin{lem}[Small Capacity Sum for Small Diameter]
\label{lem:overview-kappa} Given graph $G$ and any $K \subseteq V(G)$,
there exists a capacity function $\kappa(v)$ such that $(K, \kappa)$ forms a capacitated vertex expander in $G$ and $\sum_{v \in V(G)} \kappa(v) = \Otil(|K|\diam_{G}(K))$, where $\diam_{G}(K) \defeq \max_{x,y \in K} \dist_{G}(x,y)$.
This bound is tight: there exist $G,K$  such that any feasible function $\kappa$ necessarily has $\sum_{v \in V(G)} \kappa(v) = \Omegahat(|K|\diam_{G}(K))$. 
\end{lem}

\begin{figure}[h!]
\centering
\includegraphics[scale=.5]{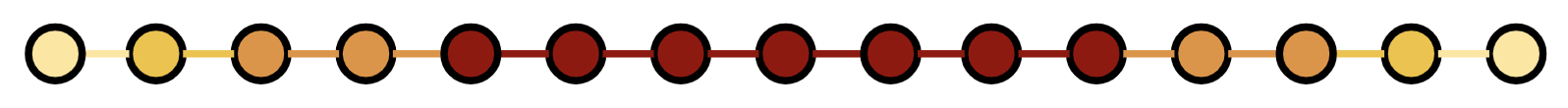}
\caption{On a path, roughly half vertices are very critical. This is no coincidence: the path graph has large diameter.}
\label{fig:overview-critical-edges-lower-bound}
\end{figure}

Unfortunately, we do not know how to compute the function $\kappa$ guaranteed by Lemma \ref{lem:overview-kappa} in near-linear time. Instead, we compute a slightly relaxed version which only guarantees expansion for relatively large cuts -- i.e. cuts where $L \cap K$ is large with respect to $K$. We can show that the pruning of Lemma \ref{lem:overview-capacity-pruning} also works with this relaxed notion of capacitated vertex expansion.

\begin{restatable}[Computing the Capacities]{lemma}{lemmaCongestionOverview}
	\label{lem:overview-kappa-relaxed} Given graph $G$ and any $K \subseteq V(G)$,
	one can compute in $\Ohat(n\diam_{G}(K))$ time a capacity function $\kappa(v)$ such that $\sum_{v \in V(G)} \kappa(v) = \Ohat(|K|\diam_{G}(K))$ and such that there are no sparse capacitated vertex cuts $(L,S,R)$ with respect to $K$ for which $|L\cap K| \geq K/n^{o(1)}$.
\end{restatable}

\subsection[Algorithm for Robust Core]{Algorithm for Robust Core (Simplified Version of Algorithm \ref{alg:Core} in \Cref{part:dynamicShortestPaths}.)}
\label{subsec:overview-RobustCoreAlgoAndAnalysis}

We later sketch a proof for Lemma \ref{lem:overview-kappa-relaxed}. But first let us show how capacitated expanders can be used to solve Robust Core (Definition \ref{def:overview-core}); see pseudocode below. 

\paragraph{Initialization of Robust Core.} First we apply Lemma \ref{lem:overview-kappa-relaxed} to compute a capacity function $\kappa$ such that $(\Kinit,\kappa)$ forms a capacitated expander in $G$. Recall that $G$ has constant degree. Since Robust Core assumes that $\diam_G(\Kinit) = \Ohat(1)$, the running time of Lemma \ref{lem:overview-kappa-relaxed} is $\Ohat(n)$ and we have $\sum_{e \in E(G)} \kappa(e) = \Theta(\sum_{v \in V(G)} \kappa(v)) = \Ohat(n)$. Using capacitated expander pruning (Lemma \ref{lem:overview-capacity-pruning}), we can maintain an expander $X$ such that $(X,\kappa)$ forms a capacitated expander in $G$. We then define our solution $K$ to Robust Core as follows: initially $K = \Kinit$, and we remove from $K$ any vertex that leaves to be $\ball_{G}(X,4d)$. $K$ clearly satisfies the decremental property of Robust Core. We can show that $K$ satisfies the diameter property because $K \subseteq \ball(X,4d)$ and $X$ itself has low diameter because (loosely speaking) $X$ forms a capacitated expander in $G$.\footnote{Technically speaking, we show that not only does $(X,\kappa)$ form an expander, but also that $\kappa$ actually yields a short-path embedding of an expander into $X$.} Finally, as long as we have that $|X| \geq |\Kinit|/n^{o(1)}$, the core $K$ satisfies the scattering property because $v$ leaves $K$ only if it leaves $\ball(X,4d)$, at which point $\ball(v,2d) \cap X = \emptyset$. We have thus shown that $K$ continues to be a valid solution to Robust Core as long as $|X|$ is large. By capacitated expander pruning (Lemma \ref{lem:overview-capacity-pruning}), $|X|$ will be sufficiently large as long as the total capacity of deleted edges satisfies $\sum_e \kappa(e) = O(n/n^{o(1)})$. 

\begin{algorithm}[!ht]
	$K \gets \Kinit \qquad n \gets |\Kinit|$\;
	\While(\tcp*[f]{Each iteration is a single phase}){$|K| \geq n^{1-o(1)}$}{
		Compute $\kappa$ as in Lemma \ref{lem:overview-kappa-relaxed} such that $(K,\kappa)$ is a capacitated expander in $G$. \tcp*[f]{Recall that the algorithm preservers the Monotonicity Invariant.}\;
		\While{$X \subseteq K$ from running $\textsc{Prune}(G,K,\kappa)$ has size $\geq n^{1-o(1)}$}{
			Maintain $\ball(X,4d)$ using an ES-tree and remove every vertex leaving $\ball(X, 4d)$ from $K$. 
			\\
			\tcp*[h]{$|X|$ can become too small after adversary deletes $\Omegahat(n)$ edge capacity. Once this happens, algorithm restarts the outer while loop with the current $K$.}
		}
	}
	Remove all vertices from $K$ and terminate. \tcp*[f]{Only executed once $|K| \leq n^{1-o(1)}$, so all remaining vertices satisfy scattering property.}
	\caption{$\Core(G,\Kinit)$}\label{alg:overview-algorithmRobustCore}
\end{algorithm}

\paragraph{Maintaining Robust Core.} At some point, however, the capacity of deleted edges will be too large, and $|X|$ may become too small. Consider the moment right before the deletion that causes $|X|$ to become too small. At this moment, $K$ is still a valid core, and hence has small (weak) diameter. Moreover, we can assume that $|K| = \Theta(\Kinit)$, since otherwise the entire core is scattered and we can terminate Robust Core; formally, we are able to show that by the scattering property, if $|K|$ becomes very small compared to $|\Kinit|$ we can simply remove every remaining vertex in $K$. The algorithm now essentially restarts the entire process above, but with $K$ instead of $\Kinit$. That is, it computes a new capacity function $\kappa$ such $(K,\kappa)$ forms a capacitated expander in $G$. Since $K$ has small diameter, the running time is again $\Ohat(n)$ and we again have $\sum_{v \in V(G)} \kappa(v) = \Ohat(n)$. The algorithm now uses capacitated pruning to maintain a new expander $X \subseteq K$ and again removes from $K$ any vertex that leaves $\ball(X,4d)$. As before, $K$ remains a valid core as long as $|X| \geq |K|/n^{o(1)}$; here we use the fact that $|K| = \Theta(|\Kinit|)$ to ensure the scattering property. Thus, by the guarantees of pruning, $K$ remains valid until the adversary deletes at least $\Omega(n/n^{o(1)})$ more edge capacity.

\paragraph{Endgame of Robust Core.} Once the adversary deletes enough edge capacity, the algorithm again computes a new function $\kappa$ for the current $K$. We refer to each such recomputation of $\kappa$ as a new \emph{phase}. The algorithm continues executing phases until eventually $|K|$ becomes much smaller than $\Kinit$; as mentioned above, the algorithm can then remove all remaining vertices from $|K|$ and terminate. 

\paragraph{Analysis of \Cref{alg:overview-algorithmRobustCore}.}We argued above that each phase requires $\Ohat(n)$ time. The only step left is thus to show that the total number of phases is $\Ohat(1)$. To see this, assume for the moment that although $\kappa$ is recomputed between phases, every $\kappa(v)$ is monotonically increasing. The argument is now that since we always maintain a core $K$ with small diameter, Lemma \ref{lem:overview-kappa-relaxed} guarantees that the function $\kappa$ we compute to make $(K,\kappa)$ a capacitated expander in $G$ always has $\sum_{v \in \Kinit} \kappa(v)= \Ohat(n)$. Since $\kappa$ is monotonically increasing, this implies that the total vertex capacity over all phases is $\Ohat(n)$, so the total edge capacity is also $\Ohat(n)$. But a phase can only terminate after at least $n/n^{o(1)}$ edge-capacity has been deleted, leading to at most $\Ohat(n) / (n/n^{o(1)}) = \Ohat(1)$ phases.

To facilitate the above analysis, our algorithm will ensure that $\kappa(v)$ is indeed monotonic. Note that the algorithm only ever changes $\kappa$ at the beginning of a new phase.

\begin{inv}[Monotonicity Invariant]
Let $\kappanew$ be the new capacity function computed at the beginning of some phase of Algorithm $\Core$ and let $\kappaold$ be the capacity function computed in the previous phase. Then, we always have $\kappaold(v) \leq \kappanew(v) \ \forall v \in V$.
\end{inv}

We now briefly outline our algorithm for computing $\kappa$ in Lemma \ref{lem:overview-kappa-relaxed}. As we will see, the monotonicity invariant naturally follows from our approach and requires no extra work to ensure. Intuitively, since $G$ is decremental, it only becomes further from a vertex expander over time, so it is not surprising that the vertices of $G$ only become more critical.

\lemmaCongestionOverview*

\paragraph{Proof sketch of Lemma
 \ref{lem:overview-kappa-relaxed}} We follow the basic framework of \emph{congestion balancing} introduced by the authors in \cite{detDiSSSPAndSCC}. But unlike in  \cite{detDiSSSPAndSCC}, we do not need to assume the initial graph is an expander. Our overall framework thus ends up being significantly more powerful, both conceptually and technically. See \Cref{subsec:analysisOfRobustCorePart2} in the main body for details.

Recall that we only maintain a relaxed expansion that applies to \emph{balanced} cuts $(L,S,R)$ for which $|L \cap K| \geq |K|/n^{o(1)}$. The high-level idea of congestion balancing is quite intuitive. We initially set the the $\kappa(v)$ to be small enough that $\sum_{v \in V} \kappa(v) = O(|K|)$.  Then, we repeatedly find an arbitrary balanced cut $(L,S,R)$ such that $\sum_{v \in S} \kappa(v) < |L \cap K|/n^{o(1)}$. If no such cut exists, then $(K, \kappa)$ forms a capacitated expander in $G$, as desired. Else, increase the expansion of this cut by doubling all vertex capacities in $S$.

Note that this approach naturally satisfies the Monotonicity Invariant. When the algorithm needs to compute a new $\kappanew$ for a new set $K$, it starts the process with the old capacities $\kappaold(v)$. If there are no sparse balanced cuts using $\kappaold(v)$, the algorithm can just set $\kappanew = \kappaold$. Otherwise the algorithm only changes capacities by doubling them, so $\kappanew \geq \kappaold$.

The crux of the analysis is showing that such a doubling step can only occur $\Ohat(\diam_{G}(K))$ times. This allows us to bound the total running time; it also guarantees that we always have $\sum_{v \in V(G)} \kappa(v) = \Ohat(|K|\diam_{G}(K))$, because it is easy to check that each doubling step increases the total capacity of $S$ by at most $|L \cap K| < |K|$. To bound the number of doubling steps, we use a potential function $\Pi(G)$ that corresponds to (loosely speaking) the min-cost embedding in $G$, where the cost of a vertex $v$ is $\log(\kappa(v))$. We are able to show that each doubling step increases $\Pi(G)$ by $\Omegahat(|K|)$, and that $\Pi(G)$ is always $\Ohat(|K|\diam_{G}(K))$, thus giving the desired bound of $\Ohat(\diam_{G}(K))$ doubling steps.

\subsection{A Hierarchy of Emulators}
We have outlined above how to solve the crucial building block Robust Core, which can in turn be used to maintain a covering of G (Definition \ref{def:overview-covering}), which allows us to achieve Goal \ref{goal:overview} -- that is, to compress hop distances by a $n^{o(1)}$ factor. But decremental SSSP can only be solved efficiently when all hop distances are small, so we need to apply this compression multiple times. In particular, we have a hierarchy of emulators, where $H_1$ compresses hop distances in $G$, $H_2$ compresses hop distances in $H_1$, and so on.  

This layering introduces several new challenges. The biggest one is that all the tools above assume a decremental graph, and even though $G$ is indeed decremental, the graphs $H_i$ may have both edge and vertex insertions. For example, the vertex set of $H_1$ also includes core vertices for each core in the covering of $G$, and when some core $C$ in $G$ becomes scattered, new cores are added to cover the vertices previously in $C$, so new core vertices and edges are added to $H_1$. Fortunately, these insertions have low impact on distances in $H_1$ because $H_1$ is emulating a decremental graph $G$. We thus refer to $H_1$ as being decremental \emph{with low-impact insertions}.
Since the algorithm for maintaining $H_2$ sees $H_1$ as its underlying graph, all of our tools must be extended to work in this setting.

The fact of emulators having low-impact insertions is a common problem in previous dynamic algorithms as well. While there exist algorithms that are able to extend the ES tree to work in such a setting (see especially \cite{henzinger2014decremental}), extending Robust Core and congestion balancing is significantly more challenging. Conceptually speaking, the main challenge lies with the scattering property: if $G$ has insertions, then $\ball(v,2d)$ can both shrink and grow, so a vertex can alternate between being scattered and unscattered. 

One of our key technical contributions is a more general framework for analyzing congestion balancing that naturally extends to graphs with low-impact insertions. At a high-level, congestion balancing from \cite{detDiSSSPAndSCC} defined a potential function $\Pi(G)$ on the input graph $G$ (see Lemma \ref{lem:overview-kappa-relaxed}). The issue is that if $G$ has insertions, then $\Pi(G)$ can actually decrease, which invalidates the analysis. To resolve this, we show that there exists a graph $\Ghat$ which is entirely decremental and yet has exactly the same vertex-cuts as $G$. We then show that the analysis of congestion balancing goes through if we instead look at $\Pi(\Ghat)$. We note that the algorithm never has to construct $\Ghat$; it is used purely for analysis. The formal analysis is highly non-trivial and we refer the reader to Section \ref{subsec:analysisOfRobustCorePart2} for more details. 

\paragraph{Returning the Path.} The hierarchy of emulators also creates unique difficulties in path-reporting (\Cref{part:augmented-queries}). We discuss this more at the end of the overview section, after we introduce the threshold-subpath queries that we need in our minimum cost-flow algorithm.

\section{Overview of Part \ref{part:minCostFlow}: Static Min-Cost Flow}
\label{sec:overview-flow}

We now outline our flow algorithm for Theorem \ref{thm:MainMinCost}. The techniques in Part \ref{part:minCostFlow} have zero overlap with those from \Cref{part:dynamicShortestPaths,part:augmented-queries}: the only relation is that Part \ref{part:minCostFlow} uses the dynamic SSSP data structure from \Cref{part:dynamicShortestPaths,part:augmented-queries} as a black box. 

\paragraph{Simplifying Assumptions.} For ease of exposition, this overview section focuses on the problem of \emph{vertex capacitated} max flow, and ignores costs entirely. 
We note that no almost-linear time algorithm is known even for this simpler problem. The extension to costs follows quite easily.

\paragraph{Notation.}
Let $G = (V,E,u)$ be the input graph, where $u(x)$ is the capacity of vertex $x$. Let $s$ be a fixed source and $t$ be a fixed sink. For any path $P$, define $\lambda(P)$ to be the minimum vertex capacity on $P$. The goal is to compute a flow vector $f \in \mathbb{R}^{E}_+$ that satisfies standard flow constraints: $\forall x \notin \{s,t\}$, $\inflow_f(x) = \outflow_f(x)$ (flow conservation) and $\inflow_f(x) \leq u(x)$ (feasibility). We define the value of $f$ to the total flow leaving $s$. Our goal is compute a $(1-\eps)$-optimal flow. 

\subsection{Existing Technique: Multiplicative Weight Updates}
We follow the framework of Garg and Koenneman for applying MWU to maximum flow \cite{garg2007faster}. We assume for simplicity that the approximation parameter $\eps$ is a constant. Loosely speaking, the framework is as follows:

\ignore{
\begin{enumerate}
	\item Initialize the flow: $f = 0$.
	\item Create a weight-function $w: V \rightarrow \mathbb{R}_+$; initially set all vertex weights to be around $1/\poly(m)$.
	\item Repeatedly compute a $(1+\eps)$-shortest path $\pi(s,t)$ with respect to the edge weights $w$
	\begin{enumerate}
		\item \label{item:overview-MWU-exit} If $w(\pi(s,t)) > 1$ then exit the loop.
		\item Else, define $\lambda = \lambda(\pi(s,t))$ and for every edge $e \in \pi(s,t)$ do: $f(e) = f(e) + \lambda$.
		\item \label{item:overview-MWU-weight} For each vertex $v \in \pi(s,t)$ do: $w(v) = w(v) \cdot \exp(\eps \lambda / u(v))$. (Loosely speaking, $w(v)$ increases by a $(1+\eps)$ factor for every $u(v)$ units of flow through $v$.)
	\end{enumerate}
	\item Return the flow $f$ scaled down by a factor of about $\Theta(\log(n))$.
\end{enumerate}

}

\begin{algorithm}[!ht]
	 Initialize the flow: $f \gets 0$.\;
	Create weight-function $w: V \rightarrow \mathbb{R}_+$; with initial vertex weights around $1/\poly(m)$.\;
	\While{\textbf{true}}{
	    Compute a $(1+\eps)$-shortest path $\pi(s,t)$ with respect to the edge weights $w$.\;
	    \If{$w(\pi(s,t)) > 1$}{\textbf{exit} while-loop.}
	    $\lambda \gets \lambda(\pi(s,t))$.\; 
	    \lForEach{edge $e \in \pi(s,t)$}{ $f(e) \gets f(e) + \lambda$.}
	    \tcc{For every $u(v)$ units of flow entering $v$, the weight $w(v)$ is increased by a $e^{\epsilon} \approx (1+\eps)$ factor}
		\lForEach{vertex $v \in \pi(s,t)$}{ $w(v) \gets w(v) \cdot \exp(\eps \lambda / u(v))$.\label{item:overview-MWU-weight} } 
	}
	\Return $f$ scaled by factor $\Theta(\log(n))$ \label{item:overview-MWU-exit}.
	\caption{$\textsc{MWU}(G,s,t)$}\label{alg:overview-algorithmMWU}
\end{algorithm}

At a very high-level, the algorithm increases the weights of vertices that receive a lot of flow relative to their capacity, so that the next shortest path is less likely to use that vertex. Using a primal-dual analysis (see e.g.~\cite{garg2007faster}), one can show that the returned flow is feasible and $(1-\eps)$ approximate.

Following the framework by Madry \cite{madry2010faster}, Chuzhoy and Khanna \cite{Chuzhoy:2019:NAD:3313276.3316320} used a dynamic SSSP data structure to avoid recomputing a new shortest $s$-$t$ path $\pi(s,t)$ from scratch with each iteration of the while loop. (A dynamic SSSP structure for edge-weighted graphs can easily be converted into one for vertex-weighted ones.) Because vertex weights only increase, a decremental SSSP data structure suffices. Note also that the MWU framework requires the data structure to work against an \emph{adaptive} adversary, because the updates to the data structure (the weight increases) depend on the $(1+\eps)$-shortest path returned by the data structure.

\paragraph{The Flow Decomposition Barrier.}
\newcommand{\pset}{\mathcal{P}}
In addition to computing the paths $\pi(s,t)$, the MWU framework also adjusts every vertex/edge on the path. Thus, if $\pset$ is the set of all $s$-$t$ paths returned by algorithm, then the total running of MWU is: [total update time of decremental SSSP] + [$\sum_{P \in \pset} |P|$]. Previous work bounds the second quantity in the following way. Say that we have weighted vertex capacities. On the one hand, each vertex $v$ receives at most $O(u(v) \log n)$ flow in total, since the flow $f/\Theta(\log(n))$ returned in step \ref{item:overview-MWU-exit} is guaranteed to be feasible. On the other hand, each path $\shortestSTPath$ sends at most $\lambda(\pi(s,t)) = \min_{v \in \shortestSTPath} u(v)$ flow which might only "fill-up" the minimizer vertex. There might thus be $O(n\log(n))$ paths in total, each of length at most $n$, so $\sum_{P \in \pset} |P| = O(n^2\log(n))$. An example where this behavior is apparent is given in Figure \ref{fig:overview-decomposition-mwu} below.

\begin{figure}[h!]
\centering
\includegraphics[scale=.33]{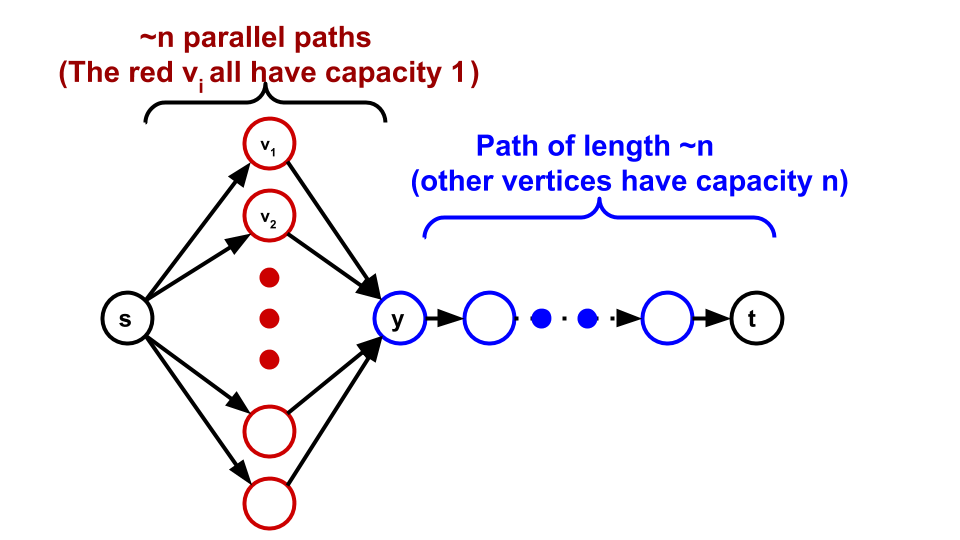}
\caption{Example of a graph with vertex capacities where running the standard MWU algorithm results in $\sum_{P \in \pset} |P| = \Theta(n^2\log(n))$.}
\label{fig:overview-decomposition-mwu}
\end{figure}

In the above figure, each path $\pi_{s,t}$ has $\lambda(\pi(s,t)) = 1$, so the algorithm only sends one unit of flow at a time. It is not hard to check that each of the red $v_i$ will be used $O(\log n)$ times, for a total of $n\log(n)$ paths; each path has length $n$, so $\sum_{P \in \pset} |P| = \Theta(n^2\log(n))$. One can similarly show that in edge-capacitated graphs, there are examples with $\sum_{P \in \pset} |P| = \Omega(mn\log(n))$. For unit edge capacities, $\sum_{P \in \pset} |P|$ is at most $O(m\log(n))$. Up to the extra $\log(n)$ factor, These bounds precisely correspond to what is known as the \emph{flow-decomposition barrier} for maximum flow \cite{goldberg1998beyond}. 

The previous state-of-the-art for adaptive decremental SSSP has total update time 
$\Otil(n^2)$ \cite{Chuzhoy:2019:NAD:3313276.3316320, ChuzhoyS20_apsp,bernstein2020fully}; plugging this into the MWU-framework gives an $\Otil(n^2)$ algorithm for approximate min-cost flow for graphs with unit edge capacities or vertex-capacitated graphs.
But these results did not lead to any improvement for edge-capacitated graphs precisely because of the flow-decomposition barrier. Similarly, our new data structure immediately yields an $\Ohat(m)$-time min-cost flow algorithm for unit-capacity graphs (itself a new result), but on its own cannot make progress in graphs with general vertex or edge capacities.

To get $\Ohat(m)$ time for general capacities, we need to modify the MWU framework. Ours is the first MWU-based algorithm for max flow to go beyond the flow-decomposition barrier. 

\subsection{Our New Approach: Beyond the Flow-Decomposition Barrier}

The basic idea of our approach is to design a new MWU-framework with the following property

\begin{inv}
	\label{inv:overview-MWU}
In our MWU framework, whenever the algorithm sends flow from $x$ to $y$ on edge $(x,y)$, it sends at least $\hat{\Omega}(u(y))$ flow.
\end{inv}

Combined with the fact that the final flow through any vertex $v$ is at most $O(u(y)\log(n))$, and the fact that MWU never cancels flow (because it does not deal with a residual graph), it is easy to see that Invariant \ref{inv:overview-MWU} guarantees that the total number of times the algorithm sends flow into any particular vertex $y$ is $\Ohat(1)$, so $\sum_{P \in \pset} |P| = \Ohat(n)$. Achieving this invariant requires making changes to the MWU-framework.

\paragraph{Pseudoflow.}
Consider Figure \ref{fig:overview-decomposition-mwu} again. Consider some path $\pi(s,t)$ chosen by the MWU algorithm. This path has $\lambda(\pi(s,t)) = 1$. The algorithm can send one of flow into some red $v_i$, but in order to preserve the invariant above, it cannot send 1 unit of flow down $(v_i,y)$. As a result, the flow we maintain is only a \emph{pseudoflow}: it is capacity-feasible, but does not obey flow conservation constraints. We will show, however, that we can couple the computed pseudo-flow to a near-optimal flow.

\begin{defn}[pseudo-optimal flow: simplified version of Definition \ref{def:feasibleFlow}]
\label{def:overview-pseudo-optimal}
We say that a pseudoflow $\flowEst$ is $(1-\eps)$-pseudo-optimal if there exists a valid flow $\flow$ such that 
\begin{itemize}
	\item $f$ is a $(1-\eps)$-optimal flow. 
	\item for every $v \in V$, $|\inflow_{\flow}(v) - \inflow_{\flowEst}(v)| \leq \eps u(v)$.	
\end{itemize}
\end{defn}

We later show that there exists a black box reduction from computing a $(1-\eps)$-optimal flow to computing a $(1-\eps)$-pseudo-optimal flow. But first, we focus this overview on computing a $(1-\eps)$-pseudo-optimal flow.

\paragraph{The Ideal Flow and the Estimated Flow.}
At each step, the algorithm will \emph{implicitly} compute a $(1+\eps)$-approximate shortest path $\pi(s,t)$, but to preserve Invariant \ref{inv:overview-MWU}, it will only add flow on some edges of $\pi(s,t)$. We denote the resulting pseudoflow $\flowEst$. To show that $\flowEst$ is $(1-\eps)$-pseudo-optimal, we will compare it to the \emph{ideal} flow $\flow$, which sends $\lambda(\pi(s,t))$ flow on every edge in $\pi(s,t)$, as in the standard MWU framework. Our approach thus needs to ensure that $\flowEst$ is always similar to $\flow$. 

\paragraph{Randomized Flow.}
Consider Figure \ref{fig:overview-decomposition-mwu} again. Say that MWU computes a long path sequence $P_1, P_2, \ldots$. For the first path $P_1$, the algorithm can simply increase $\flowEst(s,v_i)$ and not send any flow on the other edges; we will still have $|\inflow_{\flow}(y) - \inflow_{\flowEst}(y)| = 1 - 0 \ll \eps u(y)$, and the same will hold for the vertices after $y$. But as more and more paths are processed, $\inflow_{\flow}(y)$ will increase, so the algorithm must eventually send flow  on $\flowEst$ through $y$. The natural solution is to send $u(y) = n$ flow on one of the edges $(v_j,y)$ after $u(y)$ paths $P_i$ go through $y$, so that  $\inflow_{\flow}(y) = \inflow_{\flowEst}(y) = u(y)$. (Vertex $v_j$ will then have much more than $u(v_j) = 1$ flow leaving it, but this is allowed by Definition \ref{def:overview-pseudo-optimal}, which only constrains inflow.) The problem is that in a more general graph there is no way to tell which paths $\pi(s,t)$ go through $y$, since the algorithm avoids looking at the paths explicitly.

To resolve this issue, we introduce randomization. For every implicit flow path $\pi(s,t)$, $\flowEst$ always sends flow $u(x)$ into every vertex $x$ on $\pi(s,t)$ with capacity $u(x) = \lambda(\pi(s,t))$, but also with probability $1/2$ it sends $u(x)$ flow into every $x$ with $u(x) \leq 2\lambda(\pi(s,t))$, with probability $1/4$ it sends $u(x)$ flow into every $x$ with $u(x) \leq 4\lambda(\pi(s,t))$, and so on. (In reality, we use an exponential distribution rather than a geometric one, and we scale all flow down by $n^{o(1)}$ to ensure concentration bounds.) It is not hard to see that the \emph{expected} flow $\flowEst(v,x)$ into $x$ is $\lambda(\pi(s,t)) = \flow(v,x)$.

\paragraph{Changes to the MWU-framework.}
Our algorithm thus makes the following changes to the MWU-framework above. Each iteration (implicitly) computes a $(1-\eps)$-approximate shortest as before, but instead of sending flow on every edge, the algorithm first picks a parameter $\gamma$ from the exponential distribution, and then in $\flowEst$ it sends $u(y)$ flow through every edge $(x,y)$ for which $u(y) \leq \gamma \lambda(\pi(s,t))$. The algorithm uses weight function $\weightFunctionEst$, which follows the same multiplicative update procedure as before, except it depends on $\flowEst$ rather than $\flow$. (The shortest path $\pi(s,t)$ in each iteration is computed with respect to $\weightFunctionEst$.)

The main difficulty in the analysis is that even though $\flowEst$ tracks $\flow$ in expectation, $\flow$ actually depends on earlier random choices in $\flowEst$, because $\flowEst$ determines the vertex weights $\weightFunctionEst$, which in turn affect the next $(1-\eps)$-approximate path $\pi(s,t)$ used in $\flow$. We are able to use concentration bounds for martingales to show that $\flowEst \sim \flow$ with high probability. We are also able to show that even though the flow $\flow$ is no longer in perfect sync with the weight function $\weightFunctionEst$, the chosen paths $\pi(s,t)$ are still good enough, and the final flow $\flow$ is $(1-\eps)$-optimal, so $\flowEst$ is $(1-\eps)$-pseudo-optimal. Finally, as mentioned above, we show a black-box conversion from computing a $(1-\eps)$-pseudo-optimal flow to computing a regular $(1-\eps)$-flow. 

For our modified algorithm to run efficiently, we need to be able to return all edges $(x,y)$ on $\pi(s,t)$ for which $u(y) \leq \gamma \cdot \lambda(\pi(s,t))$, in time proportional to the number of such edges. We are able to extend our data structure from \Cref{part:dynamicShortestPaths} to answer such queries (see below); the MWU algorithm then uses this data structure as a black box.

\paragraph{$(1-\eps)$-Optimal Flow from $(1-\eps)$-Pseudo-Optimal Flow.} Re-inspecting \Cref{def:overview-pseudo-optimal}, we observe that for vertices where $\inflow_{\flow}(v) \sim u(v)$, the second property $|\inflow_{\flow}(v) - \inflow_{\flowEst}(v)| \leq \eps u(v)$ implies that we have a $(1+\eps)$-multiplicative approximation of the amount of in-flow for $v$. Unfortunately, the in-flow of $v$ might be significantly lower than $u(v)$. But if $\inflow_{\flowEst}(v) \ll u(v)$,
the same property implies that $\inflow_{\flow}(v) \ll u(v)$, so most of the capacity of $v$ is not required for producing a $(1-\eps)$-optimal flow. We therefore suggest a technique that we call capacity-fitting, where we repeatedly use our algorithm for pseudo-optimal flow to reduce the total vertex capacities by a factor of roughly $2$. We terminate with a pseudo-flow that has (loosely speaking) the following property: for each vertex $v$, either $\inflow_{\flowEst}(v) \sim u(v)$ or
the capacity of $v$ is negligible. Once this property is achieved, we can route the surplus flow in the pseudo-flow by scaling the graph appropriately and then computing a single instance of regular maximum flow (only edge capacities, no costs) using the algorithm of \cite{Sherman17}.

\paragraph{Comparison to Previous Work.}
There have been several recent papers that avoid updating every weight within the MWU framework by using a randomized threshold to maintain an estimator instead \cite{chekuri2018randomized, chekuri2020fast, chekuri2020fast2}. The main difference of our algorithm is that to overcome the flow-decomposition barrier, we need to maintain an estimator not just of the weights but of the solution (i.e. the flow) itself. This introduces several new challenges:  we need a modified analysis of the MWU framework that allows us to compare the estimated flow $\flowEst$ with the ideal flow $\flow$; our MWU algorithm only computes a pseudoflow $\flowEst$, which then needs to be converted into a real flow; and in order to update $\flowEst$ efficiently, we need to introduce the notion of threshold-subpath queries and show that our new decremental SSSP data structure can answer them efficiently. 

\section{Overview of Part \ref{part:augmented-queries}: Threshold-Subpath Queries}
In order to use it in the min-cost flow algorithm of Part \ref{part:minCostFlow}, we need our SSSP data structure to handle the following augmented path queries. 

\begin{definition}[Informal Version of Definition \ref{def:pathReportingSSSP}]
Consider a decremental weighted graph $G$ where each edge $(u,v)$ has a \emph{fixed} steadiness $\sigma(u,v) \in \{1,2,\ldots,\tau\}$, with $\tau = o(\log(n))$. Note that while weights in $G$ can increase over time, the $\sigma(u,v)$ never change. For any path $\pi$, let $\sigma_{\le j}(\pi) = \{(u,v) \in \pi \mid \sigma(u,v) \leq j\}$. We say that a decremental SSSP data structure can answer \emph{threshold-subpath queries} if the following holds:
\begin{itemize}
	\item At all times, every vertex $v$ corresponds to some $(1+\epsilon)$-approximate $s$-$v$ path $\pi(s,v)$; we say that the data structure \emph{implicitly} maintains $\pi(s,v)$.
	\item Given any query($v,j$), the data structure can return $\sigma_{\le j}(\pi(s,v))$ in time $|\sigma_{\le j}(\pi(s,v))|$; crucially, the path $\pi(s,v)$ must be the same regardless of which $j$ is queried. (Note that query($v,\tau$) corresponds to a standard path query.)
\end{itemize} 
\end{definition}

We briefly outline how threshold-subpath queries are used by our min-cost flow algorithm. Recall that in our modified framework, each iteration of MWU implicitly computes a $(1+\eps)$-approximate shortest path $\pi(s,t)$, but instead of modifying \emph{all} the edges on $\pi(s,t)$, it picks a random threshold $\gamma$ and only looks at edges $(x,v)$ on $\pi(s,t)$ for which $u(v) \leq \gamma \lambda(\pi(s,t))$. We thus want a data structure that returns all such low-capacity edges in time proportional to their number. This is exactly what a threshold-subpath query achieves. Here, $\pi(s,t)$ corresponds to the path implicitly maintained by the data structure. Every edge steadiness $\sigma(x,v)$ is a function of $u(v)$, and thus remains fixed throughout the MWU algorithm. Loosely speaking, for some $\eta = n^{o(1)}$, if $u(v) \in [1,\eta)$ then $\sigma(x,v) = 1$, if $u(v) \in [\eta,\eta^2)$ then $\sigma(x,v) = 2$, and so on. (The actual function is a bit more complicated and $\sigma(x,v)$ can also depend on the \emph{cost} of vertex $v$, not just the capacity.)  Since the buckets increase geometrically, the number of possible steadiness level $\tau$ will be small. Note that because each steadiness captures a range of capacities, when we use the data structure in our MWU algorithm, we only achieve the slightly weaker guarantee that we return edges $(x,v)$ on $\pi(s,t)$ for which $u(v) \lesssim \gamma \lambda(\pi(s,t))$; this weaker guarantee works essentially as well for our analysis.

We show in Part \ref{part:augmented-queries} that our SSSP data structure from Part \ref{part:dynamicShortestPaths} can be extended to handle threshold-subpath queries, while still having $\Ohat(m)$ total update time. We briefly outline our techniques below.

\paragraph{Techniques.} Threshold-subpath queries introduce several significant challenges. Recall that the algorithm iteratively computes emulators $G = H_0, H_1, \ldots H_q$, where each edge of $H_i$ corresponds to a short path in $H_{i-1}$, and the final emulator $H_q$ is guaranteed to have small hop distances. The algorithm can then estimate the $s$-$v$ distance by computing the shortest path in $H_q$. It is not too hard to ``unfold'' the path in $H_q$ into a path $\pi(s,v)$ in the graph $G$ by successively moving down the emulators. But to answer augmented path queries efficiently, we need to avoid unfolding emulator edges for which the corresponding path in $G$ does not contain any low-steadiness edges. We thus need a way of determining, for every emulator edge, the minimum steadiness in its unfolded path in $G$; we refer to this as the steadiness of the emulator edge.

The issue is that if each edge in $H_i$ corresponds to an \emph{arbitrary} $(1-\eps)$-approximate path in $H_{i-1}$, then the steadiness of emulator edges will be extremely unstable, and impossible to maintain efficiently. We overcome this problem by carefully defining, for each emulator edge in $H_i$, a specific \emph{critical} path in $H_{i-1}$ corresponding to $(x,y)$, which ensures that the steadiness of $(x,y)$ is robust, and allows us to maintain the entire hierarchy efficiently. 

A second challenge is that any edge $(u,v) \in E$ may participate in many emulator edges, with the result that when we unfold the emulator edges, the resulting path in $G$ might not be simple -- i.e. it might contain many copies of an edge $(u,v)$. Through a careful analysis of our emulator hierarchy, we are able to show that any path achieved via unfolding is close-to-simple, in that every $(u,v) \in E$ appears at most $n^{o(1)}$ times. We then show that MWU can be extended to handle such close-to-simple paths. See Part \ref{part:augmented-queries} for details.

%% file: prelim.tex
In this part, we give the proof for our main result: a deterministic decremental SSSP data structure in almost-linear time.

\mainSSSPResult*

\paragraph{Remark:}
In Part \ref{part:dynamicShortestPaths}, we focus exclusively on answering approximate \emph{distance} queries. Extending the data structure to return an approximate shortest path in time $|\pi(s,t)| n^{o(1)}$ is not too difficult but requires some additional work. We do not spell out the details because these path queries are a special case of the more powerful (and much more involved) augmented path queries detailed in Part \ref{part:augmented-queries}.

We start by providing the necessary preliminaries for the part and then provide a brief overview introducing the main components used in our proof and give a road map for the rest of the part.

\section{Preliminaries}
\label{sec:prelim}

\paragraph{Graphs.} We let a graph $H$ refer to a weighted, undirected graph with vertex set denoted by $V(H)$ of size $n_H$, edge set $E(H)$ of size $m_H$ and weight function $w_H : E(H) \rightarrow \mathbb{R}_{> 0}$. We define the aspect ratio $W$ of a graph to be the ratio of the largest to the smallest edge-weight in the graph.

We say that $H$ is a \emph{dynamic} graph if it is undergoing a sequence of edge deletions and insertions and edge weight changes (also referred to as updates), and refer to \emph{version} $t$ of $H$, or $H$ at \emph{stage} $t$ as the graph $H$ obtained after the first $t$ updates have been applied. We say that a dynamic graph $H$ is \emph{decremental} if the update sequence consists only edge deletions and edge weight increases. For a dynamic graph $H$, we let $m_H$ refer to the total number of edges in $H$ in all updates (we assume that the update sequence is finite). 

In this article, we denote the (decremental) input graph by $G=(V,E,w)$ with $n = |V|$ and $m = |E|$. In all subsequent definitions, we often use a subscript to indicate which graph we refer to, however, when we refer to $G$, we often omit the subscript.

\paragraph{Basic Graph Properties.} For any graph $H$, and any vertex $v \in V(H)$, we let $E(v)$ denote the set of edges incident to $v$. For any set $S \subseteq V(H)$, we let $E(S) = \bigcup_{v \in V} E(v)$. Finally, for any two disjoint sets $A,B$ we let $E(A,B)$ denote all edges with one endpoint in $A$, the other in $B$.

We let $\deg_H(v)$ denote the degree of $v$, i.e. the number of edges incident to $v$. If the graph is weighted, we let $\vol_H(v)$ denote the weighted degree or volume of vertex $v$, i.e. $\vol_H(v) \defeq \sum_{e \in E(v)} w_H(e)$. For $S \subseteq V(H)$, we also use $\deg_H(S)$ ($\vol_H(S)$) to denote the sum over the degrees (volume) of all vertices in $S$. 
If $H$ is dynamic, we define the \emph{all-time degree} of $v$ 
to be the total number of edges that are ever incident to $v$ over the entire update sequence of $H$. 
(An edge $(u,v)$ that is inserted, deleted and inserted again, contributes twice to the all-time degree of $v$).

\paragraph{Functions.} Say that we have a function $f: D \rightarrow \mathbb{R}$ for some domain $D$. Given any $S \subseteq D$ we often use the following short-hand: $f(S) \defeq \sum_{x \in S} f(x)$. For example, the definitions of $\vol_H(S)$ and $\deg_H(S)$ above follow this short-hand, and $w_H(E(A,B))$ denotes the sum of edge-weights in $E(A,B)$.

\paragraph{Expanders.} 
Let $H$ be a graph with positive real weights $w_H$. Let $0 < \phi < 1$ be the expansion parameter. We say that $H$ is a $\phi$-expander if for every $S \subset V(H)$ we have that $w_H(E(S,V(H)\setminus S)) \geq \phi \min \{ \vol_H(S) \ , \ \vol_H(V \setminus S) \}$.

\paragraph{Distances and Balls.} We let $\dist_H(u,v)$ denote the distance from vertex $u$ to vertex $v$ in a graph $H$ and denote by $\pi_{u,v, H}$ the corresponding shortest path (we assume uniqueness by implicitly referring to the lexicographically shortest path). We also define distances more generally for sets of vertices, where for any sets $X,Y \subseteq V(H)$, we denote by $\dist_H(X,Y) = \min_{u \in X, v \in Y} \dist_H(u,v)$ (whenever $X$ or $Y$ are singleton sets, we sometimes abuse notation and simply input the element of $X$ or $Y$ instead of using set notation).  

We define the \emph{ball} of radius $d$ around a vertex $v$ as $\ball_{H}(v,d)=\{w\mid\dist_{H}(v,w)\le d\}$
and the ball of radius $d$ around a set $X \subset V$ as $\ball_{H}(X,d)=\{w\mid\dist_{H}(X,w)\le d\}$. We say that a set $X$ w.r.t. a decremental graph $H$ is a \emph{decremental set} if at each stage of $H$, $X$ forms a subset of its previous versions. If $H$ is decremental, then for any $X \subseteq V$, we have that $\ball_{H}(X,d)$ is a decremental set, since distances can only increase over time in a decremental graph.

Finally, given any graph $H$ and a set $X \subset V(H)$, we define weak diameter $\diam_H(X) \defeq \max_{u,v \in X} \dist_H(u,v)$.

\paragraph{Hypergraphs.} In this part, we also use the generalization of graphs to \emph{hypergraphs} (but we will point out explicitly whenever we use a hypergraph). Let $H=(V,E)$ be a hypergraph, i.e. elements $e$ in $E$, called \emph{hyperedges}, are now sets of vertices, i.e. $e\subseteq V$ (possibly of size larger than two). We say that two vertices $u,v\in V$
are \emph{adjacent} if there is a hyperedge $e\in E$ containing both
$u$ and $v$. If $v\in e$, then $v$ is \emph{incident} to $e$.
For any vertex set $S\subseteq V$, the \emph{subhypergraph $H[S]$
induced by $S$} (or the \emph{restriction of $H$ to $S$}) is such
that $V(H[S])=S$ and $E(H[S])=\{e\cap S\mid e\in E\}$. That is,
each edge of $H[S]$ is an edge from $H$ restricted to $S$. The
\emph{total edge size} of $H$ is denoted by $|H|=\sum_{e\in E}|e|$.

Let $(L,S,R)$ be a partition of $V$ where $L,R\neq\emptyset$. We
say that $(L,S,R)$ if a \emph{vertex cut }of $H$ if, for every $u\in L$
and $v\in R$, $u$ and $v$ are not adjacent in $H$. Let $\kappa:V\rightarrow\mathbb{R}_{\ge0}$
be vertex capacities of vertices in $H$. The size of the cut $(L,S,R)$
is $\kappa(S)=\sum_{u\in S}\kappa(u)$.

The \emph{incidence graph} of $H$ denoted by $H_{bip}=(V\cup E,E_{bip})$
is a bipartite graph where $E_{bip}=\{(v,e)\in V\times E\mid v\in e\}$.
This bipartite view will be especially useful for implementing flow algorithms
on hypergraphs. Note that $|E_{bip}|=|H|$.

We say that a sequence of vertices ${v_{1},\ldots,v_{k}}$ form a
\emph{path} in $H$ if each pair of vertices $v_{i},v_{i+1}$ are
adjacent in $H$. We define the length of path ${v_{1},\ldots,v_{k}}$
to be $k-1$ and for any vertices $u,v$ in $H$ we define $\dist(u,v)$
to be the length of the shortest $u-v$ path in $H$, with $\dist(u,v)=\infty$
if there is no $u$-$v$ path in $H$. Given any vertex set $K\subseteq V(H)$,
we say that $\diam_{H}(K)\leq d$ if for every pair of vertices $u,v\in K$
we have that $\dist(u,v)\leq d$.

\paragraph{Dynamic Hypergraphs.} We subsequently deal with a dynamic hypergraph $H$. We model updates by edge deletions/insertions to the incidence graph $H_{bip}$. This corresponds to increasing/decreasing
the size of some hyperedge $e$ in $H$, or adding/removing a hyperedge
in $H$ entirely. One subtle detail that we use implicitly henceforth is that when we shrink or increase a hyperedge $e$ then this does not result in a new version $e$ but rather refers to the same edge at a different time step. This is important when we consider the all-time degree which is the total number of hyperedges that a vertex $v$ is ever contained in.

\paragraph{Embedding.} In this article, we view an embedding $\cP$ in a hypergraph $H$ as a collection of paths
in its corresponding bipartite graph representation $H_{bip}$. For any $v \in V$, we let $\cP_v$ be the set of paths in $\cP$ that contain the vertex $v$. With each path $P \in \cP$, we associate a value $\val(P) > 0$. We then say that the embedding $\cP$ has \emph{vertex congestion with respect to vertex capacities $\kappa$} at most $c$  if for every vertex $v \in H$, $\sum_{P \in \cP_v} \val(P) \leq c \cdot \kappa(v)$. We say that the embedding $\cP$ has length $\emph{len}$ if every path $P \in \cP$ consists of at most $\emph{len}$ edges. Further, we associate with each embedding $\cP$ into $H$, a weighted (multi-)graph $W$ taken over the same vertex set $V(H)$ and with an edge $(u,v)$ of weight $w(u,v) = \val(P)$ for each $u$-$v$ path $P$ in $\cP$. We say that $\cP$ \emph{embeds} $W$ into $H$ and say that $W$ is the \emph{embedded graph} or the \emph{witness} corresponding to $\cP$. %

\paragraph{Rounding Shorthand.} For any number $n$ and $k$, let $\roundup nk=\ceiling{n/k}\cdot k$
denote the integer obtained by rounding $n$ up to the nearest multiple of $k$. 

\paragraph{Parameters.} Throughout the part we refer to three global parameters: 
 $\phicmg=1/2^{\Theta(\log^{3/4} n)} = n^{o(1)}$, $\epswit = \phicmg / \log^2(n)$ and $\delta_{\mathrm{scatter}} = C \epswit$ for a large enough constant $C$. $\phicmg$ is first used in Theorem \ref{thm:CMG}, $\epswit$ in Lemma \ref{lem:embedwitness} and $\delta_{\mathrm{scatter}}$ in Definition \ref{def:Core}. 

\paragraph{A Formal Definition of a Decremental SSSP Data Structure.} In order to avoid restating the guarantees of a Decremental SSSP data structure throughout the part multiple times, we give the following formal definition.

\begin{defn}
	[$\SSSP$]\label{def:SSSP}A decremental $\SSSP$ data structure $\SSSP(G,s,\eps)$
	is given a decremental graph $G=(V,E)$, a fixed source vertex $s\in V$,
	and an accuracy parameter $\eps\ge0$. Then, it explicitly maintains
	distance estimates $\dtil(v)$ for all vertices $v\in V$ such that
	$\dist_{G}(s,v)\le\dtil(v)\le(1+\epsilon)\dist_{G}(s,v)$.
\end{defn}

\paragraph{Simplifying reduction.} 
We will use the following simplifying reduction which allows us to
assume that out input graph $G$ throughout this part has bounded degree
and satisfies other convenient properties. We give a proof of the proposition below in \Cref{subsec:appendixProofOfsimplify path}.
\begin{prop}
	\label{prop:simplify-1}
	Suppose that there is a data structure
$\SSSP(H,s,\eps)$ that only works if $H$ satisfies
the following properties: 
\begin{itemize}
\item $H$ always stays connected. 
\item Each update to $H$ is an edge deletion (not an increase in edge weight). 
\item $H$ has maximum degree $3$. 
\item $H$ has edge weights in $[1,n_{H}^{4}]$. 
\end{itemize}
Suppose $\SSSP(H,s,\epsilon)$ has $T_{\SSSP}(m_{H},n_{H},\epsilon)$
total update time where $m_{H}$ and $n_{H}$ are numbers of initial
edges and vertices of $H$. Then, we can implement $\SSSP(G,s,O(\eps))$
where $G$ is an arbitrary decremental graph with $m$ initial edges
that have weights in $[1,W]$
using total update time of $\tilde{O}\left(m/\epsilon^{2}+T_{\SSSP}(O(m \log W),2m,\epsilon)\right)\cdot\log(W).$
\end{prop}

%% file: components.tex
\section{Main Components}

In this section, we introduce the main components of our data structure. Although the part is self-contained, this section will be considerably more intuitive if the reader is familiar with the overview section \ref{sec:overviewPartII}. 

As pointed out in \Cref{subsec:existingTechniquesInOverview}, our data structure constructs a layering where each layer aims at compressing the graph $G$ further in order to compute the approximate SSSP distances up to a certain distance threshold. To make this notion of approximate SSSP up to a threshold precise, we introduce Approximate Balls in \Cref{subsec:approxBalls}. Next, we define the main building block: a Robust Core data structure that maintains a low-diameter vertex set with large approximation in \Cref{subsec:robustCore}.

With these two ingredients in place, we can introduce our most involved concept, a decremental graph Covering, formally in \Cref{subsec:coveringMainComps}. This concept forms the core of our data structure. 

\label{sec:components}

\begin{wrapfigure}{r}{0.5\textwidth}
  \begin{center}
    \includegraphics[width=0.45\textwidth]{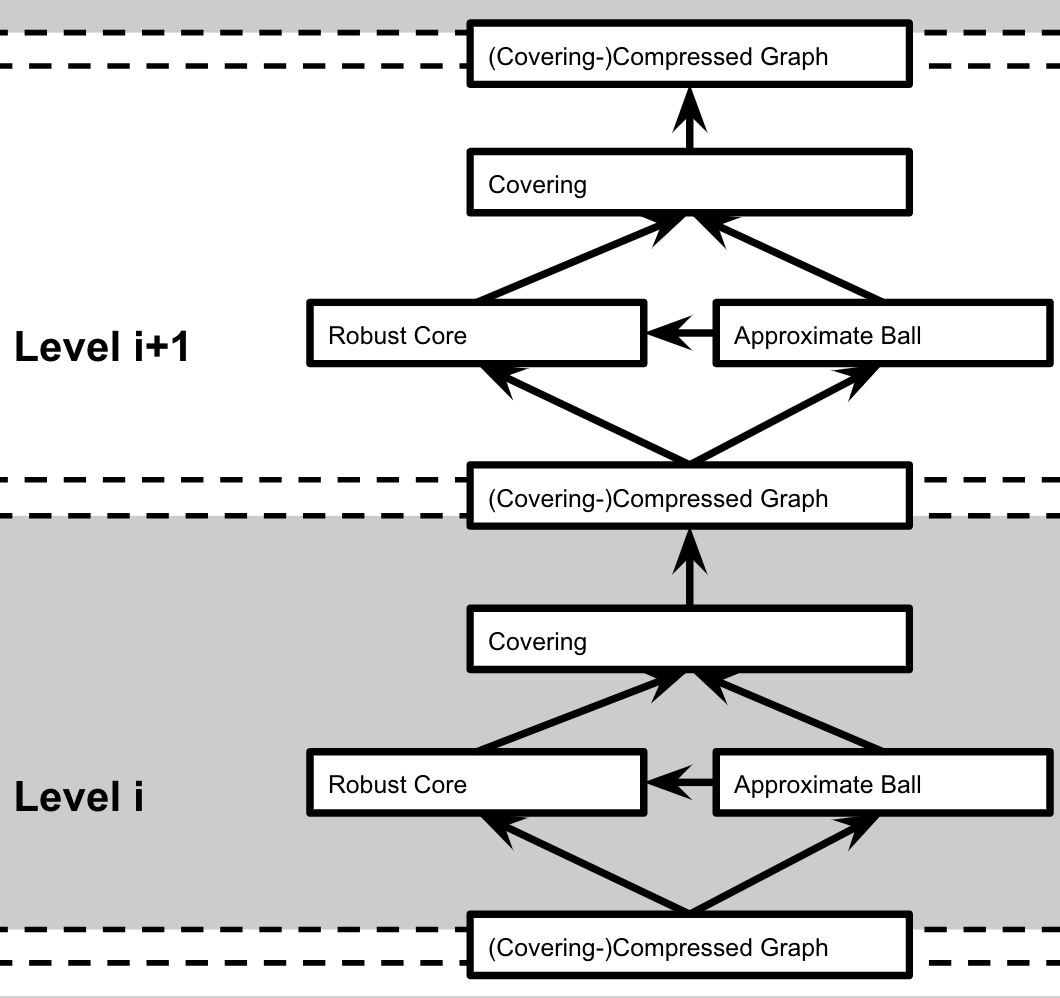}
  \end{center}
  \caption{Relations between main components.}
  \label{fig:mainCompsOrga}
\end{wrapfigure}

Finally, we show how to use the Covering as described in \Cref{subsec:existingTechniquesInOverview} to compress the graph $G$. We make the act of compression formal by introducing the concepts of Covering-Compressed Graphs and Compressed Graphs in \Cref{subsec:coreCompressedGraphMainComp}. The reason we require these notions is to give a formal interface for the next higher level where Robust Cores and Approximate Balls are maintained on the Covering-Compressed/ Compressed Graph to further compress $G$. Still, we associate each (Covering-)Compressed Graph with the level where its underlying Covering is maintained.

We summarize the relations between these components in \Cref{fig:mainCompsOrga}. 

\subsection{Approximate Ball}
\label{subsec:approxBalls}

\begin{defn}
\label{def:Apxball}An \emph{approximate ball} data structure $\Apxball(G,S,d,\eps)$
is given a decremental graph $G=(V,E)$, a decremental source set
$S\subseteq V$, a distance bound $d>0$, and an accuracy parameter
$\eps\ge0$. Then, it explicitly maintains distance estimates $\dtil(v)$
for all vertices $v\in V$ such that 
\begin{enumerate}
\item \label{enu:Apxball:overestimate}$\dtil(v)\ge\dist_{G}(S,v)$,
\item \label{enu:Apxball:approx}if $v\in\ball_{G}(S,d)$, then $\dtil(v)\le(1+\epsilon)\dist_{G}(S,v)$,
\item \label{enu:Apxball:monotone}Each $\dtil(v)$ may only increase through
time.
\end{enumerate}
\end{defn}

For convenience, we slightly abuse the notation and denote $\Apxball(G,S,d,\eps)=\{v\mid\dtil(v)\le(1+\eps)d\}$
as the set of all vertices $v$ whose distance estimate $\dtil(v)$
is at most $(1+\eps)d$. We think of this set as the set that the data
structure maintains. The next proposition relates the approximate ball to the exact ball.

\begin{prop}
\label{prop:ApxBall}We have $\ball_{G}(S,d)\subseteq\Apxball(G,S,d,\eps)\subseteq\ball_{G}(S,(1+\epsilon)d)$.
Moreover, $\Apxball(G,S,d,\eps)$ is a decremental set.
\end{prop}

\begin{proof}
If $v\in\ball_{G}(S,d)$, then by \Cref{enu:Apxball:approx} of
\Cref{def:Apxball} we have $\dtil(v)\le(1+\epsilon)\dist_{G}(S,v)\le(1+\eps)d$.
So $v\in\Apxball(G,S,d,\eps)$. For the other direction, if $v\in\Apxball(G,S,d,\eps)$,
we have $\dtil(v)\le(1+\eps)d$. Since $\dtil(v)\ge\dist_{G}(S,v)$
by \Cref{enu:Apxball:overestimate} of \Cref{def:Apxball}, we
have $\dist_{G}(S,v)\le(1+\eps)d$ and so $v\in\ball_{G}(S,(1+\epsilon)d)$.
Finally, $\Apxball(G,S,d,\eps)$ is a decremental set because $\dtil(v)$
may only increase through time.
\end{proof}
A classic ES-tree data structure \cite{EvenS} immediately gives a fast implementation for $\Apxball$ for the small distance regime.
\begin{prop}
[\cite{EvenS}]\label{prop:ES tree}We can implement $\Apxball(G,S,d,0)$
in $O(|\ball_{G}(S,d)|\cdot d)$ time.
\end{prop}

\begin{rem}
	\label{rem:apx-ball-monotonic}
Given any static input graph $G$, and static set $S$, we define $T_{\Apxball}(G,S,d,\eps)$ to refer to the worst-case total update time required by our data structure $\Apxball(G',S',d,\eps)$ for any decremental graph $G'$ initially equal to $G$, and  decremental set $S'$ initially equal to $S$. We also sometimes abuse notation and let $G,S$ be a decremental graph and set respectively, in which case we only refer to their initial versions in $T_{\Apxball}(G,S,d,\eps)$. 

Note that this definition of update time, allows us to immediately conclude that for any graphs $G$ and $G'$, and sets $S$ and $S'$ where $G \subseteq G'$ and $S \subseteq S'$, we have 
\[
T_{\Apxball}(G,S,d,\eps) \leq T_{\Apxball}(G',S',d,\eps)
\]
since any worst-case instance incurring  $T_{\Apxball}(G,S,d,\eps)$ can be emulated by deleting $G' \setminus G$ and $S' \setminus S$ from $G'$ and $S'$ in the first stage respectively. This allows us to state times more compactly and combine bounds. Note that the above in fact also implies $T_{\Apxball}(G,S,d,\eps) \leq T_{\Apxball}(G,S,d',\eps)$ for any $d \leq d'$.

We also assume that $T_{\Apxball}(G, S, d, \epsilon) = \Omega(|\ball_G(S, d)|)$ which is true throughout the part. 
\end{rem}

\subsection{Robust Core}
\label{subsec:robustCore}

Given a set $K\subseteq V$ of vertices of a graph $G=(V,E)$, we
informally call $K$ a \emph{core set} if its weak diameter $\diam_{G}(K)$ is small.
That is, every pair of vertices in $K$ are close to each other. In the definition below, recall that $\scatter = n^{o(1)}$ is a global variable set in Section \ref{sec:prelim}. For intuition, think of $\stretch$ also as $n^o(1)$.
\begin{defn}
\label{def:Core}A \emph{robust core }data structure $\Core(G,\Kinit,D)$
with a scattering parameter $\scatter\in(0,1)$ and a stretch
$\stretch\ge1$ is given
\begin{itemize}
\item a decremental graph $G=(V,E)$, and
\item an initial core set $\Kinit\subset V(G)$ where $\diam_{G}(\Kinit)\le D$
when the data structure is called initially%
\end{itemize}
and maintains a decremental set $K\subseteq\Kinit$ called \emph{core
set} until $K=\emptyset$ such that 
\begin{enumerate}
\item \label{cond:core:scatter}\textbf{(Scattered):} for each vertex $v\in\Kinit\setminus K$,
we have $|\ball_{G}(v,2D)\cap\Kinit|\le(1-\scatter)\cdot|\Kinit|$,
and 
\item \label{cond:core:stretch}\textbf{(Low stretch):} $\diam_{G}(K)\le\stretch\cdot D$.
\end{enumerate}
\end{defn}

For convenience, we sometimes slightly abuse the notation and denote
the maintained core set $K=\Core(G,\Kinit,D)$. Also, we introduce $T_{\Core}(G,\Kinit,D)$ to refer to the total update time required by our data structure implementing $\Core(G,\Kinit,D)$.

\subsection{Covering}
\label{subsec:coveringMainComps}

As mentioned before, the key ingredients of Approximate Ball and Robust Core can now be used to define a Covering that we can implement efficiently. This is key building block of our interface.
\begin{defn}
	\label{def:Covering}Let $G=(V,E)$ be a decremental graph and $\eps \le 1/3$. A $(d,k,\eps,\stretch,\Delta)$-covering
	$\cC$ of $G$ is a collection of vertex sets called \emph{cores} where
	each core $C\in\cC$ is associated with other sets called the \emph{cover,}
	\emph{shell, }and \emph{outer-shell} of $C$ denoted by $\cover(C)$,
	$\shell(C)$, $\oshell(C)$, respectively. We have the following 
	\begin{enumerate}
		\item Each core $C$ is assigned a level $\Clevel(C)=\ell\in[0,k-1]$. All
		cores $C$ from the same level are vertex disjoint.
		\item For each level $\ell$ we define $d_{\ell}\defeq d\cdot(\frac{\stretch}{\eps})^{\ell}$ and have
		\begin{enumerate}
			\item $C=\Core(G,\Cinit,d_{\ell})$ with stretch at most $\stretch$,
			and $\Cinit$ denotes $C$ when initialized in $\cC$.
			\item $\textsc{cover}(C)=\Apxball(G,C,4d_{\ell},\eps)$ and $\shell(C)=\Apxball(G,C,\frac{\stretch}{4\eps}d_{\ell},\eps)$.
			\item $\oshell(C)=\ball_{G}(C,\frac{\stretch}{3\eps}d_{\ell})$.
		\end{enumerate}
		\item For every vertex $v\in V$, at all times there is a core $C$ where $v\in\cover(C)$.
		We say $v$ is \emph{covered} by $C$.
		\item At all times, each vertex $v\in V$ can ever be in at most
		$\Delta$ many outer-shells. That is, the total number of cores $C$	that $v\in\oshell(C)$ over the whole update sequence is at most $\Delta$.
	\end{enumerate}
	We call $d$ the distance scale, $\stretch$ the stretch parameter,
	$k$ the level parameter, $\eps$ the accuracy parameter, and $\Delta$
	the outer-shell participation bound.
\end{defn}

We note that the notion of outer-shells will be important later for path-reporting data structures, more specifically, in \Cref{lem:path simple} and \Cref{lem:core path simple}. The following observation reveals basic structures of cores in the covering.
\begin{prop}
	For each core $C\in\cC$, the sets $C$, $\cover(C)$, $\shell(C)$
	and $\oshell(C)$ are decremental. Moreover, $C\subseteq\cover(C)\subseteq\shell(C)\subseteq\oshell(C)$. 
\end{prop}

\begin{proof}
	$C$ is decremental by the guarantee of $\Core$. As $C$ is decremental,
	$\cover(C)$ and $\shell(C)$ are decremental by \Cref{prop:ApxBall}.
	Also, $\oshell(C)$ must be decremental as both $G$ and $C$ are
	decremental. The moreover part follows by \Cref{prop:ApxBall} and
	the fact that $(1+\eps)\frac{\stretch}{4\eps}\le\frac{\stretch}{3\eps}$
	as $\eps\le1/3$. 
\end{proof}

\subsection{(Covering-)Compressed Graphs}
\label{subsec:coreCompressedGraphMainComp}

Given a covering $\cC$ of $G$, we can define a natural bipartite
graph $H_{\cC}$ associated with the covering $\cC$. We call this graph a Covering-Compressed Graph.

\begin{defn}
[Covering-Compressed Graph]\label{def:core compress}Let $\cC$ be a
$(d,k,\eps,\stretch,\Delta)$-covering of a graph $G=(V,E)$ at any
point of time. A \emph{weighted covering-compressed} graph of $\cC$ denoted
by $H_{\cC}=(V\cup\cC,E')$ is a bipartite graph where $E'=\{(v,C)\in V\times\cC\mid v\in\shell(C)\}$.
For each edge $e'=(v,C)\in E'$, the weight is $w_{\cC}(e')=\roundup{\stretch\cdot d_{\Clevel(C)}+\dtil^{C}(v)}{\eps d}$
where $\dtil^{C}(v)$ is the distance estimate of $\dist_{G}(C,v)$
from the instance of $\Apxball$ that maintains $\shell(C)$. An\emph{
(unweighted) covering-compressed} graph $H_{\cC}$ of $\cC$ is defined
exactly the same but each edge in $H_{\cC}$ is unweighted.
\end{defn}

In other words, the unweighted core compressed graph $H_{\cC}$ is
an incidence graph of the hypergraph on vertex set $V$ where, for
each core $C$, there is a hyperedge $e$ containing all vertices
in $\shell(C)$. Intuitively, if $v \in \shell(C)$, then $w_{\cC}(e')$ corresponds to the distances from $e'$ to a vertex inside $C$: $\dtil^{C}(v)$ corresponds to the distance from $v$ to the core $C$, while by the guarantees of $\Core$ (Definition \ref{def:Core}),  $\stretch\cdot d_{\Clevel(C)}$ is an upper bound on the diameter of $C$.
\begin{rem}
\label{rem:covering-compressed from covering}The correspondence between
the covering $\cC$ and the (weighted and unweighted) covering-compressed
graph $H_{\cC}$ of $\cC$ is straightforward. Given an algorithm
that maintains $\cC$, we can assume that it also maintains $H_{\cC}$
for us as well.
\end{rem}

When we implement $\Core$ data structure, we will exploit the covering-compressed
graph via a simple combinatorial property. Hence, we abstract this
property out via a concept called a\emph{ compressed graph}.

\begin{defn}
[Compressed Graph]\label{def:compressed graph}Let $G=(V,E)$ be
a decremental graph. We say that an unweighted hypergraph $H$ is
a \emph{$(d,\gamma,\Delta)$-compressed graph} of $G$ with distance
scale $d$, gap parameter $\gamma$, and maximum all-time degree $\Delta$
if the following hold: 
\begin{itemize}
\item if $\dist_{G}(u,v)\le d$, then $u$ and $v$ are adjacent in $H$.
\item if $\dist_{G}(u,v)>d\cdot\gamma$, then $u$ and $v$ are not adjacent
in $H$.
\item Throughout the update sequence on $G$, for each $v\in V$, the total
number of edges in $H$ ever incident to $v$ is at most $\Delta$.
\end{itemize}
\end{defn}

Recall that every unweighted bipartite graph represents some unweighted
hypergraph. The following shows that the hypergraph view of any covering-compressed
graph is indeed a compressed graph.
\begin{prop}
[A Covering-Compressed Graph is a Compressed Graph]\label{prop:covering-compressed is compressed}Let $\cC$ be an $(d,k,\eps,\stretch,\Delta)$-covering of a graph $G$ where $6\le\frac{\stretch}{4\eps}$.
Let $H_{\cC}$ be a covering-compressed graph of $\cC$. Then, the hypergraph
view of $H_{\cC}$ is a $(d,\gamma,\Delta)$-compressed graph of $G$
where $\gamma=(\stretch/\eps)^{k}$.
\end{prop}

\begin{proof}
Consider any $u,v \in V(G)$ with $\dist_{G}(u,v)\le d$.
Let $C\in\cC$ be a core that covers $u$, i.e., $u\in\cover(C)=\Apxball(G,C,4d_{\ell},0.1)$.
We claim that $v\in\shell(C)$. Let $\ell=\Clevel(C)$. We have $\dist_{G}(C,v)\le\dist_{G}(C,u)+\dist_{G}(u,v)\le4d_{\ell}\cdot(1.1)+d_\ell \le6d_{\ell}\le d_{\ell}(\frac{\stretch}{4\eps})$.
As $\cover(C)\subseteq\shell(C)$,
both $u,v\in\shell(C)$ and thus $u$ and $v$ are adjacent in $H$.
Next, suppose that $u$ and $v$ are adjacent in $H$. Then, for some
$\ell$, there is a level-$\ell$ core $C$ where $u,v\in\shell(C)=\Apxball(G,C,\frac{\stretch}{4\eps}d_{\ell},0.1)$.
So $\dist_{G}(u,v)\le2\cdot\frac{\stretch}{4\eps}d_{\ell}\cdot1.1\le\frac{\stretch}{\eps}d_{k-1}=d_{k}=d\gamma$.
Lastly, as every vertex $v\in V$ can ever be in at most $\Delta$
shells, the maximum all-time degree of $v\in V$ is at most $\Delta$.
\end{proof}
There is a trivial way to construct a $(1,1,O(1))$-compressed graph
of a bounded-degree graph with integer edge weights (recall that $G$ is such a graph by the simplifying assumption in \Cref{prop:simplify-1}):
\begin{prop}
[A Trivial Compressed Graph]\label{prop:compressed base}Let $G$
be a bounded-degree graph $G$ with integer edge weights. Let $G_{unit}$
be obtained from $G$ by removing all edges with weight greater than one.
Then, $G_{unit}$ is a $(1,1,O(1))$-compressed graph of $G$.
\end{prop}

We will use the above trivial compressed graph in the base case of
our data structure for very small distance scale.

\subsection{Organization of the Part}

In the remaining sections, we first present in \Cref{sec:Core} an algorithm to maintain a Robust Core since it is conceptually the most interesting component. We then show how to implement the Covering in \Cref{sec:part2Covering}, which is the key building block of our interface and also requires several new ideas. In \Cref{sec:ball}, we show how to implement Approximate Balls. This section is rather technical and follows well-known techniques.

Finally, we combine the components and set up the layering of our data structure in \Cref{sec:part2PuttingItTogether}.

%% file: robust_core.tex
\section{Implementing Robust Cores}
\label{sec:Core}

In this section, we show how to implement a robust core data structure $\Core$ for distance scale $D$, given a compressed graph for distance scale $d \ll D$. We introduced Robust Cores already in \Cref{def:overview-core} in the overview for the special case of the theorem below when $D = n^{o(1)}$.

\begin{theorem}
[Robust Core]\label{thm:Core}Let $G$ be an $n$-vertex bounded-degree
decremental graph. Suppose that a $(d,\gamma,\Delta)$-compressed
graph $H$ of $G$ is explicitly maintained for us. We can implement
a robust core data structure $\Core(G,\Kinit,D)$ with scattering
parameter $\scatter=\Omegatil(\phicmg)$ and stretch $\stretch_{\core}=\Otil(\gamma/\phicmg^{3})$
and total update time of \[\Otil\left(T_{\Apxball}(G,\Kinit,32D\log n,0.1)\Delta^{2}(D/d)^{3}/\phicmg^{2}\right).
\]
\end{theorem}
\begin{rem}
We assume here that only edge deletions
incident to $\ball_{G}(\Kinit,32D\log n)$ in the initial graph are forwarded to the Robust Core data structure. When we use multiple Robust Core data structures later on the same graph $G$, we assume that updates are scheduled effectively to the relevant Robust Core data structure. We point out that such scheduling is extremely straight-forward to implement and therefore henceforth implicitly assumed.
\end{rem}

\subsection{Algorithm}

For this section, we remind the reader of the intuition provided for Robust Core provided in the overview \Cref{subsec:robustCore} which provided the simplified Pseudo-Code \ref{alg:overview-algorithmRobustCore}. We present the full Pseudo-Code for Robust Core in \Cref{alg:Core}. We now discuss the algorithm in detail and state the formal guarantees that the various subprocedures achieve.

\begin{algorithm}[!ht]
\KwIn{A $(d,\gamma,\Delta)$-compressed
graph $H$ of $G$ that is explicitly maintained for us, a set $\Kinit \subseteq V(H)$, and a parameter $D \geq d$.}
    Construct $(H,G, \Kinit, d, D, \gamma, \Delta)$-heavy-path-augmented graph $\Hhat$.\label{lne:heavyPathGraph} \tcp*{see Def \ref{def:Hhat}}
    $\Vhat \gets V(\Hhat)$; $K \gets \Kinit$; $\gamma_{\size} \gets \frac{1}{4} \cdot |\Vhat|/|\Kinit|$.\; 
    $\forall v \in \Kinit, \kappa(v) \gets 2; \forall w \in \Vhat \setminus \Kinit, \kappa(w) \gets 1/\gamma_{\size}$. \label{enu:init cap}  \tcp*{$\sum_{v \in \Vhat} \kappa(v) = O(|\Kinit|)$.} 
    \tcp{As long as there exists a large core in $\Kinit$.}
	\While(\label{lne:whileKisbigCOre}){ $\CertifyCore(G,\Kinit,2D,\epswit/2)$ returns Core $K'$}{
	    \tcp{While low-diameter graph has some sparse cut, double the cut weight.}
    	\While(\label{lne:sparseCutWhileLoop}){$\EmbedCore(\Hhat,K',\kappa)$ returns a vertex cut of $(L,S,R)$ in $\Hhat$}{
           \lForEach(\label{enu:doubling step}){cut-vertex $v \in S$}{$\kappa(v)\gets2\kappa(v)$.}
           \tcc{To ensure the technical side condition in Claim \ref{clm:sideConditions}.}
            Let $w$ be an arbitrary vertex from $S$ maximizing $\kappa(w)$; pick an arbitrary $w' \neq w$ from $\Vhat$ and do $\kappa(w') \gets \max\{\kappa(w'), \kappa(w)\}$. \label{lne:technicalSideCondition}
        }
        \tcp{Let $\cP$ be the embedding that $\EmbedCore(\cdot)$ returned; and $W$ the corresponding witness. Let $W_{multi}$ be the multi-graph version of $W$. }
        $(\cP, W) \gets \EmbedCore(\Hhat,K',\kappa)$. \label{enu:embedwitness}\;
        Let (unweighted) multi-graph $W_{multi}$ be derived from $W$ by adding $w(e)/\gamma_{\size}$ copies for each $e \in E(W)$ (Recall $w(e)/\gamma_{\size}$ is an integer by guarantees of \Cref{lem:embedwitness}) .\label{line:unweighted witness}\;
        
        \tcp{Maintain the witness $W_{multi}$ of $\Hhat$ until a lot of capacity is deleted.}
		\While(\label{lne:whileExpanderWitnessIsLarge}){$X \subseteq K'$ from running $\Prune(W_{multi},\phicmg)$ has size $\geq |\Kinit|/2$}{
			Maintain $\Apxball(G,X, 4D,0.1)$ and remove every leaving vertex from $K$. \label{lne:maintainXBall}
		}
	}
	$K \gets \emptyset$; \Return
	\caption{$\Core(G,\Kinit,D)$}\label{alg:oalgorithmRobustCore}\label{alg:Core}
\end{algorithm}

\paragraph{Constructing $\Hhat$ (\Cref{lne:heavyPathGraph}).} 

The algorithm starts by constructing a special graph $\Hhat$ that can be thought of as being the $(d,\gamma,\Delta)$-compressed
graph $H$ that is maintained for us, restricted to the set $B^{init}$ with the addition of some missing edges from $G$, where $B^{init} = \ball_{G}(\Kinit,32D\log n)$ is the \emph{static} set of vertices that are in the ball around $\Kinit$ in the initial graph $G$. We define $\Hhat$ formally below.

\begin{defn}
[Heavy-Path Augmented Hypergraph]\label{def:Hhat}
Given a $(d,\gamma,\Delta)$-compressed
graph $H$ of $G$ that is explicitly maintained for us, a set $\Kinit \subseteq V(H)$, and a parameter $D \geq d$. 

Then, let $\hat{E} \gets \{e \in E(G[B^{init}]) \;|\; d < w(e) \le 32 D \cdot \log n \}$. Let $\Phat$ be a collection of \emph{heavy} paths where each edge $e=(u,v) \in \hat{E}$ corresponds to a $u$-$v$ path $\Phat_{e} \in \Phat$ consisting of $\ceiling{w(e)/d}$ edges. Define $\Hhat$ be the union of $H[B^{init}]$ and all heavy paths $\Phat$ (where internal vertices to each $\Phat_{e}$ are added as new vertices). We then say that a graph $\Hhat$ is the $(H,G, \Kinit, d, D, \gamma, \Delta)$-heavy-path-augmented graph. Note that $\Hhat$ is an unweighted graph.
\label{enu:define Hhat general}
\end{defn}

The intuition for the heavy-path-augmented graph is quite simple. We would like to ensure that for any edge $(u,v) \in G$ with $w(u,v) \leq 32D\log(n)$, $u$ and $v$ are also nearby in $\Hhat$. If $w(u,v) \leq d$ then $u$ and $v$ are adjacent in $H \subseteq \Hhat$ by definition of $H$ being a $(d,\gamma,\Delta)$-compressed graph. If $d < w(u,v) \leq 32D\log(n)$ then there exists a heavy path $\hat{P}$ from $u$ to $v$ with at most $O(D/d)$ edges.

Since we only deal with a single $(H,G, \Kinit, d, D, \gamma, \Delta)$-heavy-path-augmented graph in the rest of this part, we use $\Hhat$ to refer to this instance throughout. (We note that we assume throughout that $\Vhat = V(\Hhat)$ is of size at least $2$ since otherwise Robust Core is trivially implemented).

\paragraph{Parameters:} In the description below, recall the three global variables we set in Section \ref{sec:prelim}. 

\paragraph{Certifying a Large Core (\Cref{lne:whileKisbigCOre}).} After some further initialization takes place where in particular we set $K$ to be equal to $\Kinit$, the main while-loop starting in \Cref{lne:whileKisbigCOre} starts by checking its condition. This task is delegated to a procedure $\CertifyCore(\cdot)$ which either returns a large set $K' \subseteq \Kinit$ of small diameter (in $G$) which is called the core $K'$, or announces that all vertices satisfy the scattered property which allows us to set $K$ to be the empty set and terminate. The proof is deferred to \Cref{sec:certify_core}.

\begin{restatable}{lem}{certifyCore}
\label{lem:certifycore}There is an algorithm $\CertifyCore(G,K,d,\epsilon)$
with the following input: an $n$-vertex graph $G=(V,E,w)$, a set
$K\subseteq V$, an integer $d>0$, and a parameter $\epsilon>0$.
In time $\ensuremath{O(\deg_{G}(\ball_{G}(K,16d\lg n))\log n)}$,
the algorithm either 
\begin{itemize}
\item \label{prop:CertifyCore - noCore}\textbf{(Scattered):} certifies
that for each $v\in K$, we have $|\ball_{G}(v,d)\cap K|\leq(1-\epsilon/2)|K|$,
or
\item \label{prop:CertifyCore - core}\textbf{(Core):} returns a subset
$K'\subseteq K$, with $|K'|\geq(1-\epsilon)|K|$ and $\diam_{G}(K')\leq16d\lg n$. 
\end{itemize}
\end{restatable}

\paragraph{Embedding the Low-Diameter Graph (\Cref{lne:sparseCutWhileLoop}-\ref{enu:embedwitness}).} If a core $K'$ is returned by $\CertifyCore(\cdot)$, then we use the procedure $\EmbedCore(\cdot)$ which either returns a large sparse vertex cut $(L,S,R)$ (with respect to $\kappa$ and $K'$) or an embedding $\cP$ that embeds a witness graph $W$ in $\Hhat$. Note that the entire reason of having the capacity function $\kappa$ in the algorithm is to repeatedly find an embedding according to $\kappa$ and to then argue about progress between two such embedding steps.

\begin{restatable}{lem}{embedwitness}
\label{lem:embedwitness}
There is an algorithm $\EmbedCore(H,K,\kappa)$ that
is given a hypergraph graph $H=(V,E)$, a terminal set $K\subseteq V$,
and $1/z$-integral vertex capacities $\kappa:V\rightarrow\frac{1}{z}\mathbb{Z}_{\ge0}$
such that $\kappa(v)\ge2$ for all terminals $v\in K$ and $\kappa(v)\le\kappa(V)/2$
for all vertices $v\in V$. (The integrality parameter $z$ will appear in the guarantees of the algorithm.) The algorithm returns either
\begin{itemize}
\item \textbf{(Cut):} a vertex cut $(L,S,R)$ in $H$ such that $\epswit|K|\le|L\cap K|\le|R\cap K|$
and $\kappa(S)\le 2|L\cap K|$, where $\epswit=\phicmg/\log^2(n)$ is
a parameter we will refer to in other parts of the paper; OR
\item \textbf{(Witness):} an embedding $\cP$ that embeds a weighted multi-graph $W$ into $H$
with the following guarantees: 
\begin{itemize}
\item $W$ is a weighted $\Omega(\phicmg)$-expander. The vertex set $V(W)$ is
such that $V(W)\subseteq K$ and $|V(W)|\ge |K|-o(|K|)$. Each edge
weight is a multiple of $1/z$, where recall that $z$ is the smallest positive integer such that $\kappa:V\rightarrow\frac{1}{z}\mathbb{Z}_{\ge0}$. The total edge
weight in $W$ is $O(|K|\log|K|)$. Also, there are only $o(|K|)$
vertices in $W$ with weighted degree $\le9/10$.
\item The length of $\cP$ and vertex congestion of $\cP$ w.r.t. $\kappa$
are at most $O(\kappa(V)\log(\kappa(V))/(|K|\epswit^{2}))$ and $O(\log|K|)$, respectively.
More precisely, each path in $\cP$ has length at most \\ $O(\kappa(V)\log(\kappa(V))/(|K|\epswit^{2}))$.
For each vertex $v\in V$, $\sum_{P\in\cP_{v}}\val(P)=O(\kappa(v)\log|K|)$
where $\cP_{v}$ is the set of paths in $\cP$ containing $v$. Moreover,
each path in $\cP$ is a simple path. 
\end{itemize}
\end{itemize}
The running time of the algorithm is $\Otil(|H|\frac{\kappa(V)}{|K|\phicmg}+z\kappa(V)/\phicmg)$,
where $|H|=\sum_{e\in E}|e|$ and $z$ is the smallest positive integer such that  $\kappa:V\rightarrow\frac{1}{z}\mathbb{Z}_{\ge0}$.
\end{restatable}

Recall here that there is an edge $(u,v)$ of weight $w(u,v)=\val(P)$ in $W$ for every $u$-$v$ path $P$ in $\cP$. Intuitively, the lemma above guarantees that $\diam_H(V(W))$ is small because the length of every path in the embedding is small, and $\diam(W)$ is small because $W$ is an expander.

In the algorithm, we invoke $\EmbedCore(\cdot)$ and if it returns a vertex cut, we double the capacity function $\kappa(v)$ for all vertices $v$ in the cut set $S$. We also update some additional vertex in \Cref{lne:technicalSideCondition}: this is just a blunt and simple way to enforce the technical side conditions of Claim \ref{clm:sideConditions}. Eventually, the doubling steps increase the potential enough to ensure that the witness graph $W$ can be embedded into $H$.

\paragraph{Maintaining the Witness and its Approximate Ball (\Cref{line:unweighted witness}-\ref{lne:maintainXBall}).} 
We start in \Cref{line:unweighted witness} by obtaining an unweighted version of $W$ which we call $W_{multi}$. This version is derived by scaling up edge weights in $W$ so that each weight becomes an integer. Then, we replace edges with weights by multi-edges each of unit weight.

The above transformation from $W$ to $W_{multi}$ is simply so that we can run the pruning subroutine below, which is restricted to unweighted graphs. Pruning allows us to maintain a large set $X$ such that $W_{multi}[X]$ (and therefore also $W[X]$) remains an expander. 

\begin{restatable}[\cite{SaranurakW19}]{lem}{PruningLemma} \label{lem:prune}There is an algorithm $\Prune(W,\phi)$
	that, given an unweighted decremental multi-graph $W=(V,E)$ that is initially a $\phi$-expander with $m$ edges, maintains a decremental
	set $X\subseteq V$ using $\Otil(m/\phi)$ total update time such
	that $W[X]$ is a $\phi/6$-expander at any point of time, and $\vol_{W}(V\setminus X)\le8i/\phi$
	after $i$ updates. 
\end{restatable}

As mentioned, we denote the maintained set after removing the pruned part by\\ $X=\Prune(W_{multi},\phi)$. Since $W_{multi}$ is only used to turn $W$ into an unweighted graph while preserving all its properties (except number of edges), we refer in all proofs straight-forwardly to $W$ and say that $W$ is pruned, even when we really mean that $W_{multi}$ is pruned.

Now, as long as a large set $X$ exists, even as $\Hhat$ and therefore $W$ undergoes edge updates, we root an approximate ball $\Apxball(G,X, 4D,0.1)$ at the decremental set $X$. For every vertex that leaves this approximate ball, we check whether it is in $K$ still, and if so we remove it from $K$. 

\subsection{Analysis}
\label{subsec:analysisOfRobustCorePart2}

Throughout the analysis section, we let $\kappafinal$ denote the vertex capacity function $\kappa$ taken when the algorithm  terminates. The following is the key
lemma in our analysis.

\begin{lem}
\label{lem:total cap}
At any point of time, the total vertex capacity
in $\Hhat$ is 
\[
\kappa(\Vhat)\le\kappafinal(\Vhat)\le O\left(|\Kinit|\frac{D}{d}\log^{2}(n)\right).
\]
\end{lem}

The first inequality holds because $\kappa(\Vhat)$ can only increase
through time by \Cref{enu:doubling step} of \Cref{alg:Core}. We defer
the proof of the second inequality to the end of this section. However, we use this lemma before to establish correctness and update time. We also use throughout that $\kappa$ is a monotonically increasing function over time, which can be seen easily from the algorithm.

\paragraph{Correctness.} We now establish the correctness, i.e. that $K$ indeed forms a Robust Core as defined in \Cref{def:Core} and parameterized in \Cref{thm:Core}.

\begin{lemma}[Correctness] \label{cor:core scatter}
At any stage of \Cref{alg:Core}, the set $\Kinit$ and $K$ satisfy
\begin{enumerate}
\item \textbf{(Scattered):} for each vertex $v\in\Kinit\setminus K$,
we have $|\ball_{G}(v,2D)\cap\Kinit|\le(1-\scatter)\cdot|\Kinit|$ where recall that $\scatter=\Theta(\epswit) = \Omegatil(\phicmg)$.

\item \textbf{(Low stretch):} $\diam_{G}(K)\le\stretch_{core}\cdot D$ where $\stretch_{\core}=\Otil(\gamma/\phicmg^{3})$.
\end{enumerate}
\end{lemma}
\textbf{(Scattered):} Observe that every vertex $v$ in $\Kinit$ is originally in $K$. Further, a vertex $v$ can only be removed from $K$ in \Cref{lne:maintainXBall}. But this in turn only occurs if $v$ has its distance estimate from $X$ larger than $4D$. Thus, $\dist_G(v, X) > 2D$ (by the approximation guarantee of \Cref{def:Apxball}). It remains to observe that by Line \ref{lne:whileExpanderWitnessIsLarge} $X \subseteq K' \subseteq \Kinit$ contains at least half the vertices in $\Kinit$. This implies $|\ball_{G}(v,2D)\cap\Kinit|\le|\Kinit|/2<(1-\scatter)\cdot|\Kinit|$. Finally, observe that prior to termination of the algorithm, we have that the while-condition in \Cref{lne:whileKisbigCOre} was false, and therefore $\CertifyCore(G, \Kinit, 2D, \epswit/2)$ announced that the entire set $\Kinit$ is scattered (see \Cref{lem:certifycore}) by the choice of $\scatter = \Omega(\epswit) = \Omegatil(\phicmg)$. This allows us to subsequently set $K = \emptyset$ and return.

\textbf{(Low stretch):} We bound the diameter of $K$ in two steps: first we bound $\diam_{\Hhat}(X) = \Otil(\frac{D}{d}/\phicmg^{3})$, then we show that $\diam_G(X) = O(\gamma d \cdot \diam_{\Hhat}(X))$. Combined, this establishes the Low Stretch Property since we enforce that vertices that leave $\Apxball(G,X,4D,0.1)$ are removed from $K$, so $\diam_G(K) = O(\diam_G(X))$.

\underline{$\diam_{\Hhat}(X) = \Otil(\frac{D}{d}/\phicmg^{3})$:} We have that by \Cref{lem:embedwitness} for $\EmbedCore(\cdot)$ that the length $\len(\cP_W)$ of the embedding $\cP_W$ of $W$ is at most $O(\kappa(\Vhat)\log(\kappa(\Vhat))/(|K|\epswit^2))$. It is not hard to check that $\log(\kappa(\Vhat)) = O(\log(n))$ because we know by \Cref{lem:total cap} that $\kappa(\Vhat) \leq \kappafinal(\Vhat)$, and it is easy to see that $\kappafinal(\Vhat)$ is polynomial in $n$ because both $|\Kinit|$ and $D$ are polynomial in $n$. We have that 
$$\len(\cP_W) = \Otil(\kappa(\Vhat)/(|K|\epswit^2)).$$

Thus, any $u$-$v$ $P$ path in $W$ can be mapped to a corresponding $u$-$v$ path in $\cP_W$ of length $O(|P| \cdot \len(\cP_W))$. This implies that $\diam_{\Hhat}(X) \leq \diam(W) \cdot\len(\cP_W)$. We further have that $W_{multi}[X]$ forms an expander, and it is further well known that the diameter of an expander is upper bounded by $O(\log n)$ over its expansion, and we therefore have $\diam(W_{multi}) = \Otil(1/\phicmg)$. Also note that since $\diam(W_{multi})$ is derived from $W$ by copying edges, we have that the same statement is true for $W$. Combining these insights, we obtain
\[
\diam_{\Hhat}(X) = \Otil(1/\phicmg) \cdot O(\len(\cP_W))=\Otil\left(\frac{D}{d}/(\phicmg\epswit^{2})\right)
\]
where the last equality is by \Cref{lem:total cap} (recall $\kappa(\Vhat) \leq \kappafinal(\Vhat)$).
As $\epswit=\Omegahat(\phicmg)$ by \Cref{lem:embedwitness}, we have
$\diam_{\Hhat}(X)=\Otil(\frac{D}{d}/\phicmg^{3})$. 

\underline{$\diam_G(X) = O(\gamma d \cdot \diam_{\Hhat}(X))$:} For any $u,v\in X \subseteq \Kinit \subseteq V$, consider a $u$-$v$ shortest path $P$ in $\Hhat$. Observe that since $u,v$ are vertices in $G$, we have that $P$ is formed from (entire) heavy paths (corresponding to edges of weight $\geq d$ in $G$) and edges in $H$.

For each heavy path $P'$ on $P$, we have that it is of length at most $O(d)$ times the original path (recall, we round the weight of the edge in $G$ by $d$ and insert a path of the corresponding length). On the other hand, any edge $(u',v')$ in $H$ has $\dist_G(u',v') \leq \gamma d$ by definition. The latter factor subsumes the former and establishes our claim.\\

We also need to prove that the side conditions of $\EmbedCore(\cdot)$ hold throughout the execution of the algorithm. The proof is deferred to \Cref{sec:proofOfSideConditions}.

\begin{restatable}[Side-Conditions]{claim}{sideConditionClaim}\label{clm:sideConditions}
Whenever the algorithm invokes $\EmbedCore(\cdot)$, we have
\begin{enumerate}
    \item $\kappa(v)\ge2$ for all terminals $v\in K$, \label{prop:sideCondition1}
    \item $\kappa(v)\le\kappa(V)/2$. \label{prop:sideCondition2}
\end{enumerate}
\end{restatable}

\paragraph{Total Update Time.} As we have proven the correctness of the algorithm, it remains to analyze the total update time.

\begin{lem}
\label{lem:bound num phase}The total number of while-loop iterations starting in \Cref{lne:whileKisbigCOre} is at most
\[
\tilde{O}\left(\Delta\frac{\kappafinal(\Vhat)}{|\Kinit|}/\phicmg\right).
\]
\end{lem}
\begin{proof}
First, we observe that the total weight of edges that are ever deleted from any of the witness graphs $W$ is at most $O(\Delta\kappafinal(\Vhat)\log(n))$. To see this, recall first that the weight of an edge $(u,v)$ in a graph $W$ (associated with embedding $\cP$) is equal to $\sum_{P \in P_{uv}} \val(P)$ where $P_{uv}$ is the set of $u$-$v$ paths in $\cP$. Now observe that whenever an edge $(v,e)\in E(\Hhat_{bip})$ of the incidence graph $\Hhat_{bip}$ of $\Hhat$ is deleted where $v\in\Vhat$ and $e\in E(H)$, the total value of the paths $P_{ve}$ containing the edge $(v,e)$ is at most $O(\kappa(v)\log(n))=O(\kappafinal(v)\log(n))$ by the guarantee on vertex congestion of $\EmbedCore(\cdot)$ from \Cref{lem:embedwitness}. Further such an edge $(v,e)$ once deleted does not occur in any future witness graph $W$. But there are at most $\Delta+3$ edges incident to $v$ in all versions of $\Hhat$ ($\Delta$ from $H$, $3$ from $G$). But this bounds the total weight ever deleted from all graphs $W$ by $O(\Delta\kappafinal(\Vhat)\log(n))$.

On the other hand, we claim that during a while-loop iteration, at least $\kappa_{del}=\phicmg|\Kinit|/20$ weight is deleted from $W$. Assume for the sake of contradiction that this is not true. Observe first that we build $W$ to initially have $|K'|-o(|K'|) \ge |\Kinit|-o(|\Kinit|)$ vertices with weighted degree
at least $9/10$ (see the while-loop condition in \Cref{lne:whileKisbigCOre} and the guarantees on $\EmbedCore(\cdot)$ from \Cref{lem:embedwitness}). But deleting $\kappa_{del}$ from $W$ causes $\Prune(W_{multi}, \phicmg)$ to ensure that set $X$ is such that $\vol_W(V(W) \setminus X) \leq 8\kappa_{del}/\phicmg = |\Kinit|/4$. This in turn implies that at most $\frac{10}{9} \cdot |\Kinit|/4 \leq |\Kinit|/3$ vertices of degree at least $9/10$ are in $V(W) \setminus X$. Therefore, $|X| \geq |\Kinit| - o(\Kinit) - |\Kinit|/3 \geq |\Kinit|/2$. But this contradicts that the while-loop iteration is over since the condition of the while-loop in \Cref{lne:whileExpanderWitnessIsLarge} is still satisfied.

By using the second claim to charge the sum from the first claim, we establish the lemma.
\end{proof}

\begin{lem}
\label{cor:bound num embed}The total number of times $\EmbedCore$
is called is at most $\tilde{O}(\Delta\frac{\kappafinal(\Vhat)}{|\Kinit|}/\phicmg)$.
\end{lem}

\begin{proof}
Every time $\EmbedCore$ returns a vertex cut $(L,S,R)$, we double
the capacity $\kappa(v)$ of every vertex $v\in S$. So the total capacity is increased
by $\kappa(S)\ge|L\cap \Kinit|\ge\epswit|\Kinit|$
by \Cref{lem:embedwitness}. Further, in \Cref{lne:technicalSideCondition}, we only further increase $\kappa$. But since $\kappafinal(\Vhat)$ is the total final capacity, we have that there can be at most $O(\frac{\kappafinal(\Vhat)}{\epswit|\Kinit|})=\Otil(\frac{\kappafinal(\Vhat)}{\phicmg|\Kinit|})$ times that $\EmbedCore(\cdot)$ returns a vertex cut.

The number of times that $\EmbedCore$ returns an embedding is at most the number of while-loop iterations which is $\tilde{O}(\Delta\frac{\kappafinal(\Vhat)}{|\Kinit|}/\phicmg)$ by \Cref{lem:bound num phase}. By summing the number of times from the two cases, the lemma holds. 
\end{proof}

\begin{lem}
\label{lem:core runtime}The total running time of \Cref{alg:Core}
is 
\[
\Otil\left(T_{\Apxball}(G,\Kinit,32D\log n,0.1)\Delta^{2}\left(\frac{D}{d}\right)^{3}/\phicmg^{2}\right).
\]
\end{lem}
\underline{Initialization:} It is straight-forward to see that the initialization (i.e. the first two lines in \Cref{alg:Core}) can be executed in  $O(|\Hhat|)$ by using an invocation of Dijkstra and some basic operations. 

\underline{A Single Iteration of the While-Loop starting in \Cref{lne:whileKisbigCOre} (Excluding $\EmbedCore(\cdot)$):} The \\ while-loop condition (and computing $K'$) in  \Cref{lne:whileKisbigCOre}, is checked using $\CertifyCore(\cdot)$ which takes $\Otil\left(\left|\ball_{G}(\Kinit,32D\log n)\right|\right)$ time by \Cref{lem:certifycore}. 

The time spent on $\Prune(W_{multi},\phicmg)$ is bound by \Cref{lem:prune} to be $O(|E(W_{multi})|/\phicmg)$. We then use the fact that $W_{multi}$ has at most $O(\frac{1}{\gamma_{\size}}|\Kinit|\log|\Kinit|)$
edges because it is derived from $W$ by making $w(e)/ \gamma_{\size}$ copies of each edge $e$ in $W$ where we established that the total weight of all edges in $W$ is $O(|\Kinit|\log|\Kinit|)$ by \Cref{lem:embedwitness}. As $\gamma_{\size}=|\Vhat|/|\Kinit|$, we thus have that $|E(W_{multi})|=\Otil(|\Vhat|)$ and therefore the time spent during a while-loop iteration on pruning is at most  $O(|E(W_{multi})|/\phicmg)=\Otil(|\Vhat|/\phicmg)$. 

Finally, we have to account for the time required to maintain $\Apxball(G,X,4D,0.1)$ which is $T_{\Apxball}(G,X,4D,0.1) \leq T_{\Apxball}(G,\Kinit,32D\log n,0.1)$, where the inequality follows from the monotonicity of $\Apxball$ in Remark \ref{rem:apx-ball-monotonic}. 

All other operations during the while-loop have time subsumed by the former procedures (or the invocations of $\EmbedCore(\cdot)$) giving total time
\begin{align}
 & \Otil\left(\left|\ball_{G}(\Kinit,32D\log n)\right|+\Otil(|\Vhat|/\phicmg)+T_{\Apxball}(G,X,4D,0.1)\right)\nonumber \\
 & =\Otil(T_{\Apxball}(G,\Kinit,32D\log n,0.1)+|\Vhat|/\phicmg),\label{eq:core within phase}
\end{align}
where we used that $T_{\Apxball}(G,\Kinit,32D\log n,0.1) = \Omega(|\ball_G(\Kinit,32D\log n)|)$, as discussed in Remark \ref{rem:apx-ball-monotonic}.

\underline{All Iterations of the While-Loop starting in \Cref{lne:whileKisbigCOre}  (Excluding $\EmbedCore(\cdot)$):} 
 As there are at most
$\Otil(\Delta\frac{\kappafinal(\Vhat)}{|\Kinit|}/\phicmg)$ while-loop iterations by \Cref{lem:bound num phase}, the total time spent (excluding time spent on $\EmbedCore$)
is 
\[\Otil(T_{\Apxball}(G,\Kinit,32D\log n,0.1)\Delta\frac{\kappafinal(\Vhat)}{|\Kinit|}/\phicmg+|\Hhat|\Delta\frac{\kappafinal(\Vhat)}{|\Kinit|}/\phicmg^{2})
\]
where we used $|\Vhat|\le|\Hhat|$ in the last term.

\underline{Time spent on $\EmbedCore(\cdot)$:} It is not hard to see that $\kappa$ is a $1/\gamma_{size}$-integral function, with $\gamma_{size} = \frac{1}{4}|\Vhat|/|\Kinit|$. Therefore, each call to $\EmbedCore(\cdot)$
in \Cref{enu:embedwitness} takes time
\[
\Otil(|\Hhat|\frac{\kappa(\Vhat)}{|\Kinit|\phicmg} + \gamma_{size} \cdot \kappa(\Vhat)/\phicmg)=\Otil(|\Hhat|\frac{\kappafinal(\Vhat)}{|\Kinit|\phicmg} + |\Vhat| \frac{\kappafinal(\Vhat)}{|\Kinit|\phicmg})
\]
because $\kappa(\Vhat)\le\kappafinal(\Vhat)$. We can assume w.l.o.g. that $|\Vhat| = O(|\Hhat|)$ since
the only way this could be false is if half the vertices of $\Vhat$ were isolated (i.e. had no incident edges), in which case a sparse cut in $\Hhat$ could trivially be found by computing connected components in $|\Hhat|$ time 
We can thus simplify the above bound to $\Otil(|\Hhat|\frac{\kappafinal(\Vhat)}{|\Kinit|\phicmg})$. Finally, we note that by \Cref{cor:bound num embed}, there are at most $\Otil(\Delta\frac{\kappafinal(\Vhat)}{|\Kinit|}/\phicmg)$
calls to $\EmbedCore(\cdot)$. Therefore, the total time spent on $\EmbedCore(\cdot)$ is at most
\[
\Otil\left(|\Hhat|\Delta\left(\frac{\kappafinal(\Vhat)}{|\Kinit|}\right)^{2}/\phicmg^{2}\right).\]

\underline{Combining Calculations:} By combining the two bounds above, the total time including the time spent on $\EmbedCore$ is at most
$$\Otil(T_{\Apxball}(G,\Kinit,32D\log n,0.1)\Delta\frac{\kappafinal(\Vhat)}{|\Kinit|}/\phicmg+|\Hhat|\Delta(\frac{\kappafinal(\Vhat)}{|\Kinit|})^{2}/\phicmg^{2}).$$
To simplify this expression, 
we have
 $\frac{\kappafinal(\Vhat)}{|\Kinit|}=O(\frac{D}{d}\log^{2}(n))$
by \Cref{lem:total cap} and also $|\Hhat|=\Otil(\Delta\frac{D}{d}\left|\ball_{G}(\Kinit,32D\log n)\right|)$ which can be verified by checking \Cref{def:Hhat} of $\Hhat$ from $H$ and $G$ (where each edge in $G$ might result in $\Otil(D/d)$ new heavy-path edges in $\Hhat$ and where we have constant degree by assumption). Therefore, the expression can be bounded by 
\[
\Otil\left(T_{\Apxball}(G,\Kinit,32D\log n,0.1)\Delta^{2}\left(\frac{D}{d}\right)^{3}/\phicmg^{2}\right).
\]
as claimed. (Here we used that $T_{\Apxball}(G,\Kinit,32D\log n,0.1) = \Omega(|\ball_G(\Kinit,32D\log n)|)$, as discussed in Remark \ref{rem:apx-ball-monotonic}.)

\paragraph{Final Total Capacity.} Finally, we bound the final total vertex capacity $\kappafinal(\Vhat)$
of $\Hhat$ as claimed in \Cref{lem:total cap}. Unfortunately, it is rather difficult to argue directly about $\Hhat$ since it is fully-dynamic. To establish our proof, we therefore rely on analyzing another graph $\Ghat$ which is used purely for analysis.

We define $\Ghat$ to be a dynamic \emph{unweighted} graph with vertex set $V(\Ghat)=\Vhat$
and the edge set $E(\Ghat)$ taken to be the union of the edges $\{(u,v)\in B^{init}\times B^{init}\mid\dist_{G}(u,v)\le d\}$ and all edges on heavy paths $\Phat$ that were also added to $\Hhat$ (recall \Cref{def:Hhat} and the definition $B^{init} = \ball_G(\Kinit,32D\log(n))$).

We first list structural properties of $\Ghat$
below:
\begin{prop}
\label{prop:structure Ghat}We have the following:
\begin{enumerate}
\item \label{enu:Ghat dec}$\Ghat$ is a decremental graph.
\item \label{enu:Ghat shrink G}For any $u,v\in\Kinit$, if $\dist_{G}(u,v)\le 32D\log n$,
then $\dist_{\Ghat}(u,v)\le 4 \cdot \lceil \dist_{G}(u,v)/d \rceil$.
\item \label{enu:Ghat weaker Hhat}If $(L,S,R)$ is a vertex cut in $\Hhat$,
then $(L,S,R)$ is also a vertex cut in $\Ghat$.
\end{enumerate}
\end{prop}
\underline{Property \ref{enu:Ghat dec}:} Observe that since $G$ is a decremental graph, distances in $G$ are monotonically increasing. Thus, the set $\{(u,v)\in B^{init}\times B^{init}\mid\dist_{G}(u,v)\le d\}$ is decremental. Further, recall that we assume that $G$ is undergoing edge deletions (no weight updates) and once an edge $e$ is deleted from $G$ its corresponding heavy path $P_e \in \Phat$ (if one is associated with $e$) is simply deleted from $\Ghat$. Thus, $\Ghat$ is a decremental graph.

\underline{Property \ref{enu:Ghat shrink G}:} Let $P$ be a shortest $u$-$v$ path in $G$. Let $E_{heavy}=\{e\in P\mid w(e)>d\}$. We can partition the path $P$ into $P=P_{1} \circ e_{1} \circ P_{2} \circ\dots\circ e_{|E_{heavy}|}\circ P_{|E_{heavy}|+1}$ where each $e_{i}\in E_{heavy}$ and each path $P_{i}$ contains only edges in $G$ with weight at most $d$. It remains to observe that we can replace 
\begin{itemize}
    \item each $u_i$-$v_i$ path $P_i$ in $G$ by finding a minimal set $S_i$ of vertices on $P_i$ with $u_i,v_i \in S_i$ such that each vertex $x$ in $S_i \setminus \{v_i\}$ is at most at distance $d$ to some vertex that occurs later on $P_i$ than $x$. Then, we can replace the path between each such two consecutive vertices by an edge in  $\{(u,v)\in B^{init}\times B^{init}\mid\dist_{G}(u,v)\le d\}$ and it is not hard to see that we use at most  $2 \cdot \lceil \dist_{G}(u_i,v_i)/d \rceil$ such edges in $\Ghat$, and
    \item each edge $e_i$ by a heavy path in $\Ghat$ consisting of $\ceiling{w(e_{i})/d}$ edges (recall heavy paths from \Cref{def:Hhat}).
\end{itemize}
It is not hard to combine the above two insights to derive the Property. We point out that above we implicitly use that all vertices on $P$ are in $B^{init}$. But this is clearly given since we assume $u,v \in \Kinit$ and $\dist_{G}(u,v)\le32D\log n$ while $B^{init}$ includes all vertices in $G$ that are ever at distance at most $32D\log n$ to any vertex in $\Kinit$.

\underline{Property \ref{enu:Ghat weaker Hhat}:} We prove the contra-positive.
Suppose that $(L,S,R)$ is not a vertex cut in $\Ghat$. That is,
there is an edge $(u,v)$ in $\Ghat$ where $u\in L$ and $v\in R$.
There are two cases. First, if $(u,v)$ is in a heavy path $\Phat$
in $\Ghat$, then $\Phat$ must appear in $\Hhat$ as well. Second,
if $\dist_{G}(u,v)\le d$, then, by \Cref{def:compressed graph}, there
is a hyperedge of a $(d,\gamma,\Delta)$-compressed graph $H$ that
contains both $u$ and $v$. Therefore, $(L,S,R)$ is not a vertex
cut in $\Hhat$.\\

We now define a powerful potential function to complete our proof. The key notion for our potential function is that of a cost of an embedding. In the definition below, it is important to observe that while we have $\kappa$ and $\gamma_{\size}$ defined by \Cref{alg:Core}, the embedding $\cP'$ can be chosen arbitrary (and in particular does not have to be $\cP$ from the algorithm). Given this definition it is straight-forward to set-up our potential function.

\begin{definition}[Cost of an Embedding]
At any point during the execution of \Cref{alg:Core}, consider $\kappa$ and $\gamma_{\size}=|\Vhat|/|\Kinit|$, and consider any embedding $\cP'$ that embeds some $W'$ into $\Ghat$. Then, we define the \emph{cost} of the embedding $\cP'$ by $c(\cP') = \sum_{v\in P, P \in \cP'} \log(\gamma_{\size}\kappa(v)) \cdot \val(P)$.
\end{definition}

\begin{defn}[Potential Function]\label{def:Pi general}
At any point during the execution of \Cref{alg:Core}, let $\mathbb{P}$
be a collection of all embeddings $\cP'$ that embed a graph $W'$ into $\Ghat$ that satisfies that
\begin{enumerate}
\item $W'$ is an unweighted star where $V(W')\subseteq\Kinit$ and
$|V(W')|\ge(1-\epswit/2)|\Kinit|$, and
\item $\diam_{\Ghat}(V(W'))\leq 256 \cdot \frac{D}{d} \cdot \log n$. 
\end{enumerate}
We then define the potential function $\Pi(\Ghat,\Kinit,\kappa) = \min_{\cP'\in\mathbb{P}}c(\cP')$ that is equal to the minimal cost achieved by any embedding in $\mathbb{P}$. Here, if $\mathbb{P}=\emptyset$, then we let $\Pi(\Ghat,\Kinit,\kappa)=\infty$. 
\end{defn}

Note that for each $\cP'\in\mathbb{P}$ and $P \in \cP'$ above, we have $\val(P) = 1$ (since $W'$ is unweighted). Also note that we do not have any guarantees on vertex congestion or length of the embeddings for any $\cP'$.

Let us now analyze the potential function $\Pi(\Ghat,\Kinit,\kappa)$ over the course of the algorithm.

\begin{prop}
\label{lem:monotone Pi}$\Pi(\Ghat,\Kinit,\kappa)\ge0$ and $\Pi(\Ghat,\Kinit,\kappa)$
can only increase through time. 
\end{prop}

\begin{proof}
$\Pi(\Ghat,\Kinit,\kappa)\ge0$ because, for all $v\in\Vhat$, we
have $\kappa(v)\ge1/\gamma_{\size}$ and so $\log(\gamma_{\size}\kappa(v))\ge0$.
As $\Ghat$ is a decremental graph by Property \ref{enu:Ghat dec} in \Cref{prop:structure Ghat}, we have that $\Pi(\Ghat,\Kinit,\kappa)$ may only increase. 
\end{proof}
\begin{prop}
\label{prop:trivial cap}For all $v\in\Vhat$, we have $\kappa(v)\le4|\Kinit|$
at any point of time.
\end{prop}

\begin{proof}
The capacity $\kappa(v)$ of $v$ can be increased only if a cut $(L,S,R)$ is returned by $\EmbedCore(\cdot)$ with $v\in S$
in \Cref{enu:doubling step}. But $\EmbedCore(\cdot)$ guarantees that $\kappa(S)\le 2|\Kinit|$ (see the Cut Property in \Cref{lem:embedwitness}). So once $\kappa(v)>2|\Kinit|$, $v$ cannot be in any future $S$. Since we double $\kappa(v)$ every time that $v$ appears in $S$, we can therefore ensure that $\kappa(v)\le 4|\Kinit|$.
\end{proof}

\begin{lem}
\label{lem:bound Pi} When the invocation of $\CertifyCore(G,\Kinit,2D,\epswit/2)$ in \Cref{lne:whileKisbigCOre} returns a Core $K'\subseteq\Kinit$, then $\Pi(\Ghat,\Kinit,\kappa)=O(|\Kinit|\frac{D}{d}\log^{2}n)$.
\end{lem}
\begin{proof}
By \Cref{lem:certifycore}, we have that such $K'$ satisfies that $|K'|\geq(1-\epswit/2)|\Kinit|$ and
$\diam_{G}(K')\leq32D\log n$. Using the latter fact, combined with Property \ref{enu:Ghat shrink G} from \Cref{prop:structure Ghat}, we have $\diam_{\Ghat}(K') \leq 4 \lceil \diam_{G}(K')/d \rceil \leq 256 \cdot \frac{D}{d} \cdot \log n$.

Using the last fact, with the guarantee on the size of $K'$, we note that picking an arbitrary vertex $u \in K'$, and letting $\cP'$ be an embedding containing for each $v \in K' \setminus \{u\}$, a shortest $u$-$v$ path $P$ in $\Ghat$ with value $\val(P) = 1$, we get that $\cP'$ must be in $\mathbb{P}$ as defined in \Cref{def:Pi general}. It is further straight-forward to see that
\begin{align*}
c(\cP') = \sum_{x\in P, P \in \cP'} \log(\gamma_{\size}\kappa(x)) < |\Kinit|\cdot\diam_{\Ghat}(K')\cdot\log(\gamma_{\size}\cdot2|\Kinit|)=O\left(|\Kinit|\frac{D}{d}\log^{2}n\right)
\end{align*}
because there are $|\Kinit|-1$ paths in $\cP'$, each path is of length at most $\diam_{\Ghat}(K')$, and each vertex $x$ has $\kappa(x)$ bound by \Cref{prop:trivial cap}. This completes the proof as $\Pi(\Ghat,\Kinit,\kappa)\le c(\cP')$.
\end{proof}
\begin{lem}
\label{lem:doubling effect}Consider when $\EmbedCore(\Hhat, K', \kappa)$ returns a vertex cut $(L,S,R)$ in $\Hhat$. Let $\kappa^{OLD}$ and $\kappa^{NEW}$ be the vertex capacities of $\Vhat$ before and after the doubling step in \Cref{enu:doubling step} and the potential increase of $\kappa(w')$ in \Cref{lne:technicalSideCondition}. Then,
\begin{enumerate}
\item $\kappa^{NEW}(\Vhat)\le\kappa^{OLD}(\Vhat)+6|L\cap K'|$. \label{prop:IncreaseInCapacity}
\item $\Pi(\Ghat,\Kinit,\kappa^{NEW})\ge\Pi(\Ghat,\Kinit,\kappa)+|L\cap K'|/3$. 
\label{prop:IncreaseInPotential}
\end{enumerate}
\end{lem}
\underline{Property \ref{prop:IncreaseInCapacity}:} We have that \Cref{enu:doubling step} leads to an increase in capacity from $\kappa^{OLD}(S)$ to $2\kappa^{OLD}(S)$ at the vertices on $S$ while the capacity at $\Vhat \setminus S$ remains unchanged. In  \Cref{lne:technicalSideCondition}, we set the capacity of $w'$ at most to the current capacity at $S$, i.e. at most $2\kappa^{OLD}(S)$. Thus, we have  $\kappa^{NEW}(\Vhat)\leq\kappa^{OLD}(\Vhat)+3\kappa^{OLD}(S)$ where $\kappa^{OLD}(S)\le 2|L\cap K'|$ by \Cref{lem:embedwitness}. 

\underline{Property \ref{prop:IncreaseInPotential}:} First, recall that by Property \ref{enu:Ghat weaker Hhat} in \Cref{prop:structure Ghat}, the vertex cut $(L,S,R)$ in $\Hhat$ is also a vertex cut in $\Ghat$. Now, given any embedding $\cP'$ from $\mathbb{P}$ (as defined in \Cref{def:Pi general}) that embeds $W'$ into $\Ghat$, we define $L' = L \cap K' \cap V(W')$ and analogously $R' = R \cap K' \cap V(W')$. Further, let $v'$ be the center of the star $W'$, then if
\begin{itemize}
    \item $v' \in L$: we have that there are at least $|R'|$ paths in $\cP'$ from $v'$ to $R'$ (in $\Ghat$). But by definition of the vertex cut $(L,S,R)$ in $\Ghat$, each of these paths must contain at least one vertex in $S$.
    \item $v' \not\in L$: then $v \in S \cup R$, but this implies that there are at least $|L'|$ paths in $\cP$ from $v'$ to $L'$, thus containing at least one vertex in $S$.
\end{itemize}
As we double the capacity of every vertex in $S$ and $\cP' \in \mathbb{P}$
is chosen arbitrarily, we have thus proven that $\Pi(\Ghat,\Kinit,\kappa)$ is increased by at least $\min\{|R'|, |L'|\}$. Thus, if we could lower bound $|R'|, |L'|$ to be of size at least $|L\cap K'|/3$, then the property would be established.

Therefore, we note that by $\EmbedCore(\Hhat, K', \kappa)$ from \Cref{lem:embedwitness}, we have $|R \cap K'| \geq |L\cap K'|\ge\epswit|K'| = \epswit|\Kinit| - o(\epswit|\Kinit|)$ where the later equality is by the while-loop condition in \Cref{lne:whileKisbigCOre}. Then, since $K' \subseteq \Kinit$ and at most $(\epswit/2)|\Kinit|$ vertices in $\Kinit$ are not in $V(W')$ (by \Cref{def:Pi general}), we further obtain that $|L \cap K' \cap V(W')| \geq |L \cap K'| - (\epswit/2)|\Kinit| > |L \cap K'|/3$ and analogously $|R \cap K' \cap V(W')| \geq |L \cap K'|/3$. The property is thus established. \\

Now, we are ready to give the upper bound on the final total vertex capacity $\kappafinal(\Vhat)$ of $\Hhat$ as claimed in \Cref{lem:total cap}.

\begin{lem}
\label{lem:total capUpperBound}
At any point of time, we have $\kappafinal(\Vhat) \leq O\left(|\Kinit|\frac{D}{d}\log^{2}(n)\right)$.
\end{lem}
\begin{proof}
Throughout the algorithm, $\kappa(\Vhat)$ is only changed in \Cref{enu:doubling step} of the algorithm after an invocation of $\EmbedCore(\cdot)$ returns a
vertex cut $(L,S,R)$ in $\Hhat$. But by \Cref{lem:doubling effect}, every time $\kappa(\Vhat)$ is increased by amount $x$, the potential $\Pi(\Ghat,\Kinit,\kappa)$ is increased by at least $x/18$. 

However, $\Pi(\Ghat,\Kinit,\kappa)$ is initially non-negative (see \Cref{lem:monotone Pi}) and never exceeds $O(|\Kinit|\frac{D}{d}\log^{2}(n))$ (by \Cref{lem:bound Pi}). Hence the total increase of $\kappa(\Vhat)$ is also bound by $O(|\Kinit|\frac{D}{d}\log^{2}(n))$, combined with the initial capacity of $\kappa(\Hhat) \leq |\Kinit| \cdot 2 + |\Vhat| \cdot 1/\gamma_{size} = O(|\Kinit|)$ (see \Cref{enu:init cap}) this establishes the Lemma.
\end{proof}

%% file: covering.tex
\section{Implementing Covering}
\label{sec:part2Covering}

Building on the previous two data structures (for Approximate Balls and Robust Cores), we are now ready to give our implementation of a Covering data structure. We recall from \Cref{def:Covering} that a $(d,k,\eps,\stretch,\Delta)$-covering $\cC$ is a dynamic collection of cores $C \in \cC$ where each core $C$ is a Robust Core such that $C=\Core(G,\Cinit,d_{\ell})$ where $\ell \in [0,k-1]$ is the level assigned to $C$ (we also write $\Clevel(C)=\ell$) and where $d_{\ell}=d\cdot(\frac{\stretch}{\eps})^{\ell}$. Observe that this implies that we always have $d\leq d_{\Clevel(C)}\leq d(\stretch/\eps)^{k-1}$ for any $C \in \cC$.

For intuition, the reader should keep in mind that we intend to use the Theorem below for $k\sim\log\log(n)$ and $\stretch$ such that $(\frac{\stretch}{\epsilon})^k=n^{o(1)}$. 

\begin{theorem}
[Covering]\label{thm:covering}Let $G$ be an $n$-vertex bounded-degree
decremental graph. Given parameters $d,k,\eps,\stretch, \scatter$ where $\eps\le0.1$, and
\begin{itemize}
\item for all $d\le d'\le d(\frac{\stretch}{\eps})^{k}$, there is a approximate ball data structure $\Apxball(G,S,d',\eps)$ with total update time $T_{\Apxball}(G,S,d',\eps)$, and
\item for all $d\le d'\le d(\frac{\stretch}{\eps})^{k-1}$, there is a robust
core data structure $\Core(G,\Kinit,d')$ with scattering parameter
at least $\scatter$ and stretch at most $\stretch$ that has total
update time $T_{\Core}(G,\Kinit,d')$.
\end{itemize}
We can maintain $(d,k,\eps,\stretch,\Delta)$-covering of $G$
with $\Delta=O(kn^{2/k}/\scatter)$ in total update time 
\[
O(kn^{1+2/k}\log(n)/\scatter+\sum_{C\in\cC^{ALL}}T_{\Core}(\stage G{t_{C}},\stage C{t_{C}},d_{\Clevel(C)})+T_{\Apxball}(\stage G{t_{C}},\stage C{t_{C}},\frac{\stretch}{4\eps}d_{\Clevel(C)},\eps))
\]
where $\cC^{ALL}$ contains all cores that have ever been initialized
and, for each $C\in\cC^{ALL}$, $t_{C}$ is the time $C$ is initialized and added to $\cC$.
We guarantee that $\sum_{C\in\cC^{ALL}}|\ball_{\stage G{t_{C}}}(\stage C{t_{C}},\frac{\stretch}{4\eps}d_{\Clevel(C)})|\le O(kn^{1+2/k}/\scatter)$.
\end{theorem}

\begin{algorithm}[!ht]
\tcc{While there exists a vertex $v\in V$ not covered by any core in $\cC$.}
\While{$\exists v, \forall C\in\cC, v \notin\cover(C)$}{
    Let $\ell$ be the smallest integer with $|\ball_{G}(v,d_{\ell+1})|\le n^{(\ell+1)/k}$.\label{lne:shellissmall}\;
    $\Cinit \gets \ball_{G}(v,d_{\ell}).$\label{step:covering:level}\;
    Maintain core set $C=\Core(G,\Cinit,d_{\ell})$,
    \label{step:covering:shell} $\cover(C)=\Apxball(G,C,4d_{\ell},\eps)$\label{step:covering:core}
and $\shell(C)=\Apxball(G,C,\frac{\stretch}{4\eps}d_{\ell},\eps)$. \;
    $\Clevel(C)=\ell$.\;
    Add core $C$ to covering $\cC$. When $C$ becomes equal to $\emptyset$ (because $\Core(\cdot)$ terminates), remove $C$ from $\cC$ and stop maintaining $\cover(C)$ and $\shell(C)$.\;
}
\caption{$\Covering(G, d,k,\eps$, $\stretch)$ \label{alg:Covering}}
\end{algorithm}

The algorithm for maintaining the covering is described in \Cref{alg:Covering}. It is rather straight-forward: whenever there is a vertex $v$ that is not covered by any core in $\mathcal{C}$, then we make $v$ (together with some vertices in the ball centered around $v$ to some carefully chosen radius) a core $C^{init}$ itself. 

We first describe the basic guarantee of \Cref{alg:Covering}.
\begin{prop}
	\label{prop:clustering basic}We have the following:
	\begin{enumerate}
		\item A level $\ell$ assigned to each core $C$ is between $0$ and $k-1$. 
		\item Every vertex is covered by some core.
		\item At any stage, all cores $C$ from the same level are vertex disjoint. 
	\end{enumerate}
\end{prop}

\begin{proof}
(1): Otherwise, there is a vertex $v$ such that $|\ball_{G}(v,d_{k})|>n$
which is impossible. (2): This follows directly from \Cref{alg:Covering}. (3): Since every core $C$ is a decremental set by the guarantee of
$\Core$, it is enough to show that whenever a core $C=\ball_{G}(v,d_{\ell})$
is initialized with level $\ell$, $C$ is disjoint from other cores
$C'$ with level $\ell$. This holds because $v$ is not covered by
any level-$\ell$ core $C'$ and so $\dist_{G}(C',v)>4d_{\ell}$.
So $\ball_{G}(v,d_{\ell})$ is disjoint from all level-$\ell$ cores
$C'$.
\end{proof}
Therefore, to show that an $(d,k,\eps,\stretch,\Delta)$-clustering
$\cC$ of $G$ is maintained, it remains the bound $\Delta$, i.e.,
the number of outer-shells each vertex $v$ can ever participate in.
To do this, we first prove an intermediate step that bounds the number
of cores a vertex can participate.

\begin{lem}
\label{lem:bound core}For each level $\ell$, each vertex $v$ can
ever participate in at most $O(n^{1/k}/\scatter)$ many level-$\ell$
cores. 
\end{lem}

\begin{proof}
We prove the lemma by charging the number of vertices in $B_{v}=\ball_{G}(v,2d_{\ell})$. 

We first observe that initially, i.e. at the first time that $v$ is added to a level-$\ell$ core, we have that $|B_{v}| \leq n^{(\ell+1)/k}$. This follows since when $v$ is first added to some $C^{init} = \ball_G(u,d_{\ell})$ in \Cref{step:covering:level} of \Cref{alg:Covering}, the ball of the core is centered at some vertex $u$. But we clearly have $\ball_G(u,d_{\ell}) \subseteq B_v$. On the other hand, the algorithm ensures by choice of $\ell$ in \Cref{lne:shellissmall} that $|\ball_G(u,d_{\ell+1})| \leq n^{(\ell+1)/k}$ but we also have that $B_v \subseteq \ball_G(u,d_{\ell+1})$ which establishes the claim.

\begin{figure}[!ht]
\centering
\includegraphics[scale=.33]{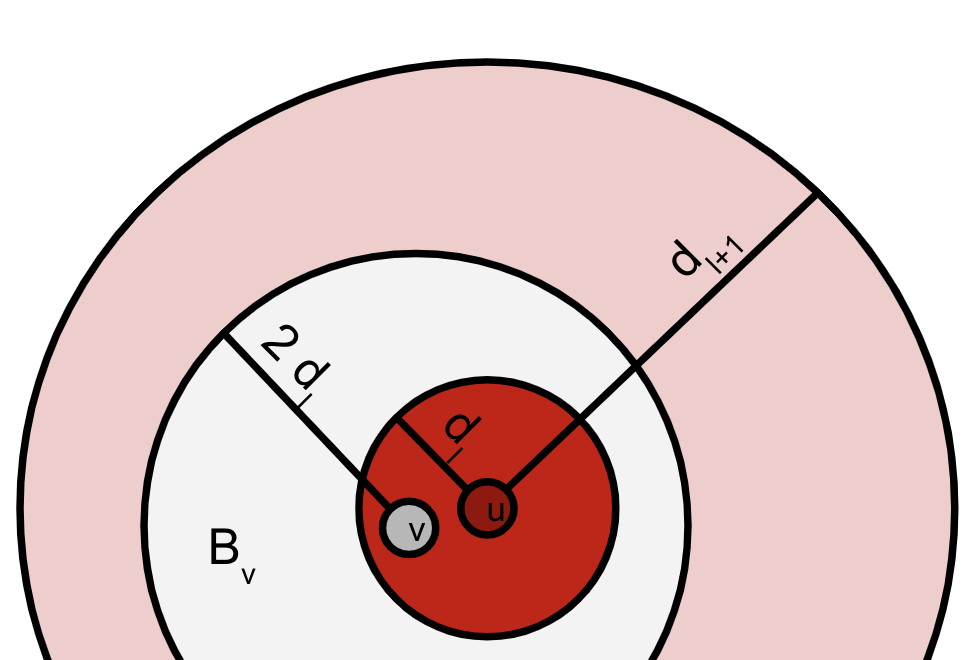}
\end{figure}

Next, recall from \Cref{prop:clustering basic} that all level-$\ell$ cores are vertex disjoint. Thus, the cores $C_1, C_2, \dots, C_{\tau}$ that $v$ participates in over time have the property that each core $C_{i+1}$ is initialized only after $v$ has left core $C_i$. Now consider some core $C_i$, that was initialized to $C_i^{init} = \ball_G(u, d_{\ell})$, i.e. the ball centered at some $u$ (as discussed above). Observe that $|\ball_G(u, d_{\ell})| > n^{\ell/k}$ by the minimality of $\ell$ (see again \Cref{lne:shellissmall}). 

But we have that $C_i^{init}$ was in $B_v$ when $v$ was added to $C_i$. Further, when $v$ leaves $C_i$, we have by the definition of $\Core$ (see in particular the Scattered Property in \Cref{def:Core}, and the parameters used in \Cref{step:covering:shell}) that only $|B_v \cap C_i^{init}| \leq (1-\scatter)|C_i^{init}|$ of the vertices in $C_i^{init}$ are still in $B_v$. Combined, this implies that $\Omega(n^{\ell/k} \scatter)$ vertices are leaving $B_v$ in between $v$ joining $C_i$ and $C_{i+1}$ for every $i$. 

Using that initially $|B_{v}| \leq n^{(\ell+1)/k}$, we derive that the number of level-$\ell$ cores that $v$ can participate in is 
\[
\tau \leq 1 +\frac{n^{(\ell+1)/k}}{\Omega(n^{\ell/k} \scatter)} = O(n^{1/k}/\scatter).
\]
\end{proof}
Now, we are ready to prove that $\Delta=O(kn^{2/k}/\scatter)$.

\begin{lem}
	\label{lem:bound shell}For each level $\ell$, each vertex $v$ can
	ever participate in at most $O(n^{2/k}/\scatter)$ many outer-shells
	of level-$\ell$ cores. Thus, over all levels, $v$ can participate
	in at most $O(kn^{2/k}/\scatter)$ many outer-shells.
\end{lem}

\begin{proof}
We again use an argument where we charge $B_v$ for a specific level $\ell$. However, this time we let the radius of $v$ be twice the radius of a shell at level $\ell$ (and also larger by a small fraction than that of an outer-shell), i.e. we define $B_v = \ball_{G}(v,\frac{\stretch}{2\eps}d_{\ell})$.

Let $C_1, C_2, \dots, C_{\tau}$ be the cores that have $v$ in their outer-shell (let them be ordered increasingly by their initialization time). Since each core $C_i$ is decremental, if $v$ is ever in the outer-shell $\oshell(C_i)$, then it is also in the outer-shell of $C_i$ upon $C_i$'s initialization. I.e. then $v \in \oshell(C_i^{init})$.

Note that when $v$ is added to the outer-shell $\oshell(C_1^{init})$ of $C_1^{init} = \ball(u, d_{\ell})$ then at that stage we also have that $|\ball(u,d_{\ell+1})| \leq n^{(\ell+1)/k}$ (by \Cref{lne:shellissmall}). But this implies that $|B_v| \leq n^{(\ell+1)/k}$ since $B_v$ can only include vertices at distance at most $\dist_G(v,u) + \frac{\stretch}{2\eps}d_{\ell} \leq d_{\ell+1}$ from $u$.

We now use a slightly more advanced charging scheme than in \Cref{lem:bound core}. To this end, consider the process where we, for every $C_i$, charge every vertex $w \in C_i^{init}$ a single credit. We note first, that by our analysis above there are at most $n^{(\ell+1)/k}$ vertices that can ever pay a credit since cores that are not fully contained in $B_v$ when $C_1$ is initialized cannot have $v$ in their outer-shell (this follows by a straight-forward application of the triangle inequality). Further, each vertex $w$ is in at most $O(n^{1/k}/\scatter)$ level-$\ell$ cores by \Cref{lem:bound core}. This bounds the total number of available credits by $O(n^{(\ell+2)/k}/\scatter)$. 

But on the other hand, each core $C_i$ at level $\ell$ has an initial set $C_i^{init}$ of size at least $n^{\ell/k}$ by minimality of $\ell$ in \Cref{lne:shellissmall} when $C_i$ is initialized. But then each such core charges at least $n^{\ell/k}$ credits in the above scheme. The bound follows. 
\end{proof}

Finally, we finish with the running time analysis.
\begin{lem}
\label{lem:time clustering}The total update time of \Cref{alg:Covering}
is at most 
\[
O(kn^{1+2/k}\log n/\scatter+\sum_{C\in\cC^{ALL}}T_{\Core}(\stage G{t_{C}},\stage C{t_{C}},d_{\Clevel(C)})+T_{\Apxball}(\stage G{t_{C}},\stage C{t_{C}},\frac{\stretch}{4\eps}d_{\Clevel(C)},\eps))
\]
where $\cC^{ALL}$ contains all cores that have ever been initialized
and, for each core $C\in\cC^{ALL}$, $t_{C}$ is the time $C$ is
initialized. We guarantee that $\sum_{C\in\cC^{ALL}}|\ball_{\stage G{t_{C}}}(\stage C{t_{C}},\frac{\stretch}{4\eps}d_{\Clevel(C)})|\le O(kn^{1+2/k}/\scatter).$
\end{lem}

\begin{proof}
To implement \Cref{alg:Covering}, for each vertex $v$, we will maintain
the lists $\core_{v}=\{C\mid v\in C\}$, $\cover_{v}=\{C\mid v\in\cover(C)\}$,
and $\shell_{v}=\{C\mid v\in\shell(C)\}$. As all cores $C$ and their
covers and shells are maintained explicitly by $\Core$ and $\Apxball$,
the time for maintaining these lists are subsumed by the total update
time of $\Core$ and $\Apxball$. Given an edge update $(u,v)$, we
only need to generate the update $(u,v)$ to all data structures $\Core$
and $\Apxball$ on the cores $C$ where $C\in\shell_{u}\cup\shell_{v}$.
By \Cref{lem:bound shell}, the total number of generated updates is
at most $O(kn^{1+2/k}/\scatter)$. 

From the collection of lists $\{\cover_{v}\}_{v\in V}$, we can report
whenever there is a vertex $v$ which is not covered by any core. 

Suppose that at time $t$ there is such a vertex $v$ and we initialize
a core $C$ with level $\ell$. In \Cref{step:covering:level}, starting
from $\ell=0$, we compute $\ball_{\stage Gt}(v_{C},d_{\ell+1})$
by running Dijkstra, and as long as $|\ball_{\stage Gt}(v_{C},d_{\ell+1})|>n^{\ell/k}$,
we set $\ell\gets\ell+1$ and continue the Dijkstra's algorithm. The
total running time is $O(|\ball_{\stage Gt}(v_{C},d_{\ell+1})|\log n)=O(|\Cinit|n^{1/k}\log n)$.
In \Cref{step:covering:core}, $\Core$ is initialized for maintaining
$C$ using $T_{\Core}(\stage Gt,\Cinit,d_{\ell})$ total update time.
In \Cref{step:covering:shell}, $\Apxball$ is initialized for maintaining
$\cover(C)$ and $\shell(C)$ using at most $2\cdot T_{\Apxball}(\stage Gt,\Cinit,\frac{\stretch}{4\eps}d_{\ell},\eps)$
total update time. We assign $t_{C}\gets t$ for this core $C$. Note
that $\Cinit=\stage C{t_{C}}$. Therefore, the total update time is
can be written as
\begin{align*}
    &O(\sum_{C\in\cC^{ALL}} [|\stage C{t_{C}}|n^{1/k}\log(n)/\scatter+T_{\Core}(\stage G{t_{C}},\stage C{t_{C}},d_{\Clevel(C)})\\
    &+T_{\Apxball}(\stage G{t_{C}},\stage C{t_{C}},\frac{\stretch}{4\eps}d_{\Clevel(C)},\eps)]).
\end{align*}
By \Cref{lem:bound core}, we have that $\sum_{C\in\cC^{ALL}}|\stage C{t_{C}}|=O(kn^{1+1/k}/\scatter)$.
Also, by \Cref{lem:bound shell}, we have 
\[
\sum_{C\in\cC^{ALL}}|\ball_{\stage G{t_{C}}}(\stage C{t_{C}},\frac{\stretch}{4\eps}d_{\Clevel(C)})|\le O(kn^{1+2/k}/\scatter)
\]
because $\ball_{\stage G{t_{C}}}(\stage C{t_{C}},\frac{\stretch}{4\eps}d_{\Clevel(C)})\subseteq\shell(\stage C{t_{C}})$
and $\shell(C)$ is decremental.
\end{proof}

%% file: MES.tex
\section{Implementing Approximate Balls}
\label{sec:ball}

In this section, we derive the $\Apxball$ data structure. Here, we use standard techniques from the literature with small adaptions to deal with our compressed graphs.

\begin{theorem}
[Approximate Ball]\label{thm:Apxball}Let $G$ be an $n$-vertex
bounded-degree decremental graph. Let $\eps\le0.1$. Suppose that
a $(d,k,\eps,\stretch,\Delta)$-covering $\cC$ of $G$ is explicitly
maintained for us. We can implement an approximate ball data structure
$\Apxball(G,S,D,50\eps)$ using $\Otil(\left|\ball_{G}(S,D)\right|\Delta\frac{D}{\eps d})+T_{\Apxball}(G,S,2(\frac{\stretch}{\eps})^{k}d,\eps)$
total update time.
\end{theorem}

\paragraph{Intuition for \Cref{thm:Apxball}.} Let us offer some intuition on the theorem above and the bounds derived. Consider the statement of the classic ES-trees (with weight rounding): Given a decremental graph $G'$ with minimum weight $\lambda$, decremental set $S$ and depth $\Lambda$, we can maintain $\Apxball(G',S,\Lambda,50\eps)$ in time $\Otil(|E(\ball_{G}(s,\Lambda))| \cdot \frac{\Lambda}{\epsilon\lambda})$. 

Now, assume for the sake of simplicity that $G$ is unweighted and that the covering-compressed graph $H_{\cC}=(V\cup\cC,E')$ is a decremental graph (i.e. that no new core needs to be added to the covering throughout the entire update sequence). Then, consider running the ES-tree from $S$ on the graph $H_{\cC}$ and run it to depth $D$. It is not hard to see that this ES-tree runs on the edge set $E(\ball_{H_{\cC}}(S,D))$ which is of size at most $|\ball_{H_{\cC}}(S,D)|\Delta$ since each vertex in $H_{\cC}$ is incident to at most $\Delta$ edges in the entire update sequence. 

To reduce the run-time by a factor of $d$, we increase all edge weights to be at least $d$. To bound the total additive error introduced by this rounding, we observe that given any vertex $t$ we can take the following $S-t$ path in $H_{\cC}$: $\pi'(S,t) = \langle v_1 = S, v_2, \dots, v_k, v_{k+1} = t \rangle$ in $H_{\cC}$ where $v_i$ is the $(i \cdot d)^{th}$ vertex on $\pi_G(S,t)$ (except for $v_k$) -- here, $\pi_G(S,t)$ is the shortest $S-t$ path in $G$. That is, every $v_i$ and $v_{i+1}$ are at distance exactly $d$ (except for $i = k$ where the distance is smaller). All but the last edge on this path already has weight $d$, so increasing edge weights to $d$ has no effect. The last edge might incur an additive error of $d$, but as long as the distance from $S$ to $t$ is at least $d/\eps$, this additive can be subsumed in a multiplicative $(1+O(\eps))$ error.

We conclude that we can run the ES-tree above in running time $O(\left|\ball_{G}(S,D)\right|\Delta\frac{D}{\eps d})$. This approach would in fact also work if $G$ was weighted, if we additionally add the edges from $G$ of weight $\geq d$ to $H_{\cC}$. The reason we need these heavy edges is that a path $\pi(S,t)$ in $G$ might have a large weight edge $(u,v)$ on the path (with edge weight $\gg d$) and $H_{\cC}$ would not guarantee that there is even a path in $H_{\cC}$ from $u$ to $v$. But instead the ES-tree could directly pick such a large edge from $G$ and include it on its path.

There are two main obstacles to the above approach. The primary obstacle is that $H_{\cC}$ is fully-dynamic and not decremental because new cores can be inserted. Intuitively, however, the insertions in $H_{\cC}$ have low impact because $H_{\cC}$ models the decremental graph $G$. In an earlier paper, Forster, Henzinger, and Nanongkai \cite{henzinger2014decremental} showed how to extend an ES tree to work in graphs with low-impact insertions; their technique is called a \emph{monotone} ES-tree (MES). We note that the MES tree is not a black-box technique: it is a general framework which has to be individually adapted to every particular graph. Most of this section is thus dedicated to proving that the MES tree works on our emulator with low-impact insertions; while this proof is quite technical, conceptually it follows the same framework as other MES proofs (see e.g. \cite{henzinger2014decremental,bernstein2016deterministic,gutenberg2020deterministic}).

The second obstacle is that the argument above incurs an additive error of $d$, so it only guarantees a good approximation when $\dist(S,t) > d/\eps$. For smaller distances, we run $\Apxball$ on a smaller distance scale, which is the source of the additional  $T_{\Apxball}(G,S,2(\frac{\stretch}{\eps})^{k}d,\eps)$ term in the theorem statement. In the final section of this part (Section \ref{sec:part2PuttingItTogether}), we use an inductive argument to argue that $T_{\Apxball}(G,S,2(\frac{\stretch}{\eps})^{k}d,\eps)$ is small, and so the running time of $\Apxball(G,S,D,50\eps)$ is in fact dominated by the first term $\Otil(\left|\ball_{G}(S,D)\right|\Delta\frac{D}{\eps d})$.

\subsection{Emulator}

Recall the covering-compressed graph $H_{\cC}=(V\cup\cC,E')$ of the covering
$\cC$ defined in \emph{\Cref{def:core compress}.} As $\cC$ is explicitly
maintained for us, we will assume that $H_{\cC}$ is explicitly maintained
for us as well by \Cref{rem:covering-compressed from covering}. 

\begin{defn}[Emulator $\Htil$]\label{def:Htil}
Given a decremental graph $G = (V,E)$, a decremental set of vertices $S \subseteq V$, depth parameters $d \leq D$ and approximation parameter $1/\polylog(n) \leq \epsilon < 1$, and a covering-compressed graph $H_{\cC}=(V\cup\cC,E')$ of $G$ of the covering $\cC$. 

We define the (static) vertex set $V^{i nit}=\ball_{\stage G0}(\stage S0,D)$. We can define the emulator $\Htil$ with weight function $\wtil$ where its edge set $\Tilde{E} = E(\Htil)$ consists of the following
\begin{enumerate}
    \item the edges $e$ that are incident to $V^{init}$ in the graph $H_{\cC}$.
    \item the edges $e \in E(G[V^{init}])$ where $d<w_G(e) \le D$, and we set $\wtil(e) =  \ceiling{w_G(e)}_{\eps d}$, and
    \item \label{enu:Htil near edge} we maintain $B_S = \Apxball(G,S,2(\frac{\stretch}{\eps})^{k}d,\eps)$ and for each vertex $v \in (B_S \cap V^{init})$, we have an edge $(s,v)$ between a universal \emph{dummy vertex $s$} and $v$ of weight $\wtil(s,v) = \ceiling{\dtil^{\near}(v)}_{\eps d}$
    where $\dtil^{\near}(v)$ denotes the distance estimate maintained by $\Apxball(G,S,2(\frac{\stretch}{\eps})^{k}d,\eps)$.
\end{enumerate}
The vertex set of $\Htil$, denoted $\Vtil = V(\Htil)$, is the union of $V^{init}$ and the set of all endpoints of $\Tilde{E}$.
\end{defn}

Here, a more explicit way of defining the vertex set of $\Htil$ is to consider the cores in $\cC$ that some vertex of $V^{init}$ in their shell (at any point), formally the collection $\cC_{refined} = \{ C \in \cC \mid \shell(C^{init}) \cap V^{init} \neq \emptyset \}$. Then, $\Vtil$ can be defined as the union of $V^{init} \cup \{s\} \cup \cC_{refined}$. Note that as $\cC$ is a fully-dynamic set, so is $\cC_{refined}$ and therefore $\Vtil$. However, since we are inducing over edges, we only add or remove vertices of degree zero.

We henceforth call the vertices in $V^{init}$, the \emph{regular} vertices. We call the vertices in $\cC_{refined}$, the \emph{core} vertices.

\begin{prop}
\label{prop:Htil vertex}We have the following:
\begin{enumerate}
\item \label{enu:Htil:vertex:reg}Regular vertices in $\Htil$ have all-time
degree at most $\Delta+O(1)$. 
\item \label{enu:Htil:vertex:core}Core vertices in $\Htil$ form an independent
set.
\end{enumerate}
\end{prop}

\begin{proof}
(1): Each regular vertex $u$ is ever incident to at most $\Delta$
core vertices by \Cref{def:Covering}. As $G$ has bounded degree and
is decremental, $u$ is ever incident to at most $O(1)$ other regular
vertices. Also, $S$ is decremental and $u$ might be incident to
$s$ only once. In total, the all-time degree of $u$ is $\Delta+O(1)$.

(2): As the covering-compressed graph $H_{\cC}$ is bipartite, core vertices
are independent in $H_{\cC}$. As we never add edges between core
vertices in $\Htil$, they are independent in $\Htil$ as well.
\end{proof}
For each edge $e\in E(\Htil)$, we let $\wtil(e)$ denote the weight
of $e$ in $\Htil$. If $(u,v)\notin E(\Htil)$, we let $\wtil(u,v)\gets\infty$.
In particular, deleting an edge $e$ in $\Htil$ is to increase the
weight $\wtil(e)$ to infinity. 
\begin{prop}
\label{prop:Htil weight}For every edge $e\in E(\Htil)$, we have
the following:
\begin{enumerate}
\item \label{enu:Htil:weight:multiple}$\wtil(e)$ is a non-negative multiple
of $\ceiling{\eps d}$. 
\item \label{enu:Htil:weight:zero}$\wtil(e)=0$ if and only if $e=(s,v)$
where $v\in S$. 
\item \label{enu:Htil:weight:increase}$\wtil(e)$ can only increase after
$e$ is inserted into $\Htil$.
\end{enumerate}
\end{prop}

\begin{proof}
(1,2): This follows directly from the construction of $\Htil$.

(3): We insert edges into $\Htil$ only when there is a new core $C$ added into the covering $\cC$ (recall that edges in $G$ do not undergo edge weight changes by the \Cref{prop:simplify-1}). For each edge $e=(v,C)\in E(H_{\cC})$
where $v\in\shell(C)$, we have that $w(e')=\roundup{\stretch\cdot d_{\Clevel(C)}+\dtil^{C}(v)}{\eps d}$
where $\dtil^{C}(v)$ is the distance estimate of $\dist_{G}(C,v)$
from the instance of $\Apxball$ that maintains $\shell(C)$. By the
guarantee of $\Apxball$, $\dtil^{C}(v)$ never decreases and hence
$w(e')$ never decreases.
\end{proof}
Let $E^{ALL}(\Htil)$ be the set of all edges ever appear in $\Htil$. 
\begin{lem}
\label{lem:Htil update}$|E^{ALL}(\Htil)|\le O(\left|\ball_{G}(S,D)\right|\Delta)$.
Moreover, the total number of edge updates (including insertions,
deletions, and weight increase) in $\Htil$ is at most $O(\left|\ball_{G}(S,D)\right|\Delta\frac{D}{\eps d})$. 
\end{lem}

\begin{proof}
The bound on $|E^{ALL}(\Htil)|$ follows directly from \Cref{prop:Htil vertex}.
For each edge, its weight can be updated at most $\ceiling{D/\eps d}$
times because (1): every edge weight $e$ is a multiple of $\eps d$
by \Cref{prop:Htil weight}(\ref{enu:Htil:weight:multiple}), (2):
$\wtil(e)$ may only increase after $e$ was inserted by \Cref{prop:Htil weight}(\ref{enu:Htil:weight:increase}),
and (3): any edge with weight more than $D$ is removed from $\Htil$.
Therefore, the total number of edge updates is $|E^{ALL}(\Htil)|\cdot\ceiling{D/\eps d}=O(\left|\ball_{G}(S,D)\right|\Delta\frac{D}{\eps d})$. 
\end{proof}

\subsection{The Algorithm: MES on the Emulator}

Our $\Apxball$ algorithm for \Cref{thm:Apxball} works as follows. 
\begin{enumerate}
\item Maintain the emulator $\Htil$ with a dummy source $s$. Let $\dtil^{\near}(u)$
be the distance estimate of $u$ maintained by $\Apxball(G,S,2(\frac{\stretch}{\eps})^{k}d,\eps)$
as described in \Cref{def:Htil}(\ref{enu:Htil near edge}).
\item Maintain the \emph{Monotone Even-Shiloach (MES) data structure} $\MES(\Htil,s,D)$
(see \Cref{alg:MES}) which maintains the distance estimates $\{\dtil(u)\}_{u\in\Vtil}$.\footnote{We note that, if there is no insertion, the described algorithm is
equivalent to the classic ES-tree algorithm \cite{EvenS}} After each edge deletion to $G$, there can be several edge updates
to $\Htil$. We feed all edge insertions to the MES data structure
before any other update generated at this time. 
\item For each regular vertex $u\in V\cap\Vtil$, we maintain $\left\{ \min\{\dtil(u),\dtil^{\near}(u)\}\right\} {}_{u\in V\cap\Vtil}$
as the distance estimates for our $\Apxball$ data structure.
\end{enumerate}
\begin{algorithm}
\SetAlgoSkip{}
	\DontPrintSemicolon
    \SetKw{KwAnd}{and}
    \SetProcNameSty{textsc}
    \SetFuncSty{textsc}
	\SetKwProg{procedure}{Procedure}{}{}
 	
	\procedure{$\textsc{Init}(\Htil)$}{
	    \lForEach{$u\in\Vtil$}{$\dtil(u)\gets\dist_{\Htil}(s,u)$.}
    }

    \procedure{$\textsc{WeightIncrease}(\Htil, (u,v))$}{
        $\updatelevel(u)$ and $\updatelevel(v)$.\label{enu:MES:update adversary}
    }
    
    \procedure{$\protect\updatelevel(u)$}{
        \If{ $\min_{v}\{\dtil(v)+\wtil(v,u)\}>\dtil(u)$}{
            $\dtil(u)\gets\min_{v}\{\dtil(v)+\wtil(v,u)\}$.\;
            \lIf{ $\dtil(u)>2D$}{ $\dtil(u)\gets\infty$. }
             $\updatelevel(v)$ for all neighbors $v$ of $u$.\label{enu:MES:update propagate}
        }
    }

\caption{$\protect\MES(\protect\Htil,s,D)$\label{alg:MES}}
\end{algorithm}

For every vertex $u\in\Vtil\setminus\{s\}$, we let $\arg\min_{v}\{\dtil(v)+\wtil(v,u)\}$
be \emph{$u$'s parent}. The set of edges between parents and children
form a tree $\Ttil$ rooted at $s$ is called the \emph{MES tree}.
In the analysis below, we do not need not the tree $\Ttil$ itself.
However, the tree $\Ttil$ will be used later for our data structure
that can report a path in \Cref{sec:ball_path}.

\subsection{Analysis of MES}

In this section, we analyze the running time of \Cref{alg:MES} and
the accuracy of the estimates $\{\dtil(u)\}_{u}$ maintained by the
MES data structure. Although the analysis is quite technical, it follows
the same template as shown by previous works that employ the MES data
structure (e.g. \cite{henzinger2014decremental,henzinger2015unifying, bernstein2016deterministic, bernstein2017deterministic, bernstein2017deterministicweighted, gutenberg2020deterministic}). 

\subsubsection{Total Update Time}

Using the standard analysis of the classic ES tree, we can bound the
total update time.

\begin{lem}\label{lem:time MES}
The total update time of $\MES(\Htil,s,D)$ is $\Otil(\left|\ball_{G}(S,D)\right|\Delta\frac{D}{\eps d})$.
\end{lem}

\begin{proof}
The initialization takes $\Otil(|E(\Htil)|)=\Otil(\left|\ball_{G}(S,D)\right|\Delta)$
by running time Dijkstra's algorithm. Each vertex $u\in\Vtil$ maintains
$\min_{v}\{\dtil(v)+\wtil(v,u)\}$ using heaps.

The algorithm calls $\updatelevel$ because of the direct edge updates
to $\Htil$ at most $U=O(\left|\ball_{G}(S,D)\right|\Delta\frac{D}{\eps d})$
time by \Cref{lem:Htil update}. Each call to $\updatelevel$ takes
only $O(1)$ time for checking the condition. Otherwise, if the algorithm
spends more time, then an estimate $\dtil(u)$ must increase. Once
$\dtil(u)$ is increased, when we spend additional $O(\deg_{\Htil}(u)\log n)$
time to update the heaps, and invoke $\updatelevel$ $\deg_{\Htil}(u)$
more times. We charge the cost for updating these heaps and the cost
for checking the condition in each call to $\updatelevel$ to the
increase of $\dtil(u)$. This charging scheme works because $\dtil(u)$
can be increased at most $O(D/\eps d)$ times. Indeed, $\dtil(u)$
is a multiple of $\eps d$ by \Cref{prop:Htil weight}(\ref{enu:Htil:weight:multiple})
and we set $\dtil(u)\gets\infty$ whenever $\dtil(u)>2D$. 

Therefore, the algorithms calls $\updatelevel$ at most $U+O(|E(\Htil)|\log n\cdot\frac{D}{\eps d})=\Otil(\left|\ball_{G}(S,D)\right|\Delta\frac{D}{\eps d})$
times, and the additional time spent when estimates are increased
is at most $O(|E(\Htil)|\log n\cdot\frac{D}{\eps d})=\Otil(\left|\ball_{G}(S,D)\right|\Delta\frac{D}{\eps d})$
time. This concludes the claim.
\end{proof}

\subsubsection{Dynamics of Distance Estimates}

In this section, we show basic properties of the distance estimates
$\{\dtil(u)\}_{u\in\Vtil}$ maintained by the MES data structure.
The analysis is genetic and so we hope that it might be useful for
future use of the MES data structure. We only need that, at each time,
all insertions to $\Htil$ are handled before other updates. The notion
of \emph{stretched} vertex will be useful here and for proving the
accuracy of the estimates later.
\begin{defn}
[Stretched Vertices]For any $u\in\Vtil\setminus\{s\}$, we say that

$u$ is \emph{stretched} if $\dtil(u)>\min_{v}\{\dtil(v)+\wtil(v,u)\}$.
If $u$ is stretched , every edge $(v,u)$ where $\dtil(u)>\dtil(v)+\wtil(v,u)$
is \emph{stretched}.
\end{defn}

Each edge deletion in $G$ generates several updates to $\Htil$.
We use the phrase ``after time $t$'' to refer to the time when
the algorithm finishes processing the $t$-th edge deletion to $G$
and all other updates to $\Htil$ generated by that deletion. Let
$\dtil_{t}(u)$ denote the distance estimate $\dtil(u)$ after time
$t$. Similarly, let $\wtil_{t}(e)$ denote the weight $\wtil(e)$
after time $t$.

The intuition of \Cref{lem:MES basic} below is that, the estimates
of non-stretched vertices ``behave'' like distances, i.e.~$\dtil_{t}(u)=\min_{v}\{\dtil_{t}(v)+\wtil_{t}(v,u)\}$.
For stretched vertices, although this is not true, their estimates
do not increase which will be helpful for proving that we never overestimate
the distances.
\begin{lem}
\label{lem:MES basic}For each vertex $u\in\Vtil\setminus\{s\}$,
we have the following: 
\begin{enumerate}
\item $\dtil_{0}(u)=\min_{v}\{\dtil_{0}(v)+\wtil_{0}(v,u)\}$.
\item \label{enu:MES:montone}$\dtil(u)$ only increases through time.
\item \label{enu:MES:lower}$\dtil_{t}(u)\ge\min_{v}\{\dtil_{t}(v)+\wtil_{t}(v,u)\}$.
\item \label{enu:MES:not stretched}If $u$ is not stretched after time
$t$, then $\dtil_{t}(u)=\min_{v}\{\dtil_{t}(v)+\wtil_{t}(v,u)\}$.
\item \label{enu:MES:stretched}If $u$ is stretched after time $t$ and
$\min_{v}\{\dtil_{t}(v)+\wtil_{t}(v,u)\}\le2D$, then $\dtil_{t}(u)=\dtil_{t-1}(u)$.
\end{enumerate}
\end{lem}

\begin{proof}
(1): At the initialization, we set $\dtil(u)=\dist_{\Htil}(s,u)$
for all $u\in\Vtil$. As $\dist_{\Htil}(s,u)=\min_{v}\{\dist_{\Htil}(s,v)+\wtil(v,u)\}$,
so $\dtil(u)=\min_{v}\{\dtil(v)+\wtil(v,u)\}$ after time $0$.

(2): $\dtil(u)$ is updated only through $\updatelevel$, which only
increases $\dtil(u)$.

(3): We say that $u$ is loose if $\dtil(u)<\min_{v}\{\dtil(v)+\wtil(v,u)\}$.
Initially, no vertex is loose by (1). At any moment, $u$ has a chance
of being loose only if, for some neighbor $v$ of $u$, $\dtil(v)$
or $\wtil(v,u)$ is increased. If this event happens, then $\updatelevel(u)$
is called by \Cref{enu:MES:update adversary,enu:MES:update propagate}
of \Cref{alg:MES}. If $u$ is indeed loose, then we set $\dtil(u)\gets\min_{v}\{\dtil(v)+\wtil(v,u)\}$
which makes $u$ not loose. Therefore, no vertex is loose after time
$t$, which implies the claim.

(4): We have $\dtil_{t}(u)\le\min_{v}\{\dtil_{t}(v)+\wtil_{t}(v,u)\}$
as $u$ is not stretched after time $t$. By combining with (3), we
are done.

(5): Let $(\bar{v},u)$ be the stretched edge after time $t$, i.e.~$\dtil_{t}(u)>\dtil_{t}(\bar{v})+\wtil_{t}(v',u)$.
Suppose for contradiction $\dtil(u)$ increases when the $t$-th edge
deletion is processed. Consider the last call to $\updatelevel(u)$
that $\dtil(u)$ is increased. Let $\dtil'(\cdot)$ and $\wtil'(\cdot)$
denote $\dtil(\cdot)$ and $\wtil(\cdot)$ at the moment when the
algorithm increases $\dtil(u)$, i.e.~when we set $\dtil'(u)=\min_{v}\{\dtil'(v)+\wtil'(v,u)\}$. 

Note that $\dtil'(\bar{v})\le\dtil_{t}(\bar{v})$ by (2). Also, $\wtil'(\bar{v},u)\le\wtil_{t}(\bar{v},u)$
because, for each time $t$, the algorithm processes all insertions
to $\Htil$ before any other updates to $\Htil$ and hence before
any call to $\updatelevel$. The remaining updates to $\Htil$ may
only increase the weight $\wtil(e)$ by \Cref{prop:Htil weight}(\ref{enu:Htil:weight:increase}).
So $\dtil'(\bar{v})+\wtil'(\bar{v},u)\le\dtil_{t}(v)+w_{t}(v,u)\le2D$.
Hence, $\dtil'(u)\le3D$ and $\dtil'(u)$ not set to $\infty$. So
we have $\dtil'(u)\le\dtil_{t}(v)+\wtil_{t}(v,u)$. As this last moment
$\dtil(u)$ is increased when the $t$-th update is processed, we
have $\dtil_{t}(u)=\dtil'(u)\le\dtil_{t}(v)+\wtil_{t}(v,u)$, which
contradicts the fact that $(v',u)$ is stretched.
\end{proof}

\subsubsection{Lower Bounds of Estimates}

In this section, we show that the estimates $\{\dtil(u)\}_{u\in\Vtil}$
are lower bounded by distances in $G$. We will prove by induction.
The proposition below handles the base case.
\begin{prop}
\label{prop:zero estimate}For any $t$, $\dtil_{t}(u)=0$ if and
only if $u\in\stage St$. 
\end{prop}

\begin{proof}
By \Cref{prop:Htil weight}(\ref{enu:Htil:weight:zero}), we have $u\in\stage S0$
iff $\dtil_{0}(u)=\dist_{\Htil}(s,u)=0$. Note that $S$ is a decremental
set. As long as $u\in S$, $\dtil(u)$ never increases otherwise $0=\dtil(s)+\wtil(s,u)>\dtil(u)$
at some point of time, which is impossible as $\dtil(u)$ never decreases
by \Cref{lem:MES basic}(\ref{enu:MES:montone}). Whenever $u$ leaves
$S$ (i.e.~$(s,u)$ is deleted from $\Htil$), then $\updatelevel(u)$
is called. As all edges incident of $\Htil$ to $u$ have positive
weight, $\dtil(u)$ will be increased and $\dtil(u)>0$ from then
forever by \Cref{lem:MES basic}(\ref{enu:MES:montone}).
\end{proof}
In \Cref{lem:MES:approx:lower} below, we prove the inductive step
on $u$ simply by applying induction hypothesis on the parent of $u$
in the MES tree. We need to lower bound the estimate of core vertices
as well (although we do not need them at the end) so that the induction
hypothesis is strong enough. 
\begin{lem}
\label{lem:MES:approx:lower}For each vertex $u\in\Vtil\setminus\{s\}$,
after time $t$, we have the following:
\begin{enumerate}
\item \label{lem:MES:approx:lower:regular} If $u$ is a core vertex corresponding to a core $C$, then $\dtil_{t}(u)\ge\dist_{\stage Gt}(\stage St,\stage Ct)$.
\item If $u$ is a regular vertex, then $\dtil_{t}(u)\ge\dist_{\stage Gt}(\stage St,u)$.
\end{enumerate}
\end{lem}

\begin{proof}
We prove by induction on $\dtil_{t}(u)$. The base case where $\dtil_{t}(u)=0$
is done by \Cref{prop:zero estimate}. It remains to consider $u\in\Vtil\setminus(\stage St\cup\{s\})$
where $\dtil_{t}(u)<\infty$. Let $v_{p}=\arg\min_{v}\{\dtil(v)+\wtil(v,u)\}$
be the parent of $u$. We have $\dtil_{t}(u)\ge\dtil_{t}(v_{p})+\wtil_{t}(v_{p},u)$
by \Cref{lem:MES basic}(\ref{enu:MES:lower}). As $\wtil_{t}(v_{p},u)>0$
by \Cref{prop:Htil weight}(\ref{enu:Htil:weight:zero}), we can lower
bound $\dtil_{t}(v_{p})$ by induction hypothesis. 

There are two main cases. If $u$ is a core vertex, then $v_{p}$
is a regular vertex by \Cref{prop:Htil vertex}(\ref{enu:Htil:vertex:core})
and since the dummy source $s$ is not incident to core vertices.
So we have
\begin{align*}
\dtil_{t}(u) & \ge\dtil_{t}(v_{p})+\wtil_{t}(v_{p},u)\\
 & \ge\dist_{\stage Gt}(\stage St,v_{p})+\roundup{\stretch\cdot d_{\Clevel(C)}+\dist_{G}(\stage Ct,v_{p})}{\eps d}\\
 & \ge\dist_{\stage Gt}(\stage St,\stage Ct)
\end{align*}
where the second inequality is by induction hypothesis and by the
edge weight of the covering-compressed graph assigned in \Cref{def:core compress}.

Now, suppose that $u$ is a regular vertex. We have three more sub-cases
because $v_{p}$ can either be a core vertex, a regular vertex, or
a dummy source vertex $s$. If $v_{p}$ is a core vertex corresponding
to a core $C_{p}$, then 
\begin{align*}
\dtil_{t}(u) & \ge\dtil_{t}(v_{p})+\wtil_{t}(v_{p},u)\\
 & \ge\dist_{\stage Gt}(\stage St,\stage{C_{p}}t)+\roundup{\stretch\cdot d_{\Clevel(C_{p})}+\dist_{G}(\stage{C_{p}}t,u)}{\eps d}\\
 & \ge\dist_{\stage Gt}(\stage St,\stage{C_{p}}t)+\diam_{\stage Gt}(\stage{C_{p}}t)+\dist_{G}(\stage{C_{p}}t,u)\\
 & \ge\dist_{\stage Gt}(\stage St,u).
\end{align*}
where the second inequality follows by the same reason as in the previous
case, and $\diam_{\stage Gt}(\stage{C_{p}}t)\le\stretch\cdot d_{\Clevel(C_{p})}$
is guaranteed by \Cref{def:Covering}. If $v_{p}$ is a regular vertex,
then we have
\[
\dtil_{t}(u)\ge\dtil_{t}(v_{p})+\wtil_{t}(v_{p},u)\ge\dist_{\stage Gt}(\stage St,v_{p})+\ceiling{w(v_{p},u)}_{\eps d}\ge\dist_{G}(\stage St,u)
\]
where second inequality is by induction hypothesis and $\wtil_{t}(v_{p},u)=\ceiling{w(v_{p},u)}_{\eps d}$
by construction of $\Htil$. Lastly, if $v_{p}=s$, then $\dtil_{t}(u)\ge\dtil_{t}(s)+\wtil_{t}(s,u)=\ceiling{\dtil_{t}^{\near}(u)}_{\eps d}\ge\dist_{G}(\stage St,u)$
because $\dtil_{t}^{\near}(u)\ge\dist_{G}(\stage St,u)$ by the guarantee
of $\Apxball(G,S,2(\frac{\stretch}{\eps})^{k}d,\eps)$.
\end{proof}

\subsubsection{Upper Bounds of Estimates}

In this section, we show that the estimates $\{\dtil(u)\}_{u\in\Vtil}$
are upper bounded by distances in $G$ within small approximation
factor. This section highly exploits the structure of $\Htil$ described
in \Cref{def:Htil}.

\begin{lem}
\label{lem:MES:approx:upper}For each vertex $u\in\Vtil\setminus\{s\}$,
after time $t$, we have the following:
\begin{enumerate}
\item \label{lem:MES:approx:upper:regular} If $u$ is a regular vertex where $\dist_{\stage Gt}(\stage St,u)\le D$, then 
\begin{equation}
\dtil_{t}(u)\le\min\left\{ \distInduc_{\stage Gt}(\stage St,u),\min_{(v,u)\in E(\stage Gt)\cap E(\stage{\Htil}t)}\left\{ \distInduc_{\stage Gt}(\stage St,v)+\wtil_{t}(v,u)\right\} \right\} \label{eq:MES bound regular}
\end{equation}
where we define $\distInduc_{\stage Gt}(\stage St,v)=\max\{\ceiling{(1+\eps)\dist_{\stage Gt}(\stage St,v)}_{\eps d},(1+50\eps)\dist_{\stage Gt}(\stage St,v)\}$.
\item If $u$ is a core vertex corresponding to a core $C$ where $(\frac{\stretch}{\eps})^{k}d<\dist_{\stage Gt}(\stage St,\stage Ct)\le D$,
then 
\begin{equation}
\dtil_{t}(u)\le(1+50\eps)\dist_{\stage Gt}(\stage St,\stage Ct)-2\stretch\cdot d_{\Clevel(C)}.\label{eq:MES bound core}
\end{equation}
\end{enumerate}
\end{lem}

\begin{proof}
For any time $t$ and any $u\in\Vtil$, we define
\[
d'_{t}(u)=\begin{cases}
\dist_{\stage Gt}(\stage St,u) & \text{if }u\text{ is a regular vertex}\\
\dist_{\stage Gt}(\stage St,\shell(\stage Ct)) & \text{if }u\text{ is a core vertex corresponds to a core }C
\end{cases}
\]
Let \emph{$d'_{t}$-order} refer to an increasing order of vertices
in $\Vtil$ according to $d'_{t}(u)$. If $d'_{t}(u)=d'_{t}(v)$ for
some regular vertex $u$ and some core vertex $v$, we let $u$ precede
$v$ in this order. We will prove the claim by induction on $t$ and
then on the $d'_{t}$-order of vertices in $\Vtil$\emph{.} 

Our strategy is to first bound $\min_{v}\{\dtil_{t}(v)+\wtil_{t}(v,u)\}$
instead of $\dtil_{t}(u)$. More formally, we will show that for regular vertices $u$ where $\dist_{\stage Gt}(\stage St,u)\le D$,
\begin{equation}
\min_{v}\{\dtil_{t}(v)+\wtil_{t}(v,u)\}\le\min\left\{ \distInduc_{\stage Gt}(\stage St,u),\min_{(v,u)\in E(\stage Gt)\cap E(\stage{\Htil}t)}\left\{ \distInduc_{\stage Gt}(\stage St,v)+\wtil_{t}(v,u)\right\} \right\} \label{eq:MES hook regular}
\end{equation}
and for core vertices $u$ corresponding to a core $C$ where $(\frac{\stretch}{\eps})^{k}d<\dist_{\stage Gt}(\stage St,\stage Ct)\le D$,
\begin{equation}
\min_{v}\{\dtil_{t}(v)+\wtil_{t}(v,u)\}\le(1+50\eps)\dist_{\stage Gt}(\stage St,\stage Ct)-2\stretch\cdot d_{\Clevel(C)}\label{eq:MES hook core}
\end{equation}
Note that, to prove \Cref{eq:MES hook regular}
and \Cref{eq:MES hook core}, we still assume that induction hypothesis
holds for $\dtil_{t}(u)$. Then, we will use \Cref{eq:MES hook regular}
and \Cref{eq:MES hook core} to prove \Cref{eq:MES bound regular} and
\Cref{eq:MES bound core}, respectively.%

\paragraph{Proving \Cref{eq:MES hook regular} for Regular Vertices.}

For any $t\ge0$, we first show that $\min_{v}\{\dtil_{t}(v)+\wtil_{t}(v,u)\}\le\distInduc_{\stage Gt}(\stage St,u)$.
If $\dist_{\stage Gt}(\stage St,u)\le2(\frac{\stretch}{\eps})^{k}d$,
then $(s,u)\in E(\Htil)$ and so 
\begin{align*}
\min_{v}\{\dtil_{t}(v)+\wtil_{t}(v,u)\} & \le\dtil_{t}(s)+\wtil_{t}(s,u)\\
 & =0+\ceiling{\dtil_{t}^{\near}(u)}_{\eps d} & \text{by construction of }\Htil\\
 & \le\ceiling{(1+\eps)\dist_{\stage{G}{t}}(\stage St,u)}_{\eps d} & \text{by }\Apxball(G,S,2(\frac{\stretch}{\eps})^{k}d,\eps)\\
 & \le\distInduc_{\stage Gt}(\stage St,u).
\end{align*}
So from now, we assume that $\dist_{\stage Gt}(\stage St,u)>2(\frac{\stretch}{\eps})^{k}d$.
The covering guarantees that there exists a level-$\ell$ core $C$
where $u\in\cover(C)$ for some $\ell\in[0,k)$. Let $v_{C}\in\Vtil$
denote the core vertex corresponding to $C$. Consider an $\stage St$-$u$
shortest path $P=(v_{1},\dots,v_{z})$ in $G$ where $v_{1}\in\stage St$
and $u=v_{z}$. There are two sub-cases whether $v_{z-1}\in\shell(C)$
or not. 
\begin{enumerate}
\item Suppose that $v_{z-1}\in\shell(C)$. Then, we can apply induction
hypothesis on $v_{C}$ because 
\[
d'_{t}(v_{C})=\dist_{\stage Gt}(\stage St,\shell(\stage Ct))<\dist_{\stage Gt}(\stage St,v_{z})=d'_{t}(u).
\]
and 
\begin{align*}
\dist_{\stage Gt}(\stage St,\stage Ct) & \ge\dist_{\stage Gt}(\stage St,u)-\dist_{\stage Gt}(\stage Ct,u)\\
 & >2(\frac{\stretch}{\eps})^{k}d-4d_{\ell}(1+\eps)>(\frac{\stretch}{\eps})^{k}d.
\end{align*}
where $\dist_{\stage Gt}(\stage Ct,u)\le4d_{\ell}(1+\eps)$ because
$u\in\cover(C)$ and $d_{\ell}\le(\frac{\stretch}{\eps})^{k-1}d$.
So, after applying induction hypothesis, we have $\dtil_{t}(v_{C})\le(1+50\eps)\dist_{\stage Gt}(\stage St,\stage Ct)-2\stretch\cdot d_{\ell}$.
By the definition of the covering-compressed graph from \Cref{def:core compress},
we have $\wtil_{t}(v_{C},u)=\roundup{\stretch\cdot d_{\ell}+\dtil_{t}^{C}(u)}{\eps d}$
where $\dtil_{t}^{C}(u)\le(1+\eps)\dist_{\stage Gt}(\stage Ct,u)$
is maintained by $\Apxball(G,C,\frac{\stretch}{4\eps}d_{\ell},\eps)$ that maintains $\shell(C)$.
We conclude by that
\begin{align*}
 & \min_{v}\{\dtil_{t}(v)+\wtil_{t}(v,u)\} \\
 & \le\dtil_{t}(v_{C})+\wtil_{t}(v_{C},u) & \text{}\\
 & \le(1+50\eps)\dist_{\stage Gt}(\stage St,\stage Ct)-2\stretch\cdot d_{\ell}+\roundup{\stretch\cdot d_{\ell}+(1+\eps)\dist_{\stage Gt}(\stage Ct,u)}{\eps d}\\
 & \le(1+50\eps)\dist_{\stage Gt}(\stage St,\stage Ct)+(1+\eps)\dist_{\stage Gt}(\stage Ct,u)+(\stretch\cdot d_{\ell}+\eps d-2\stretch\cdot d_{\ell})\\
 & \le(1+50\eps)\dist_{\stage Gt}(\stage St,u)\\
 & \le\distInduc_{\stage Gt}(\stage St,u).
\end{align*}
\item Suppose that $v_{z-1}\notin\shell(C)$.
We have $\dtil_{t}(v_{z-1})\le\distInduc_{\stage Gt}(\stage St,v_{z-1})$
by induction hypothesis. Also, note that $w_{t}(v_{z-1},v_{z})\ge\dist_{\stage Gt}(\stage St,v_{z-1})-\dist_{\stage Gt}(\stage St,v_{z})>\frac{\stretch}{4\eps}d_{\ell}-4d_{\ell}\cdot(1+\eps)>d$
because $v_{z-1}\notin\shell(C)$ but $v_{z}\in\cover(C)$. So, by
construction of $\Htil$, we have $(v_{z-1},v_{z})\in E(G)\cap E(\Htil)$
with weight $\wtil_{t}(v_{z-1},v_{z})=\ceiling{w_{t}(v_{z-1},v_{z})}_{\eps d}\le w_{t}(v_{z-1},v_{z})+\eps d<(1+\eps)w_{t}(v_{z-1},v_{z}).$
We conclude by \Cref{lem:MES basic}(\ref{enu:MES:not stretched})
that 
\begin{align*}
\min_{v}\{\dtil_{t}(v)+\wtil_{t}(v,u)\} & \le\dtil_{t}(v_{z-1})+\wtil_{t}(v_{z-1},v_{z})\\
 & \le\distInduc_{\stage Gt}(\stage St,v_{z-1})+(1+\eps)w_{t}(v_{z-1},v_{z})\\
 & \le\max\left\{ (1+\eps)\dist_{\stage Gt}(\stage St,u)+\eps d,(1+50\eps)\dist_{\stage Gt}(\stage St,u)\right\} \\
 & =(1+50\eps)\dist_{\stage Gt}(\stage St,u)=\distInduc_{\stage Gt}(\stage St,u).
\end{align*}
where the last line is because $\dist_{\stage Gt}(\stage St,u)>2(\frac{\stretch}{\eps})^{k}d$.
\end{enumerate}
In both cases, we have $\min_{v}\{\dtil_{t}(v)+\wtil_{t}(v,u)\}\le\distInduc_{\stage Gt}(\stage St,u)$ as desired. 

Lastly, we also need to show that $$\min_{v}\{\dtil_{t}(v)+\wtil_{t}(v,u)\}\le\min_{(v,u)\in E(\stage Gt)\cap E(\stage{\Htil}t)}\left\{ \distInduc_{\stage Gt}(\stage St,v)+\wtil_{t}(v,u)\right\}.$$
Consider any $(v,u)\in E(\stage Gt)\cap E(\stage{\Htil}t)$. If $\dist_{\stage Gt}(\stage St,v)\ge\dist_{\stage Gt}(\stage St,u)$,
then, trivially, we have that $\min_{v}\{\dtil_{t}(v)+\wtil_{t}(v,u)\}\le\distInduc_{\stage Gt}(\stage St,u)\le\distInduc_{\stage Gt}(\stage St,v)+\wtil_{t}(v,u)$.
Otherwise, if $\dist_{\stage Gt}(\stage St,v)<\dist_{\stage Gt}(\stage St,u)$,
then, by applying induction hypothesis on $v$, we again have $\min_{v}\{\dtil_{t}(v)+\wtil_{t}(v,u)\}\le\distInduc_{\stage Gt}(\stage St,v)+\wtil_{t}(v,u)$.

\paragraph{Proving \Cref{eq:MES hook core} for Core Vertices.}

Suppose $u$ is a core vertex corresponding to a level-$\ell$ core
$C\in\cC$. As $\shell(C)=\Apxball(G,C,\frac{\stretch}{4\eps}d_{\ell},\eps)$
and $\dist_{\stage Gt}(\stage St,\stage Ct)>(\frac{\stretch}{\eps})^{k}d>\frac{\stretch}{4\eps}d_{\ell}\cdot(1+\eps)$
for all $\ell<k$, we have $\shell(\stage Ct)\cap\stage St=\emptyset$.
Consider the $\stage St$-$\stage Ct$ shortest path $P=(v_{1},\dots,v_{z})$
in $G$ where $v_{1}\in\stage St\setminus\shell(\stage Ct)$ to $v_{z}\in\stage Ct$.
Let $i$ be the first index that $v_{i}\in\shell(\stage Ct)$. Note
that $i>1$. Note that we can apply induction hypothesis on $v_{i}$
because $d'_{t}(v_{i})=\dist_{\stage Gt}(\stage St,v_{i})=\dist_{\stage Gt}(\stage St,\shell(\stage Ct))=d'_{t}(u)$
and because $v_{i}$ is a regular vertex and $u$ is a core vertex.
There are two cases:
\begin{enumerate}
\item If $w_{t}(v_{i-1},v_{i})<\frac{\stretch}{10\eps}d_{\ell}$, then we
have 
\begin{align*}
\dist_{\stage Gt}(v_{i},\stage Ct) & \ge\dist_{\stage Gt}(\stage Ct,v_{i-1})-w_{t}(v_{i-1},v_{i})\\
 & >\frac{\stretch}{4\eps}d_{\ell}-\frac{\stretch}{10\eps}d_{\ell}>\frac{\stretch}{8\eps}d_{\ell} & \text{as }v_{i-1}\notin\shell(\stage Ct).
\end{align*}
We conclude 
\begin{align*}
	& \min_{v}\{\dtil_{t}(v)+\wtil_{t}(v,u)\}\\
	& \le\dtil_{t}(v_{i})+\wtil_{t}(v_{i},u)\\
	& \le\distInduc_{\stage Gt}(\stage St,v_{i})+\roundup{\stretch\cdot d_{\ell}+(1+\eps)\dist_{\stage Gt}(v_{i},\stage Ct)}{\eps d} & \text{by IH on }v_{i}\\
	& \le\ceiling{(1+50\eps)\dist_{\stage Gt}(\stage St,v_{i})}_{\eps d}+\roundup{\stretch\cdot d_{\ell}+(1+\eps)\dist_{\stage Gt}(v_{i},\stage Ct)}{\eps d} & \text{by definition of }\distInduc\\
	& \le(1+50\eps)\dist_{\stage Gt}(\stage St,v_{i})+(1+50\eps)\dist_{\stage Gt}(v_{i},\stage Ct)\\
	& \,\,\,\,\,\,-49\eps\dist_{\stage Gt}(v_{i},\stage Ct)+\stretch\cdot d_{\ell}+2\eps d\\
	& <(1+50\eps)\dist_{\stage Gt}(\stage St,\stage Ct)-49\eps\cdot\frac{\stretch}{8\eps}\cdot d_{\ell}+\stretch\cdot d_{\ell}+2\eps d\\
	& \le(1+50\eps)\dist_{\stage Gt}(\stage St,\stage Ct)-2\stretch\cdot d_{\ell}
\end{align*}
\item If $w_{t}(v_{i-1},v_{i})\ge\frac{\stretch}{10\eps}d_{\ell}$%
,
then $(v_{i-1},v_{i})\in E(\stage Gt)\cap E(\stage{\Htil}t)$ and
$\wtil_{t}(v_{i-1},v_{i})=\ceiling{w_{t}(v_{i-1},v_{i})}_{\eps d}$.
We have 
\begin{align*}
	\dtil_{t}(v_{i}) & \le\distInduc_{\stage Gt}(\stage St,v_{i-1})+\wtil_{t}(v_{i-1},v_{i}) & \text{by IH on }v_{i}\\
	& \le\ceiling{(1+50\eps)\dist_{\stage Gt}(\stage St,v_{i-1})}_{\eps d}+\ceiling{w_{t}(v_{i-1},v_{i})}_{\eps d} & \text{by definition of }\distInduc\\
	& \le(1+50\eps)\dist_{\stage Gt}(\stage St,v_{i-1})+(1+50\eps)w_{t}(v_{i-1},v_{i})\\
	& \,\,\,\,\,\,-50\eps\cdot w_{t}(v_{i-1},v_{i})+2\eps d\\
	& \le(1+50\eps)\dist_{\stage Gt}(\stage St,v_{i})-5\stretch\cdot d_{\ell}+2\eps d
\end{align*}
We conclude 
\begin{align*}
 & \min_{v}\{\dtil_{t}(v)+\wtil_{t}(v,u)\} \\
 & \le\dtil_{t}(v_{i})+\wtil_{t}(v_{i},u)\\
 & \le(1+50\eps)\dist_{\stage Gt}(\stage St,v_{i})-5\stretch\cdot d_{\ell}+2\eps d+\roundup{\stretch\cdot d_{\ell}+(1+\eps)\dist_{\stage Gt}(v_{i},\stage Ct)}{\eps d}\\
 & \le(1+50\eps)\dist_{\stage Gt}(\stage St,\stage Ct)-5\stretch\cdot d_{\ell}+2\eps d+\stretch\cdot d_{\ell}+\eps d\\
 & \le(1+50\eps)\dist_{\stage Gt}(\stage St,\stage Ct)-2\stretch\cdot d_{\ell}.
\end{align*}
where the second inequality is by the definition of covering-compressed
graph from \Cref{def:core compress} which says that $\wtil_{t}(v_{i},u)=\roundup{\stretch\cdot d_{\ell}+\dtil_{t}^{C}(v_{i})}{\eps d}$
where $\dtil_{t}^{C}(v_{i})\le(1+\eps)\dist_{\stage Gt}(\stage Ct,v_{i})$
is maintained by $\Apxball(G,C,\frac{\stretch}{4\eps}d_{\ell},\eps)$. 
\end{enumerate}
In both cases, we have shown that $\min_{v}\{\dtil_{t}(v)+\wtil_{t}(v,u)\}\le(1+50\eps)\dist_{\stage Gt}(\stage St,\stage Ct)-2\stretch\cdot d_{\ell}$
as desired.

\paragraph{Bounding $\protect\dtil_{t}(u)$.}

If $u$ is not a stretched vertex after time $t$, then by \Cref{lem:MES basic}(\Cref{enu:MES:not stretched})
$\dtil_{t}(u)\le\min_{v}\{\dtil_{t}(v)+\wtil_{t}(v,u)\}$. Therefore,
\Cref{eq:MES bound regular} and \Cref{eq:MES bound core} follow immediately
from \Cref{eq:MES hook regular} and \Cref{eq:MES hook core}, respectively.

Now, suppose $u$ is a stretched vertex after time $t$. \Cref{eq:MES hook regular}
and \Cref{eq:MES hook core} imply that $\min_{v}\{\dtil_{t}(v)+\wtil_{t}(v,u)\}\le2D$
because we assume $\dist_{\stage Gt}(\stage St,u)\le D$ if $u$ is
a regular vertex and $\dist_{\stage Gt}(\stage St,\stage Ct)\le D$
if $u$ is a core vertex. So by \Cref{lem:MES basic}(\Cref{enu:MES:stretched}),
we have $\dtil_{t}(u)=\dtil_{t-1}(u)$. 

So, to prove that \Cref{eq:MES hook regular} and \Cref{eq:MES hook core}
hold at time $t$, it is enough to prove that the right hand side
of both \Cref{eq:MES hook regular} and \Cref{eq:MES hook core} do
not decrease from time $t-1$ to $t$. This is true for \Cref{eq:MES hook regular}
because, for every edge $e\in E(G)\cap E(\Htil)$, $\wtil(e)=\ceiling{w(e)}_{\eps d}$
and the edge weight $w(e)$ in $G$ never decrease. Also, $G$ and
$S$ are decremental and $\dist_{G}(S,u)$ never decreases, which
in turn means that $\distInduc_{G}(S,u)$ never decreases. This is
true for \Cref{eq:MES hook core} because $C$ is also a decremental
set, and so $\dist_{G}(S,C)$ never decreases. This completes the
proof.
\end{proof}

\subsection[Proof of Theorem for Approximate Ball]{Proof of \Cref{thm:Apxball}}
\label{sec:ball conclude}

Finally, we conclude the proof of \Cref{thm:Apxball} by showing that
the all requirements of \\$\Apxball(G,S,D,50\eps)$ are satisfied and
then analyzing the running time of the algorithm.

\paragraph{Correctness.}

Recall that $\left\{ \min\{\dtil^{\near}(u),\dtil(u)\}\right\} _{u\in\Vtil}$
are the estimates maintained by the algorithm. For any vertex $u\in V\setminus\Vtil$,
we implicitly set $\dtil(u)\gets\infty$ and do not spend any time
maintaining it. This is justified because $\Vtil$ contains $\ball_{\stage G0}(\stage S0,D)$
and $\ball_{G}(S,D)$ is a decremental set, so $u\notin\ball_{\stage Gt}(\stage St,D)$
for any $t$.

Now, for each $u\in V\cap\Vtil$, by the guarantee of $\Apxball(G,S,2(\frac{\stretch}{\eps})^{k}d,\eps)$
and \Cref{lem:MES:approx:lower}, we first have that $\min\{\dtil^{\near}(u),\dtil(u)\}\ge\dist_{\stage Gt}(\stage St,u)$.
Secondly, if $u\le2(\frac{\stretch}{\eps})^{k}d$, then $\dtil^{\near}(u)\le(1+\eps)\dist_{G}(S,u)$
by $\Apxball(G,S,2(\frac{\stretch}{\eps})^{k}d,\eps)$. Otherwise,
if $2(\frac{\stretch}{\eps})^{k}d<\dist_{G}(S,u)\le D$, then $\dtil(u)\le\distInduc(S,u)=(1+50\eps)\dist_{G}(S,u)$
by \Cref{lem:MES:approx:upper}. So we satisfy Property \ref{enu:Apxball:approx}
of \Cref{def:Apxball}. Thirdly, both $\dtil^{\near}(u)$ and $\dtil(u)$
never decrease by \Cref{lem:MES basic}(\ref{enu:MES:montone}). Therefore,
$\left\{ \min\{\dtil^{\near}(u),\dtil(u)\}\right\} _{u\in V}$ satisfies
all three conditions of \Cref{def:Apxball}.

\paragraph{Total Update Time.}

As the underlying graph $G$ and the covering-compressed graph $H_{\cC}$
are explicitly maintained for us, we can maintain the emulator $\Htil$
in time $U+T_{\Apxball}(G,S,2(\frac{\stretch}{\eps})^{k}d,\eps)$
where $U$ denote the total number of edge updates to $\Htil$. By
\Cref{lem:Htil update}, $U=O(\left|\ball_{G}(S,D)\right|\Delta\frac{D}{\eps d})$.
Next, the total update time for maintaining the MES data structure
is $\Otil(\left|\ball_{G}(S,D)\right|\Delta\frac{D}{\eps d})$ by \Cref{lem:time MES}. Therefore,
in total, our algorithm takes $\Otil(\left|\ball_{G}(S,D)\right|\Delta\frac{D}{\eps d})+T_{\Apxball}(G,S,2(\frac{\stretch}{\eps})^{k}d,\eps)$
time.

%% file: together_distance_new.tex
\section{Putting Distance-Only Components Together}
\label{sec:part2PuttingItTogether}

In this section, we show how all our data structures fit together.
The main data structures were $\Apxball$(\Cref{def:Apxball}), $\Core$(\Cref{def:Core}),
and Covering(\Cref{def:Covering}).

\begin{theorem}
\label{thm:main no distance}For any $n$ and $\eps\in(\phicmg^{3},0.1)$,
let $G=(V,E)$ be a decremental bounded-degree graph with $n$ vertices
and edge weights are from $\{1,2,\dots,W=n^{5}\}$. Let $S\subseteq V$
be any decremental set. We can implement $\Apxball(G,S,\eps)$ that
has $\Ohat(n)$ total update time. 
\end{theorem}

There are \emph{distance scales} $D_{0}\le D_{1}\le\dots\le D_{\distScale}$
where $D_{i}=(nW)^{i/\distScale}$ and $\distScale=c_{\distScale}\lg\lg\lg n$
for some small constant $c_{\distScale}>0$. We will implement our
data structures for $\distScale$ many levels. Recall that $\phicmg=1/2^{\Theta(\log^{3/4})}=\Omegahat(1)$.
For $0\le i\le\distScale$, we set 
\begin{align*}
k_{i} & =(\lg\lg n)^{3^{i}}\\
\gamma_{i} & =1/\phicmg^{8k_{i+1}}\text{ and }\gamma_{-1}=1\\
\eps_{i} & =\eps/50^{\distScale-i}\\
\stretch_{i} & =\gamma_{i-1}\cdot\log^{c_{\stretch}}n/\phicmg^{3}\\
\Delta_{i} & =\Theta(k_{i}n^{2/k_{i}}/\phicmg)
\end{align*}
where we let $c_{\stretch}$ be a large constant to be determined later. The
parameters are defined in such that way that 
\[
n^{1/\distScale},\frac{D_{i}}{D_{i-1}},n^{1/k_{i}},\gamma_{i},1/\eps_{i},\stretch_{i},\Delta_{i}=\Ohat(1)
\]
for all $0\le i\le\distScale$. %
To exploit these parameters, we need
more fine-grained properties which are summarized below: 
\begin{prop}
\label{prop:para}For large enough $n$ and for all $0\le i\le\distScale$,
we have that 
\begin{enumerate}
\item \label{enu:para:k}$\lg\lg n\le k_{i}\le(\lg^{1/100}n)$, 
\item \label{enu:para:eps}$\phicmg^{4}\le\eps_{i}\le\eps$ 
\item \label{enu:para:gamma}$\gamma_{i}=2^{O(\lg^{3/4+1/100}n)}$, 
\item \label{enu:para:D ratio}$D_{i}/D_{i-1}\le n^{6/\distScale}$, 
\item \label{enu:para:gap and distance scale}$\gamma_{i}\le D_{i}/D_{i-1}$,
and 
\item \label{enu:para:key}$\gamma_{i}\ge(\gamma_{i-1}\cdot\frac{1}{\phicmg^{8}})^{k_{i}}\ge(\frac{\stretch_{i}}{\eps_{i}})^{k_{i}}$. 
\end{enumerate}
\end{prop}

\begin{proof}
(1): We have $k_{i}=(\lg\lg n)^{3^{i}}\le(\lg\lg n)^{3^{\distScale}}\leq(\lg\lg n)^{(\lg\lg n)^{1/100}}\le\lg^{1/100}n$
as $\distScale=c_{\distScale}\lg\lg\lg n$ and $c_{\distScale}$ is
a small enough constant.

(2): It is clear that $\eps_{i}\le\eps$. For the other direction,
note that in the assumption of \Cref{thm:main no distance}, we have
$\eps\ge\phicmg^{2}$. So $\eps_{i}\ge\eps/50^{\distScale}\ge\phicmg^{3}/50^{\Theta(\lg\lg\lg n)}\ge\phicmg^{4}$
because $\phicmg=1/2^{\Theta(\lg^{3/4}n)}$.

(3): As $\frac{1}{\phicmg}=2^{\Theta(\lg^{3/4}n)}$and $\gamma_{i}=1/\phicmg^{8k_{i+1}}$,
we have from property \Cref{enu:para:k} of this proposition that $\gamma_{i}=2^{O(\lg^{3/4+1/100}n)}$.

(4): We have $D_{i}/D_{i-1}=(nW)^{1/\distScale}$. Since $W=n^{5}$,
we have $D_{i}/D_{i-1}\le n^{6/\distScale}$.

(5): As $D_{i}/D_{i-1}\ge n^{1/\distScale}\ge2^{\Theta(\lg n/\lg\lg\lg n)}$,
by (\ref{enu:para:gamma}) we have $\gamma_{i}\le(D_{i}/D_{i-1})$ when
$n$ is large enough.

($8$): We have $\gamma_{i}\ge(\gamma_{i-1}\cdot\frac{1}{\phicmg^{8}})^{k_{i}}$
because 
\[
\gamma_{i}=(\frac{1}{\phicmg^{8}})^{k_{i+1}}\ge(\frac{1}{\phicmg^{8}})^{(k_{i}^{2}+k_{i})}=((\frac{1}{\phicmg^{8}})^{k_{i}}\cdot\frac{1}{\phicmg^{8}})^{k_{i}}=(\gamma_{i-1}\cdot\frac{1}{\phicmg^{8}})^{k_{i}}
\]
where the inequality holds is because $k_{i+1}=(\lg\lg n)^{3^{i+1}}=(\lg\lg n)^{3^{i}\cdot2}\times(\lg\lg n)^{3^{i}}\ge(\lg\lg n)^{3^{i}\cdot2}+(\lg\lg n)^{3^{i}}=k_{i}^{2}+k_{i}$
for all $i\ge0$. For the second inequality, we have 
\[
\frac{\stretch_{i}}{\eps_{i}}=\frac{\gamma_{i-1}(\log^{c_{\stretch}}n)/\phicmg^{3}}{\eps_{i}}\le\frac{\gamma_{i-1}}{\phicmg^{8}}
\]
because $(\log^{c_{\stretch}}n)\le1/\phicmg$ and $\eps_{i}\ge\phicmg^{4}$
by (\ref{enu:para:eps}). Therefore, $\gamma_{i}\geq(\gamma_{i-1}\cdot\frac{1}{\phicmg^{8}})^{k_{i}}\ge(\frac{\stretch_{i}}{\eps_{i}})^{k_{i}}$.

\end{proof}

Before we prove our main Lemma by induction, we recall the Figure from the beginning of the section to provide the reader with a high-level overview of how components are connected.

\begin{figure}[!ht]
\begin{center}
\includegraphics[width=0.58\textwidth]{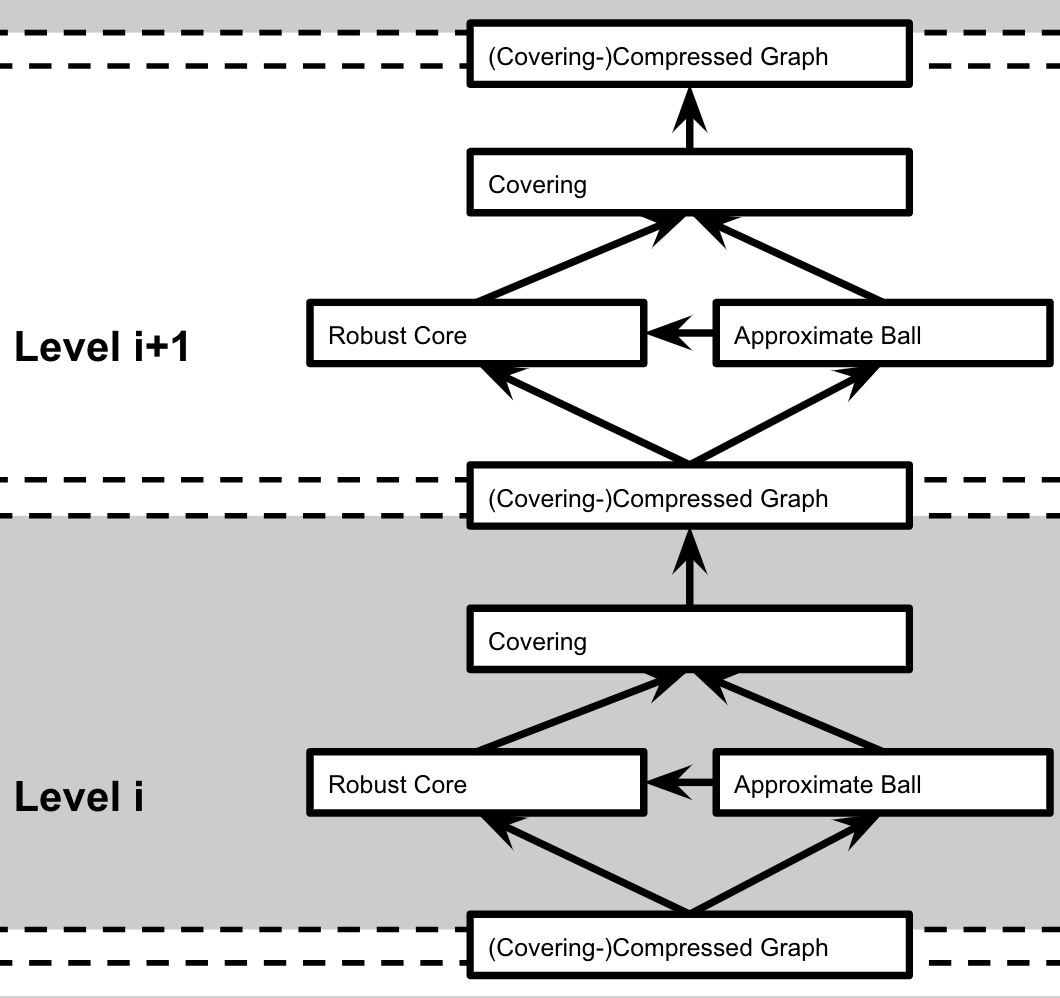}
\end{center}
\caption{An overview of the layers and their dependencies.}
\end{figure}

\begin{lem}
\label{lem:main induction}For every $0\le i\le\distScale$, we can
maintain the following data structures: 
\begin{enumerate}
\item \label{enu:induction ball}$\Apxball(G,S,d',\eps_{i})$ for any $d'\le d'_{i}\defeq32D_{i+1}\log(n)/\eps_{i}$
using total update time of 
\[
\Otil(\left|\ball_{G}(S,d')\right|\frac{n^{2/k_{0}+12/\distScale}}{\phicmg\eps_{0}^{2}})=\Ohat(\left|\ball_{G}(S,d')\right|).
\]
\item $\Core(G,\Kinit,d')$ for any $d'\le D_{i+1}$ using total update
time of 
\[
\left|\ball_{G}(\Kinit,32d'\log n)\right|\poly(\frac{n^{1/k_{0}+1/\distScale}}{\phicmg\eps_{0}})=\Ohat(\left|\ball_{G}(\Kinit,32d'\log n)\right|)
\]
with scattering parameter $\scatter=\Omegatil(\phicmg)$ and stretch
at most $\stretch_{i}$. 
\item $(D_{i},k_{i},\eps_{i},\stretch_{i},\Delta_{i})$-covering using total
update time of $\Otil(n\cdot\poly(\frac{n^{1/k_{0}+1/\distScale}}{\phicmg\eps_{0}}))=\Ohat(n)$. 
\end{enumerate}
For all $i>0$, we assume by induction that a $(D_{i-1},k_{i-1},\eps_{i-1},\stretch_{i-1},\Delta_{i-1})$-covering
of $G$ is already explicitly maintained. 
\end{lem}

\begin{proof}
(1): We prove by induction on $i$ that $T_{\Apxball}(G,S,d',\eps_{i})\le\left|\ball_{G}(S,d')\right|\cdot(i+1)\cdot\frac{n^{2/k_{0}+12/\distScale}}{\phicmg\eps_{0}^{2}}\cdot(\log n)^{c'}$
for any $d'\le d'_{i}$ where $c'$ is some large enough constant.
For $i=0$, we have by \Cref{prop:ES tree} that 
\begin{align*}
T_{\Apxball}(G,S,d',\eps_{i}) & \le O(\left|\ball_{G}(S,d')\right|d')\\
 & \le\left|\ball_{G}(S,d')\right|\cdot O(D_{1}\log(n)/\eps_{1})\\
 & \le\left|\ball_{G}(S,d')\right|\cdot(i+1)\cdot\frac{n^{2/k_{0}+12/\distScale}}{\phicmg\eps_{0}^{2}}\cdot(\log n)^{c'}.
\end{align*}
For $i>0$, we assume $d'>d'_{i-1}\defeq32D_{i}\log(n)/\eps_{i-1}$
otherwise we are done by induction hypothesis. As $(D_{i-1},k_{i-1},\eps_{i-1},\stretch_{i-1},\Delta_{i-1})$-covering
is already explicitly maintained by the induction hypothesis, by \Cref{thm:Apxball}, we can maintain
$\Apxball(G,S,d',\eps_{i})$ where $\eps_{i}=50\eps_{i-1}$ using
total update time of 

\begin{align*}
T_{\Apxball}(G,S,d',\eps_{i})&\le\Otil(\left|\ball_{G}(S,d')\right|\Delta_{i-1}\frac{(32D_{i+1}\log(n)/\eps_{i})}{\eps_{i-1}D_{i-1}})\\
&+T_{\Apxball}\left(G,S,2\left(\frac{\stretch_{i-1}}{\eps_{i-1}}\right)^{k_{i-1}}D_{i-1},\eps_{i-1}\right).
\end{align*}
We will show that $2(\frac{\stretch_{i-1}}{\eps_{i-1}})^{k_{i-1}}D_{i-1}\le d'_{i-1}$
so that we can apply induction hypothesis on \\$T_{\Apxball}(G,S,2(\frac{\stretch_{i-1}}{\eps_{i-1}})^{k_{i-1}}D_{i-1},\eps_{i-1})$.
To see this, note that $D_{i}\ge\gamma_{i}D_{i-1}\ge(\frac{\stretch_{i}}{\eps_{i}})^{k_{i}}D_{i-1}$
by \Cref{prop:para}(\ref{enu:para:gap and distance scale},\ref{enu:para:key}).
So $d'_{i-1}=\frac{32\log n}{\eps_{i-1}}D_{i}\ge\frac{32\log n}{\eps_{i-1}}\cdot(\frac{\stretch_{i}}{\eps_{i}})^{k_{i}}D_{i-1}\ge2(\frac{\stretch_{i-1}}{\eps_{i-1}})^{k_{i-1}}D_{i-1}$where
the last inequality is because $k_{i}\ge k_{i-1}$ and $\frac{\stretch_{i}}{\eps_{i}}\ge\frac{\stretch_{i-1}}{\eps_{i-1}}$
(because $\frac{\stretch_{i}}{\stretch_{i-1}}\ge50=\frac{\eps_{i}}{\eps_{i-1}}$).
Therefore, by \Cref{prop:para}(\ref{enu:para:D ratio}), the bound
on $T_{\Apxball}(G,S,d',\eps_{i})$ is at most 
\begin{align*}
 & \left|\ball_{G}(S,d')\right|\frac{n^{2/k_{i-1}+12/\distScale}}{\phicmg\eps_{i-1}^{2}}\cdot(\log n)^{c'}+T_{\Apxball}(G,S,d'_{i-1},\eps_{i-1})\\
 & \le\left|\ball_{G}(S,d')\right|\frac{n^{2/k_{i-1}+12/\distScale}}{\phicmg\eps_{0}^{2}}\cdot(\log n)^{c'}+i\left|\ball_{G}(S,d'_{i-1})\right|\cdot\frac{n^{2/k_{0}+12/\distScale}}{\phicmg\eps_{0}^{2}}\cdot(\log n)^{c'} & \text{by IH}\\
 & \le(i+1)\left|\ball_{G}(S,d')\right|\frac{n^{2/k_{0}+12/\distScale}}{\phicmg\eps_{0}^{2}}\cdot(\log n)^{c'} & \text{as }d'_{i-1}<d'
\end{align*}
which completes the inductive step.

(2): For $i=0$, we have that a $(1,1,O(1))$-compressed graph of
$G$ can be trivially maintained by \Cref{prop:compressed base}. By
\Cref{thm:Core}, we can implement $\Core(G,\Kinit,d')$ with scattering
parameter $\scatter=\Omegatil(\phicmg)$ and stretch at most $\Otil(1/\phicmg^{3})\le\stretch_{0}$
(by definition of $\stretch_{0}$) with total update time 
\[
\Otil\left(T_{\Apxball}(G,\Kinit,32d'\log n,0.1)(D_{1})^{3}/\phicmg^{2}\right)=\Otil\left(|\ball_{G}(\Kinit,32d'\log n)|D_{1}{}^{4}/\phicmg^{2}\right)
\]
by \Cref{prop:ES tree}.

For $i>0$, given that a $(D_{i-1},k_{i-1},\eps_{i-1},\stretch_{i-1},\Delta_{i-1})$-covering
is explicitly maintained, by \Cref{prop:covering-compressed is compressed},
we can automatically maintain a $(D_{i-1},\gamma_{i-1},\Delta_{i-1})$-compressed
graph where $\gamma_{i-1}\ge(\stretch_{i-1}/\eps_{i-1})^{k_{i-1}}$
by \Cref{prop:para}(\ref{enu:para:key}).

By \Cref{thm:Core},
we can maintain $\Core(G,\Kinit,d')$ with $\scatter=\Omegatil(\phicmg)$
and $\Otil(\gamma_{i-1}/\phicmg^{3})\le\stretch_{i}$ (by definition
of $\stretch_{i}$) with total update time 
\begin{align*}
&\Otil\left(T_{\Apxball}(G,\Kinit,32d'\log n,0.1)\Delta_{i-1}^{2}(D_{i+1}/D_{i-1})^{3}/\phicmg^{2}\right) & \\&=\left|\ball_{G}(\Kinit,32d'\log n)\right|\poly(\frac{n^{1/k_{0}+1/\distScale}}{\phicmg\eps_{0}})
\end{align*}
by (1).

(3): Recall that the algorithm from \Cref{thm:covering} for maintaining
a $(D_{i},k_{i},\eps_{i},\stretch_{i},\Delta_{i})$-covering of $G$
assumes, for all $D_{i}\le d'\le D_{i}(\frac{\stretch_{i}}{\eps_{i}})^{k}$,
$\Core$ and $\Apxball$ data structures with input distance parameter
$d'$. By (1) and (2), we can indeed implement these data structures
for any distance parameter $d'\le D_{i+1}$. Since $D_{i}(\frac{\stretch_{i}}{\eps_{i}})^{k_{i}}\le D_{i}\gamma_{i}\le D_{i+1}$
by \Cref{prop:para}(\ref{enu:para:gap and distance scale},\ref{enu:para:key}),
the assumption is satisfied.

So, using \Cref{thm:covering}, we can maintain a $(D_{i},k_{i},\eps_{i},\stretch_{i},\Delta_{i})$-covering
of $G$ with $\Delta_{i}=\Theta(k_{i}n^{2/k_{i}}/\scatter)$ in total
update time of 
\begin{align*}
O(k_{i}n^{1+2/k_{i}}\log n/\scatter+\sum_{C\in\cC^{ALL}}T_{\Core}(\stage G{t_{C}},\stage C{t_{C}},d_{\Clevel(C)})\\
+T_{\Apxball}(\stage G{t_{C}},\stage C{t_{C}},\frac{\stretch_{i}}{4\eps_{i}}32d_{\Clevel(C)},\eps_{i}))
\end{align*}
where $\cC^{ALL}$ contains all cores that have ever been initialized
and, for each $C\in\cC^{ALL}$, $t_{C}$ is the time $C$ is initialized.
By plugging in the total update time of $\Apxball$ from (1) and $\Core$
from (2), the total update time for maintaining the covering is 
\begin{align*}
\Otil(\frac{n^{1+2/k_{i}}}{\scatter}+\sum_{C\in\cC^{ALL}}\left|\ball_{\stage G{t_{C}}}(\stage C{t_{C}},32d_{\Clevel(C)}\log n)\right|\poly(\frac{n^{1/k_{0}+1/\distScale}}{\phicmg\eps_{0}})+\\
\left|\ball_{\stage G{t_{C}}}(\stage C{t_{C}},\frac{\stretch_{i}}{4\eps_{i}}d_{\Clevel(C)})\right|\frac{n^{2/k_{0}+12/\distScale}}{\phicmg\eps_{0}^{2}}).
\end{align*}
As it is guaranteed by \Cref{thm:covering} that
\[\sum_{C\in\cC^{ALL}}|\ball_{\stage G{t_{C}}}(\stage C{t_{C}},\frac{\stretch_{i}}{4\eps_{i}}d_{\Clevel(C)})|\le O(k_{i}n^{1+2/k_{i}}/\scatter),
\]
the above expression simplifies to $\Otil(n\cdot\poly(\frac{n^{1/k_{0}+1/\distScale}}{\phicmg\eps_{0}}))$. 
\end{proof}
By constructing all the data structures from level $i=0$ to $\distScale$,
we can conclude \Cref{thm:main no distance}.

%% file: prelim_path.tex
In this part of the paper, we augment the decremental $\SSSP$ data
structure from the previous part to support \emph{threshold-subpath
queries}, which returns a subset of edges in a path. To precisely
describe the properties of queries, we introduce the notion \emph{steadiness}. 

\paragraph{Steadiness and simpleness.} All graphs in this part can be described as follows. A graph $G=(V,E,w,\sigma)$
is such that, each edge $e$ has weight $w(e)$ and has integral \emph{steadiness}
$\sigma(e)\in[\sigma_{\min},\sigma_{\max}]$. We call $\sigma_{\min}$
and $\sigma_{\max}$ the \emph{minimum and maximum steadiness} of
$G$, respectively. For any \emph{multi}-set $E'\subseteq E$ and
$j$, we let $\sigma_{\le j}(E')=\{e\in E'\mid\sigma(e)\le j\}$ contain
all edges from $E'$ of steadiness at most $j$. We let $\sigma_{\le j}(G)=G[\sigma_{\le j}(E)]$
denote the subgraph of $G$ induced by the edge set $\sigma_{\le j}(E)$.
We define $\sigma_{\ge j}(E'), \sigma_{>j}(E'), \sigma_{<}(E')$ and $\sigma_{\ge j}(G), \sigma_{> j}(G), \sigma_{< j}(G)$ similarly.

A path $P$ is \emph{$\beta$-simple} if each vertex appears in $P$
at most $\beta$ times. We say that a path is \emph{$\beta$-edge-simple
}if each edge appears in $P$ at most $\beta$ times. Note that a
$1$-simple path is a simple path in the standard sense. Observe that
any $\beta$-simple path is $\beta$-edge simple. When $P$ is a (non-simple)
path, $\sigma_{\le j}(P)$ is a \emph{multi-set} containing all occurrences
of edges with steadiness at most $j$ in $P$.

\paragraph{The Path-reporting Data Structure.} Now, we are ready to define the augmented version of the decremental $\SSSP$
data structure from (\Cref{def:SSSP}) that \emph{supports threshold-subpath
queries}. The outputs of threshold-subpath queries are always of the
form $\sigma_{\le j}(P)$. We remind the reader of our application of these queries, where we sample "sensitive" edges (i.e. those that are almost filled by the flow we want to send along the path) more often than "steady" edges. Let us now state a definition and the main result of this section.

\begin{restatable}{defn}{pathReportingSSSP}\label{def:pathReportingSSSP}
A path-reporting decremental $\SSSP$ data
structure $\SSSP^{\pi}(G,s,\eps,\beta,q)$ is a decremental $\SSSP$
data structure $\SSSP(G,s,\eps)$ with the following additional guarantee:
\begin{itemize}
\item For each vertex $v$, $v$ is associated with a $\beta$-edge-simple
$s$-$v$ path $\pi(s,v)$ in $G$ of length at most $(1+\epsilon)\dist_{G}(s,v)$.
We say that $\pi(u,v)$ is \emph{implicitly maintained }by $\SSSP^{\pi}(G,s,\eps,\beta,q)$.
\item Given any vertex $v$ and a steadiness index $j$, the data structure
returns $\sigma_{\le j}(\pi(s,v))$ in $(|\sigma_{\le j}(\pi(s,v))|+1)\cdot q$
worst-case time. (We emphasize that $\pi(s,v)$ is independent from
$j$.)
\end{itemize}
\end{restatable}

\begin{theorem}
\label{thm:main SSSP path}Given an undirected decremental graph $G=(V,E,w,\sigma)$
with $n$ vertices and $m$ initial edges that have weights from $\{1,2,\dots,W\}$
and steadiness from $\{0,\dots,\sigma_{\max}\}$ where $\sigma_{\max}=o(\log^{3/4}n)$,
a fixed source vertex $s\in V$, and any $\eps>\phicmg$, we can implement
$\SSSP^{\pi}(G,s,\eps,\beta,q)$ in $\Ohat(m\log W)$ total update
time such that the edge-simpleness parameter is $\beta=\Ohat(1)$
and query-time overhead is $q=\Ohat(1)$.
\end{theorem}

\section{Preliminaries on Path-reporting Data Structures}

Let $P=(u,\dots,v)$ and $Q=(v,\dots,w)$ be paths that share an endpoint
at a vertex $v$. We let $P\circ Q$ or sometimes $(P,Q)$ denote
the concatenation of $P$ and $Q$. The union $\sigma_{\le j}(P)\cup\sigma_{\le j}(Q)$
is always a multi-set union of $\sigma_{\le j}(P)$ and $\sigma_{\le j}(Q)$. 

We will use the following simplifying reduction which allows us to
assume that out input graph throughout this part has bounded degree
and satisfies other convenient properties. The proof is shown in \Cref{subsec:appendixProofOfsimplify path}.

\begin{restatable}{prop}{propSimplifyingAssumptOnSSSP}
\label{prop:simplify path}Suppose that there is a data structure
$\SSSP^{\pi}(H,s,\eps,\beta,q)$ that only works if $H$ satisfies
the following properties: 
\begin{itemize}
\item $H$ always stays connected. 
\item Each update to $H$ is an edge deletion (not an increase in edge weight). 
\item $H$ has maximum degree $3$. 
\item $H$ has edge weights in $[1,n_{H}^{4}]$ and edges steadiness $[0,\sigma_{\max}+1]$. 
\end{itemize}
Suppose $\SSSP^{\pi}(H,s,\eps,\beta,q)$ has $T_{\SSSP^{\pi}}(m_{H},n_{H},\epsilon)$
total update time where $m_{H}$ and $n_{H}$ are numbers of initial
edges and vertices of $H$. Then, we can implement $\SSSP^{\pi}(G,s,O(\eps),O(\beta \log (Wn)),O(q))$
where $G$ is an arbitrary decremental graph with $m$ initial edges
that have weights in $[1,W]$ and steadiness in $[0,\sigma_{\max}]$
using total update time of $\tilde{O}\left(m/\epsilon^{2}+T_{\SSSP^{\pi}}(O(m \log W),2m,\epsilon)\right)\cdot\log(W).$
\end{restatable}

%% file: components_path.tex
\section{Main Path-Reporting Components}

\label{sec:components_path}

Below, we describe our main path-reporting data structures. They are
all natural extensions of the data structures listed in \Cref{sec:components}
so that they can support threshold-subpath queries, which return edges
with small steadiness in a path.
\begin{defn}
\label{def:Apxball path}A \emph{path-reporting approximate ball }data
structure $\Apxball^{\pi}(G,S,d,\eps,\beta)$ is an approximate ball
data structure $\Apxball(G,S,d,\eps)$ with the following additional
guarantee:
\begin{itemize}
\item For each vertex $v\in\ball_{G}(S,d)$, $v$ is associated with a $\beta$-simple
$S$-$v$ path $\pi(S,v)$ in $G$ of length at most $(1+\epsilon)\dist_{G}(S,v)$.
We say that $\pi(S,v)$ is \emph{implicitly maintained }by \\ $\Apxball^{\pi}(G,S,d,\eps,\kappa)$.
\item Given any vertex $v\in\ball_{G}(S,d)$ and a steadiness index $j$,
the data structure returns the multi-set $\sigma_{\le j}(\pi(S,v))$. (We emphasize
that $\pi(S,v)$ is independent from $j$.)
\end{itemize}
\end{defn}

Similar to \Cref{def:Apxball}, we slightly abuse the notation and
denote $\Apxball^{\pi}(G,S,d,\eps,\beta)=\{v\mid\dtil(v)\le(1+\eps)d\}$
as the set of all vertices $v$ whose distance estimate $\dtil(v)$
is at most $(1+\eps)d$.
\begin{defn}
\label{def:Core path}A \emph{path-reporting robust core} data structure
$\Core^{\pi}(G,\Kinit,d,\beta)$ with a stretch parameter $\stretch$
is a robust core data structure $\Core(G,\Kinit,d)$ with the following
additional guarantee:
\begin{itemize}
\item For each pair of vertices $(u,v)\in K\times K$ where $K\subseteq\Kinit$
is the maintained core set, the pair $(u,v)$ is associated with a
$\beta$-simple $u$-$v$ path $\pi(u,v)$ of length at most $\stretch\cdot d$.
We say that $\pi(u,v)$ is \emph{implicitly maintained }by $\Core^{\pi}(G,\Kinit,d)$.
\item Given a pair $(u,v)\in K\times K$ and a steadiness index $j$, the
algorithm returns $\sigma_{\le j}(\pi(u,v))$. (We emphasize that
$\pi(u,v)$ is independent from $j$.)
\end{itemize}
\end{defn}

We note that path-reporting $\Core^{\pi}$ is indeed stronger than
distance-only $\Core$. 
\begin{rem}
\label{rem:core path stronger}While the distance-only $\Core$ with
stretch $\stretch$ only guarantees that $\diam_{G}(K)\le\stretch\cdot d$,
the path-reporting $\Core^{\pi}$ is stronger as it implicitly maintains
paths $\pi(u,v)$ for all $u,v\in K$ of length at most $\stretch\cdot d$
which certifies that $\diam_{G}(K)\le\stretch\cdot d$.
\end{rem}

\begin{defn}
\label{def:Covering path}A \emph{path-reporting $(d,k,\eps,\stretch,\Delta,\beta)$-covering}
$\cC$ of a decremental graph $G$ is a $(d,k,\eps,\stretch,\Delta)$-covering
$\cC$ such that the distance-only $\Apxball$ and $\Core$ are replaced
by the path-reporting $\Apxball^{\pi}$ and $\Core^{\pi}$, respectively.
More precisely, for each level-$\ell$ core $C\in\cC$, we have $C=\Core^{\pi}(G,\Cinit,d_{\ell},\beta)$
with stretch at most $\stretch$, $\textsc{cover}(C)=\Apxball^{\pi}(G,C,4d_{\ell},\eps,\beta)$
and $\shell(C)=\Apxball(G,C,\frac{\stretch}{4\eps}d_{\ell},\eps,\beta)$.
\end{defn}

Next, we define a path-reporting version of compressed graphs. Recall
from \Cref{def:compressed graph} that a compressed graph is formally
a hypergraph.

\begin{defn}
\label{def:compressed graph path}A \emph{path-reporting $(d,\gamma,\Delta,\beta)$-compressed
graph} $H$ of a decremental graph $G$ is a $(d,\gamma,\Delta)$-compressed
graph $H$ and there is the following data structure: 
\begin{itemize}
\item For each adjacent pair of vertices $(u,v)$ in $H$, the pair $(u,v)$
is associated with a $\beta$-simple $u$-$v$ path $\pi(u,v)$ of
length at most $\gamma\cdot d$. We say that $\pi(u,v)$ is \emph{implicitly
maintained }by $H$.
\item Given an adjacent pair $(u,v)$ and a steadiness index $j$, the algorithm
returns $\sigma_{\le j}(\pi(u,v))$. (We emphasize that $\pi(u,v)$
is independent from $j$.)
\end{itemize}
\end{defn}

All four path-reporting components above are defined in such a way
that they are as strong as their distance-only counterparts. This
will be very important because it allows us to replace all distance-only
components in algorithms by their path-reporting counterparts without
violating any guarantees. The lemma below makes this point precise.
\begin{lem}
\label{prop:path-reporting stronger}We have the following:
\begin{enumerate}
\item $\Apxball^{\pi}(G,S,d,\eps,\beta)$ satisfies all requirement of $\Apxball(G,S,d,\eps)$.
\item $\Core^{\pi}(G,\Kinit,d,\beta)$ with stretch $\stretch$ satisfies
all requirement of $\Core(G,\Kinit,d)$ with stretch $\stretch$.
\item A path-reporting $(d,k,\eps,\stretch,\Delta,\beta)$-covering $\cC$
of $G$ is a (distance-only) $(d,k,\eps,\stretch,\Delta)$-covering
of $G$.
\item A \emph{path-reporting $(d,\gamma,\Delta,\beta)$-compressed graph}
$H$ of $G$ is a $(d,\gamma,\Delta)$-compressed graph of $G$.
\end{enumerate}
\end{lem}

\begin{proof}
For (1), this follows from definitions. For (2), this follows from
definitions and \Cref{rem:core path stronger}. For (3), this follows
by (1) and (2) because as path-reporting coverings are the same as
distance-only ones except that $\Apxball$ and $\Core$ are replaced
by $\Apxball^{\pi}$ and $\Core^{\pi}$, respectively. For (4), this
follows from definitions.
\end{proof}
The query time of threshold-subpath queries are measured as follows:
\begin{defn}
[Query-time Overhead]We say a path-reporting data structure has \emph{$(q_{\phi},q_{\path})$
query-time overhead} if, given any query with steadiness index $j$
and $P$ is the path that should be returned if $j=\infty$, then
$\sigma_{\le j}(P)$ is returned in at most $q_{\phi}$ time if $\sigma_{\le j}(P)=\emptyset$
and in at most $|\sigma_{\le j}(P)|\cdot q_{\path}$ time otherwise.
For path-reporting covering $\cC$, we say that $\cC$ has $(q_{\phi},q_{\path})$
query-time overhead, if all $\Apxball^{\pi}$ and $\Core^{\pi}$ that
are invoked for maintaining $\cC$ have query-time overhead at most
$(q_{\phi},q_{\path})$.
\end{defn}

By replacing distance-only $\Apxball$ and $\Core$ in \Cref{thm:covering}
for maintaining a distance-only covering in with path-reporting $\Apxball^{\pi}$
and $\Core^{\pi}$ which are stronger by \Cref{prop:path-reporting stronger},
we immediately obtain the following theorem analogous to \Cref{thm:covering}. 
\begin{theorem}
\label{thm:covering path}Let $G$ be an $n$-vertex bounded-degree
decremental graph. Given parameters $(d,k,\eps,\stretch)$ where $\eps\le0.1$,
we assume the following:
\begin{itemize}
\item for all $d\le d'\le d(\frac{\stretch}{\eps})^{k-1}$, there is $\Core^{\pi}(G,\Kinit,d',\beta)$
with scattering parameter at least $\scatter$ and stretch at most
$\stretch$ that has total update time $T_{\Core^{\pi}}(G,\Kinit,d',\beta)$,
and
\item for all $d\le d'\le d(\frac{\stretch}{\eps})^{k}$, there is $\Apxball^{\pi}(G,S,d',\eps,\beta)$
with total update time of\\ $T_{\Apxball^{\pi}}(G,S,d',\eps,\beta)$.
\end{itemize}
Then, we can maintain a path-reporting $(d,k,\eps,\stretch,\Delta,\beta)$-covering
of $G$ with $\Delta=O(kn^{2/k}/\scatter)$ in total update time 
\begin{align*}
O(kn^{1+2/k}\log(n)/\scatter+\sum_{C\in\cC^{ALL}}T_{\Core^{\pi}}(\stage G{t_{C}},\stage C{t_{C}},d_{\Clevel(C)},\beta)\\
+T_{\Apxball^{\pi}}(\stage G{t_{C}},\stage C{t_{C}},\frac{\stretch}{4\eps}d_{\Clevel(C)},\eps,\beta))
\end{align*}
where $\cC^{ALL}$ contains all cores that have ever been initialized
and, for each $C\in\cC^{ALL}$, $t_{C}$ is the time $C$ is initialized.
We guarantee that $\sum_{C\in\cC^{ALL}}|\ball_{\stage G{t_{C}}}(\stage C{t_{C}},\frac{\stretch}{4\eps}d_{\Clevel(C)})|\le O(kn^{1+2/k}/\scatter)$. 
\end{theorem}

The following is analogous to \Cref{prop:covering-compressed is compressed}.
\begin{prop}
[A Covering-Compressed Graph is a Compressed Graph (Path-reporting version)]\label{prop:covering-compressed is compressed path}Let
$\cC$ be a path-reporting $(d,k,\eps,\stretch,\Delta,\beta)$-covering
of a graph $G$. Let $H_{\cC}$ be the covering-compressed graph of $\cC$
and $H$ be the hypergraph view of $H_{\cC}$. Then $H$ is a path-reporting
$(d,\gamma,\Delta,3\beta)$-compressed graph of $G$ where $\gamma=(\stretch/\eps)^{k}$.
If the query-time overhead of $\cC$ is $(q_{\phi},q_{\path})$, then
$H$ has query-time overhead of $(3q_{\phi},q_{\path}+2q_{\phi})$.
\end{prop}

\begin{proof}
\Cref{prop:covering-compressed is compressed} already implies $H$ is
a (distance-only) $(d,\gamma,\Delta)$-compressed graph. It remains
to define a $3\beta$-simple $u$-$v$ path $\pi(u,v)$ of length
at most $\gamma\cdot d$ for every vertices $u$ and $v$ adjacent
in $H$, and then show a data structure that, given $(u,v)$ and a
steadiness index $j$, returns $\sigma_{\le j}(\pi(u,v))$. 

Consider vertices $u$ and $v$ adjacent in $H$ via a hyperedge $e$.
There is a level-$\ell$ core $C\in\cC$, for some $\ell\in[0,k-1]$,
corresponding to the hyperedge $e$ such that $u,v\in\shell(C)$.
Recall that $d_{\ell}=d\cdot(\frac{\stretch}{\eps})^{\ell}$ from
\Cref{def:Covering}. We have $\Apxball^{\pi}(G,C,\frac{\stretch}{4\eps}d_{\ell},\eps,\beta)=\shell(C)$
implicitly maintains $\beta$-simple paths $\pi(u,C)=(u,\dots,u_{C})$
and $\pi(C,v)=(v_{C},\dots,v)$ of length at most $(1+\eps)\frac{\stretch}{4\eps}d_{\ell}$
where $u_{C},v_{C}\in C$. Also, $\Core^{\pi}(G,\Cinit,d_{\ell},\beta)=C$
implicitly maintains a $\beta$-simple path $\pi(u_{C},v_{C})$ of
length at most $\stretch\cdot d_{\ell}$. We define $\pi(u,v)=\pi(u,C)\circ\pi(u_{C},v_{C})\circ\pi(C,v)$.
This path is clearly $3\beta$-simple and has length at most $2(1+\eps)\frac{\stretch}{4\eps}d_{\ell}+\stretch\cdot d_{\ell}\le\frac{\stretch}{\eps}d_{k-1}=d\gamma$,
as desired. 

By \Cref{rem:covering-compressed from covering}, given the covering $\cC$,
we will assume that the correspondences between each hyperedge $e\in E(H)$
and the corresponding core $C\in\cC$ is always maintained for us.
Given $(u,v)$ and a steadiness index $j$, we can straight-forwardly query
$\Apxball^{\pi}(G,C,\frac{\stretch}{4\eps}d_{\ell},\eps,\beta)$ and
$\Core^{\pi}(G,\Cinit,d_{\ell},\beta)$ to obtain 
\[\sigma_{\le j}(\pi(u,v))=\sigma_{\le j}(\pi(u,C))\cup\sigma_{\le j}(\pi(u_{C},v_{C}))\cup\sigma_{\le j}(\pi(C,v)).
\]
If $\sigma_{\le j}(\pi(u,v))=\emptyset$, this takes at most $3q_{\phi}$
time. Otherwise, this takes at most $|\sigma_{\le j}(\pi(u,v))|\cdot q_{\path}+2q_{\phi}\le|\sigma_{\le j}(\pi(u,v))|\cdot(q_{\path}+2q_{\phi})$
time.
\end{proof}
Next, we note that \Cref{prop:compressed base} generalizes to its path-reporting
version immediately.
\begin{prop}
[A Trivial Path-reporting Compressed Graph]\label{prop:compressed base path}Let
$G$ be a bounded-degree graph $G$ with integer edge weights. Let
$G_{unit}$ be obtained $G$ by removing all edges with weight greater
than one. Then, $G_{unit}$ is a path-reporting $(d=1,\gamma=1,\Delta=O(1),\beta=1)$-compressed
graph of $G$ with query-time overhead of $(q_{\phi}=1,q_{\path}=1)$.
\end{prop}

Lastly, we give a straightforward implementation for path-reporting
approximate ball data structure based on the classic ES-tree \cite{EvenS}.
\begin{prop}
[Path-reporting ES-tree]\label{prop:ES tree path}We can implement
$\Apxball^{\pi}(G,S,d,\eps=0,\beta=1)$ in $O(|\ball_{G}(S,d)|\cdot d\log n)$
total update time with $(O(\log n),O(d\log n))$ query-time overhead.
\end{prop}

\begin{proof}
This can be done by explicitly maintaining an ES-tree $T$ rooted
at $S$ to up distance $d$ in $O(|\ball_{G}(S,d)|\cdot d)$ total
update time. For every $v\in\ball_{G}(S,d)$, we define $\pi(S,v)$
as the simple $S$-$v$ path in $T$ which has length exactly $\dist_{G}(S,v)$.
We also implement a link-cut tree \cite{SleatorT83} on top of the
ES-tree $T$ so that, given any $v\in\ball_{G}(S,d)$, we can obtain
the minimum steadiness $j_{v}$ of edges in $\pi(S,v)$ in time $O(\log n)$.
Maintaining the link-cut tree only increases the total update time
by a factor $O(\log n)$. Given $v\in\ball_{G}(S,d)$ and a steadiness
index $j$, we check $j<j_{\min}$. If $j<j_{\min}$, then we know
$\sigma_{\le j}(\pi(S,v))=\emptyset$ and we return $\emptyset$ in
$O(\log n)$ time. If $j\ge j_{\min}$, then we know $\sigma_{\le j}(\pi(S,v))\neq\emptyset$
and so we just explicitly list all edges in $\pi(S,v)$ which contains
at most $d$ edges and return $\sigma_{\le j}(\pi(S,v))$ in $O(d+\log n)=O(|\sigma_{\le j}(\pi(S,v))|d\log n)$
time.\footnote{Note that using the link-cut tree, we can in fact list edges $\sigma_{\le j}(\pi(S,v))$
in $|\sigma_{\le j}(\pi(S,v))|\cdot O(\log^2 n)$ time so that the query-time
overhead is $(O(\log n),O(\log^2 n))$, but we do not need to optimize
this factor.}
\end{proof}

%% file: ball_path.tex
\section{Implementing Path-reporting Approximate Balls}

\label{sec:ball_path}

In this section, we show how to implement path-reporting approximate
ball data structures $\Apxball^{\pi}$ for distance scale $D$. We
will assume that a path-reporting covering $\cC$ for distance scale
$d\ll D$ is given for us. Then, the algorithm exploits three more
components as a subroutine: (1) path-reporting $\Apxball^{\pi}$ for
smaller distance scale $\approx d$, similar to how it is done for
the distance-only version, (2) path-reporting $\Apxball^{\pi}$ for
distance scale $D$ but the smaller graph $G^{\peel}$, and (3) distance-only
$\Apxball$ for distance scale $D$ on $G^{\peel}$ but with good
accuracy guarantee. This is why the total update time of the three
components appears in \Cref{eq:time ball_path} below.
\begin{theorem}
[Path-reporting Approximate Ball]\label{thm:ball_path} \label{thm:ballpath} Let $G$ be
an $n$-vertex bounded-degree decremental graph with steadiness between
$[\sigma_{\min},\sigma_{\max}]$. Let $G^{\peel}=\sigma_{>\sigma_{\min}}(G)$
be obtained from $G$ by removing edges with steadiness $\sigma_{\min}$.
Let $\eps\le1/500$ and $\eps^{\peel}\le1$. Suppose that a path-reporting
$(d,k,\eps,\stretch,\Delta,\beta)$-covering $\cC$ of $G$ is explicitly
maintained for us. Then, we can implement a path-reporting approximate
ball data structure $\Apxball^{\pi}(G,S,D,300\eps+\eps^{\peel},8\beta\Delta)$
using total update time
\begin{align}
 & \Otil(\left|\ball_{G}(S,D)\right|\Delta\frac{D}{\eps d})+T_{\Apxball^{\pi}}(G,S,2(\frac{\stretch}{\eps})^{k}d,\eps,\beta)+\label{eq:time ball_path}\\
 & T_{\Apxball^{\pi}}(G^{\peel},S,D,\eps^{\peel},8\beta\Delta)+T_{\Apxball}(G^{\peel},S,D,\eps).\nonumber 
\end{align}
Let $(q_{\phi},q_{\path})$ bound the query-time overhead of both
$\Apxball^{\pi}(G,S,2(\frac{\stretch}{\eps})^{k}d,\eps,\beta)$ and
$(d,k,\eps,\stretch,\Delta,\beta)$-covering $\cC$. Let $(q_{\phi}^{\peel},q_{\path}^{\peel})$
bound the query-time overhead of $\Apxball^{\pi}(G^{\peel},S,D,\eps^{\peel},8\beta\Delta)$,
Then, the data structure has query-time overhead of
\[
(q_{\phi}^{\peel}+O(1),\max\{q_{\path}^{\peel}+O(1),q_{\path}+O(\frac{D}{\eps d})\cdot q_{\phi}\}).
\]
\end{theorem}

The rest of this section is for proving \Cref{thm:ball_path}. In \Cref{sec:ball_path:datastructure},
we describe data structures for maintaining the distance estimate
$\dtil(v)$ for all $v\in\ball_{G}(S,D)$ and for additionally supporting
threshold-subpath queries, and then we analyze the total update time.
Based on the maintained data structure, in \Cref{sec:ball_path:define_path},
we define the implicitly maintained paths $\pi(S,v)$ for all $v\in\ball_{G}(S,D)$
as required by \Cref{def:Apxball path} of $\Apxball^{\pi}$. Finally,
we show an algorithm that answers threshold-subpath queries in \Cref{sec:ball_path:query_subpath}.

\subsection{Data Structures}

\label{sec:ball_path:datastructure}

\paragraph{Data structures on $G$. }

We maintain the distance estimates $\dtil(v)$ and the MES-tree $\Ttil$
using the same approach as in the distance-only algorithm from \Cref{sec:ball}.
The only difference is that we replace the distance-only components
with the path-reporting ones. 

More specifically, given the \emph{path-reporting} $(d,k,\eps,\stretch,\Delta,\beta)$-covering
$\cC$, let $H_{\cC}$ be the covering-compressed graph w.r.t.~$\cC$
(recall \Cref{def:core compress}). Then, we maintain the emulator
$\Htil$ based on $H_{\cC}$ as described in \Cref{def:Htil} but we
replace the distance-only $\Apxball(G,S,2(\frac{\stretch}{\eps})^{k}d,\eps)$
with the path-reporting $\Apxball^{\pi}(G,S,2(\frac{\stretch}{\eps})^{k}d,\eps,\beta)$
in \Cref{enu:Htil near edge} of \Cref{def:Htil}. For each $v\in\Apxball^{\pi}(G,S,2(\frac{\stretch}{\eps})^{k}d,\eps,\beta)$,
$\Apxball^{\pi}(G,S,2(\frac{\stretch}{\eps})^{k}d,\eps,\beta)$ maintains
the distance estimate $\dtil^{\near}(v)$ and implicitly maintains
an $(1+\eps)$ approximate $S$-$v$ shortest path $\pi^{\near}(S,v)$.

Now, given the emulator $\Htil$ with a dummy source $s$, we use
exactly the same algorithm $\MES(\Htil,s,D)$ from \Cref{alg:MES}
to maintain the MES-tree $\Ttil$ on $\Htil$, and let $\dtil^{\MES}(v)$
denote the distance estimate of $v$ maintained by $\MES(\Htil,s,D)$.
Recall that $\Ttil$ is defined as follows: for every vertex $u\in V(\Htil)\setminus\{s\}$,
$u$'s parent in $\Ttil$ is $\arg\min_{v}\{\dtil^{\MES}(v)+\wtil(v,u)\}$.
Then, we maintain $\dtil(v)=\min\{\dtil^{\near}(v),\dtil^{\MES}(v)\}$
for each $v\in V(\Htil)$. Note that, we used slightly different notations
in \Cref{sec:ball}; we said that the algorithm maintains $\min\{\dtil^{\near}(v),\dtil(v)\}$
for each $v$, but in \Cref{sec:ball} $\dtil(v)$ was used to denote
$\dtil^{\MES}(v)$. So the outputs from both sections are equivalent
objects. 

We observe that our slight modification does not change the accuracy
guarantee of the distance estimates.
\begin{lem}
\label{lem:dtil path correct}For $v\in\ball_{G}(S,D)$, $\dist_{G}(S,v)\le\dtil(v)\le(1+50\eps)\dist_{G}(S,v)$.
\end{lem}

\begin{proof}
The only changes in the algorithm from \Cref{sec:ball} are to replace
the distance-only $(d,k,\eps,\stretch,\Delta)$-covering $\cC$ with
the path-reporting $(d,k,\eps,\stretch,\Delta,\beta)$-covering $\cC$,
and to replace the distance-only $\Apxball(G,S,2(\frac{\stretch}{\eps})^{k}d,\eps)$
with the path-reporting $\Apxball^{\pi}(G,S,2(\frac{\stretch}{\eps})^{k}d,\eps,\beta)$.
As shown in \Cref{prop:path-reporting stronger}, these path reporting
data structures are stronger than their distance-only counterparts.
Therefore, all the arguments in \Cref{sec:ball} for proving the accuracy
of $\dtil(v)$ still hold. 
\end{proof}

\paragraph{Data structures on $G^{\protect\peel}$.}

Next, let $G^{\peel}=\sigma_{>\sigma_{\min}}(G)$ be obtained from
$G$ by removing edges with steadiness $\sigma_{\min}$. We recursively
maintain the distance-only $\Apxball(G^{\peel},S,D,\eps)$ and let
$\dtil^{\peel}(v)$ denote its distance estimate for the shortest $S$-$v$ path in $G^{\peel}$.
We also recursively maintain the path-reporting $\Apxball^{\pi}(G^{\peel},S,D,\eps^{\peel},8\beta\Delta)$
and let $\pi^{\peel}(S,v)$ denote its implicitly maintained approximate
$S$-$v$ shortest path in $G^{\peel}$. We emphasize that the approximation
guarantee on $\dtil^{\peel}(v)$ depends on $\eps$ and not on $\eps^{\peel}$.

This completes the description of the all data structures for \Cref{thm:ball_path}.
We bound the total update time as specified in \Cref{thm:ball_path}
below.
\begin{lem}
The total update time is 
\begin{align*}
 & \Otil(\left|\ball_{G}(S,D)\right|\Delta\frac{D}{\eps d})+T_{\Apxball^{\pi}}(G,S,2(\frac{\stretch}{\eps})^{k}d,\eps,\beta)+\\
 & T_{\Apxball^{\pi}}(G^{\peel},S,D,\eps^{\peel},8\beta\Delta)+T_{\Apxball}(G^{\peel},S,D,\eps).
\end{align*}
\end{lem}

\begin{proof}
As the covering $\cC$ is explicitly maintained for us, we do not
count its update time. Using the exactly same analysis as in the last
paragraph of \Cref{sec:ball conclude}, the total update time for maintaining
$\{\dtil(v)\}_{v}$ is $\Otil(\left|\ball_{G}(S,D)\right|\Delta\frac{D}{\eps d})+T_{\Apxball^{\pi}}(G,S,2(\frac{\stretch}{\eps})^{k}d,\eps,\beta)$.
Note that we replace $T_{\Apxball}(\cdot)$ with $T_{\Apxball^{\pi}}(\cdot)$.
Lastly, the data structures on $G^{\peel}$ take $T_{\Apxball^{\pi}}(G^{\peel},S,D,\eps^{\peel},8\beta\Delta)+T_{\Apxball}(G^{\peel},S,D,\eps)$
time by definition.
\end{proof}

\subsection{Defining The Implicitly Maintained Paths}

\label{sec:ball_path:define_path}

In this section, for each $v\in\ball_{G}(S,D)$, we define an approximate
$S$-$v$ shortest path $\pi(S,v)$ using \Cref{alg:define ball_path}. More precisely, we let $\pi(S,v)$ be defined as the path that would be returned if we run \Cref{alg:define ball_path} at the current stage (the algorithm is deterministic, so the query always returns the same path on a fixed input). We explicitly emphasize that these paths $\pi(S,v)$ are \emph{not} maintained
explicitly, but they are unique and fixed through the stage and they are completely independent from the steadiness
index $j$ in the queries. 

\begin{algorithm}
\If(\label{enu:condition path}){$\dtil^{\peel}(v)\le(1+50\eps)^{3}\cdot\dtil(v)$}{\Return $\pi^{\peel}(S,v)$
implicitly maintained by $\Apxball^{\pi}(G^{\peel},S,D,\eps^{\peel},8\beta\Delta)$. \label{enu:assign path next}}

\If(\label{enu:assign path near}){$(s,v)\in E(\Htil)$}{\Return $\pi^{\near}(S,v)$ implicitly maintained by $\Apxball^{\pi}(G,S,2(\frac{\stretch}{\eps})^{k}d,\eps,\beta)$.}

Let $\Ptil_{v}=(s=u_{0},u_{1}\dots,u_{z}=v)$ be the unique $s$-$v$
path in the MES tree $\Ttil$.\;

\ForEach( \label{enu:path edge}){$e\in\Ptil_{v}$}{
    \If(\label{enu:subpath near}){$e=(s,u)$}{$\pi_{e}\gets\pi^{\near}(S,u)$
implicitly maintained by $\Apxball^{\pi}(G,S,2(\frac{\stretch}{\eps})^{k}d,\eps,\beta)$.}
    \If(\label{enu:subpath edge}){$e\in E(G)$}{$\pi_{e}\gets\{e\}$.}
    \If(\label{enu:subpath shell}){$e\in E(H_{\cC})$ where $e=(u,u')$, $u$ corresponds to a core $C$ and $u'$ is a regular vertex}{ $\pi_{e}\gets\pi(C,u')$ implicitly maintained by $\Apxball^{\pi}(G,C,\frac{\stretch}{4\eps}d_{\Clevel(C)},\eps,\beta)=\shell(C)$
in the covering $\cC$.}
}

\ForEach(\label{enu:path core}){ $u_{i}\in\Ptil_{v}$ where $u_{i}$ corresponds to a core $C$}{
    Let $u'_{i},u''_{i}\in C$ be such that $\pi_{(u_{i-1},u_{i})}=(u_{i-1},\dots,u'_{i})$
and $\pi_{(u_{i},u_{i+1})}=(u'_{i},\dots,u_{i+1})$. \;
    \label{enu:subpath core} $\pi_{u_{i}}\gets\pi(u'_{i},u''_{i})$ implicitly maintained by $\Core^{\pi}(G,\Cinit,d_{\ell_{\core}(C)},\beta)=C$ in the covering $\cC$.
}

\label{enu:assign path general}Order the paths from \Cref{enu:path edge}
as $\pi_{(u_{0},u_{1})},\pi_{(u_{1},u_{2})},\dots,\pi_{(u_{z-1},u_{z})}$
and then, for each path $\pi_{u_{i}}$ from \Cref{enu:path core},
insert $\pi_{u_{i}}$ between $\pi_{(u_{i-1},u_{i})}$ and $\pi_{(u_{i},u_{i+1})}$.\;
\Return $\pi(S,v)$ as the concatenation of all these ordered
paths.
\caption{\label{alg:define ball_path}Computing $\pi(S,v)$ for each $v\in\protect\ball_{G}(S,D)$.}
\end{algorithm}

Below, we show that each path $\pi(S,v)$ defined by \Cref{alg:define ball_path}
satisfies the requirement from \Cref{def:Apxball path}: it is an approximate
$S$-$v$ shortest path in $G$ (\Cref{lem:path approx}) and it guarantees
bounded simpleness (\Cref{lem:path simple}). 
\begin{lem}
\label{lem:path approx}For every $v\in\ball_{G}(S,D)$, we have the
following: 
\begin{enumerate}
\item \label{enu:path approx:next}If $\dtil^{\peel}(v)\le(1+50\eps)^{3}\cdot\dtil(v)$,
then $\pi(S,v)$ is a $(1+300\eps+\eps^{\peel})$-approximate $S$-$v$
shortest path in $G$. 
\item \label{enu:path approx:here}If $\dtil^{\peel}(v)>(1+50\eps)^{3}\cdot\dtil(v)$,
then $\pi(S,v)$ is a $(1+50\eps)^{2}$-approximate $S$-$v$ shortest
path in $G$
\end{enumerate}
\end{lem}

\begin{proof}
$\pi(S,v)$ is indeed an $S$-$v$ path in $G$ because the subpaths
of $\pi(S,v)$ are ordered and concatenated at \Cref{enu:assign path general}
such that their endpoints meet, and one endpoint of $\pi(S,v)$ is
$v$ and another is in $S$. Below, we only need to bound the total
weight $w(\pi(S,v))$ of the path $\pi(S,v)$. 

If $\dtil^{\peel}(v)\le(1+50\eps)^{3}\dtil(v)$, then $\pi(S,v)\gets\pi^{\peel}(S,v)$
is assigned at \Cref{enu:assign path next}. Therefore, we have
\begin{align*}
w(\pi(S,v)) & \le(1+\eps^{\peel})\dist_{G^{\peel}}(S,v) & \text{by }\Apxball^{\pi}(G^{\peel},S,D,\eps^{\peel},8\beta\Delta)\\
 & \le(1+\eps^{\peel})\dtil^{\peel}(v) & \text{by }\Apxball(G^{\peel},S,D,\eps)\\
 & \le(1+\eps^{\peel})(1+50\eps)^{3}\dtil(v) & \text{by \Cref{enu:condition path}}\\
 & \le(1+\eps^{\peel})(1+50\eps)^{4}\dist_{G}(S,v) & \text{by \Cref{lem:dtil path correct}}\\
 & \le(1+300\eps+\eps^{\peel})\dist_{G}(S,v).
\end{align*}

Next, if $\dtil^{\peel}(v)>(1+50\eps)^{3}\dtil(v)$, then we have
two cases. Suppose $\pi(S,v)\gets\pi^{\near}(S,v)$ is assigned at
\Cref{enu:assign path near}. Then, $w(\pi(S,v))=w(\pi^{\near}(S,v))\le(1+\eps)\dist_{G}(S,v)$
by the guarantee of $\Apxball^{\pi}(G,S,2(\frac{\stretch}{\eps})^{k}d,\eps,\beta)$.
Otherwise, $\pi(S,v)$ must be assigned at \Cref{enu:assign path general}.
Recall that $\wtil(e)$ denotes the weight of $e$ in the emulator
$\Htil$. It suffices to show that $w(\pi(S,v))\le(1+\eps)\cdot\sum_{e\in\Ptil_{v}}\wtil(e)$
and $\sum_{e\in\Ptil_{v}}\wtil(e)\le(1+50\eps)\dist_{G}(S,v)$ because
they imply that $w(\pi(S,v))\le(1+50\eps)^{2}\dist_{G}(S,v)$. Below,
we prove each inequality one by one. 

To prove $w(\pi(S,v))\le(1+\eps)\cdot\sum_{e\in\Ptil_{v}}\wtil(e)$,
observe that $\pi(S,v)$ is a concatenation of subpaths of the following
three types: (1) $\pi_{e}$ where $e\in E(G)$, (2) $\pi_{(u_{i-1},u_{i})}\circ\pi_{u_{i}}\circ\pi_{(u_{i},u_{i+1})}$
where $u_{i}$ corresponds to a core $C$ and $(u_{i-1},u_{i}),(u_{i},u_{i+1})\in E(H_{\cC})$,
and (3) $\pi_{(s,u_{1})}$ where $s$ is the dummy source $s$. For
a type-1 subpath, we have that $w(\pi_{e})=w(e)\le\ceiling{w(e)}_{\eps d}=\wtil(e)$
by \Cref{def:Htil} of $\Htil$. For a type-2 subpath, we have
\begin{align*}
 & w(\pi_{(u_{i-1},u_{i})}\circ\pi_{u_{i}}\circ\pi_{(u_{i},u_{i+1})})\\
 & \le(1+\eps)\dist_{G}(u_{i-1},C)+\stretch\cdot d_{\ell_{\core}(C)}+(1+\eps)\dist_{G}(C,u_{i+1})\\
 & \le(1+\eps)\cdot(\wtil(u_{i-1},u_{i})+\wtil(u_{i},u_{i+1}))
\end{align*}
where the first inequality is by the guarantee of $\Apxball^{\pi}$
and $\Core^{\pi}$ with stretch $\stretch$ that maintain $\shell(C)$
and $C$, respectively, and the second inequality follows from weight
assignment of edges in the covering-compressed graph $H_{\cC}$, see \Cref{def:core compress}.
For a type-3 subpath, $\Apxball^{\pi}(G,S,2(\frac{\stretch}{\eps})^{k}d,\eps,\beta)$
guarantees that 
\[
w(\pi_{(s,u_{1})})\le(1+\eps)\dist_{G}(S,u_{1})\le(1+\eps)\ceiling{d^{\near}(S,u_{1})}_{\eps d}=(1+\eps)\cdot\wtil(s,u_{1})
\]
where the equality is by \Cref{def:Htil} of $\Htil$. Observe that
each term in $\sum_{e\in\Ptil_{v}}\wtil(e)$ is charged only once
by each subpath of $\pi(S,v)$. Therefore, we indeed have $w(\pi(S,v))\le(1+\eps)\cdot\sum_{e\in\Ptil_{v}}\wtil(e)$.

To prove $\sum_{e\in\Ptil_{v}}\wtil(e)\le(1+50\eps)\dist_{G}(S,v)$,
observe that $\sum_{e\in\Ptil_{v}}\wtil(e)\le\dtil^{\MES}(v)$ by
\Cref{lem:MES:approx:lower}(\ref{lem:MES:approx:lower:regular}).
On the other hand, \Cref{lem:MES:approx:upper}(\ref{lem:MES:approx:upper:regular})
says that $\dtil^{\MES}(v)\le\max\{\ceiling{(1+\eps)\dist_{G}(S,v)}_{\eps d},(1+50\eps)\dist_{G}(S,v)\}=(1+50\eps)\dist_{G}(S,v)$.
The equality is because $\dist_{G}(S,v)\ge2(\frac{\stretch}{\eps})^{k}d\ge d$,
which holds because $(s,v)\notin E(\Htil)$, i.e.~$v\notin\Apxball^{\pi}(G,S,2(\frac{\stretch}{\eps})^{k}d,\eps,\beta)$.
\end{proof}
\begin{lem}
\label{lem:path simple}For every $v\in\ball_{G}(S,d)$, the path
$\pi(S,v)$ is $(8\beta\Delta)$-simple.
\end{lem}

\begin{proof}
First, note that if we set $\pi(S,v)=\pi^{\peel}(S,v)$ at \Cref{enu:assign path next}
or $\pi(S,v)=\pi^{\near}(S,v)$ at \Cref{enu:assign path near}, then
$\pi(S,v)$ is $(8\beta\Delta)$-simple by the definition of $\Apxball^{\pi}(G^{\peel},S,D,\eps^{\peel},8\beta\Delta)$
and $\Apxball^{\pi}(G,S,2(\frac{\stretch}{\eps})^{k}d,\eps,\beta)$.
Now, suppose that $\pi(S,v)$ is assigned at \Cref{enu:assign path general}.
We claim two things. First, each subpath that was concatenated into
$\pi(S,v)$ is a $\beta$-simple path. Second, every vertex $u$ can
participate in at most $8\Delta$ such subpaths of $\pi(S,v)$. This
would imply that $\pi(S,v)$ is $(8\beta\Delta)$-simple as desired.

To see the first claim, we consider the four cases of the subpath
of $\pi(S,v)$: First, from \Cref{enu:subpath near}, the subpath $\pi_{(s,u_{i})}$
is $\beta$-simple by the definition of $\Apxball^{\pi}(G,S,2(\frac{\stretch}{\eps})^{k}d,\eps,\beta)$.
Second, from \Cref{enu:subpath edge}, the subpath $\pi_{e}=\{e\}$
where $e\in\Ptil_{v}\cap E(G)$ is clearly $1$-simple. Third and
forth, from \Cref{enu:subpath shell} and \Cref{enu:subpath core},
the subpaths $\pi_{e}$ and $\pi_{u_{i}}$ are $\beta$-simple because
of the simpleness parameter of the covering $\cC$

To see the second claim, consider any vertex $u\in V(G)$. Clearly,
$u$ can participate in at most 1 subpath from \Cref{enu:subpath near}
as $\pi_{(s,u_{i})}$ is the only path generated from this step. Next,
$u$ can participate in at most $2$ subpaths from \Cref{enu:subpath edge}
because $\Ptil_{v}$ is a simple path in $\Htil$ and thus $u$ can
be in at most $2$ edges from $\Ptil_{v}\cap E(G)$. The last case
counts the subpaths from both \Cref{enu:subpath shell} and \Cref{enu:subpath core}.
For any $u_{i}\in V(\Htil)$ corresponding to a core $C$, if $u$
appears in any path from $\pi_{(u_{i-1},u_{i})},\pi_{u_{i}},\pi_{(u_{i},u_{i+1})}$,
then we claim $u\in\oshell(C)$. But $u$ can be in at most $\Delta$
outer-shells by \Cref{def:Covering}. Hence, $u$ can appear in at
most $3\Delta$ subpaths from \Cref{enu:subpath shell} and \Cref{enu:subpath core}.
In total, $u$ appears in at most $3\Delta+3\le8\Delta$ subpaths
of $\pi(S,v)$. The claim below finishes the proof:
\begin{claim}
\label{claim:in outershell}If $u$ appears in $\pi_{(u_{i-1},u_{i})},\pi_{u_{i}}$
or $\pi_{(u_{i},u_{i+1})}$, then $u\in\oshell(C)$.
\end{claim}

\begin{proof}
According to \Cref{def:Covering} and \Cref{def:Covering path}, the
paths $\pi_{(u_{i-1},u_{i})}$ and $\pi_{(u_{i},u_{i+1})}$ have length
at most $(1+\eps)\cdot\frac{\stretch}{4\eps}d_{\ell_{\core}(C)}$,
and the path $\pi_{u_{i}}$ has length at most $\stretch\cdot d_{\ell_{\core}(C)}$.
As each of these paths has an endpoint in $C$, so $u\in\ball_{G}(C,(1+\eps)\cdot\frac{\stretch}{4\eps}d_{\ell_{\core}(C)})\subseteq\oshell(C)$.\footnote{This inclusion is actually the only reason we introduce the notion
of outer-shell. If we could argue that $u\in\ball_{G}(C,\frac{\stretch}{4\eps}d_{\ell_{\core}(C)})$,
then we would have concluded $u\in\shell(C)$. We do not need $\oshell(C)$
else where.}
\end{proof}
\end{proof}
To conclude, from \Cref{lem:path approx} and \Cref{lem:path simple},
for each $v\in\ball_{G}(S,D)$, $\pi(S,v)$ is indeed a $(8\beta\Delta)$-simple
$(1+300\eps+\eps^{\peel})$-approximate $S$-$v$ shortest path in
$G$ as required by $\Apxball(G,S,D,300\eps+\eps^{\peel},8\beta\Delta)$.

\subsection{Threshold-Subpath Queries}

\label{sec:ball_path:query_subpath}

In this section, we describe in \Cref{alg:return ball_path} below
how to process the threshold-subpath query that, given a vertex $v\in\ball_{G}(S,D)$
and a steadiness index $j$, returns $\sigma_{\le j}(\pi(S,v))$ consisting
of all edges of $\pi(S,v)$ with steadiness at most $j$.

\begin{algorithm}
\lIf{$j<\sigma_{\min}$}{\Return $\emptyset$. \label{enu:return empty}}
$\answer_{(v,j)} \gets \emptyset$.\;

\If{$\dtil^{\peel}(v)\le(1+50\eps)^{3}\cdot\dtil(v)$}{
    \label{enu:return path next}\Return $\sigma_{\le j}(\pi^{\peel}(S,v))$
by querying $\Apxball^{\pi}(G^{\peel},S,D,\eps^{\peel},8\beta\Delta)$.}

\If(\label{enu:return path near}){$(s,v)\in E(\Htil)$}{\Return $\sigma_{\le j}(\pi^{\near}(S,v))$ by querying $\Apxball^{\pi}(G,S,2(\frac{\stretch}{\eps})^{k}d,\eps,\beta)$.}

Let $\Ptil_{v}=(s=u_{0},u_{1}\dots,u_{z}=v)$ be the unique $s$-$v$
path in the MES tree $\Ttil$.\;

\ForEach{$e\in\Ptil_{v}$}{
    \If{$e=(s,u)$}{ $\answer_{(v,j)}\gets\answer_{(v,j)}\cup\sigma_{\le j}(\pi^{\near}(S,u))$
by querying $\Apxball^{\pi}(G,S,2(\frac{\stretch}{\eps})^{k}d,\eps,\beta)$.}
    \lIf{ $e\in E(G)$}{ $\answer_{(v,j)}\gets\answer_{(v,j)}\cup\sigma_{\le j}(\{e\})$.}
    \If{ $e\in E(H_{\cC})$ where $e=(u,u')$ and $u$ corresponds to a
core $C$ and $u'$ is a regular vertex}
    { $\answer_{(v,j)}\gets\answer_{(v,j)}\cup\sigma_{\le j}(\pi(C,u'))$
by querying $\Apxball^{\pi}(G,C,\frac{\stretch}{4\eps}d_{\Clevel(C)},\eps,\beta)$.}
}

\ForEach{$u_{i}\in\Ptil_{v}$ where $u_{i}$ corresponds to a core $C$}{ 
    Let $u'_{i},u''_{i}\in C$ be such that $\pi_{(u_{i-1},u_{i})}=(u_{i-1},\dots,u'_{i})$
    and $\pi_{(u_{i},u_{i+1})}=(u'_{i},\dots,u_{i+1})$. \;
     $\answer_{(v,j)}\gets\answer_{(v,j)}\cup\sigma_{\le j}(\pi(u'_{i},u''_{i}))$
    by querying $\Core^{\pi}(G,\Cinit,d_{\ell_{\core}(C)},\beta)$.\;
}

\label{enu:return path general}\Return  $\answer_{(v,j)}$
\caption{\label{alg:return ball_path}Returning $\sigma_{\le j}(\pi(S,v))$,
given $v\in\protect\ball_{G}(S,D)$ and a steadiness index $j$}
\end{algorithm}

We first observe that \Cref{alg:return ball_path} returns the correct
answer. This follows straightforwardly because all the steps of \Cref{alg:return ball_path}
are analogous to the ones in \Cref{alg:define ball_path}
except that we just return $\emptyset$ if we first find that $j<\sigma_{\min}$.
\begin{prop}
Given $v\in\ball_{G}(S,D)$ and a steadiness index $j$, \Cref{alg:return ball_path}
returns $\sigma_{\le j}(\pi(S,v))$ where $\pi(S,v)$ is defined in
\Cref{alg:define ball_path}.
\end{prop}

\begin{proof}
There are four steps that \Cref{alg:return ball_path} may return.
At \Cref{enu:return empty}, we have $\sigma_{\le j}(\pi(S,v))=\emptyset$
as $j<\sigma_{\min}$. At \Cref{enu:return path next}, we have $\sigma_{\le j}(\pi^{\peel}(S,v))=\sigma_{\le j}(\pi(S,v))$
by \Cref{enu:assign path next} of \Cref{alg:define ball_path}. At
\Cref{enu:return path near}, we have $\sigma_{\le j}(\pi^{\near}(S,v))=\sigma_{\le j}(\pi(S,v))$
by \Cref{enu:assign path near} of \Cref{alg:define ball_path}. Finally,
at \Cref{enu:return path general}, observe that $\answer_{(v,j)}$
is simply a multi-set union of all edges of steadiness at most $j$
from all subpaths from $\pi(S,v)$ defined in \Cref{alg:define ball_path}.
So $\answer_{(v,j)}=\sigma_{\le j}(\pi(S,v))$ as well.
\end{proof}
The following simple observation will help us bound the query time.
\begin{prop}
\label{prop:path nonempty}If $\dtil^{\peel}(v)>(1+50\eps)^{3}\dtil(v)$
and $j\ge\sigma_{\min}$, then $\sigma_{\le j}(\pi(S,v))\neq\emptyset$.
\end{prop}

\begin{proof}
First, observe that $\dist_{G^{\peel}}(S,v)>(1+50\eps)^{2}\cdot\dist_{G}(S,v)$
because 
\begin{align*}
\dist_{G^{\peel}}(S,v) & \ge\frac{1}{(1+\eps)}\cdot\dtil^{\peel}(v) & \text{by }\Apxball(G^{\peel},S,D,\eps)\\
 & >(1+50\eps)^{2}\cdot\dtil(v) & \text{by assumption}\\
 & \ge(1+50\eps)^{2}\cdot\dist_{G}(S,v) & \text{by \Cref{lem:dtil path correct}.}
\end{align*}
This implies that \emph{every} $(1+50\eps)^{2}$-approximate $S$-$v$
shortest path in $G$ must contains some edge with steadiness $\sigma_{\min}$.
By \Cref{lem:path approx}(\ref{enu:path approx:here}), $\pi(S,v)$
is such a $(1+50\eps)^{2}$-approximate shortest path. So $\sigma_{\le j}(\pi(S,v))\neq\emptyset$
as $j\ge\sigma_{\min}$.
\end{proof}
Finally, we bound the query time of the algorithm. Recall that $(q_{\phi},q_{\path})$
bounds the query-time overhead of both $\Apxball^{\pi}(G,S,2(\frac{\stretch}{\eps})^{k}d,\eps,\beta)$
and $(d,k,\eps,\stretch,\Delta,\beta)$-covering $\cC$, and $(q_{\phi}^{\peel},q_{\path}^{\peel})$
bounds the query-time overhead of $\Apxball^{\pi}(G^{\peel},S,D,\eps^{\peel},8\beta\Delta)$.
Below, we show that our algorithm has $(q_{\phi}^{\peel}+O(1),\max\{q_{\path}^{\peel}+O(1),q_{\path}+O(\frac{D}{\eps d})\cdot q_{\phi}\})$
query-time overhead as required by \Cref{thm:ball_path}.
\begin{lem}
Given any $v\in\ball_{G}(S,D)$ and $j$, \Cref{alg:return ball_path}
takes $q_{\phi}^{\peel}+O(1)$ time if $\sigma_{\le j}(\pi(S,v))=\emptyset$.
Otherwise, it takes $|\sigma_{\le j}(\pi(S,v))|\cdot\max\{q_{\path}^{\peel}+O(1),q_{\path}+O(\frac{D}{\eps d})\cdot q_{\phi}\}$
time.
\end{lem}

\begin{proof}
Suppose that $\sigma_{\le j}(\pi(S,v))=\emptyset$. \Cref{prop:path nonempty}
implies that either $j<\sigma_{\min}$ or $\dtil^{\peel}(v)\le(1+50\eps)^{3}\dtil(v)$.
Therefore, \Cref{alg:return ball_path} must return either at \Cref{enu:return empty}
or \Cref{enu:return path next} both of which takes at most $q_{\phi}^{\peel}+O(1)$
time. 

Suppose $\sigma_{\le j}(\pi(S,v))\neq\emptyset$. If \Cref{alg:return ball_path}
returns at \Cref{enu:return path next} or \Cref{enu:return path near},
then the total time is $|\sigma_{\le j}(\pi(S,v))|\cdot\max\{q_{\path}^{\peel},q_{\path}\}+O(1)$.
Otherwise, \Cref{alg:return ball_path} returns at \Cref{enu:return path general}
and so the algorithm basically just makes $O(|\Ptil_{v}|)=O(\frac{D}{\eps d})$
queries to $\Apxball^{\pi}$ and $\Core^{\pi}$ data structures maintained
inside the covering $\cC$, and one query to $\Apxball^{\pi}(G,S,2(\frac{\stretch}{\eps})^{k}d,\eps,\beta)$.
This takes $|\sigma_{\le j}(\pi(S,v))|\cdot q_{\path}+O(\frac{D}{\eps d})\cdot q_{\phi}$
time. Since $|\sigma_{\le j}(\pi(S,v))|\ge1$, in any case, the total
time is at most 
\[
|\sigma_{\le j}(\pi(S,v))|\cdot\max\{q_{\path}^{\peel}+O(1),q_{\path}+O(\frac{D}{\eps d})\cdot q_{\phi}\}.
\]
\end{proof}

%% file: robust_core_path.tex
\section{Implementing Path-reporting Robust Cores}

\label{sec:core_path}

In this section, we show how to implement path-reporting robust core
data structures $\Core^{\pi}$ for distance scale $D$. We will assume
that a path-reporting compressed-graph $H$ for distance scale $d\ll D$
is given for. Unlike the algorithm for the distance-only $\Core$,
here we need to further assume that $H$ is defined from a path-reporting
covering $\cC$ with small outer-shell participation bound $\Delta$,
so that we can bound the simpleness of the maintained paths.
\begin{theorem}
[Path-reporting Robust Core]\label{thm:Core_path}\label{thm:Corepath} Let $G$ be an
$n$-vertex bounded-degree decremental graph. Suppose that a path-reporting
$(d,\gamma,\Delta,\beta)$-compressed graph $H$ of $G$ is explicitly
maintained for us. Moreover, we assume that either $H=G_{unit}$ as
defined in \Cref{prop:compressed base} or $H$ is defined from a path-reporting
covering $\cC$ with the outer-shell participation bound $\Delta$
via \Cref{prop:covering-compressed is compressed path}. Assuming that
$D\ge d\gamma$, we can implement a path-reporting robust core data
structure $\Core^{\pi}(G,\Kinit,D,7\lenapsp\Delta\beta)$ with scattering
parameter $\scatter=\Omegatil(\phicmg)$ and stretch $\stretch=\Otil(\gamma\lenapsp/\phicmg^{2})$
and total update time of 
\[
\Otil\left(T_{\Apxball^{\pi}}(G,\Kinit,\stretch\cdot D ,0.1,\beta)\Delta^{2}(D/d)^{3}\lenapsp/\phicmg\right)
\]
where $\lenapsp=\exp(\Theta(\log^{7/8}m))=n^{o(1)}$ is a parameter
that will be used later in \Cref{sec:together_path}. Let $(q_{\phi},q_{\path})$
bound the query-time overhead of both $\Apxball^{\pi}(G,S,\stretch \cdot D ,0.1,\beta)$
and the $(d,\gamma,\Delta,\beta)$-compressed graph $H$. Then, the
data structure has query-time overhead of
\[
(4q_{\phi},q_{\path}+\Otil(\frac{D}{d}\lenapsp/\phicmg^{2})\cdot q_{\phi}).
\]
\end{theorem}

The rest of this section is for proving \Cref{thm:Core_path}. The
organization is analogous to that of \Cref{sec:ball_path}. In \Cref{sec:core_path:data},
we describe data structures for maintaining the core set $K$ and
for supporting threshold-subpath queries, and then we analyze the
total update time. In \Cref{sec:core_path:define}, we define the implicitly
maintained paths $\pi(u,v)$ for all $u,v\in K$ as required by \Cref{def:Core path}
of $\Core^{\pi}$. Finally, we show an algorithm for answering threshold-subpath
queries in \Cref{sec:core_path:query}.

\subsection{Data Structures}

\label{sec:core_path:data}

In this section, we describe data structures needed for the $\Core^{\pi}$
data structure. First, we will need the following extension of the
expander pruning algorithm $\Prune$ from \Cref{lem:prune} that is
augmented with an \emph{all-pair-short-paths oracle} on the remaining
part of the expander. 
\begin{lem}
[Theorem 3.9 of \cite{ChuzhoyS20_apsp}]\label{thm:oracle}There is
an algorithm $\Prune^{\pi}(W,\phi)$ that, given an unweighted decremental
multi-graph $W=(V,E)$ that is initially a $\phi$-expander with $m$
edges where $\phi\ge\phicmg$, maintains a decremental set $X\subseteq V$
using $O(m\lenapsp)$ total update time such that $W[X]$ is a $\phi/6$-expander
at any point of time, and $\vol_{W}(V\setminus X)\le8i/\phi$ after
$i$ updates. Moreover, given a pair of vertices $u,v\in X$ at any
time, the algorithms returns a simple $u$-$v$ path in $W[X]$ of
length at most $\lenapsp$ in $O(\lenapsp)$ time.\footnote{This lemma is obtained by setting the parameter $q=O(\log^{1/8}m)$
in Theorem 3.9 of \cite{ChuzhoyS20_apsp}.}
\end{lem}

To describe the data structure, we simply replace the distance-only
components inside the $\Core$ data structure with the path-reporting
ones as follows:
\begin{itemize}
\item Replace the distance-only $(d,\gamma,\Delta)$-compressed graph from
the assumption of \Cref{thm:Core} by the path-reporting $(d,\gamma,\Delta,\beta)$-compressed
graph. 
\item Replace $\Prune(W_{multi},\phi)$ from \Cref{lne:whileExpanderWitnessIsLarge} of \Cref{alg:Core}
by $\Prune^{\pi}(W_{multi},\phi)$ from \Cref{thm:oracle} that support all-pair-short-paths
queries.
\item In addition to maintaining $\Apxball(G,X,4D,0.1)$ from \Cref{lne:maintainXBall} of \Cref{alg:Core}, we also maintain $\Apxball^{\pi}(G,X,\stretch\cdot D,0.1,\beta)$.
\end{itemize}
Finally, let $B^{\pi}$ contain all vertices $v$ whose distance
estimate maintained by $\Apxball^{\pi}(G,X,\stretch\cdot D,0.1,\beta)$
is at most $\frac{\stretch}{10}\cdot D$. So, $\ball_{G}(X,\frac{\stretch}{10}\cdot D)\subseteq B^{\pi}\subseteq\ball_{G}(X,1.1\cdot\frac{\stretch}{10}\cdot D)$.
We maintain an edge with minimum steadiness among all edges in $G[B^{\pi}]$
with weight at most $32D\log n$, denoted by $e_{\min}$. If there
are many edges with minimum steadiness, we break tie arbitrarily but
consistently through time (for example, we can fix an arbitrary order of edges and
let $e_{\min}$ be the first edge satisfied the condition). This completes
the description of the data structure.

With the above small modification, the maintained core set $K\subseteq\Kinit$
still guarantees the scattering property. (We prove the stretch property
later in \Cref{lem:core path len}.)
\begin{lem}
Let $\scatter=\Omegatil(\phicmg)$. At any point of time, $|\ball_{G}(v,2D)\cap\Kinit|\le(1-\scatter)\cdot|\Kinit|$
for all $v\in\Kinit\setminus K$. 
\end{lem}

\begin{proof}
\Cref{prop:path-reporting stronger} implies that we can replace the
distance-only components in $\Core$ with the stronger path-reporting
components because the guarantees of the outputs of these path-reporting
components never become weaker. Therefore, structural statements including
\Cref{cor:core scatter} from \Cref{sec:Core} still hold.
\end{proof}
The total update time after modification is slightly slower. Compared
to the running time of \Cref{thm:Core}, we replace a factor of $1/\phicmg$
by a factor of $\lenapsp$ and replace $T_{\Apxball}(\cdot)$ by $T_{\Apxball^{\pi}}(\cdot)$.
\begin{lem}
	The total update time is $\Otil\left(T_{\Apxball^{\pi}}(G,\Kinit,\stretch\cdot D,0.1,\beta)\Delta^{2}(D/d)^{3}\lenapsp/\phicmg\right)$.
\end{lem}

\begin{proof}
	Note that we assume the path-reporting $(d,\gamma,\Delta,\beta)$-compressed
	graph is maintained explicitly for us and so we do not count its update
	time. The proof of this lemma is the same as in the proof of \Cref{lem:core runtime}
	except that we replace $\Prune(W_{multi},\phicmg)$ whose total update
	time is $\Otil(|E(W_{multi})|/\phicmg)$ by $\Prune^{\pi}(W_{multi},\phicmg)$
	whose total update time is $\Otil(|E(W_{multi})|\lenapsp)$. Following
	exactly the same calculation in \Cref{lem:core runtime}, the total
	update time is 
	
	\[
	\Otil\left(T_{\Apxball}(G,\Kinit,32D\log n,0.1)\Delta^{2}(D/d)^{3}\lenapsp/\phicmg\right)
	\]
	basically by replacing a factor of $O(1/\phicmg)$ by $O(\lenapsp)$.
	However, since in addition to maintaining $\Apxball(G,X,4D,0.1)$
	from \Cref{alg:Core}, we also maintain $\Apxball^{\pi}(G,X,\stretch\cdot D,0.1,\beta)$.
	Following the same calculation, the total update time becomes
	
	\[
	\Otil\left(T_{\Apxball^{\pi}}(G,\Kinit,\stretch\cdot D,0.1,\beta)\Delta^{2}(D/d)^{3}\lenapsp/\phicmg\right).
	\]
	Note that $e_{\min}$ can be maintained using a heap and the total
	update time can be charged to the time spent by $\Apxball^{\pi}(G,X,\stretch\cdot D,0.1,\beta)$.
\end{proof}

From the above, we have proved the scattering property and bounded
the total update time of the algorithm for \Cref{thm:Core_path}.

\subsection{Defining The Implicitly Maintained Paths}

\label{sec:core_path:define}

In this section, for each pair of vertices $u,v\in K$, we define
a $u$-$v$ path $\pi(u,v)$ using \Cref{alg:define core_path}. We
emphasize that these paths $\pi(u,v)$ are \emph{not} maintained explicitly
and they are completely independent from the steadiness index $j$
in the queries. 
See \Cref{fig:core_path} for illustration.

\begin{algorithm}
Let $e_{\min}=(a,b)$ be the edge with minimum steadiness among all edges in $G[B^{\pi}]$ with weight at most $32D\log n$. \;

Set $\pi_{u},\pi_{v},\pi_{a},\pi_{b}$ as $\pi(X,u),\pi(X,v),\pi(X,a),\pi(X,b)$,
respectively, which are implicitly maintained by $\Apxball^{\pi}(G,X,\stretch \cdot D,0.1,\beta)$.\;

Let $u',v',a',b'\in X$ be such that $\pi_{u}=(u,\dots,u'),\pi_{v}=(v,\dots,v'),\pi_{a}=(a,\dots,a'),\pi_{b}=(b,\dots,b')$.\;

Let $\pi_{ua}^{W}$ be the $u'$-$a'$ path in $W$ obtained by querying
$\Prune^{\pi}(W_{multi},\phicmg)$.\;

Let $\widehat{\pi}_{ua}$ be the $u'$-$a'$
path in $\Hhat$ obtained by concatenating, for all embedded edges
$e\in\pi_{ua}^{W}$, the corresponding path $P_{e}$ in $\Hhat$.
By \Cref{def:Hhat} of $\Hhat$, we can write $\widehat{\pi}_{ua}=(\widehat{p}_{1},\dots,\widehat{p}_{t})$
where each $\widehat{p}_{i}$ is either a heavy-path or $\widehat{p}_{i}=(z,z')$
where $z$ and $z'$ are adjacent by a hyperedge in $H$. \label{enu:core path in Hhat}\;

\For{$i = 1$ {\normalfont{up to}} $t$}{
    \tcc{\textbf{(Hyper-edge)}}
    \If{$\widehat{p}_{i}=(z,z')$ where $z$ and
$z'$ are adjacent by a hyperedge in $H$}{
        $p_{i}\gets\pi^{H}(z,z')$ implicitly maintained by $H$
    }
    \tcc{\textbf{(Heavy-path)}}
    \If{$\widehat{p}_{i}=(z,\dots,z')$ is a heavy path}
    {$p_{i}\gets(z,z')$ where $(z,z')\in E(G)$.}
}
$\pi_{ua} \gets (p_{1},\dots,p_{t})$. \label{enu:core path G}\;

Let $\pi_{bv}^{W}$, $\widehat{\pi}_{bv}$ and $\pi_{bv}$ be the
$b'$-$v'$ path in $W,$ $\Hhat$ and $G$, respectively, analogous
to $\pi_{ua}^{W},\widehat{\pi}_{ua},\pi_{ua}$.\;

\Return $\pi(u,v)=(\pi_{u},\pi_{ua},\pi_{a},\{(a,b)\},\pi_{b},\pi_{bv},\pi_{v})$.
\caption{Computing $\pi(u,v)$ for each pair $u,v\in K$.}
\label{alg:define core_path}
\end{algorithm}

\begin{figure}
\begin{centering}
\includegraphics[width=0.4\paperwidth]{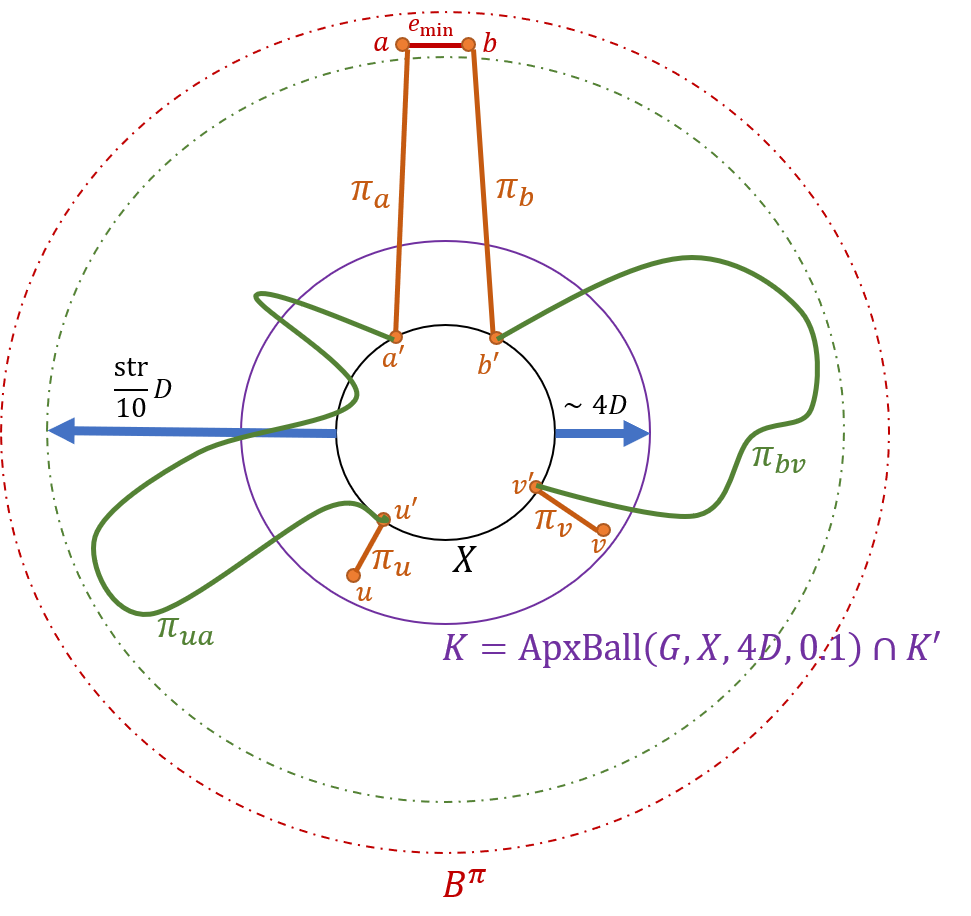}
\par\end{centering}
\caption{\label{fig:core_path}Illustration of $\pi(u,v)$ returned by \Cref{alg:define core_path}}
\end{figure}
Before analysis the properties of $\pi(u,v)$, we first argue that
it is indeed well-defined.
\begin{prop}
	\label{prop:core path well defined}For each pair $u,v\in K$, the
	path $\pi(u,v)$ defined by \Cref{alg:define core_path} is well-defined
	and is a $u$-$v$ path in $G$.
\end{prop}

\begin{proof}
	We have $u,v\in K\subseteq\Apxball(G,X,4D,0.1)$ by Line \ref{lne:maintainXBall}
	of \Cref{alg:Core}. Hence, we also have that $u,v\in\Apxball^\pi(G,X,\stretch\cdot D,0.1,\beta)$
	and so $\pi_{u}$ and $\pi_{v}$ are well-defined. By definition of
	$e_{\min}$, we have $a,b\in \Apxball^\pi(G,X,\stretch\cdot D,0.1,\beta)$.
	Hence, $\pi_{a}$ and $\pi_{b}$ are well-defined too. Lastly, as
	$u',v',a',b'\in X$, the paths $\pi_{ua}^{W}$ and $\pi_{bv}^{W}$
	can be queried from $\Prune^{\pi}(W_{multi},\phicmg)$. Then, $\widehat{\pi}_{ua}$
	and $\widehat{\pi}_{bv}$ are can be defined from $\pi_{ua}^{W}$ and $\pi_{bv}^{W}$ because of the embedding
	$\cP_{W}$ of $W$. By construction of $\Hhat$, the paths $\pi_{ua}$
	and $\pi_{bv}$ are well-defined as well. 
	
	Since the endpoints of $\pi_{u},\pi_{ua},\pi_{a},\{(a,b)\},\pi_{b},\pi_{bv},\pi_{v}$
	are $(u,u')$, $(u',a')$, $(a',a)$, $(a,b)$, $(b,b')$, $(b',v')$,
	$(v',v)$, respectively, we have $\pi(u,v)$ is indeed a $u$-$v$
	path. As all subpaths of $\pi(u,v)$ are well-defined, $\pi(u,v)$ is well-defined too.
\end{proof}

Next, we introduce notations about more fine-grained structure of
the path $\pi_{ua}$. (It is symmetric for $\pi_{bv}$.) 
Consider \Cref{enu:core path in Hhat} of \Cref{alg:define core_path}
where we have $\widehat{\pi}_{ua}=(\widehat{p}_{1},\dots,\widehat{p}_{t})$.
If $\widehat{p}_{i}$ is a heavy path, then we say that $\widehat{p}_{i}$
is of type \emph{heavy-path}. Otherwise, we say that $\widehat{p}_{i}$
is of type \emph{hyper-edge}. Recall that $\cP_{W}$ is the embedding
of $W$ into $\Hhat$. We can write $\pi_{ua}^{W}=(e_{1},\dots,e_{|\pi_{ua}^{W}|})$
and $\widehat{\pi}_{ua}=(P_{e_{1}},\dots,P_{e_{|\pi_{ua}^{W}|}})$
where each $P_{e_{j}}\in\cP_{W}$ is the path in the embedding corresponding
to $e_{j}\in E(W)$. As $P_{e_{j}}$ is a subpath of $\widehat{\pi}_{ua}$
and has endpoints in $V(G)$, we can write $P_{e_{j}}=(\widehat{p}_{j,1},\dots,\widehat{p}_{j,t_{j}})$
as a subsequence of $(\widehat{p}_{1},\dots,\widehat{p}_{t})$. Observe
that we have
\begin{align}
	\widehat{\pi}_{ua} & =(\widehat{p}_{1},\dots,\widehat{p}_{t})=(\widehat{p}_{1,1},\dots,\widehat{p}_{1,t_{1}},\widehat{p}_{2,1},\dots,\widehat{p}_{2,t_{2}},\dots,\widehat{p}_{|\pi_{ua}^{W}|,1},\dots,\widehat{p}_{|\pi_{ua}^{W}|,t_{|\pi_{ua}^{W}|}})\label{eq:pihat_ua}\\
	\pi_{ua} & =(\pi_{1},\dots\pi_{t})=(p_{1,1},\dots,p_{1,t_{1}},p_{2,1},\dots,p_{2,t_{2}},\dots,p_{|\pi_{ua}^{W}|,1},\dots,p_{|\pi_{ua}^{W}|,t_{|\pi_{ua}^{W}|}})\label{eq:pi_ua}
\end{align}
where $p_{j,k}$ is a path in $G$ corresponding to the path $\widehat{p}_{j,k}$
in $\Hhat$ assigned at \Cref{enu:core path G}. 
We emphasize that the path $(\widehat{p}_{j,1},\dots,\widehat{p}_{j,t_{j}})$ is not the same as $\widehat{p}_{j}$; $(\widehat{p}_{j,1},\dots,\widehat{p}_{j,t_{j}})$
is just some subsequence of the sequence $(\widehat{p}_{1},\dots,\widehat{p}_{t})$.
We will usually use subscript $i$ for $\widehat{p}_{i}$ and subscript
$(j,k)$ for $\widehat{p}_{j,k}$.
For each subpath $\widehat{p}_{j,k}$
of $P_{e_{j}}$, if $\widehat{p}_{j,k}$ is of type hyper-edge, then
we say that the corresponding path $p_{j,k}$ is of type hyper-edge
as well. Otherwise, $p_{j,k}$ is of type heavy-path.

We will below argue the correctness of the path $\pi(u,v)$ defined
by \Cref{alg:define core_path}. We start by bounding the length of
$\pi(u,v)$. We first bound the length of $\pi_{ua}$ and $\pi_{bv}$ which
is the only non-trivial case. The moreover part of the statement below
will be used in the next subsection.
\begin{prop}
	\label{lem:core path still near}\label{prop:expander path}We can
	choose the polylogarithmic factor in $\stretch=\Otil(\gamma\lenapsp/\phicmg^{2})$
	so that the following holds. The length of $\pi_{ua}$ is at most
	$w(\pi_{ua})=\Otil(D\gamma\lenapsp/\phicmg^{2})\le\frac{\stretch}{10}\cdot D$.
	Moreover, $\pi_{ua}$ is contained inside $G[\ball_{G}(X,\frac{\stretch}{10}\cdot D)]$
	and every edge of $\pi_{ua}$ has weight at most $32D\log n$. Symmetrically,
	the same holds for $\pi_{bv}$.
\end{prop}

\begin{proof}
	We show the argument only for $\pi_{ua}$ because the argument is
	symmetric for $\pi_{bv}$. Recall that $\ensuremath{\len(\cP_{W})}$
	is the maximum number of edges inside paths in $\cP_{W}$. We have
	$\len(\cP_{W})=\ensuremath{O(\frac{\kappa(\Vhat)}{|\Kinit|\epswit^{2}})=\Otil(\frac{D}{d}/\phicmg^{2})}$
	because (1) \Cref{lem:embedwitness} guarantees that $\len(\cP_{W})=O(\frac{\kappa(\Vhat)}{|\Kinit|\epswit^{2}})$
	as $K'\subseteq\Kinit$ and $\epswit=\phicmg/\log^{2}(n)$ was defined
	in the same lemma, and (2) we have $\kappa(\Vhat)\le\Otil(|\Kinit|\frac{D}{d})$
	by \Cref{lem:total cap}. To bound $w(\pi_{ua})$, observe that the
	path $\widehat{\pi}_{ua}$ contains at most 
	\begin{equation}
		|\widehat{\pi}_{ua}|\le|\pi_{ua}^{W}|\cdot\len(\cP_{W})=\Otil(\frac{D}{d}\lenapsp/\phicmg^{2})\label{eq:len pihat_ua}
	\end{equation}
	edges in $\Hhat$ where $|\pi_{ua}^{W}|\le\lenapsp$ by \Cref{thm:oracle}.
	Next, we will show $w(\pi_{ua})=O(|\widehat{\pi}_{ua}|\cdot d\gamma)$
	which in turn implies that $w(\pi_{ua})=\Otil(D\gamma\lenapsp/\phicmg^{2})$.
	It suffices to show that $w(p_{i})\le O(|\widehat{p}_{i}|\cdot d\gamma)$
	for each $1\le i\le t$. There are two cases. 
	\begin{itemize}
		\item If $p_{i}$ is of type hyper-edge, then $w(p_{i})=w(\pi^{H}(z,z'))\le d\gamma$
		by the guarantee of the path-reporting $(d,\gamma,\Delta,\beta)$-compressed
		graph $H$. So $w(p_{i})\le d\gamma=d\gamma|\widehat{p}_{i}|$ as
		$\widehat{p}_{i}=(z,z')$. Also, note that each edge in $\pi^{H}(z,z')$
		obviously has weight at most $w(\pi^{H}(z,z'))\le d\gamma\le32D\log n$
		by the assumption in \Cref{thm:Core_path}.
		\item If $p_{i}$ is of type heavy-path, then $\widehat{p}_{i}=(z,\dots,z')$
		is a heavy path and we have $|\widehat{p}_{i}|=\ceiling{w(z,z')/d}$
		by the construction of $\Hhat$. So $w(p_{i})=w(z,z')=O(|\widehat{p}_{i}|d)$
		again. Also, by construction of $\Hhat$ (see \Cref{def:Hhat}), we
		have $(z,z')\in E(G)$ and $w(z,z')\le32D\log n$.
	\end{itemize}
	As we can freely choose the polylogarithmic factor in the definition
	of $\stretch=\Otil(D\gamma\lenapsp/\phicmg^{2})$, we can choose it so that $w(\pi_{ua})\le\frac{\stretch}{10}\cdot D$.
	As both endpoints of $\pi_{ua}$ are inside $X$, we have that $\pi_{ua}$
	is contained inside $G[\ball_{G}(X,\frac{\stretch}{10}\cdot D)]$.
	From the analysis of the two cases above, we also have that every edge in $\pi_{ua}$ 	has weight at most $32D\log n$. 
\end{proof}
Now, we can bound the length of $\pi(u,v)$ and, hence, bounding the
stretch of $\Core^{\pi}$ (as required by \Cref{def:Core path}). 
\begin{lem}
	\label{lem:core path len}For each pair $u,v\in K$, the path $\pi(u,v)$
	defined by \Cref{alg:define core_path} has length at most $\stretch\cdot D$.
\end{lem}

\begin{proof}
	We only need bound the length of each path in $\{\pi_{u},\pi_{ua},\pi_{a},\{(a,b)\},\pi_{b},\pi_{bv},\pi_{v}\}$.
	We have $\pi_{u},\pi_{v}$ have length at most $1.1\cdot4D\le5D$
	because $u,v\in K\subseteq\Apxball(G,X,4D,0.1)$ by Line \Cref{lne:maintainXBall}
	of \Cref{alg:Core}. Next, 	by definition
	of $e_{\min}$, $\pi_{a},\pi_{b}$ have length at most $1.1\cdot\frac{\stretch}{10}\cdot D$ and $w(a,b)\le32D\log n$. Lastly, $\pi_{ua}$ and
	$\pi_{bv}$ have length at most $\frac{\stretch}{10}\cdot D$ by \Cref{prop:expander path}.
	In total, $w(\pi(u,v))\le2(5+1.1\cdot\frac{\stretch}{10}+\frac{\stretch}{10})\cdot D+32D\log n\le\stretch\cdot D$.
\end{proof}

Next, we bound the simpleness of $\pi(u,v)$.
\begin{lem}
\label{lem:core path simple}For each pair $u,v\in K$, the path $\pi(u,v)$
defined by \Cref{alg:define core_path} is $7\lenapsp\Delta\beta$-simple.
\end{lem}

\begin{proof}
The main task is to prove that $\pi_{ua}$ is $3\lenapsp\Delta\beta$-simple
(and the argument for $\pi_{bv}$ is analogous). Given this fact,
as $\pi_{u},\pi_{a},\pi_{b},\pi_{v}$ are $\beta$-simple by $\Apxball^{\pi}(G,X,\stretch\cdot D,0.1,\beta)$
and $\{(a,b)\}$ is trivially $1$-simple, the simpleness of $\pi(u,v)=(\pi_{u},\pi_{ua},\pi_{a},\{(a,b)\},\pi_{b},\pi_{bv},\pi_{v})$
can be at most $6\lenapsp\Delta\beta+4\beta+1\le7\lenapsp\Delta\beta$.

Now, we show that $\pi_{ua}$ is $3\lenapsp\Delta\beta$-simple. For
each subpath $p_{j,k}$ of $\pi_{ua}$ from \Cref{eq:pi_ua}, note
that $p_{j,k}$ is a $\beta$-simple path in $G$ because we have
either $p_{j,k}=\pi^{H}(z,z')$ is $\beta$-simple by simpleness guarantee
of $H$ or $p_{j,k}=\{(z,z')\}$ where $(z,z')\in E(G)$ is trivially
$1$-simple. The key claim is that, for any vertex $x\in V(G)$ and
index $j$, the number of subpaths from $\{p_{j,k}\}_{k}$ that $x$
can participate is at most $\Delta+2$ (i.e.~$|\{k\mid x\in p_{j,k}\}|\le\Delta+2$).
As $|\pi_{ua}^{W}|\le\lenapsp$ by \Cref{thm:oracle}, this would imply
that $\pi_{ua}$ has simpleness at most $\lenapsp(\Delta+2)\beta\le3\lenapsp\Delta\beta$
as desired. We finish by proving the claim:
\begin{claim}
For any vertex $x\in V(G)$ and index $j$, $|\{k\mid x\in p_{j,k}\}|\le\Delta+2$.
\end{claim}

\begin{proof}
From the assumption of \Cref{thm:Core_path}, there are two cases:
either $H=G_{unit}$ defined in \Cref{prop:compressed base} or $H$
is defined from a path-reporting covering $\cC$ via \Cref{prop:covering-compressed is compressed path}.
In both cases, we will use the fact that $P_{e_{j}}$ is a simple
path in $\Hhat$ guaranteed by \Cref{lem:embedwitness}.

Suppose that $H=G_{unit}\subseteq G$. We claim that $(p_{j,1},\dots,p_{j,t_{j}})$
is a simple path in $G$ and so $|\{k\mid x\in p_{j,k}\}|\le2$. The
claim holds because, for each subpath $\widehat{p}_{j,k}$ of $P_{e_{j}}$,
if $\widehat{p}_{j,k}=(z,z')$ is of type hyper-edge, then $p_{j,k}=\widehat{p}_{j,k}=(z,z')\in E(G)$,
and if $\widehat{p}_{j,k}=(z,\dots,z')$ is of type heavy-path, then
$p_{j,k}=(z,z')\in E(G)$. As $P_{e_{j}}=(\widehat{p}_{j,1},\dots,\widehat{p}_{j,t_{j}})$
is simple, the path $(p_{j,1},\dots,p_{j,t_{j}})$ must be simple
as well. 

Next, suppose that $H$ is defined from a path-reporting covering
$\cC$ via \Cref{prop:covering-compressed is compressed path}. We first
argue that $|\{k\mid x\in p_{j,k}$ and $p_{j,k}$ is of type heavy-path$\}|\le2$.
To see this, observe that all type-heavy-path $p_{j,k}$ form a collection
of disjoint simple paths in $G$, which is a subgraph of $G$ with
degree at most 2. This is because each heavy path $\widehat{p}_{j,k}=(z,\dots,z')$
in $\Hhat$ corresponds to $p_{j,k}=(z,z')$ in $G$ but $P_{e_{j}}=(\widehat{p}_{j,1},\dots,\widehat{p}_{j,t_{j}})$
is simple. So $x$ can appear in at most 2 type-heavy-path paths $p_{j,k}$.
It remains to show that $|\{k\mid x\in p_{j,k}$ and $p_{j,k}$ is of type hyper-edge$\}|\le\Delta$.
As $P_{e_{j}}$ is a simple path in $\Hhat$, each type-hyper-edge
$\widehat{p}_{j,k}$ must correspond to a \emph{unique} core $C_{j,k}$
from the covering $\cC$. Suppose that $x\in p_{j,k}=\pi^{H}(z,z')$.
By \Cref{prop:covering-compressed is compressed path}, we have $\pi^{H}(z,z')=\pi^{\cC}(z,C_{j,k})\circ\pi^{\cC}(z_{C},z'_{C})\circ\pi^{\cC}(C_{j,k},z')$
where $z_{C},z'_{C}\in C_{j,k}$, $\pi^{\cC}(z,C_{j,k})=(z,\dots,z_{C})$
and $\pi^{\cC}(C_{j,k},z')=(z'_{C},\dots,z')$ are implicitly maintained
by $\Apxball^{\pi}$ that maintains $\shell(C_{j,k})$ and $\pi^{\cC}(z_{C},z'_{C})$
is implicitly maintained by $\Core^{\pi}$ that maintains the core
$C_{j,k}$ in the covering $\cC$. By \Cref{claim:in outershell} (with
different notations), we have that $x\in\oshell(C_{j,k})$. Therefore,
the outer-shell participation bound $\Delta$ of $\cC$ implies that
$x$ can appear in at most $\Delta$ type-hyper-edge paths $p_{j,k}$
as desired.
\end{proof}
\end{proof}
\Cref{lem:core path len} and \Cref{lem:core path simple} together
imply that $\pi(u,v)$ indeed satisfies all conditions required by
$\Core^{\pi}(G,\Kinit,D,7\lenapsp\Delta\beta)$ with stretch $\stretch$
as required by \Cref{def:Core path}. 

\subsection{Threshold-Subpath Queries}

\label{sec:core_path:query}

In this section, we describe in \Cref{alg:return core_path} below
how to process the threshold-subpath query that, given a pair of vertices
$u,v\in K$ and a steadiness index $j$, return $\sigma_{\le j}(\pi(u,v))$
consisting of all edges of $\pi(u,v)$ with steadiness at most $j$.

\begin{algorithm}
Let $e_{\min}=(a,b)$ be the edge with minimum steadiness among all edges in $G[B^{\pi}]$ with weight at most $32D\log n$.\;

\label{enu:query ball}Set $\sigma_{\le j}(\pi_{u}),\sigma_{\le j}(\pi_{v}),\sigma_{\le j}(\pi_{a}),\sigma_{\le j}(\pi_{b})$
as $\sigma_{\le j}(\pi(X,u))$, $\sigma_{\le j}(\pi(X,v))$, $\sigma_{\le j}(\pi(X,a))$,
$\sigma_{\le j}(\pi(X,b))$, respectively, by querying $\Apxball^{\pi}(G,X,\stretch \cdot D,0.1,\beta)$.\;

Let $u',v',a',b'\in X$ be such that $\pi_{u}=(u,\dots,u')$, $\pi_{v}=(v,\dots,v')$,
$\pi_{a}=(a,\dots,a')$, $\pi_{b}=(b,\dots,b')$.\;

\If(\label{enu:return j small}){ $j<\sigma(e_{\min})$}{\Return
$\sigma_{\le j}(\pi(u,v))=\sigma_{\le j}(\pi_{u})\cup\sigma_{\le j}(\pi_{v})\cup\sigma_{\le j}(\pi_{a})\cup\sigma_{\le j}(\pi_{b})$.}
 
\label{enu:query expander}Let $\pi_{ua}^{W}$ be the $u'$-$a'$
path in $W$ obtained by querying $\Prune^{\pi}(W_{multi},\phicmg)$.\;
Let $\widehat{\pi}_{ua}$ be the $u'$-$a'$ path in $\Hhat$ obtained
by concatenating, for all embedded edges $e\in\pi_{ua}^{W}$, the
corresponding path $P_{e}$ in $\Hhat$. By \Cref{def:Hhat} of $\Hhat$,
we can write $\widehat{\pi}_{ua}=(\widehat{p}_{1},\dots,\widehat{p}_{t})$
where each $\widehat{p}_{i}$ is either a heavy-path or $\widehat{p}_{i}=(z,z')$
where $z$ and $z'$ are adjacent by a hyperedge in $H$. \;
\For{$i=0$ \normalfont{up to} $t$}{
    \tcc{\textbf{(Hyper-edge)}}
    \If{$\widehat{p}_{i}=(z,z')$ where $z$ and
$z'$ are adjacent by a hyperedge in $H$}{
        $\sigma_{\le j}(p_{i})\gets\sigma_{\le j}(\pi^{H}(z,z'))$ by querying the data structure on $H$.
    }
    \tcc{\textbf{(Heavy-path)}}
    \If{$\widehat{p}_{i}=(z,\dots,z')$ is a heavy path}{$\sigma_{\le j}(p_{i})\gets\sigma_{\le j}((z,z'))$
where $(z,z')\in E(G)$.}
}
$\sigma_{\le j}(\pi_{ua})=\sigma_{\le j}(p_{1})\cup\dots\cup\sigma_{\le j}(p_{t})$ \label{enu:query compressed graph}\;
Let $\pi_{bv}^{W}$, $\widehat{\pi}_{bv}$ and $\sigma_{\le j}(\pi_{bv})$
be analogous to $\pi_{ua}^{W},\widehat{\pi}_{ua},\sigma_{\le j}(\pi_{bv})$,
respectively.\;
\label{enu:return general}\Return $\sigma_{\le j}(\pi_{u})\cup\sigma_{\le j}(\pi_{ua})\cup\sigma_{\le j}(\pi_{a})\cup\{(a,b)\}\cup\sigma_{\le j}(\pi_{b})\cup\sigma_{\le j}(\pi_{bv})\cup\sigma_{\le j}(\pi_{v})$.
\caption{\label{alg:return core_path}Returning $\sigma_{\le j}(\pi(u,v))$,
given $u,v\in K$ and a steadiness index $j$}
\end{algorithm}
\begin{lem}
Given $u,v\in K$ and a steadiness index $j$, \Cref{alg:return core_path}
returns $\sigma_{\le j}(\pi(u,v))$ where $\pi(u,v)$ is defined in
\Cref{alg:define core_path}.
\end{lem}

\begin{proof}
Observe that all steps in \Cref{alg:return core_path} are completely
analogous to the steps in \Cref{alg:define core_path} except that
we collect only edges with steadiness at most $j$ into the answer
and we add \Cref{enu:return j small} for efficiency. Thus, we indeed
have $\sigma_{\le j}(\pi(u,v))=\sigma_{\le j}(\pi_{u})\cup\sigma_{\le j}(\pi_{ua})\cup\sigma_{\le j}(\pi_{a})\cup\{(a,b)\}\cup\sigma_{\le j}(\pi_{b})\cup\sigma_{\le j}(\pi_{bv})\cup\sigma_{\le j}(\pi_{v})$
and the answer is correct if \Cref{alg:return core_path} returns at
\Cref{enu:return general}. Next, recall that $e_{\min}$ is defined
as the edge with minimum steadiness among all edges in $G[B^{\pi}]$ with weight at most $32D\log n$.
As $\ball_G(X,\frac{\stretch}{10}\cdot D) \subseteq B^{\pi}$, this edge set also contains the
whole path of $\pi_{ua}$ and $\pi_{bv}$ by the ``moreover'' part of \Cref{lem:core path still near}.
Thus, if $j<\sigma(e_{\min})$, then $\sigma_{\le j}(\pi_{ua})=\emptyset$
and $\sigma_{\le j}(\pi_{bv})=\emptyset$. So if \Cref{alg:return core_path}
returns at \Cref{enu:return j small}, then the answer is correct as
well.
\end{proof}
Recall that $(q_{\phi},q_{\path})$ bounds the query-time overhead
of both $\Apxball^{\pi}(G,X,\stretch \cdot D,0.1,\beta)$ and $(d,\gamma,\Delta,\beta)$-compressed
graph $H$. We will show the query-time overhead for our $\Core^{\pi}$
data structure is $(4q_{\phi},q_{\path}+\Otil(\frac{D}{d}\lenapsp/\phicmg^{2})\cdot q_{\phi})$
as required by \Cref{thm:Core_path}.
\begin{lem}
Given $u,v\in K$ and a steadiness index $j$, \Cref{alg:return core_path}
takes $4q_{\phi}$ time if $\sigma_{\le j}(\pi(u,v))=\emptyset$.
Otherwise, it takes at most $|\sigma_{\le j}(\pi(u,v))|\cdot(q_{\path}+\Otil(\frac{D}{d}\lenapsp/\phicmg^{2})\cdot q_{\phi})$
time.
\end{lem}

\begin{proof}
If $\sigma_{\le j}(\pi(u,v))=\emptyset$, then \Cref{alg:return core_path}
must return at \Cref{enu:return j small} (otherwise $e_{\min}=(a,b)\in\sigma_{\le j}(\pi(u,v))$).
In this case, we just query $\Apxball^{\pi}(G,X,\stretch \cdot D,0.1,\beta)$ four
times which takes at most $4q_{\phi}$ time. 

Now suppose that $\sigma_{\le j}(\pi(u,v))\neq\emptyset$. At \Cref{enu:query ball}
we make path-query to $\Apxball^{\pi}(G,X,\stretch \cdot D,0.1,\beta)$ at most
$4$ times. At \Cref{enu:query expander}, it takes $O(\lenapsp)$
time to obtain $\pi_{ua}^{W}$. Constructing $\widehat{\pi}_{ua}$
takes $|\widehat{\pi}_{ua}|\le\Otil(\frac{D}{d}\lenapsp/\phicmg^{2})$
time by \Cref{eq:len pihat_ua}. At \Cref{enu:query compressed graph},
the algorithm makes at most $|\widehat{\pi}_{ua}|$ queries to $(d,\gamma,\Delta,\beta)$-compressed
graph $H$ (for the hyper-edge case) and spends additional $O(|\widehat{\pi}_{ua}|)$
time (for the heavy-path case) to obtain $\sigma_{\le j}(\pi_{ua})$.
We do the same to obtain $\pi_{bv}^{W},\widehat{\pi}_{bv}$ and $\sigma_{\le j}(\pi_{bv})$.
In total the running time is at most 
\[
|\sigma_{\le j}(\pi(u,v))|\cdot q_{\path}+(4+|\widehat{\pi}_{ua}|)q_{\phi}+O(\lenapsp+|\widehat{\pi}_{ua}|)\le|\sigma_{\le j}(\pi(u,v))|\cdot(q_{\path}+\Otil(\frac{D}{d}\lenapsp/\phicmg^{2})\cdot q_{\phi})
\]
where the first term is the total query time to both $\Apxball^{\pi}(G,X,\stretch \cdot D,0.1,\beta)$
and $H$ when they return non-empty subpaths, the second term is the
total query time that $\Apxball^{\pi}(G,X,\stretch \cdot D,0.1,\beta)$ and $H$
when they return an empty set. The inequality holds is because $|\sigma_{\le j}(\pi(u,v))|\ge1$.
\end{proof}

%% file: together_path_2.tex
\section{Putting Path-reporting Components Together}

\label{sec:together_path}

In this section, we show how to recursively combine all path-reporting
data structures including $\Apxball^{\pi}$(\Cref{def:Apxball path}),
$\Core^{\pi}$ (\Cref{def:Core path}), and path-reporting Covering
(\Cref{def:Covering path}) to obtain the desired decremental path-reporting
$\SSSP^{\pi}$ data structure. The goal of this section is to prove
the following theorem. 
\begin{theorem}
\label{thm:main path}For any $n$ and $\eps\in(\phicmg,1/500)$,
let $G=(V,E)$ be a decremental bounded-degree graph with $n$ vertices,
edge weights from $\{1,2,\dots,W=n^{5}\}$, and edge steadiness from
$\{0,\dots,\sigma_{\max}\}$ where $\sigma_{\max}=o(\log^{3/4}n)$.
Let $S\subseteq V$ be any decremental set. We can implement $\Apxball^{\pi}(G,S,\eps,\Ohat(1))$
that has $\Ohat(n)$ total update time and query-time overhead of
$(O(\log^{2}n),\Ohat(1))$. 
\end{theorem}

As $\SSSP^{\pi}$ is a special case of $\Apxball^{\pi}$ when the
source set $S=\{s\}$, by applying the reduction from \Cref{prop:simplify
path}, we immediately obtain \Cref{thm:main SSSP path}, the main
result of this part of the paper. It remains to prove \Cref{thm:main
path}.

Define $G_{j}=\sigma_{\ge j}(G)$ for each $j\in\{0,\dots,\sigma_{\max}\}$.
Note that $G_{0}=G$ and $G_{\sigma_{\max}+1}=\emptyset$. There are
\emph{distance scales} $D_{0}\le D_{1}\le\dots\le D_{\distScale}$
where $D_{i}=(nW)^{i/\distScale}$ and $\distScale=c_{\distScale}\lg\lg\lg n$
for some small constant $c_{\distScale}>0$. We will implement our
data structures for $\distScale$ many levels. Recall that $\phicmg=1/2^{\Theta(\log^{3/4}n)}=\Omegahat(1)$
and $\lenapsp=2^{\Theta(\log^{7/8}n)}=\Ohat(1)$. For $0\le i\le\distScale$
and $0\le j\le\sigma_{\max}$, we set 
\begin{align*}
k_{i} & =(\lg\lg n)^{3^{i}}\\
\gamma_{i} & =\lenapsp^{2k_{i+1}}\text{ and }\gamma_{-1}=1\\
\eps_{i,j} & =\eps/(600^{\distScale-i}\cdot2^{j})\text{ and }\eps_{i,\sigma_{\max}+1}=0\\
\stretch_{i} & =\gamma_{i-1}\cdot\lenapsp\log^{c_{\stretch}}n/\phicmg^{2}\\
\Delta_{i} & =\Theta(k_{i}n^{2/k_{i}}/\phicmg)\\
\beta_{i} & =\beta_{i-1}\cdot21\lenapsp\Delta_{i-1}\text{ and }\beta_{0}=7\lenapsp
\end{align*}
where we define $c_{\stretch}$ to be a large enough constant. We
also define parameters related to query-time overhead, for $0\le i\le\distScale$
and $0\le j\le\sigma_{\max}$, as follows:

\begin{align*}
q_{\phi}^{(i,j)} & =c_{q}(\sigma_{\max}-j+1)\log n\\
Q_{\phi}^{(i,j)} & =c_{q}(\sigma_{\max}-j+1)12^{i}\log n\\
\overhead_{\path} & =n^{2/\distScale}\cdot12^{\distScale}\cdot2^{c_{q}\cdot\log^{8/9}n}\\
q_{\path}^{(i,j)} & =(2i+1)\cdot(\sigma_{\max}-j+1)\cdot\overhead_{\path}\\
Q_{\path}^{(i,j)} & =2(i+1)\cdot(\sigma_{\max}-j+1)\cdot\overhead_{\path}
\end{align*}
where $c_{q}$ is a large enough constant. The parameters are defined
in such that way that 
\[
n^{1/\distScale},\frac{D_{i}}{D_{i-1}},n^{1/k_{i}},\gamma_{i},1/\eps_{i,j},\stretch_{i},\Delta_{i},\beta_{i},q_{\phi}^{(i,j)},Q_{\phi}^{(i,j)},q_{\path}^{(i,j)},Q_{\path}^{(i,j)}=\Ohat(1)
\]
for all $0\le i\le\distScale$ and $0\le j\le\sigma_{\max}$. However,
we will need a more fine-grained property of them as described below. 
\begin{prop}
\label{prop:para path}For large enough $n$ and for all $0\le i\le\distScale$
and $0\le j\le\sigma_{\max}$, we have that 
\end{prop}

\begin{enumerate}
\item \label{enu:para:k path}$\lg\lg n\le k_{i}\le(\lg^{1/100}n)$, 
\item \label{enu:para:eps path}$\phicmg^{2}\le\eps_{i,j}\le\eps$ and $\eps_{i,j}=300\eps_{i-1,j}+\eps_{i,j+1}$
(in particular, $\eps_{i,j}\ge\eps_{i-1,j},\eps_{i,j+1}$ and $\eps_{i,j}\ge\eps_{i,\sigma_{\max}}\ge\eps_{0,\sigma_{\max}}$), 
\item \label{enu:para:gamma path}$\gamma_{i},\stretch_{i}=2^{O(\lg^{8/9}n)}$, 
\item \label{enu:para:D ratio path}$D_{i}/D_{i-1}\le n^{6/\distScale}$, 
\item \label{enu:para:gap and distance scale path}$\gamma_{i}\le D_{i}/D_{i-1}$, 
\item \label{enu:para:key path}$\gamma_{i}\ge(\gamma_{i-1}\lenapsp^{2})^{k_{i}}\ge(\frac{\stretch_{i}}{\eps_{i,\sigma_{\max}}})^{k_{i}}$,
and 
\item \label{enu:para:beta}$\beta_{i}=2^{O((\log^{7/8}n)(\log\log\log n))}\cdot n^{O(1/\lg\lg n)}$. 
\end{enumerate}
\begin{proof}
(1): We have $k_{i}=(\lg\lg n)^{3^{i}}\le(\lg\lg n)^{3^{\distScale}}\leq(\lg\lg n)^{(\lg\lg n)^{1/100}}\le\lg^{1/100}n$
as $\distScale=c_{\distScale}\lg\lg\lg n$ and $c_{\distScale}$ is
a small enough constant.

(2): It is clear that $\eps_{i,j}\le\eps$. For the other direction,
note that in the assumption of \Cref{thm:main no distance}, we
have $\eps\ge\phicmg$. So $\eps_{i}\ge\eps/(600^{\distScale}\cdot2^{\sigma_{\max}})\ge\phicmg/2^{\Theta(\lg\lg\lg n+\sigma_{\max})}\ge\phicmg^{2}$
because $\phicmg=1/2^{\Theta(\lg^{3/4}n)}$ and $\sigma_{\max}=o(\log^{3/4}n)$.
Next, we have $300\eps_{i-1,j}+\eps_{i,j+1}=\eps(\frac{300}{600^{\distScale-(i-1)}\cdot2^{j}}+\frac{1}{600^{\distScale-i}\cdot2^{j+1}})=\eps(\frac{300/600}{600^{\distScale-i}\cdot2^{j}}+\frac{1/2}{600^{\distScale-i}\cdot2^{j}})=\eps(\frac{1}{600^{\distScale-i}\cdot2^{j}})=\eps_{i,j}$.

(3): As $\lenapsp=2^{\Theta(\lg^{7/8}n)}$ and $\gamma_{i}=\lenapsp^{2k_{i+1}}$,
we have from \Cref{enu:para:k} of this proposition that $\gamma_{i}=2^{O(\lg^{7/8+1/100}n)}=2^{O(\log^{8/9}n)}$.
Also, as $\stretch_{i}=\gamma_{i-1}\cdot\lenapsp\log^{c_{\stretch}}n/\phicmg^{2}$,
we have $\stretch_{i}=2^{O(\log^{8/9}n)}$ too.

(4): We have $D_{i}/D_{i-1}=(nW)^{1/\distScale}$. Since $W=n^{5}$,
we have $D_{i}/D_{i-1}\le n^{6/\distScale}$.

(5): As $D_{i}/D_{i-1}\ge n^{1/\distScale}\ge2^{\Theta(\lg n/\lg\lg\lg n)}$,
by \Cref{enu:para:gamma} we have $\gamma_{i}\le(D_{i}/D_{i-1})$
when $n$ is large enough.

(6): We have $\gamma_{i}\ge(\gamma_{i-1}\lenapsp^{2})^{k_{i}}$ because
\[
\gamma_{i}=\lenapsp^{2}{}^{k_{i+1}}\ge\lenapsp^{2(k_{i}^{2}+k_{i})}=(\lenapsp^{2k_{i}}\cdot\lenapsp^{2})^{k_{i}}=(\gamma_{i-1}\lenapsp^{2})^{k_{i}}
\]
where the inequality holds is because $k_{i+1}=(\lg\lg n)^{3^{i+1}}=(\lg\lg n)^{3^{i}\cdot2}\times(\lg\lg n)^{3^{i}}\ge(\lg\lg n)^{3^{i}\cdot2}+(\lg\lg n)^{3^{i}}=k_{i}^{2}+k_{i}$
for all $i\ge0$. For the second inequality, we have 
\[
\frac{\stretch_{i}}{\eps_{i,\sigma_{\max}}}=\frac{\gamma_{i-1}\lenapsp\log^{c_{\stretch}}n/\phicmg^{2}}{\eps_{i,\sigma_{\max}}}\le\frac{\gamma_{i-1}\lenapsp}{\phicmg^{6}}\le\gamma_{i-1}\lenapsp^{2}
\]
because $(\log^{c_{\stretch}}n)\le1/\phicmg$, $\eps_{i,\sigma_{\max}}\ge\phicmg^{2}$
by \Cref{enu:para:eps path}, and $\lenapsp\ge\poly(1/\phicmg)$
when $n$ is large enough. Therefore, $\gamma_{i}\geq(\gamma_{i-1}\lenapsp^{2})^{k_{i}}\ge(\frac{\stretch_{i}}{\eps_{i,\sigma_{\max}}})^{k_{i}}$.

(7): We have $\beta_{i}\le\beta_{\distScale}=\prod_{i=0}^{\distScale}\Otil(\lenapsp n^{2/k_{i}}/\phicmg)=2^{O((\log^{7/8}n)(\log\log\log n))}\prod_{i=0}^{\distScale}n^{2/k_{i}}$
by definition of $\lenapsp,\phicmg$ and $\distScale$. As $\sum_{i=0}^{\dist}1/(\lg\lg n)^{3^{i}}=O(1/\lg\lg n)$,
we have $\prod_{i=0}^{\distScale}n^{2/k_{i}}=n^{O(1/\lg\lg n)}$.
Therefore, $\beta_{i}=2^{O((\log^{7/8}n)(\log\log\log n))}\cdot n^{O(1/\lg\lg n)}=\Ohat(1)$. 
\end{proof}
As the path-reporting $\Apxball^{\pi}$will call the distance-only
$\Apxball$ as a black-box, we will need the following bound. 
\begin{prop}
\label{prop:ball time helper}For any $d'\le nW$ and $\eps'$ where
$\eps'\ge\eps_{0,\sigma_{\max}}$, we have $T_{\Apxball}(G,S,d',\eps')=\Otil(\left|\ball_{G}(S,d')\right|\frac{n^{2/k_{0}+12/\distScale}50^{\distScale}}{\phicmg\eps_{0,\sigma_{\max}}^{2}})$. 
\end{prop}

\begin{proof}
This follows from \Cref{thm:main no distance} when we set the accuracy
parameter $\eps'\gets\eps_{0,\sigma_{\max}}$ (we use $\eps'$ instead
of $\eps$ to avoid confusion). Note that $\eps_{0,\sigma_{\max}}\ge\phicmg^{2}$
satisfying \Cref{thm:main no distance}. In the proof of \Cref{thm:main
no distance}, there are parameters $\eps'_{\distScale}=\eps'=\eps_{0,\sigma_{\max}}$
and $\eps'_{0}=\eps'/50^{\distScale}=\eps_{0,\sigma_{\max}}/50^{\distScale}$.
From \Cref{enu:induction ball} when $i=\distScale$, as $d'\le nW\le d'_{\distScale}$,
we have $T_{\Apxball}(G,S,d',\eps'_{\distScale})\le\Otil(\left|\ball_{G}(S,d')\right|\frac{n^{2/k_{0}+12/\distScale}}{\phicmg(\eps'_{0})^{2}})=\Otil(\left|\ball_{G}(S,d')\right|\frac{n^{2/k_{0}+12/\distScale}50^{\distScale}}{\phicmg\eps{}_{0,\sigma_{\max}}^{2}})$. 
\end{proof}
Now, we are ready to state the key inductive lemma that combines everything
together.
\begin{lem}
\label{lem:main path induction}For every $0\le i\le\distScale$ and
$0\le j\le\sigma_{\max}$, we can maintain the following data structures: 
\begin{enumerate}
\item \label{enu:induction ball path}$\Apxball^{\pi}(G_{j},S,d',\eps_{i,j},\beta_{i})$
for any $d'\le d'_{i,j}\defeq\stretch_{i}\cdot D_{i+1}/\eps_{i,j}$
using total update time of 
\[
\left|\ball_{G_{j}}(S,d')\right|\poly\left(n^{1/k_{0}+1/\distScale}2^{\distScale+\sigma_{\max}+\log^{8/9}n}\right)=\Ohat(|\ball_{G_{j}}(S,d')|)
\]
with query-time overhead at most $(q_{\phi}^{(i,j)},q_{\path}^{(i,j)})$. 
\item \label{enu:induction core path}$\Core^{\pi}(G_{j},\Kinit,d',\beta_{i})$
for any $d'\le D_{i+1}$ using total update time of 
\[
\left|\ball_{G_{j}}(\Kinit,\stretch_{i}d')\right|\poly\left(n^{1/k_{0}+1/\distScale}2^{\distScale+\sigma_{\max}+\log^{8/9}n}\right)=\Ohat(|\ball_{G_{j}}(\Kinit,32d'\log n)|)
\]
with scattering parameter $\scatter=\Omegatil(\phicmg)$, stretch
at most $\stretch_{i}$, and query-time overhead at most $(Q_{\phi}^{(i,j)},Q_{\path}^{(i,j)})$. 
\item $(D_{i},k_{i},\eps_{i,j},\stretch_{i},\Delta_{i},\beta_{i})$-covering
of $G_{j}$ using total update time of 
\[
\Otil(n\cdot\poly\left(n^{1/k_{0}+1/\distScale}2^{\distScale+\sigma_{\max}+\log^{8/9}n}\right))=\Ohat(n)
\]
with query-time overhead at most $(Q_{\phi}^{(i,j)},Q_{\path}^{(i,j)})$. 
\end{enumerate}
For all $i>0$, we assume by induction that a $(D_{i-1},k_{i-1},\eps_{i-1,j},\stretch_{i-1},\Delta_{i-1})$-covering
of $G_{j}$ is already explicitly maintained for every $0\le j\le\sigma_{\max}$. 
\end{lem}

The rest of the section is for proving \Cref{lem:main path induction}.
Before proving \Cref{lem:main path induction}, we prove the main
theorem (\Cref{thm:main path}) using it.

\paragraph{Proof of \Cref{thm:main path}.}

We apply \Cref{lem:main path induction} for $i=\distScale$ and
$j=0$. Recall that $G=G_{j}$, $\eps=\eps_{\distScale,0}$, $\beta_{\distScale}=\Ohat(1)$
by \Cref{prop:para path}(\ref{enu:para:beta}), $q_{\phi}^{(\distScale,0)}=c_{q}\sigma_{\max}\log n=O(\log^{2}n)$
and $q_{\path}^{(i,j)}=(2\distScale+1)\cdot\sigma_{\max}\cdot\overhead_{\path}=\Ohat(1)$.
Therefore, \Cref{lem:main path induction} gives us $\Apxball^{\pi}(G,S,d',\eps,\Ohat(1))$
data structure with $\Ohat(|\ball_{G}(S,d')|)=\Ohat(n)$ total update
time and $(O(\log^{2}n),\Ohat(1))$ query-time overhead as desired.

\subsection{Bounds for $\protect\Apxball^{\pi}$}

The proof is by induction on $i$ (starting from $0$ to $\distScale$)
and then on $j$ (starting from $\sigma_{\max}$ to $0$). We will
show that 
\[
T_{\Apxball^{\pi}}(G_{j},S,d',\eps_{i,j},\beta_{i})=(4^{i}2^{(\sigma_{\max}-j)})\cdot\left|\ball_{G_{j}}(S,d')\right|\cdot n^{2/k_{0}+12/\distScale}50^{\distScale}2^{c'(\log^{8/9}n)}
\]
for any $d'\le d'_{i,j}$ where $c'$ is some large enough constant,
which implies the claimed bound of $\left|\ball_{G_{j}}(S,d')\right|\poly\left(n^{1/k_{0}+1/\distScale}2^{\distScale+\sigma_{\max}+\log^{8/9}n}\right)$.

\paragraph{Base Cases ($i=0$).}

For $i=0$ and any $j\in[0,\sigma_{\max}]$, the path-reporting ES-tree
from \Cref{prop:ES tree path} has total update time at most 
\begin{align*}
T_{\Apxball^{\pi}}(G_{j},S,d',\eps_{0,j},\beta_{0}) & \le O(\left|\ball_{G_{j}}(S,d')\right|d'\log n)\\
 & \le\left|\ball_{G_{j}}(S,d')\right|\cdot O(D_{1}\frac{\stretch_{0}\log n}{\eps_{0,\sigma_{\max}}})\\
 & \le\left|\ball_{G_{j}}(S,d')\right|\cdot(4^{i}2^{(\sigma_{\max}-j)})\cdot n^{2/k_{0}+12/\distScale}50^{\distScale}2^{c'(\log^{8/9}n)}.
\end{align*}
and query-time overhead of $(O(\log n),O(d'\log n))\le(q_{\phi}^{(i,j)},q_{\path}^{(i,j)}).$

\paragraph{The Inductive Step.}

Below, we assume that $i>0$ and $j<\sigma_{\max}$. (The proof for
another base case when $j=\sigma_{\max}$ is exactly the same as below
but simpler, because we can ignore all terms related to $G_{j+1}$
as $G_{\sigma_{\max}+1}=\emptyset$.) We assume $d'>d'_{i-1,j}\defeq\stretch_{i-1}D_{i}/\eps_{i-1,j}$
otherwise we are done by induction hypothesis.

\uline{Total Update Time:} As path-reporting $(D_{i-1},k_{i-1},\eps_{i-1,j},\stretch_{i-1},\Delta_{i-1},\beta_{i-1})$-covering
of $G_{j}$ is already explicitly maintained, we can implement $\Apxball^{\pi}(G_{j},S,d',\eps_{i,j},\beta_{i})$
where $\eps_{i,j}=300\eps_{i-1,j}+\eps_{i,j+1}$ and $\beta_{i}\ge8\Delta_{i-1}\beta_{i-1}$
via \Cref{thm:ballpath} using total update time of 
\begin{align*}
 & T_{\Apxball^{\pi}}(G_{j},S,d',\eps_{i,j},\beta_{i})\\
\le & \Otil(\left|\ball_{G_{j}}(S,d')\right|\Delta_{i-1}\frac{(\stretch_{i}\cdot D_{i+1}/\eps_{i,j})}{\eps_{i-1,j}D_{i-1}})+T_{\Apxball^{\pi}}(G_{j},S,2(\frac{\stretch_{i-1}}{\eps_{i-1,j}})^{k_{i-1}}D_{i-1},\eps_{i-1,j},\beta_{i-1})+\\
 & T_{\Apxball^{\pi}}(G_{j+1},S,d',\eps_{i,j+1},\beta_{i})+T_{\Apxball}(G_{j+1},S,d',\eps_{i-1,j}).
\end{align*}
We will prove that $2(\frac{\stretch_{i-1}}{\eps_{i-1,j}})^{k_{i-1}}D_{i-1},\eps_{i-1,j}\le d'_{i-1,j}$
so that we can apply induction hypothesis on $T_{\Apxball^{\pi}}(G_{j},S,2(\frac{\stretch_{i-1}}{\eps_{i-1,j}})^{k_{i-1}}D_{i-1},\eps_{i-1,j},\beta_{i-1})$.
To see this, note that $D_{i}\ge\gamma_{i}D_{i-1}\ge(\frac{\stretch_{i}}{\eps_{i,\sigma_{\max}}})^{k_{i}}D_{i-1}$
by \Cref{prop:para path}(\ref{enu:para:gap and distance scale path},
\ref{enu:para:key path}). So 
\[
d'_{i-1,j}=\frac{\stretch_{i-1}}{\eps_{i-1,j}}D_{i}\ge\frac{\stretch_{i-1}}{\eps_{i-1,j}}\cdot(\frac{\stretch_{i}}{\eps_{i,\sigma_{\max}}})^{k_{i}}D_{i-1}\ge2(\frac{\stretch_{i-1}}{\eps_{i-1,j}})^{k_{i-1}}D_{i-1}
\]
where the last inequality is because $\frac{\stretch_{i-1}}{\eps_{i-1,j}}\ge2$,
$k_{i}\ge k_{i-1}$ and $\frac{\stretch_{i}}{\eps_{i,\sigma_{\max}}}\ge\frac{\stretch_{i}}{\eps_{i,j}}\ge\frac{\stretch_{i-1}}{\eps_{i-1,j}}$
(because $\frac{\stretch_{i}}{\stretch_{i-1}}\ge600=\frac{\eps_{i,j}}{\eps_{i-1,j}}$).
Also, to apply induction hypothesis on $T_{\Apxball^{\pi}}(G_{j+1},S,d',\eps_{i,j+1},\beta_{i})$,
we note that $d'\le d'_{i,j}\le d'_{i,j+1}$ because $\eps_{i,j}\ge\eps_{i,j+1}$
by \Cref{prop:para path}(\ref{enu:para:eps path}). Therefore,
the bound on $T_{\Apxball^{\pi}}(G_{j},S,d',\eps_{i,j},\beta_{i})$
is at most

\begin{align*}
 & \Otil(\left|\ball_{G_{j}}(S,d')\right|n^{2/k_{i-1}+12/\distScale}2^{c''(\log^{8/9}n)})+T_{\Apxball^{\pi}}(G_{j},S,d'_{i-i,j},\eps_{i-1,j},\beta_{i-1})+\\
 & T_{\Apxball^{\pi}}(G_{j+1},S,d_{i,j+1},\eps_{i,j+1},\beta_{i})+T_{\Apxball}(G_{j+1},S,d',\eps_{i-1,j}).\\
\le & \left|\ball_{G_{j}}(S,d')\right|n^{2/k_{0}+12/\distScale}2^{c''(\log^{8/9}n)}+(4^{i-1}2^{(\sigma_{\max}-j)})\cdot\left|\ball_{G_{j}}(S,d')\right|\cdot n^{2/k_{0}+12/\distScale}50^{\distScale}2^{c'(\log^{8/9}n)}\\
 & (4^{i}2^{(\sigma_{\max}-(j+1))})\cdot\left|\ball_{G_{j+1}}(S,d')\right|\cdot n^{2/k_{0}+12/\distScale}50^{\distScale}2^{c'(\log^{8/9}n)}+\left|\ball_{G_{j+1}}(S,d')\right|n^{2/k_{0}+12/\distScale}50^{\distScale}2^{c''(\log^{8/9}n)} & \text{for a constant }c''\\
\le & (2+2\cdot(4^{i-1}2^{(\sigma_{\max}-(j+1))}+4\cdot(4^{i-1}2^{(\sigma_{\max}-(j+1))})\times\left|\ball_{G_{j}}(S,d')\right|n^{2/k_{0}+12/\distScale}50^{\distScale}2^{c'(\log^{8/9}n)}\\
\le & 4^{i}2^{(\sigma_{\max}-j)}\left|\ball_{G_{j}}(S,d')\right|n^{2/k_{0}+12/\distScale}50^{\distScale}2^{c'(\log^{8/9}n)}
\end{align*}
where the first inequality is by induction hypothesis and by \Cref{prop:ball
time helper}, and the second inequality is because $G_{j+1}\subseteq G_{j}$
and $c'\ge c''$ as $c'$ is chosen to be large enough. This completes
the inductive step for update time.

\uline{Query-time Overhead:} Since $2(\frac{\stretch_{i-1}}{\eps_{i-1,j}})^{k_{i-1}}D_{i-1}\le d'_{i-1,j}$
and by induction hypothesis, we have $\Apxball^{\pi}(G_{j},S,2(\frac{\stretch_{i-1}}{\eps_{i-1,j}})^{k_{i-1}}D_{i-1},\eps_{i-1,j},\beta_{i-1})$
has query-time overhead at most $(q_{\phi}^{(i-1,j)},q_{\path}^{(i-1,j)})\le(Q_{\phi}^{(i-1,j)},Q_{\path}^{(i-1,j)})$.
Also, the path-reporting $(D_{i-1},k_{i-1},\eps_{i-1,j},\stretch_{i-1},\Delta_{i-1},\beta_{i-1})$-covering
of $G_{j}$ has query-time overhead at most $(Q_{\phi}^{(i-1,j)},Q_{\path}^{(i-1,j)})$
by induction hypothesis. Lastly, the query-time overhead of $\Apxball^{\pi}(G_{j+1},S,d',\eps_{i,j+1},\beta_{i})$
is at most $(q_{\phi}^{(i,j+1)},q_{\path}^{(i,j+1)})$ because $d'\le d'_{i,j}\le d'_{i,j+1}$.
Therefore, by \Cref{thm:ballpath}, the query-time overhead of $\Apxball^{\pi}(G_{j},S,d',\eps_{i,j},\beta_{i})$
is at most 
\[
(q_{\phi}^{(i,j+1)}+O(1),\max\{q_{\path}^{(i,j+1)}+O(1),Q_{\path}^{(i-1,j)}+O(\frac{d'}{\eps_{i,j}D_{i-1}})\cdot Q_{\phi}^{(i-1,j)}\})\le(q_{\phi}^{(i,j)},q_{\path}^{(i,j)}).
\]
To see why the inequalities hold, we assume that $c_{q}$ is a large
enough constant. So, we have

\[
q_{\phi}^{(i,j+1)}+O(1)\le c_{q}(\sigma_{\max}-j)\log n+c_{q}\le q_{\phi}^{(i,j)}.
\]
Also, 
\[
q_{\path}^{(i,j+1)}+O(1)\le(2i+1)\cdot(\sigma_{\max}-j)\cdot\overhead_{\path}+c_{q}\le(2i+1)\cdot(\sigma_{\max}-j+1)\cdot\overhead_{\path}=q_{\path}^{(i,j)}.
\]
Finally, 
\begin{align*}
Q_{\path}^{(i-1,j)}+O(\frac{d'}{\eps_{i,j}D_{i-1}})\cdot Q_{\phi}^{(i-1,j)} & \le2i\cdot(\sigma_{\max}-j+1)\cdot\overhead_{\path}+O(\frac{n^{2/\distScale}\stretch_{i}}{\eps_{0,\sigma_{\max}}})\cdot(c_{q}\sigma_{\max}12^{\distScale}\log n)\\
 & \le2i\cdot(\sigma_{\max}-j+1)\cdot\overhead_{\path}+\overhead_{\path}\\
 & =(2i+1)\cdot(\sigma_{\max}-j+1)\cdot\overhead_{\path}=q_{\path}^{(i,j)}
\end{align*}
where $\overhead_{\path}=n^{2/\distScale}\cdot12^{\distScale}\cdot2^{c_{q}\cdot\log^{8/9}n}$.

\subsection{Bounds for $\protect\Core^{\pi}$}

The proof is by induction on $i$ (starting from $0$ to $\distScale$)
and we can fix any $j$.

\paragraph{Base Cases ($i=0$).}

For $i=0$ and any $j\in[0,\sigma_{\max}]$, we have that a path-reporting
$(1,1,O(1),1)$-compressed graph of $G_{j}$ can be trivially maintained
by \Cref{prop:compressed base path}. By \Cref{thm:Corepath}
and since $\beta_{0}\ge7\lenapsp$ (by definition of $\beta_{0}$),
we can implement $\Core^{\pi}(G_{j},\Kinit,d',\beta_{i})$ with scattering
parameter $\scatter=\Omegatil(\phicmg)$ and stretch at most $\Otil(\gamma\lenapsp/\phicmg^{2})\le\stretch_{0}$
(by definition of $\stretch_{0}$) with total update time 
\begin{align*}
 & \Otil\left(T_{\Apxball^{\pi}}(G_{j},\Kinit,\stretch_{0}d',0.1,\beta_{0})(D_{1})^{3}\lenapsp/\phicmg\right)\\
 & =\left|\ball_{G_{j}}(\Kinit,\stretch_{0}d')\right|\poly\left(n^{1/k_{0}+1/\distScale}2^{\log^{8/9}n}\right)
\end{align*}
by the ES-tree from \Cref{prop:ES tree path}. As the query-time
overhead of the $(1,1,O(1),1)$-compressed graph is $(1,1)$ by \Cref{prop:compressed
base path} and, by \Cref{prop:ES tree path}, the query-time overhead
of $\Apxball^{\pi}(G_{j},S,4d',0.1,\beta_{i})$ is at most $(O(\log n),O(d'\log n))$.
The query-time overhead of \\
 $\Core^{\pi}(G_{j},\Kinit,d',\beta_{i})$ is at most $(O(\log n),\Otil(n^{1/k_{0}}\lenapsp/\phicmg^{2}))\le(Q_{\phi}^{(i,j)},Q_{\path}^{(i,j)})$.

\paragraph{The Inductive Step.}

\underline{Total Update Time:} For $i>0$ and and any $j\in[0,\sigma_{\max}]$,
given that a path-reporting $(D_{i-1},k_{i-1},\eps_{i-1,j},\stretch_{i-1},\Delta_{i-1},\beta_{i-1})$-covering
of $G_{j}$ is explicitly maintained, by \Cref{prop:covering-compressed
is compressed path}, we can automatically maintain a $(D_{i-1},\gamma_{i-1},\Delta_{i-1},3\beta_{i-1})$-compressed
graph where $\gamma_{i-1}\ge(\stretch_{i-1}/\eps_{i-1,j})^{k_{i-1}}$
by \Cref{prop:para}(\ref{enu:para:key}) and because $\eps_{0,\sigma_{\max}}\le\eps_{i-1,j}$.
By \Cref{thm:Corepath} and since $\beta_{i}\ge7\lenapsp\Delta_{i-1}\cdot(3\beta_{i-1})$,
we can maintain $\Core^{\pi}(G_{j},\Kinit,d',\beta_{i})$ with $\scatter=\Omegatil(\phicmg)$
and $\Otil(\gamma_{i-1}\lenapsp/\phicmg^{2})\le\stretch_{i}$ (by
definition of $\stretch_{i}$) with total update time 
\begin{align*}
 & \Otil\left(T_{\Apxball^{\pi}}(G_{j},\Kinit,\stretch_{i}d',0.1)\Delta_{i-1}^{2}(D_{i+1}/D_{i-1})^{3}\lenapsp/\phicmg\right)\\
 & =\left|\ball_{G_{j}}(\Kinit,\stretch_{i}d')\right|\poly\left(n^{1/k_{0}+1/\distScale}2^{\distScale+\sigma_{\max}+\log^{8/9}n}\right)
\end{align*}
by \Cref{enu:induction ball path} of \Cref{lem:main path induction}.

\underline{Query-time Overhead:} By \Cref{enu:induction ball path}
of \Cref{lem:main path induction}, the query-time overhead of $\Apxball^{\pi}(G_{j},S,\stretch_{i}d',0.1,\beta_{i})$
is at most $(q_{\phi}^{(i,j)},q_{\path}^{(i,j)})$. By induction hypothesis,
the path-reporting $(D_{i-1},k_{i-1},\eps_{i-1,j},\stretch_{i-1},\Delta_{i-1},\beta_{i-1})$-covering
has query-time overhead of $(Q_{\phi}^{(i-1,j)},Q_{\path}^{(i-1,j)})$
and so the query-time overhead of the $(D_{i-1},\gamma_{i-1},\Delta_{i-1},3\beta_{i-1})$-compressed
graph is at most $(3Q_{\phi}^{(i-1,j)},Q_{\path}^{(i-1,j)}+2Q_{\phi}^{(i-1,j)})$
by \Cref{prop:covering-compressed is compressed path}. Let $Q'_{\phi}=\max\{q_{\phi}^{(i,j)},3Q_{\phi}^{(i-1,j)}\}$
and $Q'_{\path}=\max\{q_{\path}^{(i,j)},Q_{\path}^{(i-1,j)}+2Q_{\phi}^{(i-1,j)}\}$.
By \Cref{thm:Corepath}, we have that the query-time overhead of
$\Core^{\pi}(G_{j},\Kinit,d',\beta_{i})$ is at most 
\[
(4Q'_{\phi},Q'_{\path}+\Otil(\frac{d'}{D_{i-1}}\lenapsp/\phicmg^{2})\cdot Q'_{\phi})\le(Q_{\phi}^{(i,j)},Q_{\path}^{(i,j)}).
\]
To see why the inequalities holds, we first note that $Q'_{\phi}=3Q_{\phi}^{(i-1,j)}$
because $q_{\phi}^{(i,j)}=c_{q}(\sigma_{\max}-j+1)\log n\le c_{q}(\sigma_{\max}-j+1)12^{i-1}\log n=Q_{\phi}^{(i-1,j)}$.
So, we have 
\[
4Q'_{\phi}=4\cdot\max\{q_{\phi}^{(i,j)},3Q_{\phi}^{(i-1,j)}\}=4\cdot3Q_{\phi}^{(i-1,j)}=Q_{\phi}^{(i,j)}.
\]
Also, we have 
\begin{align*}
 & Q'_{\path}+\Otil(\frac{d'}{D_{i-1}}\lenapsp/\phicmg^{2})\cdot Q'_{\phi}\\
 & \le\max\{q_{\path}^{(i,j)},Q_{\path}^{(i-1,j)}+2Q_{\phi}^{(i-1,j)}\}+\Otil(\frac{n^{2/\distScale}\lenapsp}{\phicmg^{2}})\cdot3Q_{\phi}^{(i-1,j)}\\
 & \le\max\{q_{\path}^{(i,j)}+\overhead_{\path},Q_{\path}^{(i-1,j)}+\overhead_{\path}\}\\
 & \le2(i+1)\cdot(\sigma_{\max}-j+1)\cdot\overhead_{\path}=Q_{\path}^{(i,j)}
\end{align*}
where $\overhead_{\path}=n^{2/\distScale}\cdot12^{\distScale}\cdot2^{c_{q}\cdot\log^{8/9}n}$.

\subsection{Bounds for Path-reporting Covering}

Recall that the algorithm from \Cref{thm:covering path} for maintaining
a path-reporting $(D_{i},k_{i},\eps_{i,j},\stretch_{i},\Delta_{i},\beta_{i})$-covering
of $G_{j}$ assumes, for all $D_{i}\le d'\le D_{i}(\frac{\stretch_{i}}{\eps_{i,j}})^{k}$,
$\Core^{\pi}$ and $\Apxball^{\pi}$ data structures with input distance
parameter $d'$. By \Cref{enu:induction ball path} and \Cref{enu:induction
core path} of \Cref{lem:main path induction}, we can indeed implement
these data structures for any distance parameter $d'\le D_{i+1}$.
Since $D_{i}(\frac{\stretch_{i}}{\eps_{i,j}})^{k_{i}}\le D_{i}\gamma_{i}\le D_{i+1}$
by \Cref{prop:para}(\ref{enu:para:gap and distance scale},\ref{enu:para:key}),
the assumption is indeed satisfied by \Cref{enu:induction ball path}
and \Cref{enu:induction core path} of \Cref{lem:main path induction}.

So, using \Cref{thm:covering path}, we can maintain a path-reporting
$(D_{i},k_{i},\eps_{i,j},\stretch_{i},\Delta_{i},\beta_{i})$-covering
of $G_{j}$ with $\Delta_{i}=\Theta(k_{i}n^{2/k_{i}}/\scatter)$ in
total update time of 
\begin{align*}
O(k_{i}n^{1+2/k_{i}}\log n/\scatter+\sum_{C\in\cC^{ALL}}T_{\Core^{\pi}}(\stage G{t_{C}}_{j},\stage C{t_{C}},d_{\Clevel(C)},\beta_{i})\\
+T_{\Apxball^{\pi}}(\stage G{t_{C}}_{j},\stage C{t_{C}},\frac{\stretch_{i}}{4\eps_{i,j}}32d_{\Clevel(C)},\eps_{i,j},\beta_{i}))
\end{align*}
where $\cC^{ALL}$ contains all cores that have ever been initialized
and, for each $C\in\cC^{ALL}$, $t_{C}$ is the time $C$ is initialized.
By plugging in the total update time of $\Apxball^{\pi}$ from \Cref{enu:induction
ball path} and $\Core^{\pi}$ from \Cref{enu:induction core path},
the total update time for maintaining the covering is 
\begin{align*}
\Otil(\frac{n^{1+2/k_{i}}}{\scatter}+\sum_{C\in\cC^{ALL}}\left|\ball_{\stage G{t_{C}}_{j}}(\stage C{t_{C}},\stretch_{i}d_{\Clevel(C)})\right|\poly\left(n^{1/k_{0}+1/\distScale}2^{\distScale+\sigma_{\max}+\log^{8/9}n}\right)+\\
\left|\ball_{\stage G{t_{C}}_{j}}(\stage C{t_{C}},\frac{\stretch_{i}}{4\eps_{i,j}}d_{\Clevel(C)})\right|\poly\left(n^{1/k_{0}+1/\distScale}2^{\distScale+\sigma_{\max}+\log^{8/9}n}\right)).
\end{align*}
As it is guaranteed by \Cref{thm:covering path} that 
\[
\sum_{C\in\cC^{ALL}}|\ball_{\stage G{t_{C}}_{j}}(\stage C{t_{C}},\frac{\stretch_{i}}{4\eps_{i,j}}d_{\Clevel(C)})|\le O(k_{i}n^{1+2/k_{i}}/\scatter),
\]
and therefore the above expression simplifies to $\Otil\left(n\cdot\poly\left(n^{1/k_{0}+1/\distScale}2^{\distScale+\sigma_{\max}+\log^{8/9}n}\right)\right)$.
As the query-time overhead of all invoked instances of $\Core^{\pi}$
and $\Apxball^{\pi}$is at most $(Q_{\phi}^{(i,j},Q_{\path}^{(i,j)})$
by \Cref{enu:induction ball path} and \Cref{enu:induction core
path} of \Cref{lem:main path induction}, the query-time overhead
of the covering $\cC$ is at most $(Q_{\phi}^{(i,j},Q_{\path}^{(i,j)})$
by definition.

%% file: flow.tex
In this part of the paper, we are concerned with the problem of maximum bounded cost flow (MBCF) and the min-cost flow problem. In both problems, the input is a graph $G=(V,E,c,u)$ where $c$ is the cost function and $u$ the capacity function both taken over edges \emph{and} vertices; and a source vertex $s$ and a sink vertex $t$. In MBCF, the algorithm is further given a cost budget $\overline{C}$. The MBCF problem is then to find the maximum feasible flow with regard to capacities and cost budget $\overline{C}$, i.e. a flow of cost at most $\overline{C}$ where no edge or vertex carries more flow than stipulated by the capacity function (for precise definitions of these properties, we refer the reader to the additional preliminary section \ref{sec:prelimFlow}). 

The main result of this section is our main theorem on flow. 

\minCostMain*

Since we can derive a $(1-\epsilon)$-approximate min-cost flow algorithm from an algorithm for MBCF by trying $\Otil(\log\log C)$ cost budget values (by performing binary search over every power of $(1+\epsilon)$ smaller than $\overline{C}$), we will focus for the rest of this section on the problem of MBCF and only return to the min-cost flow in the final discussion. We now state our final result for the MBCF problem.

\begin{theorem}\label{thm:MBCFFinal}
For any $\epsilon > 1/\polylog(n)$, given an undirected graph $G=(V,E,c,u)$, a source vertex $s$ and a sink vertex $t$, and a cost budget $\overline{C}$. Let $OPT_{G, \overline{C}}$ be the maximum value of any $s$-$t$ feasible flow of cost at most $\overline{C}$.

Then, there exists an algorithm that can return a feasible flow $\flow$ and is of value at least $(1-\epsilon)OPT_{G, \overline{C}}$. The algorithm can compute $\flow$ in time 
$m^{1+o(1)}$ and runs correctly with probability at least $1-n^{-7}$.
\end{theorem}

We derive the result stated in \Cref{thm:MBCFFinal} by a series of reductions. We start this section by stating some additional preliminaries and defining some crucial concepts. In \Cref{sec:explainProblemGiveOverview}, we then discuss the problem of MBCF in more detail and state formally our reductions which provides a roadmap for the rest of this chapter.

We recommend the reader to read the overview in Section \ref{sec:overview-flow} before reading the rest of this section, as it contains high-level intuition for our overall approach.

\section{Additional Preliminaries}
\label{sec:prelimFlow}

We sometimes use $\textsc{Exp}(x)$ in place of $e^x$ to avoid clutter. We use $\lceil x \rceil_\alpha$ to denote $x$ rounded up to the nearest power of $\alpha$.

\paragraph{Flows and Cuts.} Throughout this section, let $G=(V,E,c,u)$ be an undirected graph with cost function $c$ and capacity function $u$ and assume that two distinguished vertices $s$ and $t$ are given along with a cost budget $\overline{C}$. As we will show at the end of the preliminary section, \emph{we can assume w.l.o.g. that $c$ and $u$ are only defined over the vertices}. We define $C$ and $U$ to be the max-min ratios of functions $c$ and $u$, respectively. For convenience, we model $G$ as a graph where all edges are bidirectional: that is, $(x,y) \in E$ iff $(y,x) \in E$ (we get $c(x,y) = c(y,x)$ and $u(x,y) = u(y,x)$). 

We say that a vector $f \in \mathbb{R}^{E}_+$ is a \emph{flow} if it assigns flow mass $f(x,y) \geq 0$ to every edge $(x,y) \in E$. Slightly non-standardly, we do not assume skew-symmetry.

\paragraph{Flow Properties.} We further define the in-flow and the out-flow at a vertex $x \in V$ by
\[
    \inflow_f(x) = \sum_{y \in V} f(y,x) \textit{ and } \outflow_f(x) = \sum_{y \in V} f(x,y).
\]
Note that flow on the anti-parallel edges $(x,y)$ and $(y,x)$ is \emph{not} canceled by this definition. 

We say that a flow $f$ satisfies flow conservation constraints, if for every $x \in V \setminus \{s,t\}$, we have $\inflow_f(x) = \outflow_f(x)$. We further say that a flow $f$ is satisfies capacity constraints (or is capacity-feasible) if for every $x \in V$, $\inflow_f(u) \leq u(x)$. 

The \emph{cost} of a flow is defined to be 
\[
c(f) = \sum_v \inflow_f(v) \cdot c(v),
\]
where $\inflow_f(v) \cdot c(v)$ captures of the cost of the flow going \emph{through} vertex $v$. Observe that in a feasible $s$-$t$ flow, the vertex $s$ on each flow path has no flow going into the vertex, and we therefore do not attribute any cost to $s$. Note also that if the flow $f$ obeys conservation constraints (but at $s$ and $t$), then $\inflow_f(v) = \outflow(v)$ precisely captures the flow through $v$. We use the definition of the cost even for flows which do not satisfy conservation constraints.

Then, we say that a flow $f$ is cost-feasible if $c(f) \leq \overline{C}$. We say a flow $f$ is a pseudo-flow, if it is capacity- and cost-feasible. We say that $f$ is a feasible flow if it is a pseudo-flow and $f$ satisfies conservation constraints. For a feasible flow $f$, we say that the \emph{value} of the flow is the amount of flow sent from $s$ to $t$, or more formally $\inflow_f(t) - \outflow_f(t)$.

\paragraph{(Near-)Optimality.} Given a graph $G$, vertices $s,t$ and a cost budget $\overline{C}$, we let $OPT_{G, \overline{C}}$ denote the maximum flow value achieved by any feasible flow. We also define a notion of near-optimal flows. 

\begin{restatable}{definition}{nearOptStatement}[Near-Optimality] \label{def:nearOpt}
For any $\epsilon > 0$, given a graph $G$, source and sink vertices $s,t \in V$ and a cost budget $\overline{C}$, we say that a flow $\flow$ is $(1-\epsilon)$-optimal if the flow $f$ is cost-feasible and of value at least $(1-\epsilon) OPT_{G, \overline{C}}$.
\end{restatable}

\paragraph{Reduction to Vertex-Capacities Only.} Finally, we formally state a reduction from graphs $G$ with mixed capacities and costs to vertex capacities only. The reduction also enforces some additional desirable properties that we henceforth assume. The proof of \Cref{prop:reductionVertexCapacities} can be found in \Cref{subsec:proofOfReductionVertexCaps}.

\begin{restatable}{proposition}{reductionVertexCapacities}
\label{prop:reductionVertexCapacities}
Given $G=(V,E,c,u)$ with as defined above with capacities and costs taken over $E \cup V$, $\overline{C}$ and $1/n < \epsilon < 1$ and $m \geq 16$. Then, there is a $G' = (V',E',c',u')$ with $s'$ and $t'$ and $\overline{C}' = 32m^4$ such that:
\begin{enumerate}
    \item $(x,y) \in E'$ iff $(y,x) \in E'$. Further, for each $(x,y) \in E'$, $c'(x,y) = 0$ and $u'(x,y) = \infty$, and \label{prop:antiparallelEdges}
    \item $c'(s')=c'(t')= 0$, and \label{prop:stAreZeroCost}
    \item $V'$ is of size $n_{G'} \leq m+n+2$, $E'$ of size $m_{G'} \leq 2m+4$, and \label{prop:sizeofGprime}
    \item for each vertex $x \in V(G')$, $c(x) \in [1, m^5] \cup \{0\}$ and $u(x) \in [1, m^5]$, and \label{prop:boundedVertexWights}
    \item there is a map $\mathcal{M}_{G' \rightarrow G}$ that maps any $(1-\epsilon)$-optimal $s'$-$t'$ flow $f'$ in $G'$ to a $(1-\epsilon)^2$-optimal $s$-$t$ flow $f$ in $G$. The flow map can be applied in $O(m)$ time and $G'$ can be computed in $O(m \log n)$ time.\label{prop:nearOptimalTrans}
\end{enumerate}
\end{restatable}

\paragraph{Exponential Distribution.} We make use of the exponential distribution with parameter $\lambda > 0$, that is we use random variables $X$ with cumulative distribution function $\Pr[X \leq x] = 1 - e^{-\lambda x}$ for all $x \geq 0$, which we denote by the shorthand $X \sim \textsc{Exp}(\lambda)$.

\paragraph{A Path-reporting SSSP Structure.} Finally, we need a data structure akin to the one defined in \Cref{def:pathReportingSSSP} and implemented by \Cref{thm:main SSSP path}. Before stating the definition, we start with some preliminaries. 
	
Here, we consider an \emph{undirected} graph $G=(V,E,w,\sigma)$ that we again model by having an edge $(x,y) \in E$ iff $(y,x) \in E$. For any path $P$ in $G$, we assume that the edges used in $P$ are directed correctly along $P$, i.e. $P$ consist of edges $(v_1,v_2), (v_2, v_3), (v_3, v_4), \dots$. For each \emph{vertex} $v$, we have a weight $w(v)$, and we define the weight of a path $P$ in $G$ induced by $w$ by $w(P) = \sum_{(u,v) \in P} w(v)$ (i.e. only the tail vertex of each edge is accounted for). 

Each edge $e \in E$ is assigned integral \emph{steadiness} $\sigma(e)\in[1,\tau]$, for some parameter $\tau$. For any \emph{multi}-set $E'\subseteq E$ and $j$, we let $\sigma_{\le j}(E')=\{e\in E'\mid\sigma(e)\le j\}$ contain all edges from $E'$ of steadiness at most $j$. A path $P$ is \emph{$\beta$-edge-simple} if each edge appears in $P$ at most $\beta$ times. When $P$ is a (non-simple) path, $\sigma_{\le j}(P)$ is a \emph{multi-set} containing all occurrences of edges with steadiness at most $j$ in $P$. 

\begin{definition}[Path-reporting SSSP]
\label{def:Path-reportingSSSP}
Given a decremental graph $G=(V,E,w,\edgeSensitivity)$, some $\tau\geq 1$ such that $\edgeSensitivity(e) \in [1, \tau]$ for each $e \in E$, a simpleness parameter $\beta \geq 1$, a source and sink vertex $s,t \in V$ with $w(s) = w(t) = 0$, a distance approximation parameter $\epsilon > 0$. Then, we say that a data structure $\SSSP^{\pi}(G,s,t,\eps,\beta)$ is a \emph{Path-reporting SSSP Structure} if 
\begin{itemize}
    \item $t$ is associated with a $\beta$-edge-simple
    $s$-$t$ path $\shortestSTPath$ in $G$ of length at most $(1+\epsilon)\dist_{G}(s,t)$.
    \item given a steadiness index $j$, the data structure
    returns $\sigma_{\le j}(\shortestSTPath)$. 
\end{itemize}
\end{definition}

We point out that the associated path $\shortestSTPath$ is fixed after every update to make sure that the path $\shortestSTPath$ does not depend on steadiness threshold $j$. That is, regardless of which $\sigma_{\le j}(\shortestSTPath)$ is queried, the underlying path $\shortestSTPath$ is always the same. 
This will be key for the correctness of our flow estimators, as the threshold $j$ will be chosen randomly, and we will then analyze the probability of each edge on $\shortestSTPath$ being in the set $\sigma_{\le j}(\shortestSTPath)$. 

For the rest of this chapter, we only refer to a single instance of a data structure as given in \Cref{def:Path-reportingSSSP}. We can thus reserve the variables $\beta$, $\SSSP^{\pi}(G,s,t,\eps,\beta)$ and $\tau$ for this specific data structure and denote throughout by $\timeAugSSSP{m}{n}{W}{\tau}{\epsilon}{\beta}{\Delta}{\Delta'}$ the total update time of this data structure where $G$ undergoes $\Delta$ edge weight increases, and $\Delta'$ is defined to be the sum of the sizes all encodings of sets $\sensitivity_{\leq j}(\shortestSTPath)$ that were queried for plus the number of queries (i.e. $\Delta'$ is the size of the query output where we say that a single bit is output if the output set is empty). $W$ denotes the max-min weight ratio of vertex weights $w(v)$. 

We later show that we can implement $\SSSP^{\pi}(G,s,t,\eps,\beta)$ from the result in \Cref{thm:main SSSP path} rather straight-forwardly, but keep abstraction of Definition \ref{def:Path-reportingSSSP} to allow for future work to use our reductions.

\section{A Roadmap to the Reductions}
\label{sec:explainProblemGiveOverview}

Let us now give a brief overview of the reductions we require to obtain our result for the Maximum Bounded Cost Flow (MBCF) problem. We remind the reader that we henceforth assume various properties of $G$ as obtained by the reduction described in \Cref{prop:reductionVertexCapacities}, in particular that $G$ only has \emph{vertex} capacities/costs.

Our goal in this part is to computed the maximum feasible flow from $s$ to $t$ whose flow value we denote by $OPT_{G, \overline{C}}$. The final result we aim for in our reduction chain is a near-optimal flow; we restate the definition from the preliminaries.

\nearOptStatement*

While our final goal is to obtain a near-optimal flow, we will require a relaxation of this notion throughout the algorithm to make progress. We therefore introduce the notion of a $(1-\epsilon)$-pseudo-optimal flow. This relaxation allows us to couple a pseudo-flow to a near-optimal flow.

\begin{definition}[Near-Pseudo-Optimality]\label{def:feasibleFlow}
For any $\epsilon > 0$, given a graph $G$, source and sink vertices $s,t \in V$ and a cost budget $\overline{C}$, we say that a pseudo-flow $\flowEst$ is a $(1-\epsilon)$-pseudo-optimal flow if there exists a flow $\flow$ such that
\begin{enumerate}
    \item $\flow$ is a $(1-\epsilon)$-optimal flow (see \Cref{def:nearOpt}), and 
    \item $\forall v \in V$: $|\inflow_{\flow}(v) - \inflow_{\flowEst}(v)| \leq \epsilon \cdot u(v)$.
\end{enumerate}
\end{definition}

In \Cref{sec:globalViaDecrSSSP} we describe how to compute a $(1-\epsilon)$-pseudo-optimal pseudo-flow $\flowEst$ using a Path-reporting SSSP data structure as described in \Cref{def:Path-reportingSSSP}. This forms the centerpiece of our reduction. We therefore extend the powerful MWU framework by Garg and Koenemann \cite{garg2007faster} to work with random estimators. While this greatly speeds up the running time of the algorithm, this will be at the cost of only producing a $(1-\epsilon)$-pseudo-optimal flow. The main concern with the $(1-\epsilon)$-pseudo-optimal flow is that after routing $\flowEst$, each vertex might have some small excess, i.e. the flow conservation constraint might be slightly violated at each vertex.

Ideally, we could use repeated computations of near-pseudo-optimal flows to route the excess since the excess vector is itself a demand vector that can be modeled as another instance of $s-t$ flow. But the coupling guaranteed by Definition \ref{def:feasibleFlow} is too weak on its own for this approach to work. 
We thus need something stronger. Instead of directly strengthening the coupling condition guarantees \Cref{def:feasibleFlow} tighter, we use a different approach: we  "fit" the instance $G$ to the flow. Note that the definition below is informal and we need some slightly stronger properties for the reduction.

\begin{definition}[Informal]\label{def:capacityFittedInstanceInformal}
For any $\epsilon > 0$, given $G = (V,E,u,c)$ and a cost budget $\overline{C}$. Then, we say that a graph $G' = (V,E,u',c')$ is a $(1-\epsilon)$-\emph{capacity-fitted instance derived from $G$} if 
\begin{enumerate}
    \item for every $v \in V$, $u'(v) \leq u(v)$, and
    \item we have $\sum_{x \in V'} u'(x) \cdot c'(x) \leq 18 \cdot \overline{C}$, and \label{prop:capacityFittedPropCostInformal}
    \item $OPT_{G', \overline{C}} \geq (1-\epsilon) \cdot OPT_{G, \overline{C}}$. \label{prop:capacityFittedPropNearOptInformal}
\end{enumerate}
\end{definition}

Loosely speaking, the graph $G'$ in the above definition has the property that the optimal flow is close to saturating most edges in the graph. More formally, Property \ref{prop:capacityFittedPropCostInformal} says that in $G'$ even if the flow saturated every vertex, the total cost would still be at most $18\overline{C}$.

We will show that, rather surprisingly, using the intermediate of a capacity-fitted instance will yield a black box conversion from any algorithm for computing a $(1-\epsilon)$-approximate pseudo-optimal flow into an algorithm for computing a $(1-\epsilon)$-optimal flow. In particular, we first show in Section \ref{sec:minCostFlowFinally} that repeated computation of pseudo-optimal flows will allow us to compute a capacity-fitted instance $G'$ of $G$. We then show in Section \ref{sec:nearOptFromNearLocally} that once we have a capacity-fitted instance, we can convert a near-pseudo-optimal flow into a near-optimal flow by using only a single call to a basic $(1+\epsilon)$-approximate max flow algorithm (only edge capacities, no costs), such as the algorithm in \cite{sherman2013nearly, kelner2014almost, peng2016approximate}\footnote{We point out that we do not require \cite{sherman2013nearly, kelner2014almost, peng2016approximate} and could also devise a recursive scheme that invokes our own algorithm again. However, the reduction to approximate max flow with edge capacities is a significantly cleaner approach.}. 

We summarize this roadmap by restating the reduction chain
\begin{align*}
    \xRightarrow{\text{via Def. \ref{def:Path-reportingSSSP}, Sec. \ref{sec:globalViaDecrSSSP}}} &(1-\epsilon)\text{-pseudo-optimal flow}\\ \xRightarrow{\text{via Sec. \ref{sec:minCostFlowFinally}}} &(1-\epsilon)\text{-capacity-fitted instance}\\ \xRightarrow{\text{via Approx. Max Flow, Sec. \ref{sec:nearOptFromNearLocally}}} &(1-\epsilon)\text{-optimal flow}.
\end{align*}
Finally, we point out that while \Cref{sec:globalViaDecrSSSP} makes deliberate use of randomness, resulting in a Monte-Carlo algorithm, we will state the remaining reductions in deterministic fashion. Only at the end, when combining the chain of reductions, we revisit the issue of success probability.

Finally, combining all the reductions above, we have a reduction from the MBCF problem in any special instance that satisfies the properties of \Cref{prop:reductionVertexCapacities} to the Path-reporting SSSP data structure from \Cref{def:Path-reportingSSSP}. Since Proposition \Cref{prop:reductionVertexCapacities} then gives a reduction from any instance of MBCF to such a special instance and we showed in \Cref{part:augmented-queries} how to construct the desired data structure, we can plug in this data structure to obtain our near-optimal algorithm for mixed-capacitated MBCF. We thus obtain the final min-cost flow algorithms of Theorems \ref{thm:MainMinCost} and \ref{thm:MBCFFinal}. See \Cref{sec:puttingItAllTogether} for more details on how all the reductions fit together.

\section{Near-pseudo-optimal MBCF via Path-reporting Decremental SSSP}
\label{sec:globalViaDecrSSSP}

The main result of this section is summarized in the theorem below. 

\begin{theorem}[Near-pseudo-optimal MBCF]\label{thm:mainMinCost2}
Given graph $G=(V,E,c,u)$, a dedicated source $s$ and sink $t$, some cost budget $\overline{C}$, any $0 < \epsilon \leq 1/768$, a positive integer $\tau = O(\log n)$, and data structure $\SSSP^{\pi}(G,s,t,\eps,\beta)$ from \Cref{def:Path-reportingSSSP}. \\
Then, procedure $\textsc{NearPseudoOptMBCF}(G=(V,E, c,u), s,t, \epsilon, \tau, \overline{C})$ given in \Cref{alg:boundedCostFlowexcess} returns $\flowEst$ such that $\flowEst_{scaled} = \frac{\flowEst}{(1+10\epsilon)\log_{1+\epsilon}\left(\frac{1+\epsilon}{\initialDifferenceMWU}\right)}$ is a $(1-\Theta(\epsilon))$-pseudo-optimal flow. The algorithm runs in time \[
\Otil(m \beta \cdot  n^{10/\tau}/ \epsilon^2) + \timeAugSSSP{m}{n}{m^{6/\epsilon}}{\tau}{\epsilon}{\beta}{\Delta}{\Delta'}
\]
where $\Delta, \Delta' = \Otil(m \beta \cdot n^{10/\tau}/ \epsilon^2)$ and runs correctly with probability at least $1-n^{-10}$.
\end{theorem}

We organize this section as follows: we first give some additional level of detail on the MBCF problem by providing an LP and dual LP for the problem. Building upon this discussion, we then introduce the reader to \Cref{alg:boundedCostFlowexcess} and give an overview of the analysis. This gives a further overview of the rest of the section which is dedicated to proving \Cref{thm:mainMinCost2}.

\subsection{LP formulation of the Vertex-Capacitated MBCF Problem}

Let us now describe a linear program (LP) that captures the MBCF problem (here we already assume w.l.o.g. that $G$ is vertex-capacitated and has $s$ and $t$ of infinite capacity and zero cost). The LP is given in  Program \ref{eq:LPprimal} where we denote by $P_{s,t}$ the set of all paths in $G$ from $s$ to $t$ and by $P_{v,s,t}$ the set of all $s$ to $t$ paths that contain the vertex $v \in V$. We remind the reader that we restrict our attention to vertex-capacitated graphs w.l.o.g. by Proposition \ref{prop:reductionVertexCapacities}.

\begin{equation} \label{eq:LPprimal}
\begin{array}{ll@{}llp{1cm}@{}l}
\text{maximize}  & \displaystyle\sum\limits_{p \in P_{s,t}} &f_{p} & & &\\
\text{subject to}& \displaystyle\sum\limits_{p \in P_{v,s,t}} & f_p &\leq &u(v) & \forall v \in V \setminus \{s,t\} \\
                 & \displaystyle\sum\limits_{p \in P_{s,t}}& c(p) \cdot f_p & \leq & \text{$\overline{C}$} &\\
                 && f_p &\geq& 0& \forall p \in P_{s,t}
\end{array}
\end{equation}
Observe that given a feasible solution $\{f_p\}$ to the LP, it is not hard to obtain a feasible flow $f$ of cost at most $\overline{C}$ as can be seen by setting $\flow(e) = \sum_{p \in P_{e,s,t}} f_p$ (the converse is true as well as can be seen from a flow decomposition). Throughout, we let $OPT_{G, \overline{C}}$ refer to the maximum value of the objective function (which is just the value of the flow $f$ from $s$ to $t$) obtained as the maximum over all feasible solutions. 

We also state the dual to the LP \ref{eq:LPprimal}:
\begin{equation} \label{eq:LPdual}
\begin{array}{ll@{}llp{1cm}@{}l}
\text{minimize}  & \displaystyle\sum\limits_{v \in V} &u(v) w_v + \overline{C} \varphi& & &\\
\text{subject to}&  \displaystyle\sum\limits_{v \in p} &w_v + \varphi c(v) &\geq &1 & \forall p \in P_{s,t}\\
                 && w_v &\geq& 0& \forall v \in V\\
                 && \varphi &\geq& 0& \\
\end{array}
\end{equation}
Intuitively, the dual LP minimizes over variables $w_v, \varphi$ which are related to capacity and cost budget, such that the metric induced by function $w'(x,y) = w_x + \varphi \cdot c(x)$ over each edge $(x,y) \in E$ ensures that any two vertices are at distance at least $1$.

In our analysis, we use weak duality to relate the two given LPs. 

\begin{theorem}[see for example \cite{boyd2004convex} for a more general proof]\label{thm:weakDuality}
We have that for any feasible solution $\{ f_p\}_p$ to the primal LP \ref{eq:LPprimal}, and any feasible solution $\{w_e\}_e, \varphi$ to the dual LP  \ref{eq:LPdual}, we have that
\[
    \sum_{p \in P_{s,t}} f_{p} \leq \sum_{v \in V} u(v) w_v + \overline{C} \varphi.
\]
\end{theorem}

\subsection{Algorithm and High-Level Analysis}

\begin{algorithm}
\caption{$\textsc{NearPseudoOptMBCF}(G=(V,E, c,u), s,t, \epsilon, \tau, \overline{C})$}
\label{alg:boundedCostFlowexcess}
\KwIn{A vertex-capacitated graph $G=(V,E, c,u)$, two vertices $s, t \in V$, a cost function $c$ and a capacity function $u$ both mapping $V \rightarrow \mathbb{R}_{>0}$, an approximation parameter $\epsilon > 0$, a positive integer $\tau = O(\log n), \beta \geq 1$, and a cost budget $\overline{C} \in \mathbb{R}^+$.}
$\flowEst \gets \mathbf{0}^E$; $\initialDifferenceMWU = m^{-1/\epsilon}$; $\costFunctionEst \gets \initialDifferenceMWU/\overline{C} $; $\flowRouted = n^{10/\tau}$; $\zeta = 3860 \log n \cdot \log_{1+\epsilon}(\frac{1+\epsilon}{\initialDifferenceMWU})$.\;
\lForEach{$v \in V$}{ 
    $\weightFunctionEst(v) \gets \initialDifferenceMWU/u(v)$.\label{lne:defineWeight} 
}
\ForEach{$e =(x,y) \in E$}{ 
    $\edgeSensitivity(e) \gets  \left\lfloor \log_{\flowRouted} \left(\min \{ \frac{\overline{C}}{m \cdot c(y)} , \frac{u(y)}{\deg(y)} \} / (\zeta \beta) \right) \right\rfloor$.\label{lne:defineSteadiness} 
}
\BlankLine
Start a data structure $\SSSP^{\pi}(G,s,t,\eps,\beta)$ as described in \Cref{def:Path-reportingSSSP} on $G$ from source vertex $s$ with weight function $\weightFunctionCombined(x) =  \weightFunctionEst(x) + \lceil \costFunctionEst  \rceil_{(1+\epsilon)}\cdot c(x)$, steadiness function $\edgeSensitivity$ and approximation parameter $\epsilon$.\label{lne:dataStructureInit}\;
\BlankLine

\tcc{Run the MWU algorithm from \cite{ChuzhoyS20_apsp} on $G$ where adding flow on edges is replaced by randomly estimating the flow on each edge.}
\While(\label{lne:whileLoopMWU}){$\sum_{v \in V} u(v) \cdot \weightFunctionEst(v) + \overline{C} \cdot \costFunctionEst < 1$}{
    Find a $(1+\epsilon)$-approximate shortest path $\shortestSTPath$ with respect to $\weightFunctionCombined$ using $\SSSP^{\pi}(G,s,t,\eps,\beta)$.\;
    $\minCapIndexMWU \gets \min \; \{\forLoopIndexMWU \;|\; \sensitivity_{\forLoopIndexMWU}(\shortestSTPath) \neq \emptyset\}$. \label{lne:computeFlowCap} \tcp*{$\Lambda$ is the minimum steadiness on $\shortestSTPath$.}
    $\randomThreshold \gets \lceil \textsc{Exp}(\log \flowRouted) \rceil$.\label{lne:selectSteadinessThreshold}\tcp*{Choose random steadiness threshold $\randomThreshold$.}
    \tcc{The loop below will simulate adding $\flowRouted^{\minCapIndexMWU}$ to every edge on $\shortestSTPath$ by updating every edge with probability scaling in its steadiness and update flow according to the inverse of the probability. We estimate the path cost similarly.}
    $\hat{c} \gets 0$.\tcp*{Estimator only for $c(\shortestSTPath)$.}
    \tcc{We iterate over all low steadiness edges (we consider them as directed on the flow path).}
    \ForEach(\label{lne:innermostForeachMWU}){$e = (x,y) \in \sensitivity_{\leq \minCapIndexMWU + \randomThreshold}(\shortestSTPath)$}{
        $\flowEst(e) \gets \flowEst(e) + \flowRouted^{\edgeSensitivity(e)}$.\label{lne:deployFlowToSensitiveEdges}\tcp*{Add flow to flow estimator.}
        $\weightFunctionEst(y) \gets \weightFunctionEst(y) \cdot \textsc{Exp}\left(\frac{\epsilon \flowRouted^{\edgeSensitivity(e)}}{u(y)}\right)$.\label{lne:increaseWeight}\tcp*{Update weight.}
        $\hat{c} \gets \hat{c} + c(y) \cdot \flowRouted^{\edgeSensitivity(e)}$.\label{lne:deployCostEstimator}\tcp*{Update cost estimate.}
    }
    $\costFunctionEst \gets \costFunctionEst\cdot \textsc{Exp}\left(\frac{\epsilon \cdot \hat{c}}{\overline{C}}\right)$.\label{lne:increaseVarPhi} \tcp*{Update the cost function.}
}
\BlankLine
\Return $\flowEst$ \label{lne:returnFlow} 
\end{algorithm}

Our algorithm follows the high-level framework of Garg and Koenneman for computing a maximum flow \cite{garg2007faster}, though with the crucial differences mention below. Although our write-up is entirely self-contained, we recommend readers unfamiliar with the MWU framework to start with the paper of Garg and Koenneman \cite{garg2007faster}, or with a more recent exposition in appendix C.1 of \cite{ChuzhoyS20_apsp}, which uses  notation that is more similar to ours. We also remind the reader that a high-level overview of our differences with the standard MWU framework can be found in Section \ref{sec:overview-flow}.

Inspired by the dual LP \ref{eq:LPdual}, \Cref{alg:boundedCostFlowexcess} defines an initial weight function $\weightFunctionEst$ over the vertices which is set to have very small values in the beginning and similarly assigns the $\costFunctionEst$ variable as small value. 

It then maintains an Path-reporting SSSP data structure on the graph $G$ with weight function $\weightFunctionCombined(x) \approx \weightFunctionEst(x) + \costFunctionEst \cdot c(x)$. Subsequently, the algorithm computes shortest-paths in metric induced by $\weightFunctionCombined$ and increases flow along \emph{some} edges on the shortest path. The combined flows will later form the flow variables for the primal solution given in LP \ref{eq:LPprimal}. Based on the flow updates, the algorithm then increases $\weightFunctionCombined(x)$ for every vertex whose in-flow was increased. We point out that in our algorithm, in stark contrast to previous algorithms, the flow is not directly added to the identified shortest path but instead we only maintain a random estimator $\flowEst(e)$ at each edge, that estimates how much flow should have been added throughout the algorithm. Based on the value of $\flowEst(e)$ where $e=(x,y)$, we increase $\weightFunctionEst(y)$ which in turn increases $\weightFunctionCombined(y)$. This in turn implies that we do not route a lot of flow through $y$ before $\weightFunctionCombined(y)$ becomes too large for $y$ to appear on an (approximate) shortest path.

Analyzing \Cref{alg:boundedCostFlowexcess} is rather involved since we have to combine the classic analysis of the multiplicative weight update framework for max flow and maximum bounded cost flow as given in \cite{garg2007faster, fleischer2000approximating, madry2010faster, Chuzhoy:2019:NAD:3313276.3316320, ChuzhoyS20_apsp} with some strong concentration bounds for the flow and the cost of the flow to get control over the heavy randomization we introduced. 

\paragraph{Notation for each Iteration.} We use the notation that a variable in the algorithm used with subscript $i$ denotes the variable after the $i^{th}$ while-iteration in our algorithm. For example, $\flowEst[i]$ denotes the (pseudo-)flow $\flowEst$ after the $i^{th}$ iteration. An overview of variables with definitions is given in \Cref{tab:versions_Variables}. We let $k$ be the number of iterations of the while-loop (hence $k$ is itself a random integer). 

\paragraph{The Pseudo-Flow and the Real Flow.} Recall that our final goal will be to show that a near-optimal pseudo-flow $\flowEst$, which is close to some near-optimal flow $f$ (see Definition \ref{def:feasibleFlow}). The flow $f$ that we will compare to is defined as follows. Let $f$ be the flow that is obtained by routing during each iteration $i^{th}$ exactly $\flowRouted^{\minCapIndexMWU[i]}$ units of flow along the approximate shortest path $\shortestSTPath[i]$ (i.e. every edge receives exactly this amount of flow). Let again flow $\flow[i]$ be the flow incurred by the paths chosen in the first $i$ iterations. Note that although $\flow$ obeys conservation constraints, still depends on $\flowEst$, because the path $\pi(s,t)$ is defined using weights $\weightFunctionEst$, which is updated according to $\flowEst$. 

\paragraph{Comparison to the Previous Approach.}
In the framework of Garg and Koenneman \cite{garg2007faster}, there is no pseudo-flow $\flowEst$. There is only the flow $\flow$, and the weight function $\weightFunction$ and cost function $c$ are updates using $\flow$ instead of $\flowEst$. 

Then, the key ideas are as follows. We first note that if we followed the pseudocode of Algorithm  \ref{alg:boundedCostFlowexcess} (but with $\flow$ instead of $\flowEst$), then the final flow $\flow$ returned  is not capacity feasible. However, it turns out that scaling the flow to obtain $\flow_{scaled} = \frac{\flow}{(1+10\epsilon)\log_{1+\epsilon}\left(\frac{1+\epsilon}{\initialDifferenceMWU}\right)}$ is sufficient to make it feasible. Intuitively, a vertex $v$ starts with very small weight in the algorithm but every time that $u(v)$ flow is added to the in-flow of $v$, the weight $\weightFunction(v)$ of the vertex is increased by a $e
^{\epsilon} \approx (1+\epsilon)$ factor (see \Cref{lne:increaseWeight}) and thus after $\sim \log_{1+\epsilon}(\initialDifferenceMWU)$ times that in-flow of roughly $u(v)$ is added to $v$, the vertex $v$ becomes too heavy in weight to appear on any shortest path and therefore no additional flow is added to $v$ and we only need to scale as pointed out above. A similar argument ensures that the flow $\flow_{scaled}$ is cost-feasible.

To ensure that the flow $\flow_{scaled}$ is a flow of almost optimal flow value, Garg and Koennemann always augment the flow along the currently shortest path $\shortestSTPath$ with regard to the weight function $\weightFunctionCombined$ (defined in \Cref{lne:dataStructureInit} where $c$ is the original cost of the edge). They then use that since the weight $(\weightFunction + \costFunction \circ c)(\shortestSTPath)$ of the shortest path represents the left-hand side value of the most violated constraint in the dual LP \ref{eq:LPdual}, that scaling $\weightFunction$ and $\costFunction$ by $1/(\weightFunction + \costFunction \circ c)(\shortestSTPath)$ gives a feasible solution to the dual LP \ref{eq:LPdual}. Using weak duality as described in \Cref{thm:weakDuality}, it is then straight-forward to obtain that 
\[
    \weightFunctionCombined(\shortestSTPath) = (\weightFunction + \costFunction \circ c)(\shortestSTPath) \leq \frac{\sum_{v \in V} u(v) \weightFunction(v) + \overline{C} \costFunction}{OPT_{G, \overline{C}}}
\]
Using this insight, Garg and Koenemann can upper bound the objective function $\objFunction = \sum_{v \in V} u(v) \weightFunction(v) + \overline{C} \costFunction$ which serves as a potential function in the analysis and obtain a near-optimal lower bound on the objective value in terms of the optimal solution. Fleischer \cite{fleischer2000approximating} later showed that one can relax the requirement of using a shortest path to using only a $(1+\epsilon)$-approximate shortest path.

\paragraph{Our Approach.}
We follow this fundamental approach of the original analysis, however, we only have an estimator $\flowEst$ of $\flow$ and correspondingly only an estimator for $\weightFunctionCombined$. Moreover, as mentioned above, $\flow$ actually depends on $\flowEst$ because the shortest path  $\shortestSTPath[i]$ added to $\flow[i]$ is defined in terms of weights $\weightFunctionCombined[i-1]$, which were induced by $\flow[i-1]$. In order to analyze flow $\flow$, our goal will be to show that before each iteration $i$, we have that $\weightFunctionCombined[i-1]$ as induced by $\flowEst[i-1]$ is within a $(1+\epsilon)$ factor of $\weightFunctionCombined[i-1]$ as induced by $\flow[i-1]$. Using some rather straight-forward arguments this implies that the next approximate shortest path $\shortestSTPath[i]$ is $(1+\epsilon)^2$-approximate with regard to the metric induced by $\weightFunctionCombined[i-1]$ as induced by $\flow[i-1]$.

To this end, we notice that in order to bound the difference in the resulting function $\weightFunctionCombined$, we are required to show very strong concentration bounds to prove that $|\inflow_{\flow}(v) - \inflow_{\flowEst}(v)| \leq u(v)$ for each $v \in V$ and $|c(\flow) - c(\flowEst)| \leq \overline{C}$. Using careful arguments, we can derive the required concentration bounds. We can then finally use the concentration bounds to recover good guarantees for the flow estimator $\flowEst$ that our algorithm returns by relating it back to $\flow$.

In \Cref{subsec:feasibility}, we show that the returned flow estimator satisfies capacity- and cost-feasibility. We then show strong concentration bounds in \Cref{subsec:conBoundCaps}. Finally, we combine these results which allows us to carry out the analysis for the correctness of the algorithm following closely the approach by  Garg and Koenemann in \Cref{subsec:correctnessMWU}. Finally, in \Cref{subsec:runningTimeMWU}, we bound the total running time of  \Cref{alg:boundedCostFlowexcess}. 

\begin{table}[!ht]
    \centering
    \begin{tabular}{c|m{13cm}}
    \hline
    $\shortestSTPath[i]$ & The $(1+\epsilon)$-approximate shortest path $\shortestSTPath$ in the $i^{th}$ iteration of the while-loop in \Cref{lne:whileLoopMWU}.\\ \hline
    $\minCapIndexMWU[i]$ & The min-capacity on the path $\shortestSTPath[i]$ found during the $i^{th}$ iteration.\\\hline
    $\randomThreshold[i]$ & The value of $\randomThreshold$ during the $i^{th}$ iteration.\\\hline
    $\flowEst[i]$ & The total flow estimator after the $i^{th}$ iteration. \\\hline
    $\hat{c}_i$ & The cost of the flow $\flowEst[i] - \flowEst[i-1]$, added during iteration $i$. \\\hline
    $\weightFunctionEst[i]$ & The function $\weightFunctionEst$ after being updated using $\flowEst[i]$. \\\hline
    $\costFunctionEst[i]$ & The function $\costFunctionEst$ after being updated using $\flowEst[i]$. \\\hline
    $\weightFunctionCombined[i]$ & The combined weight function obtained from $\weightFunctionEst[i]$ and $\costFunctionEst[i]$.\\ \hline
    $\flow[i]$ & The flow obtained from routing $\flowRouted^{\minCapIndexMWU[j]}$ units of flow along each edge on $\shortestSTPath[j]$ for each $j \leq i$.\\ \hline
    $k$ & Number of iterations of the while-loop.\\ \hline
    \end{tabular}
    \caption{Variables depending on $i$, the iteration of the while-loop in \Cref{lne:whileLoopMWU}.}
    \label{tab:versions_Variables}
\end{table}

\subsection{Capacity- and Cost-Feasibility of the Returned Flow after Scaling}
\label{subsec:feasibility}

We start by proving that the flow estimator $\flowEst$ after scaling is a capacity- and cost-feasible flow. In order to obtain this feasibility result, we first upper bound the maximum amount of flow added to the in-flow of a vertex in a single while-loop iteration and analogously the maximum additional cost we add to the flow.

\begin{claim}
For any iteration $1 \leq i \leq k$,  we have that
\begin{align}
   & |\inflow_{\flowEst[i]}(v) - \inflow_{\flowEst[i-1]}(v)| \leq u(v) / \zeta &  \forall v \in V \label{clm:smallFlowIncrementInEachIteration}\\
   &|c(\flowEst[i]) - c(\flowEst[i-1])| = \hat{c}_i \leq \overline{C}/\zeta. \label{clm:smallCostOfEveryPath} &
\end{align}
\end{claim}
\begin{proof}
\textbf{\Cref{clm:smallFlowIncrementInEachIteration}:} In each iteration $i$, we add flow along a single $\beta$-edge-simple path $\shortestSTPath[i]$. Let some edge $e = (x,v)$ be on $\shortestSTPath[i]$ (possibly multiple times). Then, by the value of the steadiness of vertex $v$ in \Cref{lne:defineSteadiness}, we have that $\edgeSensitivity(e) \leq \log_{\flowRouted} \left( \frac{u(v)}{\zeta\beta \cdot \deg(v)}\right)$. But this implies that for each occurrence of $e$ on the path $\shortestSTPath[i]$, it occurs once in the foreach-loop starting in \Cref{lne:innermostForeachMWU} and then we add flow at most $\flowRouted^{\edgeSensitivity(e)} \leq u(v)/(\zeta\beta \cdot \deg(v))$ to  
\[
|\inflow_{\flowEst[i]}(v) - \inflow_{\flowEst[i-1]}(v)|.
\]
But since by definition of $\beta$-edge-simple paths, the path $\shortestSTPath[i]$ contains every such edge $e$ at most $\beta$ time, and since there are at most $\deg(v)$ such edges adding to the in-flow of $v$, the total contribution to the in-flow of $v$ of $\shortestSTPath[i]$ is at most $u(v)/(\zeta\beta \cdot \deg(v)) \cdot \beta \cdot \deg(v) = u(v)/\zeta$.

\textbf{\Cref{clm:smallCostOfEveryPath}:} First observe that at the beginning of every iteration of the while-loop $\hat{c}$ is initialized to $0$, and whenever flow is added to $\flowEst$ during iteration $i$ in \Cref{lne:deployFlowToSensitiveEdges}, we immediately add the cost of the added flow in \Cref{lne:deployCostEstimator} to $\hat{c}$. When the iteration terminates, we have that $\hat{c}_i$ is equal to the cost of the flow added during the iteration $i$ of the while-loop. Thus, the equality $|c(\flowEst[i]) - c(\flowEst[i-1])| = \hat{c}_i$ holds.

For the inequality $\hat{c}_i \leq \overline{C}/\zeta$, observe that by definition of $\sigma(e)$, for any edge $e = (x,y) \in E$, we have that $\edgeSensitivity(e) \leq \log_{\flowRouted}\left(\frac{ \overline{C}}{m c(y) \zeta \beta}\right)$. Thus, every time the foreach-loop starting in \Cref{lne:innermostForeachMWU} features the edge $e$, it adds cost $c(y) \cdot \flowRouted^{\edgeSensitivity(e)} \leq c(y) \cdot \frac{\overline{C}}{m c(y) \zeta \beta } = \frac{\overline{C}}{m \beta \zeta}$ to $\hat{c}_i$. Since each edge $(x,y)$ can occur at most $\beta$ times on the path $\shortestSTPath[i]$ (by definition of $\beta$-edge-simple paths), and since there can be at most $m$ edges on the path, we can bound the total cost added by $\frac{\overline{C}}{m \beta\zeta} \cdot \beta m = \frac{\overline{C}}{\zeta}$, as desired.
\end{proof}

\begin{claim}\label{clm:capacityFeasible}
The flow $\flowEst_{scaled} = \frac{\flowEst}{(1 + 10\epsilon) \cdot \log_{1+\epsilon}\left(\frac{1+\epsilon}{\initialDifferenceMWU}\right)}$ returned by \Cref{alg:boundedCostFlowexcess} in \Cref{lne:returnFlow} is capacity- and cost-feasible.
\end{claim}
\begin{proof} \textbf{Capacity-feasible:} Fixing a vertex $v \in V \setminus \{s,t\}$ and a while-loop iteration $i$. Then, it is not hard to see that we have 
\[
\weightFunctionEst[i](v) = \frac{\initialDifferenceMWU}{u(v)} \cdot \textsc{Exp}\left(\frac{\epsilon \cdot \inflow_{\flowEst[i]}(v)}{u(v)} \right).
\]
This follows since every time flow is added with $v$ on the flow path, the function $\weightFunctionEst(v)$ is multiplied by $\textsc{Exp}\left(\frac{\epsilon F }{u(v)} \right)$ where $F$ is the amount of in-flow that is added to $v$ due to the new flow path (by \Cref{lne:defineWeight} and \Cref{lne:increaseWeight}). 

Further, we claim that $\weightFunctionEst(v) \leq (1 + \epsilon)$ at the end of the algorithm. To see this observe first that once $v$ has $\weightFunctionEst(v) \geq 1$, the while-loop starting in \Cref{lne:whileLoopMWU}, has its condition violated and therefore ends (here we use that $u(v) \geq 1$ by Proposition \ref{prop:reductionVertexCapacities}). Thus, at the beginning of the last while-loop iteration, we must have had $\weightFunctionEst(v) < 1$. Thus, only a single last path might further increase $\weightFunctionEst(v)$. Let us assume that $v$ is on the last path selected since otherwise we are done. But by \Cref{clm:smallFlowIncrementInEachIteration}, a single iteration can add at most  $u(v) / \zeta$ to the in-flow of $v$.

We therefore have that the final weight $\weightFunctionEst[k](v)$ is at most
\[
\weightFunctionEst[k](v) \leq \weightFunctionEst[k-1](v) \cdot \textsc{Exp}\left(\frac{\epsilon u(v)}{\zeta \cdot u(v)} \right) < 1 \cdot e^{\epsilon}.
\]

Taking the logarithm on both sides, we obtain
\[
\log\left(\frac{\initialDifferenceMWU}{u(v)}\right) + \frac{\epsilon \cdot \inflow_{\flowEst}(v) }{u(v)} \leq \epsilon \iff \inflow_{\flowEst}(v) \leq  u(v) \cdot \left(1-\log\left(\frac{\initialDifferenceMWU}{u(v)}\right)/\epsilon\right). 
\]
It remains to observe that by assumption $u(v) \leq m^5$ (see \cref{prop:reductionVertexCapacities}) and the definition of $\initialDifferenceMWU = m^{-1/\epsilon}$,
\begin{align*}
1-\log\left(\frac{\initialDifferenceMWU}{u(v)}\right)/\epsilon& \leq 1-\log\left(\frac{\initialDifferenceMWU}{m^5}\right)/\epsilon
= 1 + (1+5\epsilon) \log(1/\initialDifferenceMWU) / \epsilon
\\
&\leq (1+ 7\epsilon) \log\left(\frac{1}{\initialDifferenceMWU}\right)/\epsilon \leq (1 + 10\epsilon) \cdot \log_{1+\epsilon}\left(\frac{1+\epsilon}{\initialDifferenceMWU}\right)
\end{align*}
where we use that $\log\left(\frac{\initialDifferenceMWU}{m^5}\right)= \log(1/m^5) + \log(\initialDifferenceMWU) = (5\epsilon + 1) \log(\initialDifferenceMWU)$, that $\log(x) = - \log(1/x)$, and that $1 \leq \log(m) = \log(1/\initialDifferenceMWU)/\epsilon$. The third inequality follows by a change of basis of the logarithm and the inequalities $\log(\frac{1+\epsilon}{\initialDifferenceMWU}) \leq (1+\epsilon) \log(\frac{1}{\initialDifferenceMWU})$, $\log(1+\epsilon) \geq \epsilon$ and $1+x \leq e^x \leq 1+x+x^2$ for $x \leq 1$. This proves capacity-feasiblity for all $v$ but for $s$ and $t$, for which the capacities are $\infty$ by assumption, which implies that the claim for them is vacuously true.

\textbf{Cost-feasible:} We observe that after iteration $i$, we have 
\[
\costFunctionEst[i] = \frac{\initialDifferenceMWU}{\overline{C}} \cdot \textsc{Exp}\left(\frac{\epsilon \cdot \sum_{j \leq i} c(\flowEst[j]- \flowEst[j-1])}{\overline{C}} \right) = \frac{\initialDifferenceMWU}{\overline{C}} \cdot \textsc{Exp}\left(\frac{\epsilon c(\flowEst[i])}{\overline{C}} \right).
\]
where we use that we do not cancel any flow between iterations to obtain the equality. Using \Cref{clm:smallCostOfEveryPath} in place of \Cref{clm:smallFlowIncrementInEachIteration}, we can follow the same proof template as for capacity-feasiblity to conclude the claim.
\end{proof}

\subsection{Strong Concentration Bounds}
\label{subsec:conBoundCaps}

Next, we would like to obtain strong concentration bounds for the difference between $\flow$ and $\flowEst$. To this end, we use a result that is akin to Chernoff Bounds while allowing for some limited dependence between the random variables (in a Martingale fashion).

\begin{theorem}[see \cite{koufogiannakis2014nearly, chekuri2018randomized}]\label{thm:onlineChernoff}
Let $X_1, X_2, \dots, X_k, Z_1, Z_2, \dots Z_k \in [0,W]$ be random variables and let $\epsilon \in [0, 1/2)$. Then, if for every $i \in \{1, 2, \dots, k\}$, 
\begin{equation}
       \mathbb{E}[X_i \;|\; X_1, \dots, X_{i-1}, Z_1, \dots, Z_i] = Z_i \label{eq:conditionToUseMartingale}
\end{equation}
then for any $\delta > 0$, we have that
\[
   \mathbb{P}\left[\left| \sum_{i=1}^k X_i - Z_i\right| \geq \epsilon \sum_{i=1}^k Z_i + \delta\right] \geq 2 \cdot (1+\epsilon)^{-\delta/W}.
\]
\end{theorem}

Note that we will have to know the value $k$ of iterations in \Cref{alg:boundedCostFlowexcess}. Here we will use a crude upper bound of $k = O(n^{20} \log^2(m)/\epsilon)$. This is straight-forward since in each iteration starting in \Cref{lne:whileLoopMWU}, we add a $1/(n^{10/\tau}\zeta\beta)$-fraction of the capacity of the min-capacity vertex on the flow path to the in-flow of the vertex.

We are now ready to prove the main result of this section.

\begin{claim}\label{clm:flowDeviationIsSmallForSingleVertexAtSingleTimeStep}
For any $v \in V$, we have
\[
    \mathbb{P}\left[\exists 0 \leq i \leq k, |\inflow_{\flow[i]}(v) - \inflow_{\flowEst[i]}(v)| \geq u(v)\right] \leq n^{-20}.
\]
\end{claim}
\begin{proof}
We prove by induction on $i$. The base case $i=0$, is true since $\flow[0]$ and $\flowEst[0]$ are initialized to $0$.

For the inductive step $i-1 \mapsto i$ for $i \geq 1$, we start by defining the two random processes $\{X_j = \inflow_{\flowEst[j]}(v) - \inflow_{\flowEst[j-1]}(v)\}_j$ and  $\{Z_j = \inflow_{\flow[j]}(v)-\inflow_{\flow[j-1]}(v)\}_j$. Here, $Z_j$ is the amount of flow added to the in-flow of $v$ in $\flow$ during the $j^{th}$ iteration of \Cref{lne:whileLoopMWU}, while $X_j$ is the amount of flow added to the in-flow of $v$ in $\flowEst$. Thus, since we never cancel flow over iterations, we have that $\sum_{j=1}^i Z_i = \inflow_{\flow[i]}(v)$ and $\sum_{j=1}^i X_i = \inflow_{\flowEst[i]}(v)$.

Now, the key statement that we need in order to invoke \Cref{thm:onlineChernoff}, is to prove Condition \ref{eq:conditionToUseMartingale}. Therefore, observe that given $Z_j$, the only randomness in determining $X_j$ stems from picking $\randomThreshold[j]$ in \Cref{lne:selectSteadinessThreshold}. Using the definition of expectation, we obtain that
\[
    \mathbb{E}[X_j \; |\; X_1, \dots, X_{j-1}, Z_1, \dots, Z_j ] = \sum_{e = (x,v) \in \shortestSTPath[j]} \mathbb{P}[\edgeSensitivity(e) \leq \minCapIndexMWU[j] + \randomThreshold[j]] \cdot \flowRouted^{\edgeSensitivity(e)}.
\]
This follows since the algorithm adds $\flowRouted^{\edgeSensitivity(e)}$ units to $X_j$ if the random threshold makes the sum $\minCapIndexMWU[j] + \randomThreshold[j]$ larger than the steadiness threshold $\sigma(e)$ for every edge $e$ on $\shortestSTPath[j]$ that enters $v$. We can then use the definition of the exponential distribution coordinate-wise which gives that $\mathbb{P}[\edgeSensitivity(e) \leq \minCapIndexMWU[j] + \randomThreshold[j]] = \mathbb{P}[\edgeSensitivity(e) - \minCapIndexMWU[j] \leq \randomThreshold[j]] = 1- (1-e^{-\log \flowRouted \cdot (\edgeSensitivity(e) - \minCapIndexMWU[j])}) = \flowRouted^{\minCapIndexMWU[j] - \edgeSensitivity(e)}$.

But this implies that 
\[
\mathbb{E}[X_j \; |\; X_1, \dots, X_{i-1}, Z_1, \dots, Z_j ] = \sum_{e = (x,v) \in \shortestSTPath[j]} \flowRouted^{\minCapIndexMWU[j] - \edgeSensitivity(e)} \cdot \flowRouted^{\edgeSensitivity(e)} = \sum_{e = (x,v) \in \shortestSTPath[j]} \flowRouted^{\minCapIndexMWU[i]} = Z_i.
\]

Finally, we can invoke \Cref{thm:onlineChernoff} where we plug in $\eta = \frac{1}{4 \cdot (1 + 10\epsilon) \cdot \log_{1+\epsilon}\left(\frac{1+\epsilon}{\initialDifferenceMWU}\right)}$ and $\delta = u(v)/2$. Recall we obtain a result of the form for our choices of $\epsilon$ and $\eta$
\[
\mathbb{P}\left[\left| \sum_{j=1}^i X_j - Z_j\right| \geq \eta \sum_{j=1}^i Z_j + u(v)/2\right] \geq 2 \cdot (1+\eta)^{-\delta/W}.
\]
We then observe that the random variables are bounded by $W \leq u(v)/\zeta$ by \Cref{clm:smallFlowIncrementInEachIteration}, so $2 \cdot (1+\eta)^{-(u(v)/2)/W} \geq 2 \cdot (1+\eta)^{-\zeta/2} \geq n^{-40}$. At the same time, we have that 
$4 \eta \sum_{j=1}^i Z_j \leq 2 \cdot u(v)$ since $\flowEst[i]$ is capacity-feasible after scaling by $4 \eta$ as shown in \cref{clm:capacityFeasible} and $|\inflow_{\flowEst[i-1]}(v) - \inflow_{\flow[i-1]}(v)| <  u(v)$ by the induction hypothesis. The claim follows by carefully taking a union-bound over all $i \leq k$. 
\end{proof}
\begin{corollary}\label{clm:flowDeviationIsSmall}
We have
\[
    \mathbb{P}\left[\exists v \in V, \exists 0 \leq i \leq k, |\inflow_{\flow[i]}(v) - \inflow_{\flowEst[i]}(v)| \geq u(v)\right] \leq n^{-10}.
\]

\end{corollary}
Following the proof template for the concentration bounds on the flow on each edge, we can get similar concentration bounds on the cost of the flow.

\begin{restatable}{claim}{concentrationOfCost}\label{clm:costDeviationIsBounded}
For any $0 \leq i \leq k$, we have 
\[
 \mathbb{P}\left[|c(\flow[i]) - c(\flowEst[i])| \geq \overline{C}\right] \leq n^{-10}.
\]
\end{restatable}
\begin{proof}
Consider the random processes $\{X_j = c(\flowEst[j]) - c(\flowEst[j-1])\}_j$ and $\{Z_j = c(\flow[j]) - c(\flow[j-1])\}_j$. Next, observe that for $j \geq 1$, we have by definition that $Z_j = c(\flow[j]) - c(\flow[j-1]) = \sum_{v \in \shortestSTPath[j]} c(v) \flowRouted^{\minCapIndexMWU[j]}$ (recall that we assume $c(s) = 0$, and that $\shortestSTPath[j]$ is a multi-set). Further, we have that 
\[
    \mathbb{E}[X_j \; |\; X_1, \dots, X_{i-1}, Z_1, \dots, Z_j ] = \sum_{e =(x,v) \in \shortestSTPath[j]} c(v) \cdot \mathbb{P}[\edgeSensitivity(e) \leq \minCapIndexMWU[j] + \randomThreshold[j]] \cdot \flowRouted^{\edgeSensitivity(e)}
\]
again since $\gamma_j$ is the only random variable not conditioned upon that determines $X_j$. But it is straight-forward to see that the right-hand side is exactly $Z_j$, using again the definition of the exponential distribution. Finally, we use this claim in an induction on $i$, to invoke at each iteration step \Cref{thm:onlineChernoff} with the same parameters as chosen above and carefully take a union bound. This concludes the proof.
\end{proof}

We henceforth condition on Claims \ref{clm:flowDeviationIsSmall} and \ref{clm:costDeviationIsBounded} holding true and treat them like deterministic results.

\subsection{Correctness of the Algorithm}
\label{subsec:correctnessMWU}

We can now use the results from the previous sections to conclude that our algorithm returns the correct solution with high probability. We start by showing that the flow $\flow$ is a near-optimal flow and then proceed by coupling $\flow$ and $\flowEst$.

\begin{claim}\label{lma:RealFlowIsFeasible}
The flow $\flow_{scaled} = \frac{\flow}{(1+24\epsilon) \log_{1+\epsilon}\left(\frac{1+\epsilon}{\initialDifferenceMWU}\right)}$ is a capacity-feasible and satisfies flow conservation constraints.
\end{claim}
\begin{proof}
We have that since $\flow_{scaled}$ is the weighted sum of $s$-to-$t$ paths that the flow conservation constraints are satisfied. To see that $\flow_{scaled}$ is capacity feasible, observe that for each $v \in V$,
\begin{align*}
\frac{\inflow_{\flow}(v)}{(1+24\epsilon)\log_{1+\epsilon}\left(\frac{1+\epsilon}{\initialDifferenceMWU}\right)} 
&\leq \frac{\inflow_{\flowEst}(v) + u(v)}{(1+24\epsilon)\log_{1+\epsilon}\left(\frac{1+\epsilon}{\initialDifferenceMWU}\right)} \\
&\leq  \frac{\inflow_{\flowEst}(v)}{(1+24\epsilon)\log_{1+\epsilon}\left(\frac{1+\epsilon}{\initialDifferenceMWU}\right)} + \frac{\epsilon u(v)}{(1+\epsilon)} \\
&  \leq \frac{\inflow_{\flowEst}(v)}{(1+2\epsilon)(1+10\epsilon)\log_{1+\epsilon}\left(\frac{1+\epsilon}{\initialDifferenceMWU}\right)} + \epsilon u(v)\\
& \leq (1-\epsilon) \cdot \frac{\inflow_{\flowEst}(v)}{(1+10\epsilon)\log_{1+\epsilon}\left(\frac{1+\epsilon}{\initialDifferenceMWU}\right)} + \epsilon u(v) &\leq u(v)
\end{align*}
where we use \Cref{clm:flowDeviationIsSmall} in the first inequality, and in the second inequality that $\log_{1+\epsilon}\left(\frac{1+\epsilon}{\initialDifferenceMWU}\right) = \log_{1+\epsilon}\left(\frac{1+\epsilon}{m^{-1/\epsilon}}\right) \geq 1/\epsilon$, in the third and forth inequality, we used $1+x \leq e^x \leq 1+x+x^2$ (for $x \leq 1$), and in the final inequality, we used that by capacity feasibility of $\flowEst$ (after scaling) as established in \Cref{clm:capacityFeasible}, we have
\[
    \frac{\inflow_{\flowEst}(v)}{(1+10\epsilon) \log_{1+\epsilon}\left(\frac{1+\epsilon}{\initialDifferenceMWU}\right)} \leq u(v).
\]
\end{proof}

\begin{lemma}\label{lma:nearOptimalRealFlow}
The flow $\flow_{scaled} = \frac{\flow}{(1+24\epsilon) \log_{1+\epsilon}\left(\frac{1+\epsilon}{\initialDifferenceMWU}\right)}$ is a $(1-\Theta(\epsilon))$-optimal flow.
\end{lemma}
\begin{proof}
We have feasibility of $\flow_{scaled}$ by \Cref{lma:RealFlowIsFeasible}.

It remains to prove that the flow value $F$ of $\flow_{scaled}$ is at least $(1-\epsilon) OPT_{G, \overline{C}}$. To this end, let us define the functions
\begin{align*}
    \weightFunction[i](v) &= \frac{\initialDifferenceMWU}{u(v)} \cdot \textsc{Exp}\left(\frac{\epsilon \cdot \inflow_{\flow[i]}(v)}{u(v)}\right) \\
    \costFunction[i] &= \frac{ \initialDifferenceMWU }{\overline{C} } \cdot \textsc{Exp}\left(\frac{\epsilon \cdot c(\flow[i])}{\overline{C}}\right)\\
    \objFunction[i] &= \sum_{v \in V} u(v) \cdot \weightFunction[i](e) + \overline{C} \cdot \costFunction[i]
\end{align*}
for all $0 \leq i \leq k$. Here, we define $\weightFunction[i]$ to be the weight function that would result if we would always use the flow $\flow$ up to update vertex weights instead of the flow estimator $\flowEst$ as is the case for $\weightFunctionEst[i]$. Analogously, $\costFunction$ is the version of $\costFunctionEst$ that is based on $\flow$ instead of $\flowEst$ and $\objFunction$ is the resulting objective value corresponding to the sum we use in the while-loop condition in \Cref{lne:whileLoopMWU}. 

We start by establishing a useful claim that relates these versions based on the flow $\flow$ instead of $\flowEst$ tightly together. 

\begin{claim}\label{clm:smallMistakeByTakenTheHatOff}
We have for any $0 \leq i \leq k$, we have that 
\begin{align*}
     \forall v \in V, &\frac{1}{(1+2\epsilon)}\weightFunctionEst[i](v) \leq \weightFunction[i](v) \leq (1+2\epsilon)\weightFunctionEst[i](v) \text{, and } \\
     & \frac{1}{(1+2\epsilon)}\costFunctionEst[i] \leq \costFunction[i] \leq (1+2\epsilon)\costFunctionEst[i].
\end{align*}
\end{claim}
\begin{proof}
We observe that by \Cref{clm:flowDeviationIsSmall}, we have $|\inflow_{\flow[i]}(v) - \inflow_{\flowEst[i]}(v)| < u(v)$ for every $v \in V$, and therefore we have
\begin{align}\label{eq:weightFunctionIteration}
\begin{split}
     \weightFunctionEst[i](v) &\leq \frac{\initialDifferenceMWU}{u(v)} \textsc{Exp}\left(\frac{\epsilon \cdot \inflow_{\flowEst[i]}(v)}{u(v)}\right) \leq  \frac{\initialDifferenceMWU}{u(v)} \textsc{Exp}\left(\frac{\epsilon (\inflow_{\flow[i]}(v) + u(v)) }{u(v)}\right) \\
     &\leq \frac{\initialDifferenceMWU}{u(v)} \left(1+2\epsilon\right)  \textsc{Exp}\left(\frac{\epsilon \inflow_{\flow[i]}(v) }{u(v)}\right) = \left(1+2\epsilon\right) \weightFunction[i](v).
\end{split}
\end{align}
where we use for the inequality that $e^x \leq 1 + x + x^2$ for $x \leq 1$, and $x^2 \leq x$ for $x \leq 1$. The remaining inequality statements can be proven by following this template and using the additional \Cref{clm:costDeviationIsBounded}.
\end{proof}

Observe that for any $i \geq 1$, for every $v \in V$ that occurs $\beta_v$ times on the path (minus one if $v = s$) that 
\begin{align}
    \begin{split}\label{eq:boundUIteration}
         \weightFunction[i](v) &= \frac{\initialDifferenceMWU}{u(v)} \cdot \textsc{Exp}\left(\frac{\epsilon \inflow_{\flow[i-1]}(v)}{u(v)}\right) \cdot \textsc{Exp}\left(\frac{\epsilon\left(\inflow_{\flow[i]}(v)-\inflow_{\flow[i-1]}(v)\right)}{u(v)}\right) \\
         &\leq \frac{\initialDifferenceMWU}{u(v)} \cdot \textsc{Exp}\left(\frac{\epsilon \inflow_{\flow[i-1]}(v)}{u(v)}\right) \cdot \left(1 + \left(\epsilon + \epsilon^2\right) \cdot  \frac{\beta_v \cdot \flowRouted^{\minCapIndexMWU[i]}}{u(v)}\right)\\
         &= \weightFunction[i-1](v) + {\initialDifferenceMWU} \cdot  \textsc{Exp}\left(\frac{\epsilon \inflow_{\flow[i-1]}(v)}{u(v)}\right) \cdot \left(\epsilon + \epsilon^2\right) \cdot \beta_v \cdot  \flowRouted^{\minCapIndexMWU[i]}\\
         &= \weightFunction[i-1](v) + \frac{\left(\epsilon + \epsilon^2\right) \cdot \beta_v \cdot  \flowRouted^{\minCapIndexMWU[i]}}{u(v)} \cdot \weightFunction[i-1](v)
     \end{split}
\end{align}
where we use for the first inequality that $e^x \leq 1 + x + x^2$ for $x \leq 1$ (which is given since our exponent is at most $\epsilon/\zeta$ by \Cref{clm:smallFlowIncrementInEachIteration}) and that $\beta_v \cdot \flowRouted^{\minCapIndexMWU[i]} = \inflow_{\flow[i]}(v)-\inflow_{\flow[i-1]}(v)$ by definition of $f$. We then rearrange terms to obtain the equalities. 

For $\costFunction[i]$, we can argue similarly that
\begin{align}\label{eq:boundCIteration}
\begin{split}
         \costFunction[i] &\leq \frac{\overline{C} }{ \initialDifferenceMWU } \cdot \textsc{Exp}\left(\frac{\epsilon \cdot c(\flow[i-1])}{\overline{C}}\right) \cdot \left(1 + \left(\epsilon + \epsilon^2\right) \cdot \frac{c(\flow[i])-c(\flow[i-1])}{\overline{C}}\right)\\
         &\leq  \costFunction[i-1] + \frac{\left(\epsilon + \epsilon^2\right) \cdot  \flowRouted^{\minCapIndexMWU[i]} \cdot c(\shortestSTPath[i])}{ \overline{C}} \cdot \costFunction[i-1]
\end{split}
\end{align}
where we use that the difference in the cost between flows $\flow[i]$ and $\flow[i-1]$ is the cost of the path $\shortestSTPath[i]$ times the value of the flow we send in iteration $i$ which is $\flowRouted^{\minCapIndexMWU[i]}$. We further use Equation \ref{clm:smallCostOfEveryPath} to ensure that we use inequality $e^x \leq 1 + x + x^2$ with $x \leq 1$. The last inequality again uses \Cref{clm:smallMistakeByTakenTheHatOff}.

Combined, we obtain that 
\begin{align}\label{eq:boundingDi}
\begin{split}
     \objFunction[i] &= \sum_{v \in V} u(v) \cdot \weightFunction[i](v) + \overline{C} \cdot \costFunction[i]\\
     &\leq \objFunction[i-1] + \sum_{v \in V} \left(\epsilon + \epsilon^2\right) \cdot \beta_v \cdot  \flowRouted^{\minCapIndexMWU[i]} \cdot \weightFunction[i-1](v) +  \left(\epsilon + \epsilon^2\right) \cdot  \flowRouted^{\minCapIndexMWU[i]} \cdot c(\shortestSTPath[i]) \cdot \costFunction[i-1]\\
     &= \objFunction[i-1] +  \left(\epsilon + \epsilon^2\right) \cdot  \flowRouted^{\minCapIndexMWU[i]} \cdot \left(\weightFunction[i-1] + \costFunction[i-1] \circ c\right)(\shortestSTPath[i]).
\end{split}
\end{align}
Let $w' = \weightFunction[i-1] + \costFunction[i-1] \circ c$, we observe that the distance $\dist_{w'}(s,t)$ from $s$ to $t$ in $G$, weighted by function $w'$, satisfies
\[
\dist_{w'}(s,t) \leq \frac{ \objFunction[i-1]}{OPT_{G, \overline{C}}}.
\] 
This follows since scaling $\weightFunction[i-1]$ and $\costFunction[i-1]$ by $1/\dist_{w'}(s,t)$ makes them a feasible solution to the dual LP given in \Cref{eq:LPdual}. Since it is a feasible solution to a minimization problem, we have that 
\[
\frac{ \objFunction[i-1]}{\dist_{w'}(s,t)} \geq OPT_{G, \overline{C}}
\]
where we used weak duality as stated in \Cref{thm:weakDuality} for the inequality to further lower bound the optimal value to the dual LP by the optimal value of the primal LP. Multiplying both sides by $\dist_{w'}(s,t)$ and dividing by $OPT_{G, \overline{C}}$ proves the statement. 

Finally, we observe that the selected path $\shortestSTPath[i]$, is a $(1+2\epsilon)^5$-approximate shortest $s$ to $t$ path with respect to $w'$. This follows since $\shortestSTPath[i]$ is selected to be a $(1+\epsilon)$-approximate shortest path in the metric determined by weight function $\weightFunctionCombined[i-1]$ by the definition of the SSSP data structure $\SSSP^{\pi}(G,s,t,\eps,\beta)$. Further, $\weightFunctionCombined[i-1]$ is a $(1+\epsilon)$ approximation of the metric induced by the weight function $(\weightFunctionEst[i-1] + \costFunctionEst[i-1] \circ c)$ (as can be seen from the rounding of $\costFunctionEst$ described in \Cref{lne:dataStructureInit}). Finally, $(\weightFunctionEst[i-1] + \costFunctionEst[i-1] \circ c)$ is a $(1+2\epsilon)^2$-approximation of $w'$ by \Cref{clm:smallMistakeByTakenTheHatOff}. 

It remains to put everything together: from the combination of 
\Cref{eq:boundingDi} and the path approximation, we obtain that
\begin{align}
    \objFunction[i] &\leq \objFunction[i-1] + 
    \left(\epsilon + \epsilon^2\right) \cdot  \flowRouted^{\minCapIndexMWU[i]}(1+2\epsilon)^5 \cdot  \frac{ \objFunction[i-1]}{OPT_{G, \overline{C}}} \\
    &\leq  \objFunction[i-1] \cdot \textsc{Exp}\left(\frac{\left(\epsilon + \epsilon^2\right) \cdot  (1+2\epsilon)^5}{OPT_{G, \overline{C}}} \cdot \flowRouted^{\minCapIndexMWU[i]} \right).
\end{align}

We finally observe that by the while-loop condition, we have that after the last iteration, we have that $\objFunction[k] \geq 1$. Since $\objFunction[0] \geq \initialDifferenceMWU m$, we therefore have that
\begin{align*}
    &1 \leq \objFunction[k] \leq \initialDifferenceMWU m \cdot \textsc{Exp}\left(\frac{\left(\epsilon + \epsilon^2\right) \cdot (1+2\epsilon)^5}{OPT_{G, \overline{C}}} \cdot \sum_{i=1}^{k} \flowRouted^{\minCapIndexMWU[i]} \right) &\iff\\
    &0 \leq \log(\initialDifferenceMWU m) + \left(\frac{\left(\epsilon + \epsilon^2\right) \cdot (1+2\epsilon)^5}{OPT_{G, \overline{C}}} \cdot \sum_{i=1}^{k} \flowRouted^{\minCapIndexMWU[i]} \right) &\iff\\
    & \sum_{i=1}^{k} \flowRouted^{\minCapIndexMWU[i]} \geq \frac{\log(\frac{1}{m\initialDifferenceMWU}) \cdot OPT_{G, \overline{C}}}{\left(\epsilon + \epsilon^2\right) \cdot (1+2\epsilon)^5} & 
\end{align*}
Noticing that the value of the flow $\flow$ is exactly $F = \sum_{i=1}^{k} \flowRouted^{\minCapIndexMWU[i]}$ and therefore the flow value of $\flow_{scaled}$ is $\flow_{scaled} = \frac{F}{(1+24\epsilon)\log_{1+\epsilon}\left(\frac{1+\epsilon}{\initialDifferenceMWU}\right)}$, we have
\begin{align*}
    \flow_{scaled} &\geq  \frac{\log(\frac{1}{m\initialDifferenceMWU}) \cdot OPT_{G, \overline{C}}}{\left(\epsilon + \epsilon^2\right) \cdot (1+2\epsilon)^5} \cdot \frac{1}{(1+24\epsilon)\log_{1+\epsilon}\left(\frac{1+\epsilon}{\initialDifferenceMWU}\right)}\\
    &= \frac{(1 - \epsilon) \log(1/\initialDifferenceMWU) \cdot OPT_{G, \overline{C}}}{\left(\epsilon + \epsilon^2\right) \cdot (1+2\epsilon)^5} \cdot \frac{\log(1+\epsilon) }{(1+24\epsilon)\log\left(\frac{1+\epsilon}{\initialDifferenceMWU}\right)}\\
    &\geq \frac{(1 - \epsilon) \cdot OPT_{G, \overline{C}} \cdot \log(1+\epsilon)}{\left(\epsilon + \epsilon^2\right)(1+24\epsilon)^7} \geq \frac{(1 - \epsilon) \cdot \epsilon \cdot OPT_{G, \overline{C}}}{\left(\epsilon + \epsilon^2\right) (1+24\epsilon)^7} \\
    &\geq \frac{(1 - \epsilon) OPT_{G, \overline{C}}}{ (1+24\epsilon)^8} \geq (1 - 192\epsilon)(1-\epsilon) OPT_{G, \overline{C}} \geq (1 - 768\epsilon) OPT_{G, \overline{C}} 
\end{align*}
where we use that $\initialDifferenceMWU = m^{-1/\epsilon}$ such that in the first equality we can use $\log(\frac{1}{m\initialDifferenceMWU}) = \log(1/\initialDifferenceMWU) - \log(m) = (1-\epsilon) \log(1/\initialDifferenceMWU)$, and for the second term that we can change basis of the logarithm using $\log_{1+\epsilon}(x) = \frac{\log(x)}{\log({1+\epsilon})}$. We then obtain the second inequality using $\log(\frac{1+\epsilon}{\initialDifferenceMWU}) \leq (1+\epsilon) \log(\frac{1}{\initialDifferenceMWU})$, the third inequality using $\log(1+\epsilon) \geq \epsilon$. In the final two inequalities, we use that for $|x| < 1$, we have $(1+x)^n \leq (1+nx)$, from the Taylor series of $1/(1+x)$, we obtain $1/(1+x) \geq (1-x)$, and finally we use that $1+x+x^2 \geq e^x \geq 1+x$ combined with the fact that $(1-384\epsilon+384\epsilon^2) \geq (1-192\epsilon)$ using our assumption that $\epsilon \leq 1/768$.
\end{proof}

It remains to show that $\flowEst$ after scaling is a $(1-\epsilon)$-pseudo-optimal flow. This proves the correctness of \Cref{thm:mainMinCost2}.

\begin{corollary}\label{cor:correctnessOfHatF}
The flow $\flowEst_{scaled} = \frac{\flowEst}{(1+10\epsilon)\log_{1+\epsilon}\left(\frac{1+\epsilon}{\initialDifferenceMWU}\right)}$ is a $(1-\Theta(\epsilon))$-pseudo-optimal flow.
\end{corollary}
\begin{proof}
Combining \Cref{lma:nearOptimalRealFlow} with \Cref{clm:capacityFeasible} immediately gives the Corollary.
\end{proof}

\subsection{Runtime Complexity of the Algorithm}
\label{subsec:runningTimeMWU}

Next, we bound the runtime of the algorithm. In this section, we use the fact that $U \leq m^5 \leq n^{10}$ by Proposition \ref{prop:reductionVertexCapacities}. This ensures that every edge $e$ has steadiness $\edgeSensitivity(e) \in [1, \tau]$ because $\edgeSensitivity(e) \leq \log_{n^{10/\tau}}(u(v)) = \frac{\log u(v)}{\log n^{10/\tau} } \leq \tau \cdot \frac{\log n^{10}}{\log n^{10} } = \tau$.

We start by giving an upper bound on the number of times that we enter the foreach-loop in \Cref{lne:innermostForeachMWU}.

\begin{claim}\label{clm:numberOfWhileLoopIterations}
The total number of edges $e = (x,y)$ that are looked at in the foreach-loop in \Cref{lne:innermostForeachMWU}, over the entire course of Algorithm \ref{alg:boundedCostFlowexcess}, is at most
\[
O(m \log (m) \zeta \beta \flowRouted/ \epsilon) = O(m \log^2(m) \cdot \beta \cdot n^{10/\tau}/ \epsilon^2).
\]
\end{claim}
\begin{proof}
We observe that for any edge $e = (x,v) \in E$, upon entering the foreach-loop, we add $\flowRouted^{\edgeSensitivity(e)}$ units of flow to $\flowEst(e)$. Recall that $\edgeSensitivity(e) = \left\lfloor \log_{\flowRouted} \left(\min \{\frac{\overline{C}}{m c(v)} , \frac{u(v)}{ \deg(v)} \} \right) /(\zeta\beta) \right\rfloor$. We distinguish two cases:
\begin{itemize}
    \item if $\edgeSensitivity(e) = \left\lfloor \log_{\flowRouted} \left(\frac{\overline{C}}{m c(v)} \right) /(\zeta\beta) \right\rfloor$: then upon adding $\flowRouted^{\edgeSensitivity(e)}$ units of flow to $\flowEst(e)$, we increase the cost of $c(\flowEst)$ by at least
    \begin{align*}
    c(v) \cdot \flowRouted^{\left\lfloor \log_{\flowRouted} \left(\frac{\overline{C} }{m c(v)} \right) /(\zeta\beta) \right\rfloor} 
    \geq c(v) \cdot \flowRouted^{\log_{\flowRouted} \left(\frac{\overline{C} }{m c(v)} \right) /(\zeta\beta) - 1} = c(v) \cdot \left(\frac{\overline{C} }{m c(v)} \right) /(\zeta\beta\flowRouted) = \frac{\overline{C} }{m \cdot \zeta\beta\flowRouted}.
    \end{align*}
    But since we have by \Cref{lma:nearOptimalRealFlow} in combination with \Cref{clm:costDeviationIsBounded} that $c(\flowEst) = O(\overline{C} \cdot \log m/\epsilon)$ and since the cost is monotonically increasing over time (because the algorithm never cancels flow), there are at most $O(m\log(m) \zeta \beta \flowRouted/\epsilon)$ such iterations.
    \item otherwise, we have $\edgeSensitivity(e) = \left\lfloor \log_{\flowRouted}(\frac{u(v)}{ \deg(v)} /(\zeta\beta))\right\rfloor$: but this implies that we increase the in-flow to $v$ by at least $\frac{u(v)}{ \deg(v)}/(\zeta\beta\flowRouted)$ (by the same argument as above). On the other hand, by \Cref{clm:capacityFeasible}, we have that $\flowEst_{scaled}$ is a capacity-feasible flow. Thus, we have for every vertex $v \in V$, $\inflow_{\flowEst}(v) \leq u(v) \cdot (1+10\epsilon)\log_{1+\epsilon}\left(\frac{1+\epsilon}{\initialDifferenceMWU}\right) = O(u(v) \cdot \log m/ \epsilon)$. Therefore, for any vertex $v$, there are at most $O(\deg(v) \cdot \log(m) \zeta \beta \flowRouted/ \epsilon)$ such iterations.
\end{itemize}
Thus, combining the two cases, we can bound the number of iterations by $O(m\log(m) \zeta \beta \flowRouted/ \epsilon)$ and plugging in the values for $\zeta$ and $\flowRouted$ gives the result.
\end{proof}

We can now establish the running time stated in \Cref{thm:mainMinCost2}.

\begin{claim}\label{clm:totalRunningTime}
The total running time of \Cref{alg:boundedCostFlowexcess} can be bound by
\[
 \Otil(m \beta \cdot n^{10/\tau}/ \epsilon^2) + \timeAugSSSP{m}{n}{m^{6/\epsilon}}{\tau}{\epsilon}{\beta}{\Delta}{\Delta'}
\]
for $\Delta, \Delta' = \Otil(m  \beta \cdot n^{10/\tau}/ \epsilon^2)$.
\end{claim}

\begin{proof}
We start by observing that up to \Cref{lne:dataStructureInit}, the algorithm uses time $O(m)$. Henceforth, we do not account for the running time used by the data structure but rather only keep track of the number of updates $\Delta$ and the number of queries plus the size of the output of the query $\Delta'$.

When we enter the while-loop, we find the current approximate shortest path from $s$ to $t$ using the data structure and find the smallest steadiness class $\sensitivity_{\forLoopIndexMWU}(\shortestSTPath)$ that is non-empty. We note that we do not compute the path $\shortestSTPath$ explicitly but rather query 
\[
\sensitivity_{\leq 1}(\shortestSTPath), \sensitivity_{\leq 2}(\shortestSTPath), \dots, \sensitivity_{\leq \minCapIndexMWU}(\shortestSTPath)
\]
until we find the first class that is non-empty (and there always exists such a class). We then select a random threshold $\randomThreshold \geq 0$. 

We note that the foreach-loop starting in \Cref{lne:innermostForeachMWU} can then be implemented in time $|\sensitivity_{\leq \minCapIndexMWU + \randomThreshold}(\shortestSTPath)|$. Since steadiness classes are nesting, we have that $\sensitivity_{\leq \minCapIndexMWU}(\shortestSTPath) \subseteq \sensitivity_{\leq \minCapIndexMWU + \randomThreshold}(\shortestSTPath)$. Since every other operation in the while-loop iteration is a constant time operation, the overall running time for a single iteration of the while loop is at most $O(|\sensitivity_{\leq \minCapIndexMWU + \randomThreshold}(\shortestSTPath)| + \tau)$. (The additive $+ \tau$ comes from the fact that in Line \ref{lne:computeFlowCap} of \ref{alg:boundedCostFlowexcess}, the algorithm might go through at most $\tau$ steadiness values $\lambda$ before it find one with $\sensitivity_{\forLoopIndexMWU}(\shortestSTPath) \neq \emptyset$.) 

Using \Cref{clm:numberOfWhileLoopIterations}, we thus obtain that the total running time of the algorithm excluding the time spent by the data structure can be bound by
\[
 \Otil(m \tau\beta \cdot n^{10/\tau}/ \epsilon^2)
\]

We further observe that such a while-loop iteration adds at most $O(|\sensitivity_{\leq \minCapIndexMWU + \randomThreshold}(\shortestSTPath)| + \tau)$ to the query parameter $\Delta'$. Note that \Cref{clm:numberOfWhileLoopIterations} upper bounds the sum of $O(|\sensitivity_{\leq \minCapIndexMWU + \randomThreshold}(\shortestSTPath)|)$ over all foreach-loop iterations and thereby over all path-queries that return a non-empty set of edges. At the same time, \Cref{clm:numberOfWhileLoopIterations} is also a trivial bound on the number of while-loop iterations (since we always visit the foreach-loop in the while-loop at least once). Since each such while-loop iteration contributes at most $O(\tau)$ queries which return an empty set of edges, we can finally bound $\Delta'$ by $ \Otil(m \tau \beta \cdot n^{10/\tau}/ \epsilon^2))$.

We can now also bound $\Delta$, the number of updates to the weight function $\weightFunctionCombined(x) =  \weightFunctionEst(x) + \lceil \costFunctionEst  \rceil_{(1+\epsilon)}\cdot c(x)$ (as defined in \Cref{lne:dataStructureInit}). To this end, we observe that $\weightFunctionCombined(x)$ is updated either if $\weightFunctionEst(x)$ or if $\lceil \costFunctionEst  \rceil_{(1+\epsilon)}$ is increased. But the former updates can be upper bounded by $O(\Delta')$ since each such update results from a single edge in the query. For the number of updates caused by $\lceil \costFunctionEst  \rceil_{(1+\epsilon)}$, we observe that each increase of $\lceil \costFunctionEst  \rceil_{(1+\epsilon)}$ results in $m$ updates to $\weightFunctionCombined$. However, since we round $\costFunctionEst$ to powers of $(1+\epsilon)$, we can bound the total number of increases of $\lceil \costFunctionEst  \rceil_{(1+\epsilon)}$ by $O(\log_{1+\epsilon}(\initialDifferenceMWU)) = O(\log m/\epsilon)$. Combined, we obtain $\Delta' = \Otil(m \tau \beta \cdot n^{10/\tau}/ \epsilon^2)$.

For the claim, it remains to use the assumption that $\tau = O(\log n)$.
\end{proof}

\section{Near-capacity-fitted instance via Near-pseudo-optimal MBCF}
\label{sec:minCostFlowFinally}

We now build upon \Cref{thm:mainMinCost2} to obtain Near-capacity-fitted instances. We start by making the definition of such an instance formal. 

\begin{definition}[Edge-Split Transformation]\label{def:edgeSplitTrans}
Given a flow instance $G=(V,E,c,u)$, we let $G' = \textsc{Edge-Split}(G)$ denote the instance derived from $G$ by splitting every edge $e =(x,y)$ in $G$ into two edges $(x,v_e)$ and $(v_e,y)$ where $v_e$ is a new vertex added to $G'$. The capacity $u'(v_e)$ of each such $v_e \in V(G') \setminus V$ is set to $U$ (the max capacity of $G$), and its cost $c'(v_e)$ to $0$. For all $v \in V(G') \cap V$, we set $u'(v) = u(v)$ and $c'(v)=c(v)$.
\end{definition}

Here, we note that if $G'$ is derived from $G$ as proposed above, and $G$ was derived using \Cref{prop:reductionVertexCapacities}, then also $G'$ satisfies the properties in \Cref{prop:reductionVertexCapacities} as can be verified straight-forwardly (except that the number of edges and vertices increases by $m$).

\begin{definition}[Near-capacity-fitted instance]\label{def:capacityFittedInstance}
For any $0 < \epsilon < 1$, given graph $G = (V,E,u,c)$ and a cost budget $\overline{C}$. Let $G' = (V', E', u', c')$ be the graph defined by $G' = \textsc{Edge-Split}(G)$. Then, we say that a graph $G'' = (V',E',u'',c')$ is a $(1-\epsilon)$-\emph{capacity-fitted instance derived from $G$} if: 
\begin{enumerate}
    \item for every $v \in V$, $u''(v) \leq u(v)$, and
    \item we have for each $v \in V$, that $\sum_{x \in \neighborhood_{G''}(v)} u''(x) \leq 18 \cdot u(v)$, where $\neighborhood_{G''}(v) = \{x \in V' | (x,y) \in E'\}$ and \label{prop:capacityFittedPropLocal}
    \item we have $\sum_{x \in V''} u''(x) \cdot c'(x) \leq 18 \cdot \overline{C}$, and \label{prop:capacityFittedPropCost}
    \item $OPT_{G'', \overline{C}} \geq (1-\epsilon) \cdot OPT_{G, \overline{C}}$. \label{prop:capacityFittedPropNearOpt}
\end{enumerate}
\end{definition}

Intuitively, the first Property ensures that every flow in $G''$ is capacity-feasible in $G'$. At the same time Property \ref{prop:capacityFittedPropLocal} ensures that for every original vertex in $V$, the vertices in its neighborhood have capacity $\sim u(v)$. Recall that these neighbors in $G''$ are the vertices resulting from edge-splits of edges incident to $v$ in $G$. This property will later be helpful to argue not only about $\inflow_f(v)$ of some flow $f$ in $G$ but also about $\outflow_f(v)$ by using the guarantees of a $(1-\epsilon)$-pseudo-optimal flow on the neighborhood of $v$. Property \ref{prop:capacityFittedPropCost} ensures that any capacity-feasible flow $f$ in $G''$ will not have large cost (w.r.t $\overline{C}$). Thus, scaling such $f$ by $\epsilon$ will imply that it is cost-feasible even in $G'$. Finally, we ensure in Property \ref{prop:capacityFittedPropNearOpt} that $G''$ still contains a large valued feasible flow.

We can now formally state the main result of this section.

\begin{restatable}{lemma}{NearPseudoOptMBCFforKnownTarget}[Near-capacity-fitted instance via Near-pseudo-optimal MBCF] \label{lma:NearPseudoOptMBCFforKnownTarget}
Given any $0 < \epsilon < 64$, given a graph $G=(V,E,c,u)$, a dedicated source $s$ and sink $t$, a cost bound $\overline{C}$. Additionally, let there be an algorithm $\mathcal{A}$ that computes a $(1-\epsilon)$-pseudo-optimal flow $\hat{g}$ in total update time $\mathcal{T}_{PseudoMBCF}(m,n,\epsilon, \overline{C})$.

Then, there exists an algorithm $\mathcal{B}$ that computes a $(1-\epsilon)$-capacity-fitted instance $G''$ in time 
\[
\Otil(m + \mathcal{T}_{PseudoMBCF}(m,n,\Theta(\epsilon/\log n), \overline{C})) .
\]
\end{restatable}

\begin{algorithm}
\caption{$\textsc{NearFeasibleCostFeasibleMBCF}(G, s,t, \epsilon, \overline{C}, \overline{U})$}
\label{alg:boundedCostFlowWithKnownFlowValue}
\KwIn{A graph $G=(V,E,c,u)$, two dedicated vertices $s, t \in V$, a cost function $c$ and a capacity function $u$ (both mapping edges in $E$ to positive reals), an approximation parameter $0 < \epsilon < 1/64$, a cost budget $\overline{C} \in \mathbb{R}^+$ and a real $\overline{U} \in [OPT_{G, \overline{C}}/2, OPT_{G, \overline{C}}]$. Also, an algorithm $\mathcal{A}$ that computes a $(1-\epsilon)$-pseudo-optimal flow.}
$(V',E',c',u') \gets \textsc{Edge-Split}(G)$.\;
$j_{max} \gets \lfloor \log_{(2-16\epsilon)}(2m^{11}/\epsilon) \rfloor$.\;
$\epsilon' = \frac{\epsilon}{20 \cdot j_{max}}$\;
\lForEach{$v \in V$}{$u_0(v) \gets \min\{ u'(v), 2\overline{U}\}$.}
\For(\label{lne:forLoopKnownObjeMBCF}){$j = 0, 1, \dots, j_{max}$}{
    $g_j \gets \mathcal{A}((V', E' ,c' ,u_j), s,t, \epsilon', \overline{C})$.\;
    $\forall v \in V', \; u_{j+1}(v) \gets \begin{cases} u_{j}(v)/2 & \text{if } \inflow_{g_j}(v) \leq u_{j}(v) / 2\\ u_{j}(v) & \text{otherwise}\end{cases}$.\;
}
\BlankLine
\tcc{Return capacity-fitted instance.}
\Return $G'' = (V', E', c, u_{j_{max}+1})$ \label{lne:returnFlowForKnownObjective} 
\end{algorithm}

For simplicity, we assume for the rest of the section that we have a $\frac{1}{2}$-approximation $\overline{U}$ of the value of the optimal MBCF solution, i.e. $\overline{U} \in [OPT_{G, \overline{C}}/2, OPT_{G, \overline{C}}]$. This guess can later be removed by guessing values for $\overline{U}$ at the cost of a multiplicative $O(\log n)$ factor (recall that $U \leq m^5$ by \Cref{prop:reductionVertexCapacities}).

We present $\mathcal{B}$ in \Cref{alg:boundedCostFlowWithKnownFlowValue}. The main idea behind the algorithm is to apply a technique that we call \emph{capacity fitting}. Loosely speaking, we halve the capacity of every vertex for which the in-flow given by $\mathcal{A}$ is smaller-equal to half its capacity. Thus, every iteration, we roughly half the capacity of all vertices until the flow has to use a constant fraction of the capacity of each vertex.

We now prove simple claims which will then allow us to conclude \Cref{lma:NearPseudoOptMBCFforKnownTarget} straight-forwardly. 

We start by the most important claim, that right away shows that even filling all edges in the graph $G''$ with flow will not induce cost far beyond the cost budget $\overline{C}$.

\begin{claim}\label{clm:reduceTotalCapacityTimeCost}
For any $0 \leq j \leq j_{max}$,
\[
    \sum_{v \in V'} c'(v) \cdot u_j(v) \leq \frac{2m^{11}}{(2-16\epsilon)^j} \cdot 2 \overline{C} / \epsilon + 10 \overline{C}.
\]
In particular, we have $\sum_{v \in V'} c'(v) \cdot u_{j_{max}}(v) \leq 18 \overline{C}$.
\end{claim}
\begin{proof}
We prove the claim by induction. For the base case $j = 0$, we observe that every vertex $v \in V$ has cost $c'(v)$ at most $m^5$ and $u_0(v) = u'(v) \leq m^5$ (see \Cref{prop:reductionVertexCapacities}). Since there are only $m+n \leq 2m$ vertices in $G'$, we can therefore deduce $\sum_{v \in V'} c'(v) \cdot u_0(v) \leq 2m \cdot m^{10}$ and we finally use that $\overline{C} \geq 1$.

Let us now prove the inductive step $j \mapsto j + 1$ for $j \geq 0$: We observe that by the induction hypothesis, we have that:
\begin{align}\label{lne:eqUpperBoundUj}
 \sum_{e \in V'} c'(v) \cdot u_j(v) \leq \frac{2m^{11}}{(2-16\epsilon)^j} \cdot 2 \overline{C} / \epsilon  +  10 \overline{C}.
\end{align}
We recall that in the $j^{th}$ iteration of the for-loop starting in \Cref{lne:forLoopKnownObjeMBCF}, we invoke algorithm $\mathcal{A}$ to define the function $u_{j+1}$ based on the near-pseudo-optimal-flow $g_j$. We observe that by assumption on $\mathcal{A}$ and \Cref{def:feasibleFlow}, there is a near-optimal flow $\flow[j]$ such that $|\inflow_{g_j}(v) - \inflow_{\flow[j]}(v)| \leq \epsilon' \cdot u_j(v)$ for all vertices $v \in V'$, and $c'(\flow[j]) \leq \overline{C}$ for the given instance.

But this implies that 
\begin{equation}
\label{eq:cost-gprime}
c'(g_j) \leq \epsilon \cdot \sum_{v \in V'} c'(v) \cdot u_j(v) + \overline{C} \leq \frac{2m^{11}}{(2-16\epsilon)^j} \cdot 2 \overline{C} + (1+ 10 \epsilon) \overline{C}.
\end{equation}
To avoid clutter, we define $T = \frac{2m^{11}}{(2-16\epsilon)^j} \cdot 2 \overline{C} +  (1+ 10 \epsilon) \overline{C}$ for further use. We then note that the capacity $u_{j+1}(v)$ of every vertex $v$ becomes $u_j(v) /2$ if the vertex has inflow $\inflow_{g_j}(v)$ less than half of its capacity $u_j(v)$. However, by the upper bound on the cost of $g_j$, we have that the vertices that contain greater-equal to half of their capacity in flow satisfy
\[
     \sum_{v \in V, \; \inflow_{g_j}(v) \geq u_j(v)/2} c'(v) \cdot u_j(v) \leq 2 T.
\]
This inequality follows from the combination of two facts. The first is that the LHS of the inequality is at most $2c(g')$, because the LHS only considers vertices $v \in V$ through which $g_j$ sends at least $u_j(v)/2$ flow. The second fact is that $c(g') \leq T$, as shown in Equation \ref{eq:cost-gprime}.

Since for the rest of the vertices, the capacity is halved in $u_{j+1}$, we have
\begin{align}
    \sum_{v \in V'} c'(v) \cdot u_{j+1}(v)
    &\leq 2T + \frac{\sum_{v \in V'} c'(v) \cdot u_{j}(v)}{2}\label{lne:eqBeforePuttingInRealValues} \\
    &\leq \frac{4m^{11}}{(2-16\epsilon)^j} \cdot 2 \overline{C} + (2+20\epsilon)\overline{C} +  \frac{2m^{11}}{(2-16\epsilon)^j \cdot 2} \cdot 2 \overline{C} / \epsilon + 5 \overline{C} \label{lne:eqPuttingInRealValues}\\
    &= (1 + 4\epsilon) \cdot  \frac{2m^{11}}{(2-16\epsilon)^j \cdot 2} \cdot 2 \overline{C} / \epsilon + (7 + 20\epsilon)\overline{C}\label{lne:eqBeforeUsingeXForError}\\
    &\leq \frac{2m^{11}}{(2-16\epsilon)^{j} \cdot 2 \cdot (1 - 8\epsilon)} \cdot 2 \overline{C} / \epsilon + (7 + 20\epsilon)\overline{C}\label{lne:eqUsingeXForError}\\
    &< \frac{2m^{11}}{(2-16\epsilon)^{j+1}} \cdot 2 \overline{C} / \epsilon + 10 \overline{C}
\end{align}
where we use \Cref{lne:eqUpperBoundUj} and the definition of $T$ to get $($\ref{lne:eqBeforePuttingInRealValues}$) \implies ($\ref{lne:eqPuttingInRealValues}$)$, then rearrange terms and use $1+x \leq e^x \leq 1+x+x^2$ for $x \leq 1$ to obtain $($\ref{lne:eqBeforeUsingeXForError}$) \implies ($\ref{lne:eqUsingeXForError}$)$. In the final inequality, we use our assumption $\epsilon < 1/32$.
\end{proof}

Using the same proof template, it is not hard to establish the following claim whose proof is deferred to  \Cref{sec:actualProofOfClmreduceLocalCapacityForPlainVertex}.

\begin{restatable}{claim}{clmReduceLocalCapacityForPlainVertex}
\label{clm:reduceLocalCapacityForPlainVertex}
For every vertex $v \in V$, any $0 \leq j \leq j_{max}$,
\[
    \sum_{x \in \neighborhood_{G'}(v)} u_j(x) \leq \frac{m}{(2-16\epsilon)^j} \cdot 2\overline{U} / \epsilon + 10 u(v).
\]
In particular, we have $\sum_{x \in \neighborhood_{G'}(v)} u_{j_{max}}(x) \leq 18 \cdot u(v)$.
\end{restatable}

Finally, we have to argue that we can route a near-optimal flow in $G'$ in the final instance $G''$. Since the claim below is straight-forward to obtain but tedious to derive, we defer its proofs to \Cref{sec:proofOfOptIsNotReducedByMuch}.

\begin{restatable}{claim}{clmOptNotReducedByMuch}\label{clm:OptIsNotReducedByMuch}
Define $G_j = (V',E',c' ,u_j)$ to be the graph that $\mathcal{A}$ is invoked upon during the $j^{th}$ iteration of the for-loop starting in \Cref{lne:forLoopKnownObjeMBCF}. Then, we have that for every $j \geq 0$, 
\[
    OPT_{G_{j+1}, \overline{C}} \geq \left(1- 10\epsilon' \right) OPT_{G_{j}, \overline{C}}.
\]
In particular, we have $OPT_{G_{j_{max}}, \overline{C}} \geq \left(1- 10\epsilon' \right)^{j_{max}} OPT_{G_0, \overline{C}} \geq (1-\epsilon)OPT_{G_0, \overline{C}}$.
\end{restatable}

We can now prove \Cref{lma:NearPseudoOptMBCFforKnownTarget}.

\NearPseudoOptMBCFforKnownTarget*
\begin{proof}
Let us first argue about correctness by establishing the properties claimed in \Cref{def:capacityFittedInstance}. It is immediate to see that $u''(v) = u_{j_{max}+1}(v) \leq u'(v)$ since our algorithm only decreases capacities. By \Cref{clm:reduceTotalCapacityTimeCost}, we also have the second property of a near-capacity-fitted instance satisfied, and by \Cref{clm:reduceLocalCapacityForPlainVertex} the third property. Finally, observe that by \Cref{clm:OptIsNotReducedByMuch}, we immediately obtain that $OPT_{G'', \overline{C}} \geq (1-\epsilon) OPT_{G', \overline{C}} = (1-\epsilon) OPT_{G, \overline{C}}$ (where we use that $OPT_{G', \overline{C}} = OPT_{G_0, \overline{C}}$ since we they only differ in capping the capacities at $2\overline{U}$ which does not affect the maximum value of any flow by definition of $\overline{U}$). It remains to bounds the running time of \Cref{alg:boundedCostFlowWithKnownFlowValue} which can be seen by straight-forward inspection of the algorithm to be $O(m\log n)$ plus $O(\log n)$ invocations of $\mathcal{A}$ (here, we also assume that $\epsilon > 1/n$). 
\end{proof} 

\section{Near-Optimal MBCF via Near-pseudo-optimal MBCF in a Near-capacity-fitted instance}
\label{sec:nearOptFromNearLocally}

Finally, we show how to obtain a near-optimal flow from a near-pseudo-optimal flow in a near-capacity-fitted instance. 

\begin{restatable}{theorem}{MBCF}\label{thm:MBCF}
For any $\epsilon > 0$, given a graph $G=(V,E,c,u)$ that satisfies the properties of \Cref{prop:reductionVertexCapacities}, a dedicated source $s$ and sink $t$ and a cost budget $\overline{C}$. Given an algorithm $\mathcal{A}$ that computes a $(1-\epsilon)$-pseudo-optimal flow in time $\mathcal{T}_{PseudoMBCF}(m,n,\epsilon,\overline{C})$ and given an algorithm $\mathcal{B}$ that computes a $(1-\epsilon)$-capacity-fitted instance $G'$ for the given flow problem in time $\mathcal{T}_{CapacityFitting}(m,n,\epsilon,\overline{C})$.

Then, there exists an algorithm $\mathcal{C}$ that computes a $(1-\epsilon)$-optimal flow $f$ in time
\[
\tilde{O}(m) + \mathcal{T}_{PseudoMBCF}(m,n,\Theta(\epsilon),\overline{C}) + \mathcal{T}_{CapacityFitting}(m,n,\Theta(\epsilon),\overline{C}).
\]
with high probability.
\end{restatable}

\paragraph{A Near-pseudo-optimal Flow in a capacity-fitted instance.} We start by invoking algorithm $\mathcal{B}$ on $G, s , t, \overline{C}$ and $\epsilon$, which returns a $(1-\epsilon)$-capacity-fitted instance $G''$. We then invoke algorithm $\mathcal{A}$ on $G''$,$s$,$t$,$\overline{C}$ and $\epsilon$ to obtain a $(1-\epsilon)$-pseudo-optimal flow $\hat{g}$. Let $g$ be the near-optimal flow that proves $\hat{g}$ to be $(1-\epsilon)$-pseudo-optimal. We assume w.l.o.g. that in $\hat{g}$, flow is only either on $(x,y)$ or $(y,x)$ for any such pair of edges in $E''$ (here we just use flow cancellations). 

\paragraph{Mapping the Flow Back to $G$.} Next, let us map the flow $\hat{g}$ back to $G$. We can firstly just apply the identity map to obtain $\hat{g}$ in $G'$. We observe that if $\hat{g}$ would satisfy flow conservation constraints in $G'$ (even only in the vertices $V' \setminus V$), then we could that the inverse of the transformation described in \Cref{def:edgeSplitTrans} to obtain $G'$ from $G$, and use it to map the flow on edges $(x,v_e),(v_e,y)$ (where $x,y \in V$ but $v_e \in V' \setminus V$) back to $(x,y)$. 

But observe that if there is positive excess at a vertex $v_e \in V' \setminus V$, where again $v_e$ is the vertex associated with edge $e = (x,y) \in E$, we can just route that excess back to $x$ and $y$ since the edges $(x, v_e)$ and $(y, v_e)$ carry all the in-flow to $v_e$ (let us assume for convention that the flow is first routed back to $x$ and then to $y$ if excess is still at $v_e$). Since this monotonically decreases flow on every edge (and thereby the in-flow to every vertex), it is easy to see that the resulting flow still satisfies capacity- and cost-feasibility constraints.

Further, the resulting flow can now be mapped straight-forwardly to $G$. We denote this flow on $G$ by $\flowEst$ and again assume w.l.o.g. that $\flowEst$ has flow either on edge $(x,y)$ or on edge $(y,x)$ but not both. 

\paragraph{Routing the Remaining Excess in $G$.} We now want to route the remaining excess in $G$. However, we first need to know the flow value from $s$ to $t$ that we want to route in $G$. We therefore simply check the in-flow at $t$, and let $F = \inflow_{\flowEst}(t)- \outflow_{\flowEst}(v)$. Next, we compute the excess vector $\textrm{ex}_{\flowEst,s,t,F}$ which is defined
\[
    \textrm{ex}_{\flowEst,s,t,F}(v) = \begin{cases}
    \inflow_{\flowEst}(v) - \outflow_{\flowEst}(v) & \text{if } v \neq s, v \neq t\\
    \inflow_{\flowEst}(v) - \outflow_{\flowEst}(v) + F & \text{if } v = s\\
    0 & \text{if } v = t
    \end{cases}
\]

Next, we want to construct a flow problem where we route a general demand $\chi \in \mathbb{R}^{V}$ (where $\sum_v \chi(v) = 0$). More precisely, we note that the vector  $\textrm{ex}_{\flowEst,s,t,F}$ is a valid demand vector. We then set up the graph to be $G''' = (V,E, u''')$, where $u'''(v) = \infty$ for any $v \in V$, and $u'''(e) = \epsilon \cdot u''(v_e)$ for each $e \in E$ where $v_e$ is again the vertex in $V(G'') \setminus V$ that is associated with edge $e$. We do not define a cost function and observe that the created instance only has edge capacities by design.

\paragraph{Feasibility Of Excess Routing.} We note that we can indeed route $\textrm{ex}_{\flowEst, \chi_{s,t,F}}$ in $G$ capacitated by $u'''$. To see this, recall that $g$ is the $(1-\epsilon)$-optimal flow certifying that $\hat{g}$ is $(1-\epsilon)$-pseudo-optimal in $G''$. Let $\flow$ be the flow on $G$ obtained by mapping $g$ to $G$ just like we mapped $\hat{g}$. Then, it is not hard to see that $\flow - \flowEst$ routes $\textrm{ex}_{\flowEst, s,t,F}$. Since each edge $e$ has $|\flow(e) - \flowEst(e)| < \epsilon \cdot u''(v_e)$ by \Cref{def:feasibleFlow}, our claim follows.

\paragraph{Using Max-Flow for Excess Routing.} We then use the following result for max flow on edge-capacitated graphs from \cite{sherman2013nearly, peng2016approximate} on $G''' = (V,E,u''')$.

\begin{theorem}[see Theorem 1.2 in \cite{sherman2013nearly}, \cite{peng2016approximate}]\label{thm:maxFlowSherman}
Given a flow instance $G''' = (V,E,u''')$ and a demand vector $\chi$ (with $\sum_v \chi(v) = 0$). Then, there exists an algorithm that returns a flow $f'''$ that obeys flow conservation constraints, and satisfies for each edge $e \in E$, $f'''(e) \leq 2 \cdot u'''(e)$. For a graph $G'''$ with polynomially bounded capacity ratio, the algorithm runs in time $\tilde{O}(m)$ and succeeds with high probability.
\end{theorem}

We denote by  $\flow = (1-80\epsilon) \left(\flowEst + \flow'''\right)$ the flow obtained by combined the flow mapped from the capacity-fitted instance and the max flow instance (after some careful scaling).

\paragraph{Feasibility of $\flow$.} From construction, it is not straight-forward to see that $\flow$ satisfies flow conservation. Further, we have for every vertex $v \in V$,
\begin{align*}
\inflow_{\flow}(v) &\leq  (1-80\epsilon)\left( \inflow_{\flowEst}(v) + \inflow_{\flow'''}(v) \right)\leq (1-80\epsilon) \left(u(v) + 2 \sum_{x \in \mathcal{N}(v)} u'''(x,v) \right)\\
&\leq (1-80\epsilon) (1+28\epsilon) u(v) \leq u(v)
\end{align*}
where we used in the second inequality the feasibility of $\flowEst$, the guarantee from \Cref{thm:maxFlowSherman} on $f'''$ to almost stipulate capacities, and Property \ref{prop:capacityFittedPropLocal} from \Cref{def:capacityFittedInstance} which implies that all edge capacities incident to $v$ sum to at most $14 \epsilon \cdot u(v)$ which is a trivial upper bound on the amount of flow routed through $v$ in $\flow'''$.

We further have that 
\[
    c(f) \leq (1-80\epsilon) \left( c(f) + c(f''') \right) \leq (1-80\epsilon) \left( \overline{C} + 28 \cdot \epsilon \overline{C} \right) \leq \overline{C}
\]
where we use in the second inequality that by Property \ref{prop:capacityFittedPropCost} of \Cref{def:capacityFittedInstance}, the sum of capacities times costs in $G''$ is bounded by $14 \epsilon \overline{C}$ and the fact that $f'''$ satisfies $f'''(e) \leq 2 u'''(e)$ by \Cref{thm:maxFlowSherman}. Combined these facts prove that $\flow$ is a feasible flow in $G$.

\paragraph{Near-Optimality.} It remains to conclude that since $OPT_{G'', \overline{C}} \geq (1-\epsilon) OPT_{G, \overline{C}}$ by Property \ref{prop:capacityFittedPropNearOpt} in \Cref{def:capacityFittedInstance} and $\hat{g}$ is $(1-\epsilon)$-pseudo-optimal in $G''$ that the pseudo-flow $\flowEst$ is $(1-\epsilon)^2$-pseudo-optimal in $G$. Thus, we have that $v(\flow) \geq (1-80\epsilon)v(\flowEst) \geq (1-80\epsilon)(1-\epsilon)^2 OPT_{G, \overline{C}}$. Rescaling $\epsilon$ by a constant factor, we obtain that $\flow$ must be a $(1-\epsilon)$-optimal flow. This concludes our analysis.

\section{Putting it all Together}
\label{sec:puttingItAllTogether}

Finally, we combine our reduction chain with the main result of \Cref{part:augmented-queries}: \Cref{thm:main SSSP path}. However, instead of using the main result of \Cref{part:augmented-queries} directly, we rather prove that it can be used straight-forwardly to implement the data structure given in \Cref{def:Path-reportingSSSP}.

\begin{restatable}{theorem}{implementationFlowSSSPviaAugSSSP}\label{thm:implementationFlowSSSPviaAugSSSP}
There exists an implementation of the data structure given in \Cref{def:Path-reportingSSSP} where for any $\tau=o(\log^{3/4} n),$ $\epsilon > 1/\polylog(n)$ and some $\beta = \Ohat(1)$, the data structure can be implemented with total running time $T_{SSSP^{\pi}}(m,n,W, \tau,\epsilon, \Delta, \Delta') = \Ohat(m \log W + \Delta + \Delta')$.
\end{restatable}
\begin{proof}[Proof Sketch.]
The proof is almost immediate from \Cref{thm:main SSSP path}, except that the data structure in \Cref{thm:main SSSP path} deals with edge weights, while we require vertex weights $w$. 
In \cite{ChuzhoyS20_apsp} a simple transformation was described by defining edge weights for each edge $(x,y)$ to be $\frac{w(x) + w(y)}{2}$. Then for any path $P$ from $s$ to $t$ where $w(s)=w(t)=0$, the weight of $P$ with regard to this edge weight function is equal to the weight in the vertex-weighted graph. Unfortunately, we cannot assume $w(s) = w(t) = 0$, but using the same idea we can create a small workaround that is presented in the \Cref{subsec:implementationFlowSSSPviaAugSSSP}.
\end{proof}

Let $\epsilon > 1/\polylog(n)$. Then, plugging in \Cref{thm:implementationFlowSSSPviaAugSSSP} into \Cref{thm:mainMinCost2}, we obtain procedure to find a $(1-\epsilon)$-pseudo-optimal flow $f$ in $G$ in total time $\Ohat(m)$ with probability at least $1-n^{-10}$.

Using this result in \Cref{lma:NearPseudoOptMBCFforKnownTarget}, we again obtain total time $\Ohat(m)$ to produce a corresponding $(1-\epsilon)$-capacity-fitted instance $G''$ with probability at least $1-n^{-9}$ (there are $O(\log^2 n)$ invocations of the algorithm in \Cref{thm:mainMinCost2}, and we can take a union bound over the failure probability and assume that $n$ is larger than some fixed constant). 

Finally, we use the reduction in \Cref{thm:MBCF} with the above running times to obtain a near-pseudo-optimal flow and a capacity-fitted instance, and obtain a $(1-\epsilon)$-optimal flow $f$ in the original graph $G$, again in total time $\Ohat(m)$ and with success probability at least $1-n^{-8}$.

To obtain a proof for our main result, \Cref{thm:MBCFFinal}, we point out that we assumed in the chain of reduction above that $G$ was derived from applying the reduction in \Cref{prop:reductionVertexCapacities}. However, since our dependency of run-time is purely in terms of $m$ (and not in $n$) this does not lead to an asymptotic blow-up. The proof therefore follows immediately.

%% file: related_work.tex
\section[Appendix of Part I]{Appendix of \Cref{part:intro}}

\subsection{Related Work}
\label{sec:relatedWork}

In addition to our discussion of previous work in \Cref{subsec:previousWork}, we also give a brief overview of related work. 

\paragraph{Dynamic SSR and SSSP in Directed Graphs.} While our article focuses on the decremental SSSP problem in undirected graphs, there is also a rich literature for dynamic SSSP in directed graphs and also for the simpler problem of single-source reachability and the related problem of maintaining strongly-connected components.

For fully-dynamic SSR/ SCC, a lower bound by Abboud and Vassilevska Williams \cite{abboud2014popular} shows that one can essentially not hope for faster amortized update time than $\Otil(m)$.

For decremental SSR/ SCC, a long line of research \cite{roditty2008improved, lkacki2013improved, henzinger2014sublinear, henzinger2015improved, chechik2016decremental, ItalianoKLS17} has recently lead to the first near-linear time algorithm \cite{bernstein2019decremental}. A recent result by Bernstein, Probst Gutenberg and Saranurak has further improved upon the classic $O(mn)$ total update time barrier to $\Ohat(mn^{2/3})$ in the deterministic setting \cite{BernsteinGS20scc}. 

While incremental SSR can be solved straight-forwardly by using a cut-link tree, the incremental SCC problem is not very well-understood. The currently best algorithms \cite{Haeupler12, Bender15} obtain total update time $\Otil(\min\{m^{1/2}, n^2\})$. Further improvements to time $\Otil(\min\{m\sqrt{n}, m^{4/3}\})$ for sparse graphs are possible for the problem of finding the first cycle in the graph \cite{BernsteinC18, bhattacharya2020improved}, the so-called cycle detection problem.

For fully-dynamic SSSP, algebraic techniques are known to lead to algorithms beyond the $\hat{\Theta}(m)$ amortized update time barrier at the cost of large query times. Sankowski was the first to give such an algorithm \cite{sankowski2005subquadratic} which originally only supported distance queries, however, was recently extended to also support path queries \cite{bergamaschi2020new}. An algorithm that further improves upon the update time/query time trade-off at the cost of an $(1+\epsilon)$-approximation was given by van den Brand and Nanongkai in \cite{brand2019dynamic}.

The decremental SSSP problem has also received ample attention in directed graphs \cite{EvenS, henzinger2014sublinear, henzinger2015improved, GutenbergW20a, bernstein2020near}. The currently best total update time for $(1+\epsilon)$-approximate decremental SSSP is $\Otil(\min\{n^2, mn^{2/3}\}\log W)$ as given in \cite{bernstein2020near}. Further,  \cite{BernsteinGS20scc} can be extended to obtain a deterministic $\Ohat(n^{2+2/3}\log W)$ total update time algorithm. 

The incremental SSSP problem has also been considered by Probst Gutenberg, Wein and Vassilevska Williams in \cite{GutenbergWW20} where they propose a $\Otil(n^2 \log W)$ total update time algorithm.

\paragraph{Dynamic APSP.} There is also an extensive literature for the dynamic all-pairs shortest paths problems. 

In the fully-dynamic setting a whole range of algorithms is known for different approximation guarantees, and for the particular setting of obtaining worst-case update times \cite{henzinger1995fully, King99, demetrescu2001fully, demetrescu2004new, roditty2004dynamic, thorup2005worst, bernstein2009fully, roditty2012dynamic,  abraham2014fully, roditty2016fully, henzinger2016dynamic, abraham2017fully, brand2019dynamic, probstWulffNilsenwcAPSP}. Most relevant to our work is a randomized $\Ohat(m)$ amortized update time algorithm by Bernstein \cite{bernstein2009fully} that obtains a $(2+\epsilon)$-approximation. An algorithm with faster update time is currently only known for very large constant approximation \cite{abraham2014fully}.

Similarly, in the decremental setting there has been considerable effort to obtain fast algorithms \cite{baswana2007improved,bernstein2011improved, abraham2013dynamic, henzinger2014decremental, henzinger2016dynamic,bernstein2016maintaining, chechik2018near, gutenberg2020deterministic, ChuzhoyS20_apsp, karczmarz2020simple, EvaldFGW20}. We explicitly highlight two contributions for undirected graphs: in \cite{henzinger2016dynamic}, the authors obtain a $O(mn\log n)$ deterministic $(1+\epsilon)$-approximate APSP algorithm (a simpler proof of which can be found in \cite{gutenberg2020deterministic}) and in \cite{chechik2018near} an algorithm is presented that for any positive integer $k$ maintains $(1+\epsilon)(2k-1)$-approximate decremental APSP in time $\Ohat(mn^{1/k} \polylog W)$. 

The incremental APSP problem has also recently been studied \cite{karczmarz2019reliable}.

\paragraph{Hopsets.} We also give a brief introduction to the literature on hopsets. Originally, hopsets were defined and used in the parallel setting in seminal work by Cohen \cite{cohen2000polylog}. However, due to their fundamental role in both the parallel and the dynamic graph setting, hopsets have remained an active area of development. Following lower bounds on the existential guarantees of hopsets \cite{AbboudBP17}, first Elkin and Neiman \cite{elkin2019hopsets} and then Huang and Pettie \cite{huang2019thorup} obtained almost optimal hopset constructions, where the latter was based on a small modification to the classic Thorup-Zwick emulators/ hopset \cite{thorup2006spanners}.

\subsection{Alternative Statement of Min-Cost Flow Result}
\label{subsec:alternativeStatement}

We can also derive the following theorem straight-forwardly from a standard reduction that applies \Cref{thm:MainMinCost} a polylogarithmic number of time (essentially, once can apply \Cref{thm:MainMinCost} recursively for $O(\log nC/\epsilon))$ times and then use a max flow algorithm to route the tiny amount of remaining demand cheaply).

\begin{theorem}\label{thm:minCostFlowConvenientForPNorm}
For any $\epsilon > 1/\polylog(n)$, consider undirected graph $G=(V,E,c,u)$ where cost function $c$ and capacity function $u$ map each edge and vertex to a non-negative real. Let $\chi \in \mathbb{R}^n$ be a demand vector. Then, there is an algorithm that in $m^{1+o(1)}\log \log C$ time returns a feasible flow $f$ that routes $\chi$ (i.e. $B^{\top} f = \chi$ where $B$ is the associated (unweighted) incidence matrix of $G$). Let $f^*$ be the feasible flow with $B^{\top} f=\chi$ such that 
\[
c(f) = \sum_{e \in E} c(e) \cdot |f_e| + \sum_{v \in V} c(v) \cdot (B^{\top} |f|)_v
\]
is minimized. Then, we can compute a flow $f$ that is feasible and satisfies $\|B^{\top} f - \chi \|_1 \leq \epsilon \cdot \|\chi\|_1$ and $c(f) \leq (1+\epsilon)c(f^*)$ in time $m^{1+o(1)}$. The algorithm runs correctly with high probability.
\end{theorem}

\subsection{Discussion of Applications}
\label{appendix:ApplicationDiscussion}

We now expand on the discussion in \Cref{subsec:applications} and add explicit reductions or pointers to papers where they are stated clearly. We discuss the applications in the same order as in \Cref{subsec:applications}. 

\paragraph{Applications of Mixed-Capacitated Min-Cost Flow.} 
\begin{itemize}
	\item A $O(\log^2 n)$-sparsest vertex cut algorithm in almost-linear time: a through explanation of the reduction to $O(\polylog(n))$ vertex-capacitated flows and presented in Lemma D.4 in the ArXiv version of \cite{ChuzhoyS20_apsp}. Their reduction in turn is based by making a rather straight-forward observation about the sparsest cut algorithm in \cite{KhandekarRV09}. 
	\item A $O(\log^3(n))$-approximate algorithm for computing tree-width (and the corresponding tree decomposition) in $\Ohat(m)$ time: a formal reduction statement is again given in the ArXiv version of \cite{ChuzhoyS20_apsp} in their Lemma D.6. They basically use straight-forwardly the result in  \cite{bodlaender1995approximating} which reduces the problem to finding sparsest vertex cuts. 
	\item A high-accuracy LP solver by Dong, Lee and Ye \cite{dong2020nearly} with running time to $\Ohat(m \cdot \textrm{tw}(G_A)^2 \log(1/\epsilon))$: the result is immediate from Theorem 1.1. in \cite{dong2020nearly} and our almost-linear time tree decomposition algorithm.
	\item We provide an informal proof of the result below in \Cref{subsec:proofOfIncreasingCostFunction}.

\begin{theorem}\label{thm:formalGenPFlow}
Given any graph $G=(V,E)$ with incidence matrix $B$, demand vector $\chi \in \mathbb{R}^n$, and differentiable cost functions $c_e,c_v : \mathbb{R}_{\geq 0} \rightarrow \mathbb{R}_{\geq 0}$ growing (super-)linearly in their input for each $e \in E$ and $v \in V$ and each $c'_i(x)$ for $i \in E \cup V, x \in \mathbb{R}_{\geq 0}$ can be computed in $\Ohat(1)$ time. Let $f^*$ be some flow minimizing 
\[
	\min_{B^\top f = \chi} c(f) = \sum_{e\in E} c_e(|f_e|) + \sum_{v \in V} c_v((B^\top |f|)_v).
\]
Then, given the above, and $\epsilon > 1/\polylog(n)$, letting $C = \frac{c(f^*)}{\min_{i \in E \cup V} c_i'(0)}$, there is an algorithm that in $m^{1+o(1)} \polylog(C)$ time, returns a flow $f$ with $\|B^\top f - \chi\|_1 \leq \epsilon \cdot \|\chi \|_1$ such that $c(f) \leq (1+\epsilon) c(f^*)$ with high probability.
\end{theorem}

Previous results for $p$-norm flow concentrated on solving the problem to high-accuracy (i.e. $O(\polylog(1/\epsilon)$ dependence on $\epsilon$) but cannot handle weights \cite{adil2019iterative,KyngPSW19}.
\end{itemize}

\paragraph{Applications of Decremental SSSP.} 
\begin{itemize}
    \item  Decremental $(1+\epsilon)$-approximate all-pairs shortest paths (APSP) in total update time $\Ohat(mn)$: the algorithm is immediate from running the $(1+\epsilon)$-approximate SSSP data structure from \Cref{thm:mainSSSPResult} from every vertex $v \in V$. On query for a distance from $u$ to $v$, one can then just query the SSSP data structure at $v$ in constant time.
     \item Decremental $\Ohat(1)$-approximate APSP with total update time $\Ohat(m)$: to this end, we use \Cref{lem:main induction} which implies that we can maintain a covering of vertices, such that for each $D$, the diameter of each core is smaller $Dn^{o(1)}$, and each core has all vertices in its SSSP ball data structure that are at distance at most $Dn^{o(1)}$ from some vertex in the core. Maintaining such a covering for every $D_i= 2^i$, for every two vertices $u,v$ at distance $D_{i-1} \leq \dist_G(u,v) \leq D_{i}$, we can locate the correct covering by testing all values of $i$, and then find $u,v$ either in the same core, or in the SSSP ball data structure located at the core of the other vertex's core.
    \item Fully-dynamic $(2+\epsilon)$ approximate all-pairs shortest paths with $\Ohat(m)$ update time: there is essentially a reduction from fully-dynamic $(2+\epsilon)$ approximate APSP to decremental $(1+\epsilon)$-approximate SSSP in Section 3.3. of \cite{bernstein2009fully}. 
\end{itemize}

\subsection{Proof of Near-Optimal Flow for Flow under any Increasing Cost Function}
\label{subsec:proofOfIncreasingCostFunction}
We now give an informal proof of the above theorem. For convenience, we assume that $\chi$ characterizes an $s$-$t$ flow of value $F \geq 1$ (we can then use a standard reduction to recover full generality of the demand vector). We also assume without loss of generality that all $c_e$ map to $0$, by using the edge splitting procedure described in \Cref{part:minCostFlow} which increases the number of vertices and edges to $O(m)$. Let us also assume that we can roughly compute $C$ (to a two approximation via binary search).

Now, given this instance, let us for each vertex $v$ discretize $c_v(x)$ by finding values $x_0, x_1, \dots, x_k$ such that all $c'_v(x)$ for $x \in [x_i, x_{i+1})$ have $c'_v(x) \in \frac{C}{\poly(m)} [(1+\epsilon)^i, (1+\epsilon)^{i+1})$. Observe that this bounds $k = O(\polylog(n))$. We note that we might not be able just from querying the function $c'_v$ to find the precise values $x_i$ but we can find $\hat{x}_i$ with $\hat{x}_i = x_i - O(\frac{1}{\poly(m)})$ by employing binary search. Since these differences are tiny, i.e. only a negligible amount of flow is mischaracterized in rounded cost, we will ignore this issue altogether.

Finally, we create a min-cost flow instance from $G$, by adding for each vertex $v \in V$,  $k = O(\polylog(n))$ copies $v_1, v_2, \dots, v_k$ to the min-cost flow instance, where each vertex $v_i$ is assigned cost $\frac{C}{\poly(m)} \cdot (1+\epsilon)^{i+1}$ and capacity $x_{i+1} - x_i$. Further, for every adjacent vertices $v,w$ in the original graph $G$, we add edges between their copies $v_i, w_j$ $\forall i,j$ of cost $0$ and infinite capacity. It is not hard to see that the resulting instance has $\Tilde{O}(m \log C)$ edges and maximum capacity $C$.

Then invoke \Cref{thm:MainMinCost} on the created instance. A proof of correctness is straight-forward.

%% file: appendix_for_overview.tex
\subsection{Proof of Lemma \ref{lem:overview-kappa}}
\label{sec:overview-proof-kappa}

In this section we prove Lemma \ref{lem:overview-kappa}, which was stated in the overview, but never explicitly proved in the main body of the paper. See Section \ref{subsec:overview-capacitated} for the lemma statement and relevant notation.

The algorithm to compute the function $\kappa$ in Lemma \ref{lem:overview-kappa} is given in the pseudocode for Algorithm \ref{alg:overview-kappa} below. The algorithm follows the basic framework of congestion balancing and is a highly simplified version of the while loop in Line \ref{enu:embedwitness} of Algorithm \ref{alg:Core}. Recall that $\diam_G(K) \defeq \max_{x,y \in K} \dist_G(x,y)$.

\begin{algorithm}
	\caption{Algorithm to construct function $\kappa$ in Lemma \ref{lem:overview-kappa}}
	\label{alg:overview-kappa}
	\KwIn{An undirected graph $G = (V,E)$ and a set $K \subseteq V$}
	\KwOut{A capacity function $\kappa: V \rightarrow \mathbb{R}_{\geq 0}$ such that $(K,\kappa)$ forms a capacitated vertex expander in $G$ and $\sum_{v \in V} \kappa(v) = \Otil(|K|\diam_G(k))$.}
	Initialize $\kappa(v) \leftarrow 1$ for all $v \in V$ \;
	\While(\label{lne:overview_kappa_while}){There exists a sparse cut $(L,S,R)$ with respect to $K,\kappa$}{
		\ForEach{$v \in S$}{
			$\kappa(v) \leftarrow 2\kappa(v)$}
	}
	\BlankLine
	\Return $\kappa$
\end{algorithm}

\paragraph{Correctness Analysis:}
We now argue that the function $\kappa$ returned by the algorithm satisfies the output guarantees. When the algorithm terminates, $(K,\kappa)$ trivially forms a vertex expander in $G$, since the while loop only terminates when no sparse cuts remains. The bulk of the proof is showing that $\sum_{v \in V} \kappa(v) = \Otil(|K|\diam_G(k))$.

We define a potential function similar (but simpler) to the one in Definition \ref{def:Pi general}.

\begin{defn}
	We define the potential function $\Pi(G,K,\kappa)$ as follows. Let $\mathbb{P}$ be a collection of all embeddings $\cP$ where $\cP$ embeds some graph $\Wstar$ into $G$ such that
	\begin{enumerate}
		\item $\Wstar$ is an unweighted star with $V(\Wstar) = K$
	\end{enumerate}
	Define the cost of each vertex $v \in V$ to be $\log(\kappa(v))$. For any path $P$ in $G$ let $c(P) = \sum_{v \in P} c(v)$. The cost of an embedding $\cP$ is  $c(\cP)=\sum_{P\in\cP}c(P)$. We define $\Pi(G,K,\kappa)=\min_{\cP\in\mathbb{P}}c(\cP)$,
	and we call the corresponding $\cP$ the \textbf{minimum cost embedding
	} into $G$. 
\end{defn}

The proof now follows the general framework of Section \ref{subsec:analysisOfRobustCorePart2}. It is not hard to check that we always have $\kappa(v) \leq n \ \forall v \in V$, since once $\kappa(v) \geq n/2$ it will never increase again, because any vertex cut $(L,S,R)$ for which $v \in S$ is by definition not a sparse cut. We thus have $c(v) \leq \log(n) \ \forall v \in V$. This in turn implies that at all times $$\Pi(G,K,\kappa) \leq |K|\cdot \diam_G(K)\cdot \log(n).$$ 
To see this, consider the star formed by picking an arbitrary vertex $v \in K$, and then letting the embedding $\cP$ contain the shortest path in $G$ from $v$ to $x$ for every $x \in K$. (Note that this path may include vertices in $G \setminus K$.) Each $v-x$ path contains at most $\diam_G(K)$ vertices by definition of diameter. We have already shown that each vertex has cost $c(v) \leq \log(n)$. Finally, there are $|K|-1 < |K|$ choices for $x$. Thus, the cost of this embedding is $\leq |K|\cdot \diam_G(K)\cdot \log(n)$.

It is also trivial to check that at the beginning of the algorithm $\Pi(G,K,\kappa) = 0$ (because $\kappa(v) = 1$ for all vertices) and that since $\kappa$ only increases, $\Pi(G,K,\kappa)$ is monotonically increasing.

Now, consider any iteration of the while loop in the algorithm that returns a sparse vertex cut $(L,S,R)$. Let $\kappa$ be the capacity function before the cut is found, and let $\kappa'$ be the capacity function after $\kappa(v)$ is doubled for all $v \in S$. Using an argument identical to that of Lemma \ref{lem:doubling effect}, it is easy to check that 
\begin{equation}
\Pi(G,K,\kappa') \geq \Pi(G,K,\kappa) + |L \cap K|/3.
\label{eqn:overview-kappa-change}
\end{equation}
The basic argument here is that at least $|L\cap K|/3$ paths in the embedding cross from $L$ to $S \cup R$ and thus go through $S$. But for each vertex in $v \in S$ we have $\kappa'(v) = 2\kappa(v)$, so $c'(v) = c(v) + 1$. This argument can be formalized using the arguments of \ref{lem:doubling effect}.

Let us again consider a sparse cut $(L,S,R)$ returned by the while loop. Since the cut was sparse, we have that $\sum_{v \in S} \kappa(v) < |L \cap K|$. This implies that $\sum_{v \in V} \kappa'(v) = \sum_{v \in V} \kappa(v) + \sum_{v \in S} \kappa(v) \leq \sum_{v \in V} \kappa(v) + |L \cap K|$. Combining this with Equation \ref{eqn:overview-kappa-change} we see that whenever $\sum_{v \in V} \kappa(v)$ increases by some $\Delta$, $\Pi(G,K,\kappa)$ increases by at least $\Delta / 3$. But we know that $\Pi(G,K,\kappa)$ increases monotonically from $0$ to $\Pi(G,K,\kappa) \leq |K|\cdot \diam_G(K)\cdot \log(n)$. These two facts combined imply that at all times $\sum_{v \in V} \kappa(v) \leq 3|K|\cdot \diam_G(K)\cdot \log(n)$, as desired.

\paragraph{Discussion of Running Time:}
Since Lemma \ref{lem:overview-kappa} is only concerned with the \emph{existence} of a function $\kappa$, we did not concern ourselves with the running time of the algorithm. In particular, we did not specify how to find the sparse cut $S$ in the while loop. Below, we briefly discuss how such an algorithm could be implemented.

It is not hard to check that the algorithm goes through if we allow some slack in our requirement of the spare cut returned in the while loop, and that with this slack the cut can be computed in polynomial time. One could perhaps even compute the cut in almost-linear time using more sophisticated techniques. But the total time to compute $\kappa$ will still not be linear because there could be many iterations of the while loop: each iteration might find a sparse cut $(L,S,R)$ with $|L\cap K| = n^{o(1)}$, in which case the number of iterations can be as large as $\Ohat(|K|)$, so the total running time would be at least $\Ohat(m|K|)$, which could be as large as $\Ohat(mn)$.

The above obstacle explains why our final algorithm settles on a function $\kappa$ with the slightly weaker guarantees of Lemma \ref{lem:overview-kappa-relaxed}. This relaxed lemma only guarantees capacitated expansion for \emph{balanced} cuts $(L,S,R)$, so the while loop always returns a sarse balanced cut $(L,S,R)$, or returns $\kappa$ if no such cut exists. This allows us to ensure that $|L \cap K| = \Omega(|K|/n^{o(1)})$ in each iteration, so the number of iterations is only $n^{o(1)}$.

\paragraph{Lower Bound:} We now prove the lower bound of the lemma. Consider the following graph $G$, with $K = V(G)$. Let $G_A, G_B$ be vertex expanders, with $n/3$ vertices each. Let $a$ be a vertex in $G_B$ and $b$ a vertex in $G_B$. The graph $G$ contains both $G_A$ and $G_B$, as well as a path $P$ from $a$ to $b$ with $n/3$ intermediate vertices. We have $|K| = |V(G)| = n$ and $\diam_G(K) = \diam(G) = \Theta(n)$. It is not hard to check that any function $\kappa$ for which $(K,\kappa)$ is a capacitated expanders must have $\kappa(v) \geq n^{1-o(1)}/3$ for all $v \in P$, so $\sum_{v \in V(G)} \kappa(v) =\Omegahat(n^2) = \Omegahat(|K|\diam_G(K))$, as stated in the lemma.

%% file: certify_core.tex
\subsection{$\CertifyCore$: Finding A Large Low-diameter Subset}
\label{sec:certify_core}

In this section, we prove  \Cref{lem:certifycore} which is restated below.

\certifyCore*

We give pseudo-code for the procedure $\textsc{CertifyOrReturnCore}(G, K, d, \epsilon)$ in \Cref{alg:certifyOrReturnCore}.

\begin{algorithm}
\caption{$\textsc{CertifyOrReturnCore}(G, K, d, \epsilon)$}
\label{alg:certifyOrReturnCore}

$K' \gets K$\;
$G' \gets G[\ball_G(K, 16 d \lg n)]$\;
\While(\label{lne:whileLoopMainOfCertifyReturnCore}){$|K'| > (1-\epsilon/2)|K|$}{
    Let $v$ be an arbitrary vertex from $K'$ \label{lne:chooseV}\;
    $i \gets 0$\;
    \While(\label{lne:whileIncreaseI}){ $\deg_{G'}(\ball_{G'}(v, 2(i+1)d)) > 2\deg_{G'}(\ball_{G'}(v, 2i \cdot d))$}{
            $i \gets i+1$    
    }
    
    \If(\label{lne:ifCaseCertify}){$|\ball_{G'}(v, 2(i+1) \cdot d) \cap K'| > (1-\epsilon/2)|K'|$}{
        $K' \gets \ball_{G'}(v, 2(i+1) \cdot d) \cap K'$ \label{lne:removeVertsFromKPrime} \;
        \Return $K'$ \label{lne:returnKPrime}
    }\Else(\label{lne:enterElseCaseToPrune}){
        $K' \gets K' \setminus \ball_{G'}(v, 2i \cdot d)$\label{lne:removeVerticesFromKPrime}\;
        $G' \gets G' \setminus \ball_{G'}(v, 2i \cdot d)$\label{lne:removeVerticesFromGraph}\;
    }
}
$K' \gets \emptyset$\;
\Return $K'$
\end{algorithm}

Here, we initially set the set $K'$ to be the full set $K$ and set the graph $G'$ to the graph $G$. Then, while there are vertices in $K'$, we choose an arbitrary vertex $v$ from $K'$ (in \Cref{lne:chooseV}). We then search the smallest non-negative integer $i$, such that $\deg_{G'}(\ball_{G'}(v, (i+1)d)) \leq 2\deg_{G'}(\ball_{G'}(v, i \cdot d))$ by repeatedly increasing $i$ if for the current value of $i$ if the property is violated (by visiting another iteration of the while-loop starting in \Cref{lne:whileIncreaseI}). Finally, when the property is satisfied, we check whether the number of vertices in $K'$, in the ball $\ball_{G'}(v, i \cdot d)$ larger than $(1-\epsilon/2)|K|$. If so, we have found a subset of $K$ of large size and small diameter and return $\ball_{G'}(v, i \cdot d) \cap K'$ to end in the first scenario of our lemma. Otherwise, we remove the vertices in $K'$ that are in the ball $\ball_{G'}(v, i \cdot d)$ from $K'$ and the edges incident to the ball from $G'$. 

Let us now analyze the procedure more carefully by proving a series of simple claims.

\begin{claim} \label{clm:halfVerticesWhenClose}
Consider an execution of the while-loop starting in \Cref{lne:whileLoopMainOfCertifyReturnCore} where $i^{final}$ is the value $i$ takes after the algorithm leaves the while-loop starting in \Cref{lne:whileIncreaseI}. Then, if we enter the else-case in \Cref{lne:enterElseCaseToPrune}, we have for every vertex $w \in \ball_{G'}(v, (2i^{final}+1) d)$ that
\[
    |\ball_{G'}(w, d) \cap K'| \leq (1-\epsilon/2)|K'|.
\]
\end{claim}
\begin{proof}
We have by the triangle inequality that $\ball_{G'}(w, d) \subseteq \ball_{G'}(v, 2(i^{final}+1) d)$. But since by the if-condition we have that $\ball_{G'}(v, 2(i^{final}+1) d)$ contains at most a $(1-\epsilon/2)$-fraction the vertices in $K'$ the claim follows.
\end{proof}

\begin{claim} \label{clm:ifNoIfThenFewVerticesInAnyBall}
If the if-case in \Cref{lne:ifCaseCertify} is not entered then
any vertex $w \in K$ has at most $(1-\epsilon/2)|K|$ vertices in $\ball_G(w, d) \cap K$.
\end{claim}
\begin{proof}
Let us first focus on vertices $w \in K$ that have some vertex $w'$ from $\ball_G(w, d)$ that is removed from $G'$ at some point of the algorithm. That is, the algorithm removes $\ball_{G'}(v,2i^{final} \cdot d)$ for some $v\in K$ and some number $i^{final}$ and $w' \in \ball_{G'}(v,2i^{final} \cdot d) \cap \ball_G(w,d)$.

For each such vertex $w$, we consider the while-loop iteration starting in \Cref{lne:whileLoopMainOfCertifyReturnCore} at the time when the algorithm first removes a vertex $w'$ from $\ball_G(w, d)$ from the graph $G'$. We observe that up to  \Cref{lne:removeVerticesFromGraph}, we have never removed a vertex $w'$ from $\ball_G(w, d)$ from $G'$ and therefore up to this point, we have that $\ball_G(w, d) = \ball_{G'}(w, d)$ (technically we also have to argue that $G'$ is initialized to an induced graph of $G$ but it is clear that none of the vertices not in $G'$ are in $\ball_G(w, d)$ either). In fact, since vertices that are removed from $K'$ are in the balls that are removed from $G'$, we have in fact that up to this point $\ball_G(w, d) \cap K = \ball_{G'}(w, d) \cap K'$. But since the else-case in \Cref{lne:enterElseCaseToPrune} is entered in the iteration where the first such vertex $w'$ exists, we have by \Cref{clm:halfVerticesWhenClose} that $|\ball_{G}(w, d) \cap K| = |\ball_{G'}(w, d) \cap K'| \leq (1-\epsilon/2)|K'| \leq (1-\epsilon/2)|K|$, as desired. Note that we can invoke \Cref{clm:halfVerticesWhenClose} because $w' \in \ball_{G'}(v,2i^{final}\cdot d)$ and so $w \in \ball_{G'}(v,(2i^{final}+1) \cdot d)$ as needed in \Cref{clm:halfVerticesWhenClose}.

Otherwise, we have by the same argument that for any vertex $w$ in $K$ that had no vertex $w'$ from $\ball_G(v, d)$ removed from $G'$ that at termination of the while-loop starting in \Cref{lne:whileLoopMainOfCertifyReturnCore}, we have $\ball_G(w, d) \cap K = \ball_{G'}(w, d) \cap K'$. But by the while-loop condition, we have that $|K'| \leq (1-\epsilon /2) |K|$ at that point which establishes our claim.
\end{proof}

This claim establishes the Core Property in \Cref{lem:certifycore}. It remains to establish the first Property in of the Lemma where we start with proving that the diameter of the core that is returned in the if-statement in \Cref{lne:ifCaseCertify} is small.

\begin{claim} \label{clm:IRemainsSmall}
The integer variable $i$ is always chosen to be at smaller $2 \lg n$.
\end{claim}
\begin{proof}
For the sake of contradiction, let us assume that there is a time when the variable $i$ takes a value larger-equal than $2 \lg n$.

We first observe that in each iteration of the while-loop starting in \Cref{lne:whileLoopMainOfCertifyReturnCore}, the variable $i$ is initialized to $0$. Further, whenever $i$ is increased by one, we have that $\deg_{G'}(\ball_{G'}(v, (i+1)d)) \leq 2\deg_{G'}(\ball_{G'}(v, i \cdot d))$. Thus, we have
that
\begin{itemize}
    \item $\deg_{G'}(\ball_{G'}(v, 0)) \geq 1$, and
    \item for every $i < 2 \lg n$. we have $\deg_{G'}(\ball_{G'}(v, (i+1)d)) > 2\deg_{G'}(\ball_{G'}(v, i \cdot d))$.
\end{itemize}
Therefore, by induction we can straight-forwardly establish that
\[
    \deg_{G'}(\ball_{G'}(v, 2 \lg n \cdot d)) > 2^{2 \lg n} \geq n^2.
\]
But this leads to a contradiction since it implies $|E(G')| < \deg_{G'}(\ball_{G'}(v, (2 \lg n) \cdot d))$.
\end{proof}

We are now well-equipped to prove \Cref{lem:certifycore} which is again restated below for convenience.

\certifyCore*

\begin{proof}
We have that if \Cref{alg:certifyOrReturnCore} returns in  \Cref{lne:returnKPrime}, then we have that the final $K'$ is of size at least $(1-\epsilon/2)|K|$ since the while-loop condition ensured that in the $|K'| > (1-\epsilon/2)|K|$ and at most a $(1-\epsilon/2)$-fraction of the vertices from $K'$ remain in $K'$ in the if-statement in \Cref{lne:ifCaseCertify} by the condition of the if-statement. Thus, there are at least $(1-\epsilon/2)(1-\epsilon/2)|K| > (1-\epsilon)|K|$ vertices in the final $K'$. Further, by only leaving vertices in $K'$ that are contained in the same ball in $G'$ of radius $2(i+1) \cdot d$ where $i < 2 \lg n$ by  \Cref{clm:IRemainsSmall}, we certainly have that the diameter of the returned set $K'$ in $G \supseteq G'$ is at most $2 \cdot 2 \cdot 2 \lg n \cdot d \leq 8d \log n$ (where we use the radius to the the center of the ball and the triangle inequality). Thus, in this case, $K'$ satisfies the Scattered Property. 

Otherwise, we have by \Cref{clm:ifNoIfThenFewVerticesInAnyBall} that every vertex in $K$ has only few vertices in $K$ in its ball of radius $d$, thus satisfying the Core Property. 

Finally, let us bound the running time. Here we observe that each while-loop iteration starting in \Cref{lne:whileLoopMainOfCertifyReturnCore} can first run Dijkstra's algorithm to compute the smallest $i$ value and all information for the rest of the while-loop by running from the chosen vertex $v$ on $G'$ to depth $2(i+1)d$. It is not hard to see that the running time of the entire loop iteration is therefore dominated by $O(\deg_{G'}(\ball_{G'}(v, 2(i+1)d)) \log n)$. However, if we do not enter the if-case, we also remove $\Omega(\deg_{G'}(\ball_{G'}(v, 2(i+1)d)))$ edges from $G'$ in the else-loop since the ball $\ball_{G'}(v, 2i \cdot d)$ has at least half the volume by choice of $i$. Thus, the total running time of all such while-loops is upper bounded by $O(\deg_G(\ball_G(K, 16 d \lg n)) \log n)$. Since the algorithm returns upon entering the if-case, it can also only use additional time $O(\deg_G(\ball_G(K, 16 d \lg n)) \log n|)$ in the iteration not considered so far. This establishes the total running time and thereby the lemma.
\end{proof}

%% file: embedwitness.tex
\subsection{$\EmbedCore$: Embedding Expanders into Hypergraphs}

In this section, we show the procedure used by Algorithm \ref{alg:Core} for either finding a sparse cut or embedding an expander into a hypergraph. The algorithm is a standard combination of flow algorithms and the cut-matching game. The only non-standard element is that we need an algorithm for finding sparse cuts in \emph{hypergraphs}, which we call $\EmbedMatching$, and which was already developed in \cite{BernsteinGS20scc}.

We now restate the Lemma $\EmbedCore$ that we aim to prove. See Theorem \ref{thm:CMG} below for the definition of parameter $\phicmg$.

\embedwitness*

\label{sec:embedwitness}

We now recap three existing lemmas that we use to prove Lemma \ref{lem:embedwitness}

\subsubsection{First Ingredient: Embedding Matchings into Hypergraphs}

Our algorithm $\EmbedCore$ uses as a subroutine an existing algorithm from \cite{BernsteinGS20scc} that is given a hypergraph $H$ and either finds a sparse cut in $H$ or embeds a perfect matching into $H$ with low congestion.

\begin{lem}[\cite{BernsteinGS20scc}]
There is an algorithm $\EmbedMatching(H,A,B,\kappa,\eps)$ that is given
a hypergraph graph $H=(V,E)$, two disjoint sets of terminals $A,B\subseteq V$
where $|A|\le|B|$, a vertex capacity function $\kappa:V\rightarrow\frac{1}{z}\mathbb{Z}_{\ge0}$
such that $\kappa(v)\ge2$ for all terminals $v\in A\cup B$ and $\kappa(v)\le\kappa(V)/2$
for all vertices $v\in V$, and a balancing parameter $\eps>0$. (The integrality parameter $z$ will appear in the guarantees of the algorithm.) Then
the algorithm returns either 
\begin{itemize}
\item \textbf{(Sparse Cut):} a vertex cut $(L,S,R)$ in $H$ such that $\min\{|L\cap A|,|R\cap B|\}\ge\eps|A|$
and $\kappa(S)\le2\min\{|L\cap A|,|R\cap B|\}$; OR 
\item \textbf{(Matching):} an embedding $\cP$ that embeds a $\frac{1}{z}$-integral
matching $M$ from $A$ to $B$ of total value at least $(1-3\eps)|A|$
into $H$ where the congestion of $\cP$ w.r.t. $\kappa$ is at most
$1$ and the length of $\cP$ is at most $\len(\cP)\le O(\kappa(V)\log(\kappa(V))/(|A|\eps^{2}))$. 
More precisely, each path in $\cP$ has length at most $O(\kappa(V)\log(\kappa(V))/(|A|\eps^{2}))$ and
for each vertex $v\in V$, $\sum_{P\in\cP_{v}}\val(P)\leq\kappa(v)$,
where $\cP_{v}$ is the set of paths in $\cP$ containing $v$.
Moreover, each path in $\cP$ is a simple path. 
\end{itemize}
The running time of the algorithm is $\Otil(|H|\frac{\kappa(V)}{\eps|A|}+z\kappa(V)/\eps)$, where $|H|=\sum_{e\in E}|e|$, and $z$ is the smallest parameter such that $\kappa$ is $z$-integral, i.e. such that $\kappa:V\rightarrow\frac{1}{z}\mathbb{Z}_{\ge0}$ 
\end{lem}

\subsubsection{Second Ingredient: Cut-matching Game }

\paragraph{Deterministic Cut-matching Game.}

The cut-matching game is a game that is played between two players,
called the \emph{cut player} and the \emph{matching player}. The game
starts with a graph $W$ whose vertex set $V$ has cardinality $n$,
and $E(W)=\emptyset$. The game is played in rounds; in each round
i, the cut player chooses a partition $(A_{i},B_{i})$ of $V$ with
$|A_{i}|\le|B_{i}|$, $|A_i| \geq |V|/2-1$ and $|B_i| \geq |V|/2 - 1$. The matching player then chooses an arbitrary
$1/z$-integral matching $M_{i}$ that matches every vertex of $A_{i}$
to some vertex of $B_{i}$. (That is, the total weight of edges in
$M$ incident to each vertex in $A_{i}$ is exactly $1$ and the total
weight of edges in $M$ incident to vertex in $B_{i}$ is at most
$1$). The edges of $M_{i}$ are then added to $W$, completing the
current round. (Note that $W$ is thus a weighted multigraph; the edges of each $M_i$ are weighted, and if $M_i$ and $M_j$ both contain an edge $(x,y)$ then for simplicity we just think of $W$ as containing two copies of $(x,y)$.) Intuitively, the game terminates once graph $W$ becomes
a $\phi$-expander, for some given parameter $\phi$. It is convenient
to think of the cut player\textquoteright s goal as minimizing the
number of rounds, and of the matching player\textquoteright s goal
as making the number of rounds as large as possible. We will use the
following theorem from \cite{ChuzhoyS20_apsp} which says that there
is a fast deterministic algorithm for the cut player that ends
this game within $O(\log n)$ rounds.
\begin{theorem}
\label{thm:CMG}
[Deterministic Algorithm for Cut Player (Theorem B.5 of \cite{ChuzhoyS20_apsp} or Theorem 7.1 of \cite{BernsteinGS20scc})]Let
$\phicmg=1/2^{\Theta(\log^{3/4}n)}$. There is a deterministic algorithm,
that, for every round $i\ge1$, given the graph $W$ that serves as
input to the $i$-th round of the cut-matching game, produces, in
time $O(zn/\phicmg)$, a partition $(A_{i},B_{i})$ of $V$ with $|A_{i}|\le|B_{i}|$,  $|A_i| \geq |V|/2-1$, $|B_i| \geq |V|/2 - 1$
such that, no matter how the matching player plays, after $R=O(\log n)$
rounds, the resulting graph W is a $\phicmg$-expander, $V(W) = V$, and every vertex in $W$ has weighted degree at least $1$. 
\end{theorem}

\subsubsection{Third Ingredient: Expander Pruning}

Finally, we restate the lemma for expander pruning

\PruningLemma*

\subsubsection{Proof of Lemma \ref{lem:embedwitness}}
Armed with the three ingredients above, we can now present the algorithm for $\EmbedCore$ from Lemma \ref{lem:embedwitness}
\newcommand{\mstar}{M^*}
\newcommand{\wu}{W_u}
Recall that $z$ is the integrality parameter of input function $\kappa$, i.e. the smallest positive integer such that $\kappa:V\rightarrow\frac{1}{z}\mathbb{Z}_{\ge0}$.

The algorithm $\EmbedCore(H,K,\kappa)$ starts by initiating the cut-matching game (Theorem \ref{thm:CMG}) on vertex set $K$.  Let $R = O(\log(|K|))$ be the maximum number of rounds in the cut-matching game. The cut player from theorem \ref{thm:CMG} provides the terminal sets $A_i,B_i$ at every round $i$. To simulate the matching player the algorithm $\EmbedCore$ will, in each round, either find a sparse cut and terminate or return a matching $M_i$ from $A_i$ to $B_i$. In particular, in round $i$ of the cut-matching game, the algorithm runs $\EmbedMatching(H,A_i,B_i,\kappa,\epswit)$, where $\epswit = \phicmg/\log^2(n)$ is the parameter from $\EmbedCore$. 

If $\EmbedMatching(H,A_i,B_i,\kappa,\eps)$ returns a cut $(L,S,R)$ then $\EmbedCore$ can return the same cut $(L,S,R)$ and terminate. 

The other case is that $\EmbedMatching(H,A_i,B_i,\kappa,\eps)$ returns a matching $\mstar_i$ from $A_i$ to $B_i$ along with a corresponding embedding $\cP_i$ of $\mstar_i$ into the graph $H$. Note that the algorithm cannot simply use $\mstar_i$ as the matching $M_i$ in the ith round of the cut-matching game because the cut matching game requires a matching $M_i$ of size exactly $|A_i|$, while $\EmbedMatching$ only guarantees that matching $\mstar_i$ has size $(1-3\eps)|A_i|$. To overcome this, the algorithm chooses an arbitrary set of ``fake'' edges $F_i \in A_i \times B_i$ such that $\mstar_i \cup F_i$ is a matching from $A_i$ to $B_i$ of size exactly $|A_i|$; the set $F_i$ can trivially be computed by repeatedly adding edges of weight $1/z$ from an (arbitrary) unsaturated vertex in $A_i$ to an (arbitrary) unsaturated vertex in $B_i$. (Adding multiple copies of the same edge corresponds to increasing the weight of that edge.) The algorithm then returns $M_i = \mstar_i \cup F_i$ inside the cut-matching game. Note that unlike the edges of $\mstar_i$, we do not embed the fake edges of $F_i$ into $G$. 

If in any round $i$ the subroutine $\EmbedMatching$ returns a cut then the algorithm terminates. Thus the only case left to consider is when in each round $i$ the algorithm returns $\mstar_i$ and $\cP_i$. Let $\mstar$ be the union of all the $\mstar_i$ and let $F$ be the union of all the $F_i$. Let $\Wstar = (V,\mstar \cup F)$. Theorem $\ref{thm:CMG}$ guarantees that $\Wstar$ is a $\phicmg = 1/n^{o(1)}$ expander. Note, however, that we cannot return $\Wstar$ as our witness because there is no path set corresponding to $F$ (we never embedded the edges in $F$). We also cannot simply remove $F$ as $\mstar$ on its own might not be an expander.

Instead, we apply expander pruning from Lemma \ref{lem:prune}. Recall that $\Wstar = (V,\mstar \cup F)$. We would like to apply pruning directly to $\Wstar$, but Lemma \ref{lem:prune} only applies to unweighted multi-graphs. Since $\EmbedMatching$ guarantees that all edge-weights in $\mstar$ are $1/z$ integral, we know that all edge weights in $W$ are also multiples of $1/z$. We can thus convert $\Wstar$ to an equivalent unweighted multigraph $\Wstar_u$ in the natural way: every edge $e \in \Wstar$ is replaced by $w(e) \cdot z$ copies of an unweighted edge. Note that $\Wstar$ has total weight $O(|K|\log(|K|))$, because it contains $R = O(\log(|K|))$ matchings, each of weight $O(|K|)$; thus $\wu$ contains $O(z|K|\log(|K|))$ edges. We now run $\Prune(\Wstar_u,\phicmg)$, where we feed in all the edges in $\Wstar_u$ corresponding to $F$ as adversarial deletions. Let $X \subset K$ be the set returned by pruning, set $\wu = \Wstar_u[X]$ and $W = \Wstar[X]$. 

We now define the embedding $\cP$ of $W$ into $H$. We will have that $\cP \subseteq \bigcup \cP_i$. Consider any edge $(u,v) \in W$. By construction of $W$, we know that $(u,v)$ comes from some $\mstar_i$; $(u,v)$ cannot come from any of the $F_i$, because all of the edges in $F$ were pruned away. Thus, we simply add to $\cP$ the path from $\cP_i$ used to embed edge $(u,v) \in \mstar_i$. 

Let us now prove that $W$ satisfies the desired properties of $\EmbedCore$. We know from the cut-matching game (Theorem \ref{thm:CMG}) that $\Wstar$ has expansion $\phicmg$, so the same holds for $\Wstar_u$, since the two graphs clearly have identical expansion. By the guarantees of pruning, $W_u$ and $W$ thus have expansion $\phicmg/6 = \Omega(\phicmg)$, as desired. It is clear by construction that $V(\Wstar) \subset K$, so $V(W) \subset K$. 

Let us now argue that $|V(W)| \geq |K| - o(|K|)$. We know that $V(\Wstar_u) = V(\Wstar) = K$. Recall that Our algorithm feeds all edges in $\Wstar_u$ that correspond to $F$ as adversarial deletions to $\Prune(\Wstar_u,\phicmg)$. It is not hard to check that the number of such deleted edges is at most $3zR\epswit |K| = O(z\phicmg|K|/\log(n))$, because each $F_i$ contains a total weight of at most $3\epswit|A| \leq 3\epswit |K|$, there are $R$ different values of $i$, and by construction the multiplicity of each edge in $\Wstar_u$ is equal to $z$ multiplied by its weight in $F$. Thus, recalling that $X \subseteq K$ is the set returned by $\Prune(\Wstar_u,\phicmg)$, we have by Lemma \ref{lem:prune} that $\vol_{\Wstar_u}(K \setminus X) = O(z|K|/\log(n))$. This implies that $\vol_{\Wstar}(K \setminus X) = O(|K|/\log(n))$; since we know from the cut-matching game (Theorem \ref{thm:CMG}) that every vertex in $W$ has weighted degree at least $1$, we have that $|K \setminus X| = O(|K|/\log(n))$, so $|V(W)| = |X| \geq |K| - o(|K|)$, as desired. 

We now argue about the weights in $W$. The fact that total edge weight in $W$ is at most $O(|K|\log(|K|))$ follows from the fact that each matching $M_i$ has weight at most $|K|$ and there are $R = O(\log(|K|))$ rounds of the cut matching game. Finally, we need to show that there are only $o(|K|)$ vertices in $W$ with weighted degree $\leq 9/10$. This follows straightforwardly from the facts that all vertices have weighted degree $\geq 1$ in $\Wstar$ (Theorem \ref{thm:CMG}) and that $W = \Wstar[X]$, where, as argued in the paragraph above, $\vol_{\Wstar}(K \setminus X) = O(|K|/\log(n))$. 

We now argue about the embedding $\cP$. The congestion follows from the fact that each $\cP_i$ embeds $\mstar_i$ with congestion 1 with respect to $\kappa$, so since there are at most $R = O(\log(|K|))$ rounds in the cut-matching game, the total congestion in $\bigcup_i \cP_i$ is at most $\log(|K|)$ with respect to $\kappa$, so the same holds for $\cP$ because $\cP \subseteq \bigcup \cP_i$. The length and simplicity of paths in $\cP$ returned by $\EmbedCore$ follow directly from the same guarantees on $\cP_i$ returned by $\EmbedMatching$.

We finally analyze the running time. The algorithm runs in $R = O(\log(K)) = O(\log(n))$ rounds. In each rounds, it runs $\EmbedMatching$ with $\eps = \epswit = \phicmg / \log^2(n)$; plugging in the guarantees of $\EmbedMatching$ we see that this runtime fits into the desired runtime of $\EmbedCore$. Each round also runs a single iteration of the cut-matching game, which requires $O(z|K|/\phicmg)$ time (Theorem \ref{thm:CMG}); this satisfies the desired runtime of the lemma because by the input guarantees of Lemma \ref{lem:embedwitness} we have $\kappa(v) \geq 2 \ \forall v \in V$ so $z|K|/\phicmg \leq z\kappa(V)/\phicmg$. It is clear that the time to construct each $F_i$ is at most $O(z|K|)$. Finally, the algorithm performs a single execution of $\Prune(\Wstar_u,\phicmg)$; by Lemma \ref{lem:prune} this requires time $\Otil(z|V|/\phicmg) = \Otil(z\kappa(V)/\phicmg)$, as desired. $\qed$

\subsection{Proof of \Cref{clm:sideConditions}}
\label{sec:proofOfSideConditions}

\sideConditionClaim*
\begin{proof}
Property \ref{prop:sideCondition1}: follows immediately from the initialization of $\kappa$ in \Cref{enu:init cap} and the fact that $\kappa$ is monotonically increasing over time.

Property \ref{prop:sideCondition2}: $\EmbedCore(\cdot)$ is only invoked in \Cref{lne:sparseCutWhileLoop}. We prove by induction on the time that $\EmbedCore(\cdot)$ is executed. Initially, we have that $|\Vhat| \geq 2$, and by the values chosen for initialization in \Cref{enu:init cap}, it is immediate that the condition is true before the first time $\EmbedCore(\cdot)$ is invoked. For the inductive step, observe that in between two invocations of $\EmbedCore(\cdot)$, the property can only be affected if the former invocation produced at cut $(L,S,R)$, prompting the algorithm to enter the while-loop in \Cref{lne:sparseCutWhileLoop}. The capacity of vertices in $\Vhat \setminus (S \cup \{w'\})$ remains unchanged during this step. Since $\kappa$ is monotonically increasing, this implies by the induction hypothesis that only one of the vertices $S \cup \{w'\}$ might violate the property. But $w \in S$ is chosen in \Cref{lne:technicalSideCondition} to have maximal capacity among vertices in $S$, and $w' \neq w$ has capacity at least as large as $w$. It remains to show that $w'$ is not violating the property. But this follows since either the capacity of $w'$ is unchanged and we can therefore use the induction hypothesis, or it is equal to the capacity of $w$ and therefore not more than half of the total capacity.
\end{proof}

%% file: appendix_path.tex
\section[Appendix of Part III]{Appendix of \Cref{part:augmented-queries}}

\subsection{Simplifying Reduction for $\SSSP$ Data Structures}
\label{subsec:appendixProofOfsimplify path}

In this section, we prove both \Cref{prop:simplify-1} and \Cref{prop:simplify path}. However, since \Cref{prop:simplify path} is a more involved version of  \Cref{prop:simplify-1}, we only prove the former one. It is straight-forward from inspecting the proof that it extends seamlessly to \Cref{prop:simplify-1}. We start by restating the theorem.

\propSimplifyingAssumptOnSSSP*

For our proof, we first state the following result which is derived by a straight-forward extension of Theorem 2.3.1. in \cite{gutenbergThesis}.

\begin{theorem}[see \cite{gutenbergThesis}]
\label{thm:gutenbergThesis}
For any $1/n < \epsilon <1/2$, given a data structure that maintains $\SSSP^{\pi}(H,s,\eps,\beta,q)$ on any graph $H$ with edge weights in $[1, n_H^4]$ in time $\mathcal{T}_{SSSP}(m_H, n_H, \epsilon)$ (where we assume that distance estimates are maintained explicitly). Then, there exists a data structure, that maintains  $\SSSP^{\pi}(G,s,6\eps,\beta,q)$ on a graph $G$ with weights in $[1, W]$ for any $W$ in time 
\[
\tilde{O}\left( m/ \epsilon + \mathcal{T}_{SSSP^{\pi}}(m, n, \epsilon)\right) \cdot \log(W).
\]
We note that all graphs $H$ on which the SSSP data structure is run upon are subgraphs of $G$ at any stage. 
\end{theorem}

We can then apply the following series of transformations of $G$ to derive \Cref{prop:simplify path}.

\paragraph{Ensuring Connectivity.} Given the decremental graph $G$, we use a Connectivity data structure (see \cite{holm2001poly, wulff2013faster}) which allows us to remove any edge deletion that disconnects the graph from the update sequence. We let $G'$ be the resulting decremental graph. We can then run the SSSP data structure only on $G'$ instead of $G$. To obtain a distance estimate from the source to some vertex $v$ in $V$, we can first query the Connectivity Data Structure on $G$ if $v$ and the source are in the same connected component. If not, we return $\infty$. Otherwise, we forward the query to the SSSP data structure and return the distance estimate. For a formal argument that this gives correct distance estimates, we refer the reader to  \cite{gutenberg2020deterministic}.

\paragraph{Edge Deletions, no Weight Increases.} For the second property, we preprocess $G'$ so that for each edge $(u,v)$ of weight $w_{G'}(u,v)$, we split $(u,v)$ into $\lceil \log Wn \rceil$ multi-edges of weight $w_{G'}(u,v), (1+\epsilon) \cdot w_{G'}(u,v), (1+\epsilon) \cdot w_{G'}(u,v)^2, \dots$ respectively. Then, an edge weight increase of $(u,v)$ to $w'_{G'}(u,v)$ can be emulated by deleting all versions of $(u,v)$ that have weight smaller $w'_{G'}(u,v)$ from the graph. It is not hard to see that the resulting decremental graph preserves all distance to a $(1+\epsilon)$-approximation, only undergoes edge deletions, not edge weight increases and has at most $O(m \log W)$ edges. We denote by $G''$ the resulting graph.

\paragraph{Ensuring Small Degree.} Given the decremental graph $G''$, we can for each vertex $v$ with degree $\deg_{G''}(v) > 3$, add $\deg_{G''}(v)$ vertices to $G''$ and connect them among each other and with $v$ by a path where we assign each edge the weight $1/n^3$. Then, we can map each edge that was originally in $G''$ and incident to $v$ to one vertex on the path. It is not hard to verify that after these transformations the resulting graph $G'$ has maximum degree $3$ and each distance is increased by at most an $\epsilon$ fraction (this follows since the original paths might now also have to visit the newly created line paths but these paths consist of at most $m \leq n^2$ edges, thus the total contribution for each vertex on the path is at most $1/n < \epsilon$ but each original edge on the path has weight at least $1$ by assumption). Note that the number of vertices in $G'$ is at most $m+n \leq 2m$ and the number of edges is at most $2m$. Also note that we can multiply all edge weight above by $n^3$ to satisfy again that all edge weights are positive integers. This increase the weight ratio to $Wn^3$. We denote by $G'''$ the resulting graph. 

\paragraph{Ensuring Small Weight Ratio.} Finally, we can apply \Cref{thm:gutenbergThesis} on $G'''$ to obtain a data structure $\SSSP^{\pi}(G''',s, 6\eps,\beta,q)$. Observe that each distance estimate maintained by this data structure from $s$ to a vertex in $V$, approximates the distance in $G$ by a factor of $(1+\epsilon)(1+6\epsilon) = (1+O(\epsilon))$.

\paragraph{Queries on $G'''$.} Finally, we discuss how to conduct path-queries. We point out that given the data structure $\SSSP^{\pi}(G''',s, 6\eps,\beta,q)$, when we conduct a path-queries, for a path $\pi(s,t)$ from $s$ to some vertex $t$, it returns edges in $G'''$ instead of $G$. However, it is rather straight-forward by going backwards through the transformations from $G$ to $G'''$, to see that each such path can be mapped back unambiguously to a $s$-$t$ path in $G$ of weight at most equal to the weight in $G'''$. 

We consider this path in $G$ that the path is mapped to and discuss how to implement the subpath-query for an index $j \in [0, \sigma_{max}]$ given, in time $O(\sigma_{\leq j}(\pi(s,t)) \cdot q)$. To this end, we do the following: for each edge $e=(u,v)$ in $G'$, we only keep the heaviest copy of $e$ in $G''$. Note that any path including such an edge copy has weight $\geq nW$, thus we can ignore all such paths in our query and are therefore ensured that no edge that was deleted from $G$ but not $G'$ appears on any path. For $G''$, we give each copy of an edge $e=(u,v)$, the steadiness $\sigma(e)$. Note that if we maintain a data structure that maintains $\beta$-edge-simple paths, then the edge $e$ and its copy are present at most $\beta \cdot \log Wn$ times. For the transformation to $G'''$, we simply give each edge that was not in the graph formerly (i.e. is used to split a vertex into multiple vertices of low degree), the highest steadiness class $\sigma_{max}+1$. Such edges, do not appear in any path query since the highest steadiness in $G$ was $\sigma_{max}$. It is straight-forward to establish that this ensures the properties stated in the Proposition.

%% file: flow_appendix.tex
\section[Appendix of Part IV]{Appendix of \Cref{part:minCostFlow}}

\subsection{Proof of \Cref{prop:reductionVertexCapacities}} \label{subsec:proofOfReductionVertexCaps}

\reductionVertexCapacities*

In order to prove the proposition, we start by computing a crude approximation to $OPT_{G, \overline{C}}$. 

\begin{claim}\label{clm:crudeButFastApproximator}
In time $O(m\log n)$, we can compute $\tilde{U}$, such that $OPT_{G, \overline{C}}/2m^2 \leq \tilde{U} \leq OPT_{G, \overline{C}}$. 
\end{claim}
\begin{proof}
In order to find such $\tilde{U}$, we use \Cref{alg:crudeApproxOpt}.

\begin{algorithm}[h!]
\ForEach{$e = (x,y) \in E$}{
    $u_{approx}(e) \gets \min\{u(e), u(x), u(y)\}$.\;
    $c_{approx}(e) \gets \max\{c(e), c(x), c(y)\}$.\;
}
Let $e_1, e_2, \dots, e_m$ be an ordering of the edges such that $c_{approx}(e_i) \leq c_{approx}(e_{i+1})$ for all $i$.\;
\For{$i \in \{1, \dots m\}$}{
    $E_i \gets \{e_1, e_2, \dots, e_i\}$.\;
    Compute a maximum spanning forest $T_i$ in $G_i = (V, E_i)$ with edges weighted by $u_{approx}$.\;
    $\lambda_i \gets \min\{\min_{e \in T_i(s,t)} u_{approx}(e), \overline{C}/(2m \cdot c_{approx}(e_i))\}$.
}
\Return $\tilde{U} = \max_i \lambda_i$
\caption{$\textsc{CrudeApproxOpt}(G, \overline{C})$}
\label{alg:crudeApproxOpt}
\end{algorithm}

Let us now carry out the analysis of correctness for the algorithm:
\begin{itemize}
    \item $\mathbf{OPT_{G, \overline{C}}/2m^2 \leq \tilde{U}:}$ We aim to show that for some iteration $i$, we have $\lambda_i \geq OPT_{G, \overline{C}}/2m^2$. 
    
    We start by observing that we have for any feasible $s$-$t$ flow $f$ in $G$ of value $OPT_{G, \overline{C}}$, that there is some $s$-$t$ path $P$ in $G$, such that each edge carries at least $OPT_{G, \overline{C}}/m$ flow. 
    
    Let $i$ be the smallest index as defined in the algorithm, such that $P$ is contained in $G_i$, i.e. let $e_i$ be the heaviest edge on $P$ in terms of $c_{approx}$. Next, observe that since $T_i$ is a maximum spanning forest with regard to $u_{approx}$, the path $T_i(s,t)$ has min-capacity larger-equal to $P$, i.e. larger-equal to $OPT_{G, \overline{C}}/m$. In particular, this means $\min_{e \in T_i(s,t)} u_{approx}(e) \geq OPT_{G, \overline{C}}/2m$.
    
    Further, routing a single unit of flow along $P$ is at cost at least $c_{approx}(e_i)$. Also, we now from above that $c_{approx}(e_i) \cdot OPT_{G, \overline{C}}/m \leq \overline{C}$. Thus, $\overline{C}/(2m \cdot c_{approx}(e_i)) \geq OPT_{G, \overline{C}}/2m^2$. This establishes the case.

    \item $\mathbf{\tilde{U} \leq OPT_{G, \overline{C}}:}$ Observe for each iteration $i$, the amount $\lambda_i$ can be routed in $G$ since the path $T_i(s,t)$ has min-capacity at least $\lambda_i$ and each of the at most $m$ edges and $n \leq m$ vertices on $T_i(s,t)$ contributes at most $\lambda_i \cdot c_{approx}(e_i) \leq (\overline{C}/(2m \cdot c_{approx}(e_i))) \cdot c_{approx}(e_i) \leq \overline{C}/2m$ cost, as desired.
\end{itemize}
Finally, let us discuss the running time of \Cref{alg:crudeApproxOpt}. We observe that the ordering of $e_1, e_2, \dots, e_m$ can be done in $O(m\log n)$ time using classic sorting algorithms. For the for-loop, we observe that we can use a dynamic tree data structure in combination with Prim's classic maximum spanning forest algorithm (see \cite{tarjan1983data,SleatorT83, cormen2009introduction}). This allows us to implement each loop iteration in only $O(\log n)$ time, since we can use the dynamic tree also to check for the min-capacity on $T_i(s,t)$ in iteration $i$ in $O(\log n)$ time. This completes the analysis.
\end{proof}

We assume henceforth that we have $\tilde{U}$ with guarantees described in \Cref{clm:crudeButFastApproximator}. We can now describe how to obtain $G'$ from $G$ as stated in \cref{prop:reductionVertexCapacities}. Throughout this section, we use the parameters $\tau_u = \tilde{U}/8m^2$ and $\tau_c = \overline{C} \cdot 8m/\tilde{U}$. Using these two parameters, we define two refined versions of $V$ and $E$ that restrict them to include only items of reasonable cost and capacity:
\begin{align*}
    V_{reasonable} &= \{v \in V \;|\; u(x) \geq \tau_u \text{ and } c(x) \leq \tau_c\}\\
    E_{reasonable} &= \{e \in E \;|\; u(e) \geq \tau_u \text{ and } c(e) \leq \tau_c\} \cap (V_{reasonable} \times V_{reasonable}).
\end{align*}

Given these preliminaries, we can now define $G'$.

\paragraph{Vertex Set $V'$.} We define $V'$, the vertex set of $G'$, to consist of the vertices in $V_{reasonable}$, two special vertices $s', t'$ and an additional vertex $v_{x,y}$ for each pair of anti-parallel edges $(x,y), (y,x) \in E_{reasonable}$ (here $v_{\{x,y\}} = v_{\{y,x\}}$). 

\paragraph{Edge Set $E'$.} We define the edge set $E'$ to be such that for each edge $(x,y) \in E_{reasonable}$, that there are two edges $(x,v_{\{x,y\}}), (v_{\{x,y\}}, y)$. Finally, we insert edges $(s,s'), (s',s), (t, t'), (t', t)$ into $E'$. 

\paragraph{Cost and Capacity Functions.} We list the edge and vertex capacities and cost in detail in the list below. Here, we define $\gamma_u = \tau_u$ and $\gamma_c = \overline{C}/ (4\tilde{U}m^2)$.
\begin{table}[!ht]
    \centering
    \begin{tabular}{c|c|c}
        Item & Cost & Capacity \\ \hline
        $e \in E'$ & $c'(e) = 0$ & $u'(e) = \infty$ \\ \hline
        $x \in \{t',s'\}$ & $c(x) = 0$ & $u'(x) = \tilde{U} \cdot 2m^2 / \gamma_u$ \\ \hline
        $x \in V_{reasonable}$ & $c'(x) = \max\{c(x) / \gamma_c, 1\}$ & $u'(x) = \min\{u(x), \tilde{U} \cdot 2m^2\} / \gamma_u$\\ \hline
        $(x,y), (y,x) \in E_{reasonable}$ & $c'(v_{\{x,y\}}) = \max\{ c(x,y) / \gamma_c, 1\}$ & $u'(v_{\{x,y\}}) = \min\{u(x,y),\tilde{U} \cdot 2m^2\} / \gamma_u$
    \end{tabular}
    \label{tab:my_label}
\end{table}

Finally, we also define $\overline{C}'=\overline{C} / (\gamma_c \cdot \gamma_u) = 32m^4$. We now prove \Cref{prop:reductionVertexCapacities}, Property by Property:
\begin{enumerate}
    \item Observe first that for any $(x,y) \in E$ with $(x,y) \in E_{reasonable}$, we also have $(y,x) \in E_{reasonable}$ by definition. Further, since we insert for each such anti-parallel edges $(x,y)(y,x)$, the edges $(x,v_{\{x,y\}}), (v_{\{x,y\}}, y), (y,v_{\{x,y\}}), (v_{\{x,y\}},x)$, and since the only other edges inserted are the edges $(s,s'), (s',s), (t, t'), (t', t)$, we have that all edges in $E'$ are anti-parallel.
    \item By definition of the cost and weight functions in $s', t'$.
    \item Each vertex in $V'$ is uniquely associated to either a vertex from $V$, or an edge from $E$, or is $s'$ or $t'$. Thus, $V'$ is of size at most $m+n+2$. Since we split each edge in $E$ into two and add these edges to $E'$, and then only add an additional $4$ edges, we have that $E'$ is of size at most $2m+4$.
    \item Since by definition of $V_{reasonable}$ and $E_{reasonable}$ all elements of these sets are mapped by $u$ to a real of size at least $\tau_u = \gamma_u$, we have that all capacities in $G'$ are at least $1$. Further, all capacities are capped at $\tilde{U} \cdot 2m^2/\gamma_u = 16m^4$.
    
    For the costs, we observe that all costs of elements in $V_{reasonable}$ and $E_{reasonable}$ are at most $\tau_c = \overline{C}\cdot 8m/\tilde{U}$ in $G$. By setting $\gamma_c = \overline{C} /(4\tilde{U}m^2)$, we further have that the largest cost in $G'$ is $\frac{\overline{C}\cdot 8m/\tilde{U}}{\overline{C}/ (4\tilde{U}m^2)} = 32m^3$. The smallest cost is at least $1$ since we set each $c'(x)$ for $x \in V'$ to be at least $1$ by definition.
    \item First, consider any feasible $s'$-$t'$ flow $f'$ in $G'$ of value $F$. Here, we can assume w.l.o.g. that only a single anti-parallel edge carries any flow. We can then construct a flow $f$ in $G$ by assigning for each $(x,y), (y,x) \in E_{reasonable}$ the flow  $f(x,y) = f'(x,v_{\{x,y\}}) \cdot \gamma_u$ and $f(y,x) = f'(v_{\{x,y\}},y) \cdot \gamma_u$. It is not hard to see that $f$ is a feasible flow in $G$ and of value $F \cdot \gamma_u$. That is, we can map each flow in $G'$ to a flow $f$ in $G$ of the same flow value (up to scaling by $\gamma_u$).
    
    It thus only remains to show that there is a feasible $s'$-$t'$ flow $f'$ of flow value at least $(1-\epsilon) \cdot OPT_{G, \overline{C}} / \gamma_u$. To this end, let $f$ be a feasible $s$-$t$ flow in $G$ of value $OPT_{G, \overline{C}}$. Further let $\mathcal{P}$ be a flow path decomposition of $f$, where each $P \in \mathcal{P}$ sends $v(P)$ flow from $s$ to $t$. 
    
    Let $\mathcal{P}'$ be the set of paths $P \in \mathcal{P}$ such that $P$ is fully contained in $G[E_{reasonable}]$. Then, construct a flow $f'$ in $G'$ by routing for each path $P \in \mathcal{P}'$, $v(P) \cdot \gamma_u \cdot (1-\epsilon/4)$ units of flow along the corresponding path in $G'$ (i.e. map each edge $(x,y)$ in $P$ to $(x, v_{\{x,y\}}), (v_{\{x,y\}},y)$ to obtain a path in $G'$). 
    
    We claim that $f'$ is a feasible $s$-$t$ flow in $G'$ of value at least $(1-\epsilon) OPT_{G, \overline{C}} / \gamma_u$ (and thus can be easily extended to a feasible $s'$-$t'$-flow of the same value). To see this, let us first observe that capacity constraints in $G'$ are equal to the ones in $G$ up to scaling by $\gamma_u$ and capping at the optimum flow value of $f$. Thus, capacity constraints are satisfied in $G'$.
    
    For cost-feasibility, we observe that the only way that costs are increased (after scaling by $\gamma_c$) is if a cost $c(x)/\gamma_c$ was so small that $c'(x)$ is rounded up to $1$. Since we scale the flow not only by $\gamma_u$ but also by $(1-\epsilon/4)$, we have that if we would not have rounded up any costs, we would obtain total cost of $f'$ in $G'$ of at most $\overline{C} / (\gamma_u \cdot \gamma_c) \cdot (1-\epsilon/4) = (1-\epsilon/4) \overline{C}'$. But rounding up small costs, results in additional cost of $f'$ of at most $\overline{U}m /\gamma_u \cdot 1 \leq 8m^3$. But we have that $\overline{C}' = 32m^4$, thus this is at most a $\epsilon/4$-fraction of $\overline{C}'$ for every $\epsilon \geq 1/n$. This establishes cost-feasibility of $f'$ in $G'$.   
    
    It remains to show that the flow value of $f'$ is large. Now, if we would have that every path in $\mathcal{P}$ would be also in $\mathcal{P}'$, then the flow $f'$ would be exactly of value $OPT_{G, \overline{C}} \cdot \gamma_u \cdot (1-\epsilon/4)$. But we now argue that every $P$ that is not in $\mathcal{P}'$ must have carried a small flow anyway since either
    \begin{itemize}
        \item the capacity of some vertex or edge $x$ on the path $P$ was smaller $\tau_u$. But note that this implies that such $P$ carried at most $\tau_u$ units of flow by capacity-feasibility. But there are at most $m$ such paths, thus the total amount of flow in $G$ along such paths is upper bounded by $m \cdot \tau_u = \tilde{U}/8m$.
        \item the cost of some vertex or edge $x$ on path $P$ was larger than $\tau_c$. But then, we have that the total amount of flow in $G$ on \emph{all} such paths can be at most $\overline{C} / \tau_c \leq \tilde{U}/8m \leq OPT_{G, \overline{C}}/8m$.
    \end{itemize}
    Combined, all paths in $\mathcal{P}$ that do not participate in $\mathcal{P}'$ carried at most a $\tilde{U}/4m \leq OPT_{G, \overline{C}}/4m$ units of flow in $f$ which is just a $(\epsilon/4$-fraction of the total flow value. 
    
    Thus, the flow $f'$ is of value at least $(1-\epsilon/4)(1-\epsilon/4) \cdot OPT_{G, \overline{C}} \cdot \gamma \geq (1-\epsilon) \cdot OPT_{G, \overline{C}} \cdot \gamma$. Thus, any $(1-\epsilon)$-optimal flow $f'$ has flow value at least $(1-\epsilon)^2 \cdot OPT_{G, \overline{C}}$ and we have a simple transformation of $f'$ to a flow $f$ in $G$ of value $(1-\epsilon)^2 \cdot OPT_{G, \overline{C}}$.
    
    The running time of applying flow map and for computing $\tilde{U}$ are rather straight-forward, the later is implied by \Cref{clm:crudeButFastApproximator}.
\end{enumerate}

\subsection{Proof of \Cref{clm:reduceLocalCapacityForPlainVertex}}
\label{sec:actualProofOfClmreduceLocalCapacityForPlainVertex}

Let us restate and prove \Cref{clm:reduceLocalCapacityForPlainVertex}.

\clmReduceLocalCapacityForPlainVertex*
\begin{proof}
We prove by induction on $j$. For $j= 0$, observe that there are at most $m$ edges incident to a vertex $v$ in $V$, and for each edge that is incident to $v$ in $G$ there is exactly one vertex in $G'$ that is in $v$'s neighborhood. Since we set $u_0(x)$ to at most $2 \overline{U}$ for each vertex in $G'$, the base case follows.

For the induction step $j \mapsto j + 1$, we observe that by the induction hypothesis we have that $\sum_{x \in \neighborhood_{G'}(v)} u_j(x) \leq \frac{m}{(2-16\epsilon)^j} \cdot 2\overline{U} / \epsilon + 10 \cdot u(v)$. Next, we observe that by assumption on $\mathcal{A}$, there is a feasible flow $f_j$ such that $|\inflow_{f_j}(z) - \inflow_{g_j}(z)| \leq \epsilon \cdot u(y)$ and we have $\outflow_{f_j}(z) = \inflow_{f_j}(z) \leq u_j(z)$ for each $z \in V'$ (and in particular for $z \in \mathcal{N}(v) \cup \{v\}$).  

Observing that each vertex $v_e \in \neighborhood_{G'}(v)$ corresponds to an edge $e = (x,y)$ in $E$, we have that $v_e$ has one in-edges $(x,v_e)$ and one out-edge $(v_e, y)$ in $G'$. Using the above facts, it is thus not hard to derive that 
\begin{align*}
     \sum_{v_e \in \neighborhood_{G'}(v), e=(x,y)} g_j(x,v_e) + g_j(v_e,y) &\leq 2 \cdot \frac{m}{(2-16\epsilon)^j} \cdot 2\overline{U} + \epsilon \cdot 20 \cdot u(v) + 2 \cdot u(v)\\
     &\leq 2 \cdot \frac{m}{(2-16\epsilon)^j} \cdot 2\overline{U} + (2+20\epsilon) \cdot u(v).
\end{align*}
Next, we observe that the total capacity of edges $(x,v)$ or $(v,x)$ in $E'$ that carry flow greater-equal to half the capacity of either $x$ or $v$ can have at most total capacity $2$ times the right-hand-side of the above equation by a simple pigeonhole-principle style argument. 

Since the rest of the capacities are halved, we thus have that
\begin{align*}
 \sum_{x \in \neighborhood_{G'}(v)} u_{j+1}(x) &\leq 2 \left(2 \cdot \frac{m}{(2-16\epsilon)^j} \cdot 2\overline{U} + (2+20\epsilon) \cdot u(v)\right)+ \frac{\frac{m}{(2-16\epsilon)^j} \cdot 2\overline{U} / \epsilon + 10 \cdot u(v)}{2}\\
 &\leq (1 + 8\epsilon) \cdot \frac{m}{(2-16\epsilon)^j \cdot 2} \cdot 2\overline{U}/\epsilon + (9+40\epsilon) \cdot u(v)\\
 & \leq \frac{m}{(2-16\epsilon)^{j+1}} \cdot 2\overline{U}/\epsilon + 10 \cdot u(v)
\end{align*}
where we use $\epsilon \leq 1/40$ and $1+x \leq e^x \leq 1+x+x^2$ for $x \leq 1$.
\end{proof}

\subsection{Proof of \Cref{clm:OptIsNotReducedByMuch}}
\label{sec:proofOfOptIsNotReducedByMuch}
Let us restate and prove \Cref{clm:OptIsNotReducedByMuch}.

\clmOptNotReducedByMuch*
\begin{proof}
Observe that in the $j^{th}$ iteration of the for-loop, we obtain a $(1-\epsilon')$-pseudo-optimal flow $g_j$ with regard to the current instance $G_{j}$ by \Cref{thm:mainMinCost2}. By \Cref{def:feasibleFlow}, this implies that there is a feasible flow $\flow[j]$ in $G_j$ of value at least $(1-\epsilon')OPT_{G_j, \overline{C}}$ with $|\inflow_{g_j}(v) - \inflow_{\flow[j]}(v)| < \epsilon \cdot u_j(v)$. 

Now, consider the flow $\flow[j]' = (1-4\epsilon')\flow[j]$, we claim that $\flow[j]'$ is feasible in $G_{j+1}$ which implies our claim, since it is straight-forward to see that
\[
OPT_{G_{j+1}, \overline{C}} \geq v(\flow[j]') = (1-4\epsilon') v(\flow[j]) \geq (1-\epsilon')(1-4\epsilon') OPT_{G_{j}, \overline{C}} \geq (1-10\epsilon') OPT_{G_{j}, \overline{C}} 
\]
where $v(\cdot)$ gives the value of the flow, and where we used the feasibility of $\flow[j]'$ in $G_{j+1}$ for the first inequality and $1+x \leq e^x \leq 1+x+x^2$ for $x \leq 1$, and $\epsilon < 1/2$ for the final inequality.

To see that $\flow[j]'$ is feasible, observe first that we have for any vertex $v$, with $\inflow_{\flow[j]'}(v) \leq u_j(v)/2$ that the flow does not violate the capacity constraint on vertex $v$ since $u_{j+1}(v) \geq u_j(v)/2 \geq \inflow_{\flow[j]'(v)}$. On the other hand, if $\inflow_{\flow[j]'}(v) \geq u_j(v)/2$, we have that the flow 
\[
\inflow_{\flow[j]}(v) \geq (1+2\epsilon) u_j(v)/2 = u_j(v)/2 + \epsilon \cdot u_j(v)
\]
where we again use $1+x \leq e^x \leq 1+x+x^2$ for $x \leq 1$, and $\epsilon < 1/2$. But since $\inflow_{g_j}(v)$ differs from $\inflow_{\flow[j]}(v)$ by at most $\epsilon \cdot u_j(v)$, we have that $\inflow_{g_j}(v) \geq u_j(v)/2$ which implies that $u_{j+1}(v) = u_j(v)$. Thus, since $\inflow_{\flow[j]'}(v) < \inflow_{\flow[j]}(v) \leq u_j(v)$, we also have that the capacity constraint is satisfied for these edges.

For the final claim, we observe that 
\[
OPT_{G_{j_{max}}, \overline{C}} \geq \left(1- 10\epsilon' \right)^{j_{max}} OPT_{G_0, \overline{C}} \geq (1-\epsilon)OPT_{G_0, \overline{C}}
\]
since $\left(1- 10\epsilon' \right)^{j_{max}} \geq e^{-20\epsilon' \cdot j_{max}} = e^{- \epsilon} \geq 1-\epsilon$ by the definition of $\epsilon'$ and $1+x \leq e^x \leq 1+x+x^2$ for $x \leq 1$, and $\epsilon < 1/64$.
\end{proof}

\subsection{Proof of \Cref{thm:implementationFlowSSSPviaAugSSSP}}
\label{subsec:implementationFlowSSSPviaAugSSSP}

\implementationFlowSSSPviaAugSSSP*
\begin{proof}
Let us take the original graph $G$ with vertex weights $w$. We create two instances of the data structure in \Cref{thm:main SSSP path}:
\begin{itemize}
    \item We first define $w'$ to be the vertex weights over $V$ such that for all $v \in V \setminus \{s,t\}$, $w'(v) = w(v)$ and $w'(s) = 0$ and $w'(t) = 2 \cdot w(t)$. We then define an edge weight function $w''(x,y) = \frac{w'(x) + w'(y)}{2}$ for $(x,y) \in E$, that takes the average weight over the endpoints of each edge. We let $\dtil''(t)$ denote the estimate maintained for the distance from $s$ to $t$ in the graph weighted by $w''$. Observe that any $1$-simple $s$ to $t$ path in $w''$ has equal weight as in $w$ (recall that the first vertex on the path does not incur any weight contributing in our definition). Thus, $\dist_w(s,t) \leq \dtil''(t) \leq (1+\epsilon) \dist_{w}(s,t)$, i.e. the distance estimate is with regard to vertex weights $w$.
    \item Next, let us define a weight function $w'''$ over the vertices, defined by $w'''(v) = w(v)$ for $v \in V \setminus \{s,t\}$, $w'''(s) = \epsilon \dtil''(t) / 4$ and $w'''(t) = 2 \cdot w(t)$ (observe that $w'''$ only differs from $w'$ in $s$). Finally, we define an edge weight function $w''''(x,y) = \frac{w'''(x) + w'''(y)}{2}$ for $(x,y) \in E$.
    
    We then run a data structure $\mathcal{E}$ as described in \Cref{thm:main SSSP path} on $w''''$ and set the approximation parameter to $\epsilon' = \epsilon/16$. Observe that the shortest $s$ to $t$ path in $w''''$ has weight at most $\dist_w(s,t) + w'''(s)/2 < (1+\epsilon/4) \dist_w(s,t)$, and that each $s$ to $t$ path in $w''''$ is of even smaller weight in $w$. 
    
    Thus, the vertex $s$ can only occur at most once on any $(1+\epsilon')$-approximate shortest path from $s$ to $t$ by the size of $w'''(s)$ (this is important since $w(s)$ might be very large). Therefore, any such path is $(1+\epsilon') (1+\epsilon/4) \leq (1+\epsilon)$-approximate with respect to $w$.
    
    We conclude that the $s$-$t$ paths maintained by $\mathcal{E}$ are $(1+\epsilon)$-approximate and using the feature of path queries straight-forwardly, we can implement a data structure as described in \Cref{def:Path-reportingSSSP}.
\end{itemize}
The update time then follows simply by using the bounds from \Cref{thm:main SSSP path}.
\end{proof}